\documentclass[runningheads,envcountsame,colorlinks=true,capitalize]{llncs}
\usepackage[T1]{fontenc}
\usepackage[normalem]{ulem}
\usepackage{graphicx}
\usepackage{S_Packages}
\usepackage{S_Macro}
\usepackage{S_Comments}
\usepackage{caption}
\usepackage{enumitem}

\newbool{versionLongue}
\setbool{versionLongue}{true}

\crefname{lemma}{Lem.}{Lems.}
\Crefname{lemma}{Lem.}{Lems.}
\crefname{appendix}{App.}{Apps.}
\Crefname{appendix}{App.}{Apps.}
\crefname{section}{Sec.}{Secs.}
\Crefname{section}{Sec.}{Secs.}
\crefname{example}{Ex.}{Exs.}
\Crefname{example}{Ex.}{Exs.}
\crefname{corollary}{Cor.}{Cors.}
\Crefname{corollary}{Cor.}{Cors.}
\crefname{figure}{Fig.}{Figs.}
\Crefname{figure}{Fig.}{Figs.}
\crefname{definition}{Def.}{Defs.}
\Crefname{definition}{Def.}{Defs.}
\crefname{theorem}{Thm.}{Thms.}
\Crefname{theorem}{Thm.}{Thms.}
\crefname{proposition}{Prop.}{Props.}
\Crefname{proposition}{Prop.}{Props.}

\DeclareCaptionLabelFormat{customLabel}{Fig \thefigure#2:}%#2}
\begin{document}
\captionsetup[subfigure]{labelformat=customLabel}%
\title{The Benefits of Diligence}

\titlerunning{The Benefits of Diligence}
\author{Victor Arrial\inst{1}%\email{arrial@irif.fr}\orcidID{0000-0002-1607-7403}
\and Giulio
Guerrieri\inst{2}%\email{g.guerrieri@sussex.ac.uk}\orcidID{0000-0002-0469-4279} 
\and Delia
Kesner\inst{1}%\email{kesner@irif.fr}\orcidID{0000-0003-4254-3129}
}
\authorrunning{V. Arrial et al.}
\institute{Université Paris Cité, CNRS, IRIF, France \and University of Sussex, Department of Informatics, Brighton, UK}
\maketitle              % typeset the header of the contribution
\begin{abstract}
This paper studies the strength of embedding \CbNText (\CBNSymb) and
\CbVText (\CBVSymb) into a unifying framework called the
\BangCalculusText (\BANGSymb). These embeddings 
enable establishing (static and dynamic) properties of \CBNSymb and
\CBVSymb through their respective counterparts in $\BANGSymb$. While
some specific static properties have been already successfully studied
in the literature, the dynamic ones are more challenging and have been
left unexplored. We accomplish that by using a standard embedding for
the (easy) \CBNSymb case, while a novel one must be introduced for the
(difficult) {\CBVSymb\ case}. Moreover, a key point of our approach is
the identification of \BANGSymb diligent reduction sequences, which
eases the preservation of dynamic properties from $\BANGSymb$ to
$\CBNSymb/\CBVSymb$. We illustrate our methodology through two
concrete applications: confluence/factorization for both \CBNSymb and
\CBVSymb are respectively derived from confluence/factorization for
\BANGSymb.
\end{abstract}

\boolfalse{inAppendix}

\section{Introduction}
\label{sec:intro}

Call-by-Name (CBN) and Call-by-Value (CBV) stand as two foundational
evaluation strategies inspiring  distinct techniques and models of
computation in the theory of programming languages and proof
assistants~\cite{Plotkin75}. Notably, most theoretical studies in the
$\lambda$-calculus still continues to focus on its CBN variant, while
CBV, the cornerstone of operational semantics for most programming
languages and proof assistants, has been less extensively explored.
This is due in particular to the CBV stipulation that an argument can
be passed to a function only when it is a \emph{value} (\ie  variable
or abstraction), making the reasoning notably challenging to grasp.
Consequently, some fundamental concepts in the theory of the
$\lambda$-calculus (e.g. denotational semantics, contextual
equivalence, solvability, Böhm trees) make subtle\,--and not entirely
understood--\,distinctions between CBN and CBV, sometimes resulting in
completely ad-hoc scenarios for CBV, not being uniform with the
corresponding notion in CBN. This is for example the case of CBV Böhm
trees~\cite{KerinecManzonettoPagani20} or the notion of substitution
in~\cite{Santo20}.

\paragraph{Unifying Frameworks.} 
Reynolds~\cite{Reynolds98} (quoted by Levy~\cite{Levy99}) advocated
for a unifying framework for CBN and CBV.  This not only
minimizes their arbitrariness, but also avoids developing and proving
distinct and independent concepts and properties for them from scratch. Indeed,
both paradigms can be encompassed into broader foundational
frameworks~\cite{LincolnMitchell92,Abramsky93,BentonBiermanPaivaHyland93,BentonWadler96,RonchiRoversi97,MaraistOderskyTurnerWadler99,Levy99,Ehrhard16,EspiritoPintoUustalu19}
that explicitly differentiate values by marking them with a
distinguished constructor. While multiple such frameworks exist, our
focus lies on the Bang
Calculus~\cite{EhrhardGuerrieri16,GuerrieriManzonetto19,BucciarelliKesnerRiosViso20}.
Inspired by Girard's Linear Logic (LL)~\cite{Girard87} and Ehrhard's
interpretation~\cite{Ehrhard16} of Levy's
Call-by-Push-Value~\cite{Levy99} into LL, the Bang Calculus is
obtained by enriching the $\lambda$-calculus with two distinguished
modalities $\oc$ and $\derSymb$. The modality $\oc$ plays a twofold
role: it marks what can be duplicated or erased during evaluation (\ie copied an
arbitrary number of times, including zero), and it freezes the
evaluation of subterms (called \emph{thunks}). The modality $\derSymb$
annihilates the effect of $\oc$. Embedding CBN or CBV into the Bang
Calculus just consists in decorating $\lambda$-terms with $\oc$ and
$\derSymb$, thus forcing one model of computation or the other one.
Thanks to these two modalities, the Bang Calculus eases the
identification of shared behaviors and properties of CBN and CBV,
encompassing both syntactic and semantic aspects of them, within a
unifying and simple framework.

\paragraph{Adequate Models of Computation.}
Both CBN and CBV were originally defined on \emph{closed} terms
(without occurrences of free variables), that are enough to model
execution of programs. However, evaluation in proof assistants must be
performed on possibly \emph{open} terms, that is, with free variables.
While open terms are harmless to CBN, the theory of the CBV
$\lambda$-calculus on open terms turns out to be much more subtle and
trickier
(see~\cite{AccattoliGuerrieri16,AccattoliGuerrieri19,AccattoliGuerrieri22}
for a detailed discussion). In particular, Plotkin's original
CBV~\cite{Plotkin75} is not \emph{adequate} for open terms, as there
exist terms that may be both \emph{irreducible} and
\emph{meaningless/unsolvable}. The non-adequacy problem in  Plotkin's
CBV calculus can be repaired by introducing a form of sharing
implemented by \emph{explicit substitutions (ES)}, together with a
notion of \emph{reduction at a
distance}~\cite{AccattoliKesner10,AccattoliKesner12}, like in the
Value Substitution Calculus~\cite{AccattoliPaolini12} (here called
\CBVSymb), a CBV variant of Accattoli and Kesner's linear substitution
calculus~\cite{Accattoli12,ABKL14} (generalizing in turn Milner's
calculus~\cite{Milner2006,KOC}). Adequacy also fails for the version
of the Bang Calculus studied
in~\cite{GuerrieriManzonetto19,FaggianGuerrieri21}, for the same
reasons as in CBV. It can be repaired again via ES and distance,
resulting in the  Distant Bang Calculus
\BANGSymb~\cite{BucciarelliKesnerRiosViso20,BucciarelliKesnerRiosViso23}.
It is then natural to also integrate  ES and distance in the CBN
specification: this gives rise to CBN substitution calculi at a distance~\cite{AccattoliKesner10,AccattoliKesner12}, here we call \CBNSymb the one in~\cite{Accattoli12}, which is adequate as the usual CBN. In summary, we
focus in this paper on a CBN calculus \CBNSymb, an adequate CBV
calculus \CBVSymb, and the  adequate unifying Distant Bang Calculus
\BANGSymb.

\paragraph{Static and dynamic.}
The literature {has} shown that {some} \emph{static}
properties of CBN and CBV, including normal forms~\cite{KesnerViso22},
quantitative typing~\cite{BucciarelliKesnerRiosViso20}, tight
typing~\cite{KesnerViso22,BucciarelliKesnerRiosViso23},
inhabitation~\cite{ArrialGuerrieriKesner23}, and denotational
semantics~\cite{GuerrieriManzonetto19}, can be inferred from their
corresponding counterparts in the (Distant) Bang Calculus by
exploiting {suitable} CBN and CBV encodings. However, retrieving
\emph{dynamic} properties from the Bang Calculus into CBN or CBV turns
out to be a more intricate task, especially in their \emph{adequate}
(distant)
variant~\cite{GuerrieriManzonetto19,BucciarelliKesnerRiosViso20,FaggianGuerrieri21,BucciarelliKesnerRiosViso23}.
Indeed, it is {easy} to obtain \emph{simulation} {(a CBN or CBV reduction sequence is always embedded into a \BANGSymb reduction sequence)}, {\emph{but}} the
converse, known as \emph{reverse simulation}, fails{:  a \BANGSymb reduction sequence from a term
in the image of the CBN or CBV embedding  may not correspond to a
valid reduction sequence in CBN or CBV (counterexample in
%\Cref{sec:CBN_CBV_Embeddings}, 
\Cref{CE: CBV Reverse Simulation})}. 
Up to these days, there are no embeddings in the
literature enjoying reverse simulation for an adequate CBV calculus,
so that it is not possible to export \mbox{dynamic properties from \BANGSymb
to both \CBNSymb and \CBVSymb}.

\paragraph{Contributions.} We first revisit and \emph{extend} the
existing static and dynamic preservation results relating \CBNSymb and
\BANGSymb,
%\cgiulio{so as to 	establish a comprehensive connection between them,}{} 
including simulation and reverse simulation, exploiting the embedding
used in
\cite{BucciarelliKesnerRiosViso20,BucciarelliKesnerRiosViso23}.
However, our primary and most significant contribution is a new
\emph{methodology} to deal with the (adequate) calculus \CBVSymb.
Indeed, we define a \emph{novel embedding} from \CBVSymb into
\BANGSymb, \emph{refining} the one of
\cite{BucciarelliKesnerRiosViso20,BucciarelliKesnerRiosViso23}, that
finely decorates terms with the modalities $\oc$ and $\derSymb$. 
%\cgiulio{In particular, since these two modalities annihilate each other, 
%	some of these decorations turn out to be redundant. A}
To avoid redundant decorations, as $\oc$ and $\derSymb$ annihilate each other, a dedicated
$\bangSymbBang$-reduction step is then applied \emph{on the fly} by the
embedding, as in~\cite{BucciarelliKesnerRiosViso20,BucciarelliKesnerRiosViso23}. But our new 
{\CBVSymb} embedding not only preserves static and dynamic properties,
but also satisfies \emph{reverse simulation}, an essential property
that was previously lacking. This achievement is realized by the
second ingredient of our new methodology, given by  the notion of
\emph{diligent sequence} in \BANGSymb, a concept standing
independently of the embeddings. Indeed, a challenge at this point is
to prove that the earlier mentioned $\bangSymbBang$-reductions have a
purely \emph{administrative} nature, and additionally, that they can
be treated \emph{diligently}, by executing all of them as soon as
possible. We call this method \emph{diligent administration}: we
consistently address all administrative steps before proceeding with
any other \emph{computational} steps. A further challenge is then to establish that
working with administrative diligence does not alter the CBN or CBV
\mbox{nature~of~evaluation}.

As explained above, reverse simulation is crucial to derive properties
for \CBNSymb and \CBVSymb from their respective properties in
\BANGSymb. We provide two main illustrative \emph{applications} of
this by studying the cases of \emphit{confluence} and
\emphit{factorization}. Confluence is a well known property, and
factorization is crucial to prove important results in
rewriting~\cite{Mitschke79,Barendregt84,Takahashi95,Wadsworth76,Accattoli12,AccattoliFaggianGuerrieri19,AccattoliFaggianGuerrieri21,GuerrieriPaoliniRonchi17,Guerrieri15,FaggianGuerrieridLiguoroTreglia22}:
we say that a reduction enjoys factorization when every reduction
sequence can be rearranged so that some specific relevant steps are
performed first. In the two last sections, we use
confluence/factorization for \BANGSymb as a basis to easily deduce
confluence/factorization for \CBNSymb and \CBVSymb. This is done by
exploiting the CBN and CBV embeddings back and forth, via reduction
simulation and reverse~simulation. Just one proof is enough for three
confluence/factorization results: it's a three-for-one deal!  
The fact that \CBNSymb and \CBVSymb confluence/factorizations can be
\emph{easily} derived from \BANGSymb confluence/factorization in
essentially the \emph{same} way is another achievement, attained
thanks to having introduced good tools, such as diligence and the new
\CBVSymb embedding.

We actually provide a first proof of factorization for \BANGSymb,
another major contribution of this paper. Factorizations in \CBNSymb
and \CBVSymb were already proved in~\cite{Accattoli12}
and~\cite{AccattoliPaolini12}, respectively, but their proofs are not
trivial, even when applying some abstract approach~\cite{Accattoli12}.

\paragraph{Road Map.}
\Cref{sec:Distant_Bang_Calculus} recalls \BANGSymb  and introduces
diligence. The \CBNSymb/\CBVSymb calculi and their embeddings are
presented in \Cref{sec:CBN_CBV_Embeddings}, together with their
corresponding (static and dynamic) preservation results.
\Cref{sec:Bang_CBN_CBV_Confluence} derives \CBNSymb/\CBVSymb
confluence from that of \BANGSymb. \Cref{sec:Factorization} proves a
factorization result for \BANGSymb, and deduces factorization for
\CBNSymb and \CBVSymb by projection. \Cref{s:conclusion} discusses
future and related work and concludes. Proofs can be found
in~\ifbool{versionLongue}{the appendix}{\victor{[?]}}.

\subsection{Basic Notions Used all along the Paper.}
\label{sec:basic-notions}

An \emphasis{abstract rewriting system ($ARS$)} $\relE$ is a set $E$
with a binary relation $\arrE$ on $E$, called \emphasis{reduction}. We
write $u \revArrE t$ if $t \arrE u$, and we denote by $\arrE^+$ (\resp
$\arrE^*$) the transitive (\resp reflexive-transitive) closure of
$\arrE$. Given $t \in E$, $t$ is an \emphasis{$\relE$-normal form}
($\relE$-{\tt NF}) if there is no $u \in E$ such that $t \arrE u$; $t$
is \emphasis{$\relE$-terminating} if there is no infinite $\arrE$
reduction sequence starting at $t$. Reduction $\arrE$ is
\emphasis{terminating} if every $t \in E$ is $\relE$-terminating;
$\arrE$ is \emphasis{diamond} if for any $t, u_1, u_2 \in E$ such that
$u_1 \revArrE t \arrE u_2$ and $u_1 \neq u_2$, there is $s\in E$ such
that $u_1 \arrE s \revArrE u_2$; $\arrE$ is \emphasis{confluent} if
$\arrE^*$ is diamond.

All reductions in this paper will be defined by a set of rewrite rules
$\SetRules$, closed by a set of contexts $\SetContexts$. A term being an instance of the left-hand side of a rewrite rule
$\rel \in \SetRules$ is called a \emphasis{$\rel$-redex}.  Given a
rule $\rel \in \SetRules$, and a context $\bangECtxt \in
\SetContexts$, we use $\bangArr_E<R>$ to denote the reduction of the
$\rel$-redex under the context $\bangECtxt$.  The reduction
$\bangErrSet_E<R>$ is the union of reductions $\bangArr_E<R>$ over
\emph{all} contexts $\bangECtxt \in \SetContexts$. In other words,
$\bangArrSet_E<R>$ is the closure of the rule $\rel$  under all the
contexts in $\SetContexts$.

\newcommand{\nameSectionII}
    {Distant Bang Calculus}

\section{The Distant Bang Calculus \BANGSymb}
\label{sec:Distant_Bang_Calculus}

We introduce the term syntax of
\BANGSymb~\cite{BucciarelliKesnerRiosViso20}. Given a countably
infinite set $\setVar$ of variables $x, y, z, \dots$, the set
$\setBangTerms$ of terms is defined inductively as follows:
\begin{equation*}
	\textbf{(Terms)} \qquad t, u, s \;\Coloneqq\; x \in \setVar
        \vsep \app{t}{u}
	    \vsep \abs{x}{t}
        \vsep \oc t
        \vsep \der{t}
        \vsep t\esub{x}{u}
\end{equation*}

The set $\setBangTerms$ includes \emphasis{variables} $x$,
\emphasis{abstractions} $\abs{x}{t}$, \emphasis{applications} $\app{t}{u}$,
\emphasis{closures} $t\esub{x}{u}$ representing a pending \emphasis{explicit
substitution} (\emph{ES}) $\esub{x}{u}$ on $t$, \emphasis{bangs} $\oc t$
and \emphasis{derelictions} $\der{t}$ (their operational meaning is
\mbox{explained~below}).

Abstractions $\abs{x}{t}$ and closures $t\esub{x}{u}$ bind the
variable $x$ in their body $t$. The set of \emphasis{free variables}
$\freeVar{t}$ of a term $t$ is defined as expected, in particular
$\freeVar{\abs{x}{t}} \coloneqq \freeVar{t} \setminus \{x\}$ and
$\freeVar{ t\esub{x}{u}} \coloneqq \freeVar{u} \cup (\freeVar{t}
\setminus \{x\})$. The usual notion of $\alpha$-conversion
\cite{Barendregt84} is extended to the whole set $\setBangTerms$, and
terms are identified up to $\alpha$-conversion, \eg
$y\esub{y}{\abs{x}{x}} =  z\esub{z}{\abs{y}{y}}$. We denote by
$t\isub{x}{u}$ the usual (capture avoiding) meta-level substitution of
$u$ for all free occurrences of $x$~in~$t$.
\emphasis{\Full contexts} $(\bangFCtxt \in \bangFCtxtSet)$,
\emphasis{surface contexts} $(\bangSCtxt \in \bangSCtxtSet)$ and
\emphasis{list contexts} $(\bangLCtxt \in \bangLCtxtSet)$, which can
be seen as terms with exactly one \emphasis{hole} $\Hole$, are
inductively defined by:
\begin{align*}
	{\emphasis{(Full Contexts)}}& &
	\bangFCtxt &\Coloneqq \Hole
		\vsep \app[\,]{\bangFCtxt}{t}
		\vsep \app[\,]{t}{\bangFCtxt}
		\vsep \abs{x}{\bangFCtxt}
		\vsep \oc \bangFCtxt
		\vsep \der{\bangFCtxt}
		\vsep \bangFCtxt\esub{x}{t}
		\vsep t\esub{x}{\bangFCtxt}
\\[-2pt]
	{\emphasis{(Surface Contexts)}}& &
	\bangSCtxt &\Coloneqq \Hole
		\vsep \app[\,]{\bangSCtxt}{t}
		\vsep \app[\,]{t}{\bangSCtxt}
		\vsep \abs{x}{\bangSCtxt}
		\vsep \der{\bangSCtxt}
		\vsep \bangSCtxt\esub{x}{t}
		\vsep t\esub{x}{\bangSCtxt}
\\[-2pt]
	{\emphasis{(List Contexts)}}& &
	\bangLCtxt &\Coloneqq \Hole \vsep
		\bangLCtxt\esub{x}{t}
\end{align*}
$\bangLCtxt$ and $\bangSCtxt$ are special cases of $\bangFCtxt$: the
hole may occur everywhere in $\bangFCtxt$, while in $\bangSCtxt$ it
cannot appear under a $\oc$. List contexts $\bangLCtxt$ are arbitrary
lists of ES, used to implement reduction at a
distance~\cite{AccattoliKesner10,AccattoliKesner12}. We write
$\bangFCtxt<t>$ for the term obtained by replacing the hole in
$\bangFCtxt$ with the term $t$ (possibly capturing the free
variables~of~$t$).
The following \emphasis{rewrite rules} are the base components of
our reductions. 
%Any term having the shape of the left-hand
%side of one of these three rules is called a
%$\rel$-\emphasis{redex}, where $\rel
%\in \{\bangSymbBeta, \bangSymbSubs, \bangSymbBang \}$.
\begin{align*}
%	\begin{array}{rrcl}
%		\textbf{(Distant Beta)}\hspace{1cm}&
%			\app{\bangLCtxt<\abs{x}{t}>}{u}
%				&\;\bangMapstoBeta\;&
%					\bangLCtxt<t\esub{x}{u}>
%	\\
%		\textbf{(Substitute Bang)}\hspace{1cm}&
%			t\esub{x}{\bangLCtxt<\oc u>}
%				&\bangMapstoSubs&
%					\bangLCtxt<t\isub{x}{u}>
%	\\
%		\textbf{(Distant Bang)}\hspace{1cm}&
%			\der{\bangLCtxt<\oc t>}
%				&\bangMapstoBang&
%					\bangLCtxt<t>
%	\end{array}
\app{\bangLCtxt<\abs{x}{t}>}{u}
				&\bangMapstoBeta
					\bangLCtxt<t\esub{x}{u}>
&
t\esub{x}{\bangLCtxt<\oc u>}
				&\bangMapstoSubs
					\bangLCtxt<t\isub{x}{u}>
&
\der{\bangLCtxt<\oc t>}
				&\bangMapstoBang
					\bangLCtxt<t>
\end{align*}
Rule $\bangSymbBeta$ (\resp $\bangSymbSubs$) is assumed to be
capture-free, so no free variable of $u$ (resp. $t$) is captured by
the context $\bangLCtxt$.
%Given the translation of the original Bang
%Calculus into linear logic proof-nets~\cite{EhrhardGuerrieri16},
%$\bangSymbBeta$-steps can be seen as \emph{multiplicative} and
%$\{\bangSymbSubs,\bangSymbBang\}$-steps as \emph{exponential}.
The rule $\bangSymbBeta$ fires a $\beta$-redex, generating an ES. The
rule $\bangSymbSubs$ fires an ES provided that its argument is
duplicable, \ie is a bang. The rule $\bangSymbBang$ uses $\derSymb$ to
erase a $\oc$. In all of these rewrite rules, the reduction acts
\emphit{at a distance}~\cite{AccattoliKesner10,AccattoliKesner12}: the
main constructors involved in the rule can be separated by a finite
---possibly empty--- list $\bangLCtxt$ of ES.  This mechanism unblocks
desired computations that otherwise would be stuck, \eg
$(\abs{x}{x})\esub{y}{w}\oc z \bangMapstoBeta x\esub{x}{\oc z}
\esub{y}{w}$.

%

%\cdelia{Given a rule $\rel \in \{\bangSymbBeta, \bangSymbSubs,
%\bangSymbBang\}$ and a context $\bangSCtxt \in \bangSCtxtSet$ (\resp
%$\bangFCtxt \in \bangFCtxtSet$), $\bangArr_S<R>$ (\resp
%$\bangArr_F<R>$) is the reduction of the $\rel$-redex under the
%context ---\ie at the position--- $\bangSCtxt$ (\resp $\bangFCtxt$).
%The reduction $\bangArrSet_S<R>$ (\resp $\bangArrSet_F<R>$) is the
%union of the reductions $\bangArr_S<R>$ (\resp $\bangArr_F<R>$) over
%all contexts $\bangSCtxt \in \bangSCtxtSet$ (\resp $\bangFCtxt \in
%\bangFCtxtSet$). In other words, $\bangArrSet_S<R>$ (\resp
%$\bangArrSet_F<R>$) is the closure of $\mapsto_\rel$ under contexts in
%$ \bangSCtxtSet$ (resp. $\bangFCtxtSet$). }
Reductions are defined, as  specified in~\cref{sec:basic-notions}, by
taking  the set of rewrite rules $\{\bangSymbBeta, \bangSymbSubs,
\bangSymbBang\}$ and the sets of contexts $\bangSCtxtSet$ and
$\bangFCtxtSet$. \emphasis{Surface reduction} is the relation
$\bangArrSet_S \,\coloneqq\, \bangArrSet_S<dB> \cup \bangArrSet_S<s!>
\cup \bangArrSet_S<d!>$, while \emphasis{\full reduction} is the
relation $\bangArrSet_F \,\coloneqq\, \bangArrSet_F<dB> \cup
\bangArrSet_F<s!> \cup \bangArrSet_F<d!>$. For example, for
$\bangSCtxt_1 = \Hole \in \bangSCtxtSet$ and $\bangFCtxt_1= \oc \Hole
\in \bangFCtxtSet \setminus \bangSCtxtSet$: $\app{(\abs{x}{\oc
\der{\oc x}})}{\oc y} \;\bangArr_{\bangSCtxt_1}<\bangSymbBeta>\; (\oc
\der{\oc x})\esub{x}{\oc y} \;\bangArr_{\bangSCtxt_1}<\bangSymbSubs>\;
\oc\der{\oc y} \;\bangArr_{\bangFCtxt_1}<\bangSymbBang>\; \oc y$. The
first two steps are $\bangArrSet_S$- and also $\bangArrSet_F$-steps,
while the last one is a $\bangArrSet_F$-step but not a
$\bangArrSet_S$-step. More generally, $\bangArrSet_S \,\subsetneq\,
\bangArrSet_F$. For instance, $\oc (\der{\oc y})$ is a \bangSetSNF but
not a \bangSetFNF\ since $\oc (\der{\oc y}) \bangArrSet_F \oc y$,
while $\oc y$ is a \bangSetFNF (and hence a \bangSetSNF too).

The $\oc$ modality plays a twofold role. First, it marks the only
subterms that can be substituted (\ie erased or arbitrarily copied):
the $\bangSymbSubs$-rule fires an ES only if there is a $\oc$ in its
argument (up to a list context). Second, it freezes (\surface)
evaluation of the term under the scope of $\oc$: surface reduction
$\bangArrSet_S$ does not reduce under $\oc$. In \full reduction
$\bangArrSet_{F}$, the $\oc$ modality looses its~freezing~behavior.

\paragraph{Diligent Administration.}
While reductions $\bangArrSet_F<dB>$ and $\bangArrSet_F<s!>$ are
actual \emph{computational} steps, reduction $\bangArrSet_F<d!>$ is
rather \emph{administrative} in nature. As we use \BANGSymb to
simulate other calculi, we need to align with the \emph{implicit
nature} of these administrative steps: this can be achieved by
executing them as soon as possible. We thus introduce a \emph{diligent
process} that reorders some reduction steps to ensure that
administrative steps are always performed as soon as \mbox{there is a
$\bangSymbBang$-redex}.

To begin, we formally introduce the concept of \emphasis{diligent
administrative} reduction sequence, characterizing sequences where
each \emph{computational} step ($\bangSymbBeta$ or $\bangSymbSubs$)
can be performed only \emph{after} all \emph{administrative} steps
($\bangSymbBang$) have been executed.

\begin{definition}[Diligent Administrative Reduction]
    The \emph{diligent administrative} \emph{\surface} (\resp
    \emph{\full}) reduction $\bangArrSet_{Sai}$ (\resp
    $\bangArrSet_{Fai}$) is a subset of the \surface (\resp \full)
    reduction obtained by restricting $\bangSymbBeta$- and
    $\bangSymbSubs$-steps to $\bangSCtxtSet<\bangSymbBang>$-normal
    forms (\resp $\bangFCtxtSet<\bangSymbBang>$-normal forms). More
    precisely, it is defined as follows:
    \begin{equation*}
        \begin{array}{rcl}
            \bangArrSet_{Sai}
                &\coloneqq& (\bangArrSet_S<dB> \cap\; \text{\bangSetSNF<d!>} \times \setBangTerms)
                    \;\;\cup\;\; (\bangArrSet_S<s!> \cap\; \text{\bangSetSNF<d!>} \times \setBangTerms)
                    \;\;\cup\; \bangArrSet_S<d!>
		\\[0.2cm]
            \bangArrSet_{Fai}
                &\coloneqq& (\bangArrSet_F<dB> \cap\; \text{\bangSetFNF<d!>} \times \setBangTerms)
                    \;\;\cup\;\; (\bangArrSet_F<s!> \cap\; \text{\bangSetFNF<d!>} \times \setBangTerms)
                    \;\;\cup\; \bangArrSet_F<d!>
        \end{array}
    \end{equation*}
	% where \bangSNF<d!> $:= \{t \in \setBangTerms \vsep t
	% \not\bangArr_S<d!>\}$ and \bangFNF<d!> $:= \{t \in \setBangTerms
	% \vsep t \not\bangArr_F<d!>\}$ 
	
	% As expected, $\bangArr*_{Sai}$ (\resp $\bangArr*_{Fai}$) denoted
	% the reflexive and transitive closure of
    % % the \surface (\resp \full) administratively synchronized reduction
	% the reduction $\bangArr_{Sai}$ (\resp $\bangArr_{Fai}$). 
\end{definition}

\begin{example}
    \label{ex:diligent}
    Consider the two surface reduction sequences $\der{\oc
    x}\esub{x}{\oc y} \bangArrSet_S<s!> \der{\oc y} \bangArrSet_S<d!>
    y$ and $\der{\oc x}\esub{x}{\oc y} \bangArrSet_S<d!> x\esub{x}{\oc
    y} \bangArrSet_S<s!> y$. The first one is not diligent
    administrative, as the step $\bangArrSet_S<s!>$ is performed in a
    term that is not \bangSetSNF<d!>. But the second one is diligent
    administrative: \mbox{$\der{\oc x}\esub{x}{\oc y}
    \bangArrSet_{Sai} x\esub{x}{\oc y} \bangArrSet_{Sai} y$}.
\end{example}

To show that every reduction sequence can be transformed into a
diligent  one (\Cref{lem:Implicitation}), we first observe that it is
%actually 
possible to perform \emph{all} administrative steps from any term:
indeed, reductions $\bangArrSet_F<d!>$ and $\bangArrSet_S<d!>$ are
\emph{terminating}, because each administrative step erase two
constructors, $\derSymb$ and $\oc$, so the term size decreases.

Some reduction sequences can be made diligent, as in
\Cref{ex:diligent}, but this is not the case for all reduction
sequences. For instance $\der{\oc x}\esub{x}{\oc y} \bangArrSet_S
\der{\oc y}$ but $\der{\oc x}\esub{x}{\oc y} \not\bangArrSet_{Sai}
\der{\oc y}$. Therefore, we focus solely on reduction sequences
reaching terms that are normal for $\bangSymbBang$. Under these conditions and by
commuting computational steps with administrative ones, we obtain the
following results:

\begin{restatable}[Diligence Process]{lemma}{RecImplicitation}
    \LemmaToFromProof{Implicitation}
    % \label{lem:Implicitation}%
    Let $t, u \in \setBangTerms$ be terms.
    \begin{itemize}
    \item  \textbf{(Surface)} \;%
        If $t \bangArrSet*_S u$ and $u$ is a \bangSetSNF<d!>, then $t
        \bangArrSet*_{Sai} u$.
    \item \hspace{0.2cm} \textbf{(Full)} \;\hspace{0.3cm}%
        If $t \bangArrSet*_F u$ and $u$ is a \bangSetFNF<d!>, then $t
        \bangArrSet*_{Fai} u$.
    \end{itemize}
\end{restatable}
% \begin{proof}
%     It is based on definitions and lemmas relegated to (and proved in)
%     \Cref{a:abstract-diligence}. The reasoning goes as follows. First,
%     an abstract notion of administrative diligent reduction is defined
%     for any family of contexts (\Cref{def:abstract-diligence}).
%     Diligence for this abstract notion is proved (\Cref{lem: One Step
%     Abstract Implicitation,lem: Abstract Implicitation}).
% %\odelia{This is possible when some local commutation diagrams hold.}
%     Finally, the properties are established, drawing upon the
%     verification of local commutation diagrams
%     %\odelia{for each of them}
%     (\Cref{lem:R=dB/s!andt->S<R>u1andt->S<d!>u2=>u1->S<d!>sandu2->S<R>s,lem:t->*S<d!>u1andt->*S<d!>u2=>u1->*S<d!>sandu2->*S<d!>s}
%     for the surface case,
%     \Cref{lem:R=dB/s!andt->F<R>u1andt->F<d!>u2=>u1->*F<d!>sandu2->F<R>s,lem:t->*F<d!>u1andt->*F<d!>u2=>u1->*F<d!>sandu2->*F<d!>s}
%     for the~full~case).
%     \qed
% \end{proof} 

\newcommand{\nameSectionIII}
    {Call-by-Name and Call-by-Value Embeddings}

\section{Call-by-Name and Call-by-Value Embeddings}
\label{sec:CBN_CBV_Embeddings}

In this section we present the call-by-name \CBNSymb
(\Cref{subsec:Call-by-Name_Calculus_and_Embedding}) and call-by-value
\CBVSymb (\Cref{subsec:Call-by-Value_Calculus_and_Embedding}) calculi,
as well as their embeddings into \BANGSymb, which preserve static
properties (Corollaries \ref{cor:cbn prop-full}.2 and %
\ref{cor:cbn prop-surface}.2 for \CBNSymb, \ref{cor:cbv prop-full}.2
and \ref{cor:cbv prop-surface}.2 for \CBVSymb) and dynamic ones
(Corollaries \ref{cor:cbn prop-full}.3 and %
\ref{cor:cbn prop-surface}.3 for \CBNSymb, \ref{cor:cbv prop-full}.3 and
\ref{cor:cbv prop-surface}.3 for \CBVSymb).

Both \CBNSymb~\cite{AccattoliKesner10,AccattoliKesner12,Accattoli12}
and \CBVSymb~\cite{AccattoliPaolini12} are specified  using ES and
action at a distance, as explained in~\Cref{sec:intro}, and they share
the same term syntax. The sets $\setCbvTerms$ of \emphasis{terms} and
$\setCbvValues$ of \emphasis{values} are inductively defined below.
\begin{equation*}
    \textbf{(Terms)} \ \ 
    t, u 	\;\Coloneqq\;  v
    \vsep \app[\,]{t}{u}
    \vsep t\esub{x}{u}
    \hspace{1cm}
    \textbf{(Values)} \ \
    v \;\Coloneqq\; x 
    \vsep \abs{x}{t}
\end{equation*}
Note that the syntax  contains neither $\derSymb$~nor~$\oc$. The
distinction between terms and values is irrelevant in \CBNSymb but
crucial in \CBVSymb. The two calculi also share the same
\emphasis{full contexts} $\genFCtxt$ and \emphasis{list contexts}
$\genLCtxt$, which can be seen as terms with exactly one
\emphasis{hole} $\Hole$ and are inductively defined below. The
differences between \CBNSymb and \CBVSymb are in the definitions of
\emph{surface} contexts and \emph{rewrite rules}.
\begin{align*}
		\emphasis{(List Contexts)}& &
		\genLCtxt&\;\Coloneqq  \Hole
		\vsep \genLCtxt\esub{x}{t}
		\\[-2pt]
		\emphasis{(Full Contexts)}& &
		\genFCtxt&\;\Coloneqq \Hole
		\vsep \app[\,]{\genFCtxt}{t}
		\vsep \app[\,]{t}{\genFCtxt}
		\vsep \abs{x}{\genFCtxt}
		\vsep \genFCtxt\esub{x}{t}
		\vsep t\esub{x}{\genFCtxt}
\end{align*}

\subsection{The Call-by-Name Calculus \CBNSymb and its embedding to \BANGSymb}
\label{subsec:Call-by-Name_Calculus_and_Embedding}

%The \emph{call-by-name $\lambda$-calculus} (called \CBNSymb\
%here) is specified  using ES and action at a distance~\cite{AccattoliKesner10,AccattoliKesner12,Accattoli12}, as explained in the introduction. The set
%$\setCbnTerms$ of \emphasis{terms} is inductively defined as follows
%(note that the syntax contains neither ${\tt der}$ nor $\oc$):
%%
%\begin{equation*}
%    \textbf{(Terms)} \ \ 
%        t, u 	\;\Coloneqq\;  x
%            \vsep \abs{x}{t}
%            \vsep \app{t}{u}
%            \vsep t\esub{x}{u}
%\end{equation*}
%
In \CBNSymb, \emphasis{surface contexts} $\cbnSCtxt \in \cbnSCtxtSet$
are defined below: the hole cannot be in the argument of an application or ES. 
To align the notations, in \CBNSymb full contexts
are denoted by $\cbnFCtxt \in \cbnFCtxtSet$ and list contexts by
$\cbnLCtxt \in \cbnLCtxtSet$.
\begin{align*}
        \emphasis{(\CBNSymb Surface Contexts)}& &
        \cbnSCtxt&\;\Coloneqq  \Hole
        \vsep \app[\,]{\cbnSCtxt}{t}
        \vsep \abs{x}{\cbnSCtxt}
        \vsep \cbnSCtxt\esub{x}{t}
\end{align*}

As explained in~\cref{sec:basic-notions}, reductions
in $\CBNSymb$ are defined by taking the set of rewrite rules
$\{\cbnSymbBeta, \cbnSymbSubs\}$ defined below and the sets of contexts
$\cbnSCtxtSet$ and $\cbnFCtxtSet$.
\begin{equation*}
	\app{\cbnLCtxt<\abs{x}{t}>}{u}
	\;\cbnMapstoBeta\;
	\cbnLCtxt<t\esub{x}{u}>
	\hspace{2cm}
	t\esub{x}{u}
	\;\cbnMapstoSubs\;
	t\isub{x}{u}
\end{equation*}

Rule \cbnSymbBeta is capture-free: no free variable of $u$ is captured
by the context $\cbnLCtxt$. The \emphasis{\CBNSymb surface reduction}
is the relation $\cbnArrSet_S \,\coloneqq\, \cbnArrSet_S<dB>
\cup\cbnArrSet_S<s>$, while the \emphasis{\CBNSymb full reduction} is
the relation $\cbnArrSet_F \,\coloneqq\, \cbnArrSet_F<dB> \cup
\cbnArrSet_F<s>$. \Eg, for $\cbnFCtxt = \abs{z}{\Hole}$, $t_0 =
\abs{z}{(\app{(\abs{x}{\app{y}{\app{x}{x}}})}{(\app{z}{z})})}
\cbnArr_F<dB> t_1 =
\abs{z}{((\app{y}{\app{x}{x}})\esub{x}{\app{z}{z}})} \cbnArr_F<s> t_2
= \abs{z}{(\app{y}{\app{(\app{z}{z})}{(\app{z}{z})}})}$.

The {\CBNSymb surface reduction} is nothing but (a non-deterministic
but diamond variant of) the well-known \emph{head} reduction.

\paragraph{Embedding \CBNSymb into \BANGSymb.}
The \emphasis{\CBNSymb embedding} $\cbnToBangAGK{\cdot} \colon
\setCbnTerms \to \setBangTerms$ from \CBNSymb to \BANGSymb, introduced
in \cite{BucciarelliKesnerRiosViso20,BucciarelliKesnerRiosViso23} and
presented below, extends Girard's one \cite{EhrhardGuerrieri16}~to~ES.
%\begin{definition}
%	\label{def:cbnAGK_Embedding}
%	The \emphasis{\CBNSymb embedding} $\cbnToBangAGK{\cdot} \colon
%	\setCbnTerms \to \setBangTerms$ is defined as follows:
	\begin{align*}
			\cbnToBangAGK{x}
                &\coloneqq
            x
		&
			\cbnToBangAGK{(\abs{x}{t})}
			&\coloneqq
			\abs{x}{\cbnToBangAGK{t}}
        &
            \cbnToBangAGK{(\app{t}{u})}
                &\coloneqq
            \app[\,]{\cbnToBangAGK{t}}{\oc\cbnToBangAGK{u}}
		&
			\cbnToBangAGK{(t\esub{x}{u})}
                &\coloneqq
            \cbnToBangAGK{t}\esub{x}{\oc\cbnToBangAGK{u}}.
	\end{align*}
%\end{definition}
As an example, $\cbnToBangAGK{(\app{y}{x})\esub{y}{z}} = (\app{y}{\oc
x})\esub{y}{\oc z}$. Note that $\cbnToBangAGK{\cdot}$ never
introduces $\derSymb$, hence $\cbnToBangAGK{t}$, and every term it
reduces to, are always a \bangSetFNF<d!> (this does not hold for the
\CBVSymb embedding,
\Cref{subsec:Call-by-Value_Calculus_and_Embedding}). In every
application and ES, $\cbnToBangAGK{\cdot}$ puts a $\oc$ in front of
their argument, which shows the two roles---called
\emph{duplicability} and \emph{accessibility}---played by $\oc$ in
this embedding: \CBNSymb duplicability means that any argument can be
duplicated (or erased), \CBNSymb accessibility means that \surface
reduction cannot take place inside arguments. Indeed, the $\oc$ seals
all subterms in argument~position.

The embedding is trivially
extended to \CBNSymb contexts by setting $\cbnToBangAGK{\Hole} =
\Hole$.
~

The static properties of this embedding have already been partially
discussed in
\cite{BucciarelliKesnerRiosViso20,BucciarelliKesnerRiosViso23}. We
will revisit and refine them (\Cref{cor:cbn prop-full,cor:cbn
prop-surface,cor:cbn prop-internal}), but our main focus lies in the
preservation of the dynamics of \CBNSymb within \BANGSymb. For that,
we first extend the embedding to rule names, by defining
$\cbnToBangAGK{\cbnSymbBeta} \;\coloneqq\; \bangSymbBeta$ and
$\cbnToBangAGK{\cbnSymbSubs} \;\coloneqq\; \bangSymbSubs$.

The reduction of a \CBNSymb redex can be effectively simulated in
\BANGSymb by reducing the corresponding redex occurring at the
translated location/context.

\begin{restatable}[\CBNSymb One-Step Simulation]{lemma}{RecCbnOneStepSimulation}
    % \LemmaToFromProof{Cbn_Simulation_One_Step}%
    \label{lem:Cbn_Simulation_One_Step}%
    Let $t, u \in \setCbnTerms$ and $\cbnFCtxt \in
    	\cbnFCtxtSet$ and $\rel \in \{\cbnSymbBeta,
    \cbnSymbSubs\}$. If $t \cbnArr_F<R> u$ then $\cbnToBangAGK{t} \bangArr_{\cbnToBangAGK{\cbnFCtxt}\bangCtxtPlug{\cbnToBangAGK{\rel}}} \cbnToBangAGK{u}$.
    \label{l:one-step-cnb}
\end{restatable}
\begin{example}
    \label{ex:cbn-simulation}
    Consider the \CBNSymb reductions $t_0 \cbnArr_F<dB> t_1$ and $t_1
	\cbnArr_F<s> t_2$ seen above with $\cbnFCtxt = \abs{z}{\Hole}$.
	Since $ \cbnToBangAGK{\cbnFCtxt} = \abs{z}{\Hole}$,
	we have
%	\begin{align*}
%		\cbnToBangAGK{t_0}
%            &= \abs{z}{(\app{(\abs{x}{\app{y}{\app{\oc x}{\oc x}}})}{\oc (\app{z}{\oc z})} )}
%	            \bangArr_{ \cbnToBangAGK{\cbnFCtxt}}<dB>
%		            \abs{z}{((\app[\,]{y}{\app{\oc x}{\oc x}})\esub{x}{\oc (\app{z}{\oc z})} )}
%                    = \cbnToBangAGK{t_1}
%		\\
%        \cbnToBangAGK{t_1}
%            &= \abs{z}{((\app[\,]{y}{\app{\oc x}{\oc x}})\esub{x}{\oc (\app{z}{\oc z})}} )
%                \bangArr_{ \cbnToBangAGK{\cbnFCtxt}}<s!>
%                    \abs{z}{(y \app{\oc(\app{z}{\oc z})}{\oc(\app{z}{\oc z})} )}
%                    = \cbnToBangAGK{t_2}
%	\end{align*}}
$\cbnToBangAGK{t_0}
= \abs{z}{(\app{(\abs{x}{\app{y}{\app{\oc x}{\oc x}}})}{\oc (\app{z}{\oc z})} )}
\allowbreak\bangArr_{\cbnToBangAGK{\cbnFCtxt}\bangCtxtPlug{\bangSymbBeta}}
\abs{z}{((\app[\,]{y}{\app{\oc x}{\oc x}})\esub{x}{\oc (\app{z}{\oc z})} )}
= \cbnToBangAGK{t_1}$
and
$\cbnToBangAGK{t_1}
%= \abs{z}{((\app[\,]{y}{\app{\oc x}{\oc x}})\esub{x}{\oc (\app{z}{\oc z})}} )
\bangArr_{\cbnToBangAGK{\cbnFCtxt}\bangCtxtPlug{\bangSymbSubs}}
\abs{z}{(y \app{\oc(\app{z}{\oc z})}{\oc(\app{z}{\oc z})} )}
= \cbnToBangAGK{t_2}$.
    \end{example}

%   \cgiulio{ In particular, note that every \CBNSymb reduction step is
%    faithfully simulated by the corresponding \BANGSymb reduction
%    step. This simulation process operates seamlessly without any
%    administrative step\cgiulio{, thus removing the need for any diligence process.}{.} 
%    We will see however (\Cref{lem:Cbv_Simulation_One_Step})
%  that the $\CBVSymb$ case is more~involved\giulio{, requiring some administrative steps}.}
{So, every \CBNSymb reduction step is simulated by the corresponding \BANGSymb reduction
	step, without the need for any administrative step.  
	Simulation of $\CBVSymb$ (\Cref{lem:Cbv_Simulation_One_Step}) is instead more~involved, requiring some further administrative steps.}

%% ~

%% In the simulation process we just described, we used \BANGSymb
%% reductions to simulate \CBNSymb reductions, or from another point of
%% view, we used \CBNSymb reductions to deduce
%%  some
%% of \BANGSymb reductions. This process can be reversed, allowing us to
%% use some of \BANGSymb reductions to foresee those of \BANGSymb

The following property, which effectively reverses the simulation
process, extends the one holding for the original Bang Calculus
(without distance) \cite{GuerrieriManzonetto19}.

\begin{restatable}[\CBNSymb One-Step Reverse Simulation]{lemma}{RecCbnOneStepReverseSimulation}
    % \LemmaToFromProof{Cbn_Reverse_Simulation_One_Step}%
    \label{lem:Cbn_Reverse_Simulation_One_Step}%
    Let $t \in \setCbnTerms$, $u' \in \setBangTerms$, $\bangFCtxt \in
    	\bangFCtxtSet$ and $\rel' \in
    \{\bangSymbBeta, \bangSymbSubs, \bangSymbBang\}$. 
    
    \vspace{-1.65\baselineskip}
    \begin{equation*} 
        \cbnToBangAGK{t} \bangArr_F<R'> u'
            \quad \implies \quad
        \left\{\begin{array}{l@{\hspace{.2cm}}l}
            \exists\; u \in \setCbnTerms,
                & \cbnToBangAGK{u} = u'      \\
         \exists\; \rel \in \{\cbnSymbBeta, \cbnSymbSubs\},
                & \cbnToBangAGK{\rel} = \rel'          \\
         \exists\; \cbnFCtxt \in \cbnFCtxtSet,
                &\cbnToBangAGK{\cbnFCtxt} = \bangFCtxt
        \end{array}\right\}
        \text{ such that } t \cbnArr_F<R> u.
    \end{equation*}%
    \label{l:reverse-simulation-cbn}
\end{restatable}

\Cref{lem:Cbn_Reverse_Simulation_One_Step} states that any \BANGSymb
step from the image $\cbnToBangAGK{t}$ of a \CBNSymb term $t$ (which
is necessarily diligent, because $\cbnToBangAGK{t}$ is a \bangFNF<d!>)
actually simulates a \CBNSymb step from $t$.
In~\Cref{ex:cbn-simulation}, $\cbnToBangAGK{t_0}$
$\bangSymbBeta$-reduces in the context $\bangFCtxt = \abs{z}{\Hole}$
to $\abs{z}{((\app[\,]{y}{\app{\oc x}{\oc x}})\esub{x}{\oc
(\app{z}{\oc z})})}$, which is indeed equal to $\cbnToBangAGK{t_1}$,
and %we have 
$t_0 \cbnArr_F<dB> t_1$ in the context $\cbnFCtxt = \abs{z}{\Hole}$ as
well, with $\cbnToBangAGK{\cbnFCtxt} = \bangFCtxt$. Note that
\Cref{lem:Cbn_Reverse_Simulation_One_Step} is vacuously true for $\rel
= \bangSymbBang$, since there is no term $t$ such that $\derSymb$
occurs in $\cbnToBangAGK{t}$.
\Cref{l:one-step-cnb,l:reverse-simulation-cbn} have some
significant~consequences:

\begin{restatable}{corollary}{RecCbnPropFull}
    \label{cor:cbn prop-full}%
    Let $t, u \in \setCbnTerms$ and $s' \in \setBangTerms$. 
    \begin{enumerate}
    \item \textbf{(Stability)}: if $\cbnToBangAGK{t} \bangArrSet*_F
        s'$ then there is $s \in \setCbnTerms$ such that
        $\cbnToBangAGK{s} = s'$.

    \item \textbf{(Normal Forms)}: $t$ is a \cbnSetFNF if and only if
        $\cbnToBangAGK{t}$ is a \bangSetFNF.

    \item \textbf{(Simulations)}: $t \cbnArrSet*_F u$ if and only if
        $\cbnToBangAGK{t} \bangArrSet*_F \cbnToBangAGK{u}$. Moreover,
        the number of $\cbnSymbBeta/\cbnSymbSubs$-steps on the left
        matches the number $\bangSymbBeta/\bangSymbSubs$-steps on the
        right.
    \end{enumerate}
\end{restatable}

These results deserve some comments. Point 1 states that the image of
the \CBNSymb embedding is \emph{stable under reduction}. However,
%it should be noted that 
it is not stable under expansion. For
instance, %we have 
$\der{\oc x} \bangArr_S x = \cbnToBangAGK{x}$,
although $\der{\oc x}$ does not belong to the embedding's image, which
only contains terms without $\derSymb$. % constructor. 
Point 2 guarantees the \emph{preservation of normal forms} in both
directions. Finally, Point 3 concerns the \emph{preservation of
reduction sequences}. It is worth highlighting that this is an
equivalence, enabling to inject reduction sequences from \CBNSymb into
\BANGSymb and project them back from \BANGSymb into \CBNSymb. This is
a key property allowing in particular to infer confluence and factorization
for $\CBNSymb$ from that for $\BANGSymb$.

The reader may wonder whether similar preservation results hold for
\surface reduction. Since it is a subreduction of full reduction,
\Cref{cor:cbn prop-full}.1 already implies stability for \surface
reduction. However, it does not imply preservation of \surface normal
forms, and only yields back and forth simulation of \surface reduction
via \full reduction, which is not exactly what we want:
$\cbnToBangAGK{t} \bangArrSet*_F \cbnToBangAGK{u}$ if $t \cbnArrSet*_S
u$, and $t \cbnArrSet*_F u$ if $\cbnToBangAGK{t} \bangArrSet*_S
\cbnToBangAGK{u}$. So let us come back to analyze the situation for
the \emph{one-step} simulation and reverse simulation.  Since \surface
contexts are special cases of \full contexts,  then $t \cbnArr_S<R> u$
implies $\cbnToBangAGK{t}
\bangArr_{\cbnToBangAGK{\cbnSCtxt}\cbnCtxtPlug{\cbnToBangAGK{R}}}
\cbnToBangAGK{u}$ by \Cref{l:one-step-cnb}. To prove that this
simulating step is actually a \surface step, we need an additional
property: that \CBNSymb \surface contexts are translated into
\BANGSymb \surface contexts (\Cref{lem:Cbn Redex Position
Stability}.1). A more subtle analysis will be required for \surface reverse
simulation: positions of \BANGSymb \surface \emphit{redexes} are
always in the image of \CBNSymb \surface contexts:

\begin{restatable}{lemma}{RecCBNCtxtStability}
    \label{lem:Cbn Redex Position Stability}%
    \mbox{}
    \begin{enumerate}
    \item \textbf{(\CBNSymb $\rightarrow$ \BANGSymb)} %
        If $\cbnSCtxt \in \cbnSCtxtSet$, then
        $\cbnToBangAGK{\cbnSCtxt} \in \bangSCtxtSet$.
    \item \textbf{(\BANGSymb $\rightarrow$ \CBNSymb)} %
        If $\bangSCtxt \in \bangSCtxtSet$ and $\cbnFCtxt \in
        \cbnFCtxtSet$ such that $\cbnToBangAGK{\cbnFCtxt} =
        \bangSCtxt$, then $\cbnFCtxt \in \cbnSCtxtSet$.
    \end{enumerate}
\end{restatable}

% The hypothesis that $u$ is a \BANGSymb redex in \Cref{lem:Cbn Redex
% Position Stability}.2 is necessary. Indeed, a \BANGSymb \surface
% context contained in the encoding of a \CBNSymb term may not be the
% image of any \CBNSymb \surface context. For example, take $t =
% \app[\,]{x}{(\app{I}{I})}$ and $\bangSCtxt = \app[\,]{x}{\Hole}$,
% where $I = \abs{z}{z}$. Then
% $\cbnToBangAGK{t} = \app[\,]{x}{\oc(\app{I}{\oc I})} = \bangSCtxt<u>$
% for $u = \oc(\app{I}{\oc I})$, but there is no $\cbnSCtxt \in
% \cbnSCtxtSet$ such that $\cbnToBangAGK{\cbnSCtxt} = \bangSCtxt$: 
% indeed,  $u$ is not a surface redex. 

Thanks to \Cref{lem:Cbn Redex Position Stability}, one-step simulation
and reverse simulation
(\Cref{l:one-step-cnb,l:reverse-simulation-cbn}) can be iterated 
to obtain the following results about \emph{surface} reduction.

\begin{restatable}{corollary}{RecCbnPropSurface}
    \label{cor:cbn prop-surface}
    Let $t, u \in \setCbnTerms$  and $s' \in \setBangTerms$.  
    \begin{enumerate}
    \item \textbf{(Stability)}: if $\cbnToBangAGK{t} \bangArrSet*_S
        s'$ then there is $s \in \setCbnTerms$ such that
        $\cbnToBangAGK{s} = s'$.
        \label{lem:cbn_prop-surface_Stability}%

    \item \textbf{(Normal Forms)}: $t \text{ is a \cbnSetSNF}$ if and
        only if $\cbnToBangAGK{t}$ is a \bangSetSNF.
        \label{lem:cbn_prop-surface_normalForms}%

    \item \textbf{(Simulations)}: $t \cbnArrSet*_S u$ if and only if
        $\cbnToBangAGK{t} \bangArrSet*_S \cbnToBangAGK{u}$. Moreover,
        the number of $\cbnSymbBeta/\cbnSymbSubs$-steps on the left
        matches the number of $\bangSymbBeta/\bangSymbSubs$-steps on
        the right.
        \label{lem:cbn_prop-surface_Simulations}%
    \end{enumerate}
\end{restatable}

Our results for \CBNSymb notably extend the ones in
\cite{BucciarelliKesnerRiosViso20,BucciarelliKesnerRiosViso23}, where
it was only shown that \cbnSetSNF translates to \bangSetSNF, and that
\CBNSymb surface reduction is simulated by \BANGSymb surface
reduction: we went further by encompassing their converses.

\subsection{The Call-by-Value Calculus \CBVSymb and its embedding into \BANGSymb}
\label{subsec:Call-by-Value_Calculus_and_Embedding}

%The \emph{call-by-value $\lambda$-calculus} (called \CBVSymb\ here) is specified using ES and action at a distance \cite{AccattoliPaolini12}, as explained in \Cref{sec:intro}. 
%The sets $\setCbvTerms$ of \emphasis{terms} and
%$\setCbvValues$ of \emphasis{values} are inductively defined below (note that the syntax  contains neither $\derSymb$~nor~$\oc$):
%%
%\begin{equation*}
%    \textbf{(Terms)} \ \ 
%        t, u 	\;\Coloneqq\;  v
%            \vsep \app[\,]{t}{u}
%            \vsep t\esub{x}{u}
%\hspace{1cm}
%    \textbf{(Values)} \ \
%        v \;\Coloneqq\; x 
%            \vsep \abs{x}{t}
%\end{equation*}
%%
In \CBVSymb, \emphasis{surface contexts} $\cbvSCtxt \in \cbvSCtxtSet$
are defined below: the hole cannot be under an abstraction. To align the notations, in \CBVSymb full contexts
are denoted by $\cbvFCtxt \in \cbvFCtxtSet$ and list contexts by
$\cbvLCtxt \in \cbvLCtxtSet$.
\begin{align*}
    \emphasis{(\CBVSymb Surface Contexts)}& &
    \cbvSCtxt&\;\Coloneqq  \Hole
    \vsep \app[\,]{\cbvSCtxt}{t}
    \vsep \app[\,]{t}{\cbvSCtxt}
    \vsep \cbvSCtxt\esub{x}{t}
    \vsep t\esub{x}{\cbvSCtxt}
\end{align*}
As explained in~\cref{sec:basic-notions}, reductions in $\CBVSymb$ are
defined by taking the set of rewrite rules $\{\cbvSymbBeta,
\cbvSymbSubs\}$ defined below and the sets of contexts $\cbvSCtxtSet$
and $\cbvFCtxtSet$.
\begin{equation*}
	\app{\cbvLCtxt\cbvCtxtPlug{\abs{x}{t}}}{u}
	\;\cbvMapstoBeta\;
	\cbvLCtxt\cbvCtxtPlug{t\esub{x}{u}}
	\hspace{2cm}
	t\esub{x}{\cbvLCtxt\cbvCtxtPlug{v}}		
	\;\cbvMapstoSubs\;
	\cbvLCtxt\cbvCtxtPlug{t\isub{x}{v}}
\end{equation*}
Rule \cbvSymbBeta (\resp\ \cbvSymbSubs) is capture-free: no free
variable of $u$ (\resp $t$) is captured by context $\cbvLCtxt$. The
\emphasis{\CBVSymb surface reduction} is the relation $\cbvArrSet_S
\,\coloneqq\, \cbvArrSet_S<dB> \cup \cbvArrSet_S<sV>$, while the
\emphasis{\CBVSymb full reduction} is the relation $\cbvArrSet_F
\,\coloneqq\, \cbvArrSet_F<dB> \cup \cbvArrSet_F<sV>$.

The calculi \CBNSymb and \CBVSymb differ in that \CBNSymb can always
fire an ES (rule $\cbnSymbSubs$), while \CBVSymb only does when the ES
argument is a value, possibly wrapped by a finite list of ES (rule
$\cbvSymbSubs$). So \eg, for $\cbvSCtxt =
(\app{\app{y}{x}}{x})\esub{x}{\Hole}$, we have:
\begin{equation}\label{ex:cbv-reduction}
	\begin{split}
        u_0
    &= \app{(\abs{x}{\app{\app{y}{x}}{x}})}{(\app{(\abs{z}{z})}{y})}
        \cbvArr_{\Hole\,\bangCtxtPlug{\cbvSymbBeta}}
        u_1
    = (\app{\app{y}{x}}{x})\esub{x}{\app{(\abs{z}{z})}{y}}
\\
        	&\cbvArr_{S<dB>}
        u'_1
    = (\app{\app{y}{x}}{x})\esub{x}{z\esub{z}{y}} 
        \cbvArr_{S<sV>}
        u_2
    = (\app{\app{y}{x}}{x})\esub{x}{y} 
        \cbvArr_{\Hole\,\bangCtxtPlug{\cbvSymbSubs}}
        u_3
    = \app{\app{y}{y}}{y}
    \end{split}
\end{equation}

Notice how reduction $\cbvArrSet_S$ unblocks
redexes: given $\delta \coloneqq \lambda z. zz$, the term
$t \coloneqq (\lambda y.  \delta) (xx) \delta$, which is
a normal form 
in Plotkin's CBV~\cite{Plotkin75}, is now non-terminating
$t \cbvArrSet_S \delta\esub{y}{xx} \delta \allowbreak\cbvArrSet_S
(zz)\esub{z}{\delta} \esub{y}{xx} \cbvArrSet_S (\delta
\delta)\esub{y}{xx} \cbvArrSet*_S (\delta \delta)\esub{y}{xx}$, as one
would expect, since $t$ is semantically equivalent to the diverging
term
$\delta\delta$~\cite{PaoliniSimonaRonchiDellaRocca99,PaoliniRonchi04,CarraroGuerrieri14,AccattoliGuerrieri16,AccattoliGuerrieri22}.

The \CBVSymb surface reduction is nothing but the well-known
\emph{weak} reduction that does not evaluate under abstractions.

\paragraph{Embedding \CBVSymb into \BANGSymb.}

\newcommand{\girardCbv}[1]{#1^{{\tt v_1}}}
\newcommand{\bkrvCbv}[1]{#1^{{\tt v_2}}}

Values (\ie, variables and abstractions) are the duplicable elements
of \CBVSymb. Girard's Call-by-Value encoding (used in
\cite{EhrhardGuerrieri16,GuerrieriManzonetto19} and noted
$\girardCbv{(\cdot)}$ here) is built upon this insight, placing a bang
in front of each variable $\girardCbv{x} = \oc x$ and abstraction
$\girardCbv{(\abs{x}{t})} = \oc\abs{x}{\girardCbv{t}}$. The encoding
of an application is $\girardCbv{(\app{t}{u})} =
\app{\der{\girardCbv{t}}}{\girardCbv{u}}$, where the $\derSymb$ is
used to enable a $\bangSymbBang$-step if $t$ (the left-hand side of
the application) is a value, so as to restore its functional role.
However, as highlighted in
\cite{BucciarelliKesnerRiosViso20,BucciarelliKesnerRiosViso23}, such a
definition fails normal forms preservation: a \CBVSymb normal form is
not necessarily encoded by  a \BANGSymb normal form, for example given
the normal term $t_0 = \app{x}{y}$ we have   $\girardCbv{t_0} =
\app[\,]{\der{\oc x}}{\oc y}$ which is not normal. Consequently,
\cite{BucciarelliKesnerRiosViso20,BucciarelliKesnerRiosViso23}
proposed an alternative encoding  (noted $\bkrvCbv{(\cdot)}$ here,
whose details are omitted for lack of space), based on the same
principle, but with an additional \emphit{super-development}: all $\bangSymbBang$-redexes appearing during the encoding on the
left of an application are eliminated \emph{on the fly}, so that the embedding $\bkrvCbv{(\cdot)}$ preserves normal forms (\eg, $\bkrvCbv{t_0} = \app[\,]{x}{\oc y}$, which is normal in \BANGSymb). 
But, as shown in \Cref{CE: CBV Reverse Simulation}, $\bkrvCbv{(\cdot)}$ breaks reverse simulation with respect to surface reduction.

\begin{figure}[t]
    \begin{equation*}
        \begin{array}{ccc}
            \app{(\abs{x}{\app{(\abs{y}{y})}{z})}}{z}                     &\quad\not\cbvArr_S\quad    &\app{(\abs{x}{y\esub{y}{z}})}{z}
        \\[0.05cm]
            \rotatebox[origin=c]{-90}{$\rightsquigarrow$}\;\bkrvCbv{\cdot}  &                           &\rotatebox[origin=c]{-90}{$\rightsquigarrow$}\;\bkrvCbv{\cdot} 
        \\[0.0cm]
            \app{(\abs{x}{\app{(\abs{y}{\oc y})}{\oc z})}}{\oc z}     &\quad\bangArr_S\quad       &\app{(\abs{x}{(\oc y)\esub{y}{\oc z})}}{\oc z}
        \end{array}
    \end{equation*}
    \caption{Counterexample to \CBVSymb reverse simulation using
    the embedding $\bkrvCbv{\cdot}$}
    \label{CE: CBV Reverse Simulation}
\end{figure}

We introduce a \emph{new} \CBVSymb embedding that preserves normal
forms and fulfills simulation \emph{and} reverse simulation (this is
one of our main contributions).%

\begin{definition}
	\label{def:cbvAGK_Embedding}
	The \emphasis{\CBVSymb embedding} $\cbvToBangAGK{\cdot} \colon
	\setCbvTerms \to \setBangTerms$ is defined as follows:
%	\begin{equation*}
%		\begin{array}{rcl}
%			\cbvToBangAGK{x}
%                &\;\;\coloneqq\;\;& 
%            \oc x
%		\\
%			\cbvToBangAGK{(\abs{x}{t})}
%                &\;\;\coloneqq\;\;& 
%            \oc\abs{x}{\oc\cbvToBangAGK{t}}
%		\\
%            \cbvToBangAGK{(\app{t}{u})}
%                &\;\;\coloneqq\;\;&
%            \left\{\begin{array}{lcl}
%                \der{\app[\,]{\bangStrip{\cbvToBangAGK{t}}}{\cbvToBangAGK{u}}}
%                &\hspace{0.25cm}&\text{if } \bangBangPred{\cbvToBangAGK{t}}
%            \\
%                \der{\app[\,]{\der{\cbvToBangAGK{t}}}{\cbvToBangAGK{u}}}
%                &\hspace{0.25cm}&\text{otherwise}
%            \end{array}\right.
%        \\
%			\cbvToBangAGK{(t\esub{x}{u})}
%                &\;\;\coloneqq\;\;& 
%            \cbvToBangAGK{t}\esub{x}{\cbvToBangAGK{u}}.
%		\end{array}
%	\end{equation*}
%where the predicate $\bangBangPred{t}$ holds if $t =
%\bangLCtxt<\oc u>$, and $\bangStrip{t} \coloneqq \left\{\begin{array}{ll}
%	\bangLCtxt<u> &\text{if } t = \bangLCtxt<\oc u> \\
%	t
%	&\text{otherwise.}
%\end{array}\right.$
%
\begin{equation*}
	\begin{array}{rlp{.3cm}rlll}
		\cbvToBangAGK{x}
		&\coloneqq
		\oc x
		& &
		\multirow{2}{*}{$\cbvToBangAGK{(\app{t}{u})}$}
		&\multirow{2}{*}{$\coloneqq \bigg\{$}
		&
		\der{\app[]{\bangLCtxt<s>}{\cbvToBangAGK{u}}}
		&\text{ if } \cbvToBangAGK{t} = \bangLCtxt<\oc s>
		\\
		\cbvToBangAGK{(\abs{x}{t})}
		&\coloneqq 
		\oc\abs{x}{\oc\cbvToBangAGK{t}}
		& & & &
\der{\app[\,]{\der{\cbvToBangAGK{t}}}{\cbvToBangAGK{u}}}
&\text{ otherwise;}
		\\
		\cbvToBangAGK{(t\esub{x}{u})}
		&\coloneqq
		\cbvToBangAGK{t}\esub{x}{\cbvToBangAGK{u}}.
		& 		
\end{array}
\end{equation*}
\end{definition}

Note in particular that, thanks to super-development,
$\cbvToBangAGK{t}$ is always a \bangSetFNF<d!>.
For instance, $\cbvToBangAGK{(\abs{z}{z})} = \oc \abs{z} {\oc
\oc z}$ and $\cbvToBangAGK{(yxx)} = \Der{\der{\der{y \oc x}}\oc x}$,
whereas
$\cbvToBangAGK{(\app{(\abs{x}{\app{y}{\app{x}{x}}})}{(\app{I}{I})})} =
\Der{\app{\big(\abs{x}{\oc \der{\app{\der{\der{y \oc x}}}{\oc
x}}}\big)}{\, \der{\app{(\abs{z}{\oc \oc z})}{\oc \abs{z}{\oc \oc
z}}}}}.$

As in the \CBNSymb embedding, the modality $\oc$ plays a \emph{twofold} role
in our new \CBVSymb embedding.
First, $\cbvToBangAGK{\cdot}$ marks with $\oc$ subterms to be considered as
values, \ie\ potentially \emph{erasable} or \emph{duplicable}. This induces the use
of super-developments in the case of applications to avoid some
administrative steps that would otherwise affect preservation of
normal forms. Second, $\cbvToBangAGK{\cdot}$ marks the positions where
surface reduction must not occur: inside values; thus it introduces a \emph{second} (internal) $\oc$ in the encoding  of abstractions
to encapsulate its body and shield it from surface computation.
Additionally, to restore access to the abstraction's body when it is
applied, a second (external) $\derSymb$ is added to the encoding of
applications. These two principles highlights the dual role of $\oc$ in \BANGSymb: enabling duplication (and erasure) as well as isolating subterms
from surface computation processes.

The \CBVSymb embedding is extended to rule names, by defining
$\cbvToBangAGK{\cbvSymbBeta} \coloneqq \bangSymbBeta$ and
$\cbvToBangAGK{\cbvSymbSubs} \coloneqq  \bangSymbSubs$. 
Similarly to \CBNSymb, we have the fundamental simulation~result~below.

\begin{restatable}[\CBVSymb One-Step Simulation]{lemma}{RecCbvOneStepSimulation}
    % \LemmaToFromProof{Cbv_Simulation_One_Step}%
    \label{lem:Cbv_Simulation_One_Step}%
    Let $t, u \in \setBangTerms$, and $\rel \in \{\cbvSymbBeta,
    \cbvSymbSubs\}$. If $t \cbvArr_F<R> u$ then there is $ \bangFCtxt
    \in \bangFCtxtSet$ such that $\cbvToBangAGK{t}
    \bangArr_F<\cbvToBangAGK{\rel}>  \bangArrSet*_F<d!>
    \cbvToBangAGK{u}$, where $\bangFCtxt$ and all contexts used for
    the steps in $\bangArrSet*_F<d!>$ can be specified using
    $\cbvToBangAGK{\cbvFCtxt}, \rel$ and $t$.
\end{restatable}

Let us see how $\bangFCtxt$ and the contexts used in the steps
$\bangArrSet*_F<d!>$ are constructed: it highlights the difference
between \Cref{lem:Cbv_Simulation_One_Step} for \CBVSymb and
\Cref{l:one-step-cnb}~for~\CBNSymb.

\begin{itemize}
\item Additional administrative steps ($\bangArrSet*_F<d!>$) may be needed
    at the end. For example, for the \CBVSymb\ steps $u_0
    \cbvArrSet_F<dB> u_1$ and $u_2 \cbvArrSet_F<sV> u_3$ seen in
    \eqref{ex:cbv-reduction},~we~have:
    \begin{align}\label{ex:cbv-simulation}
        \cbvToBangAGK{u_0}
            &=
            \Der{\app{\big(\abs{x}{\oc \der{\app{\der{\der{y \oc x}}}{\, \oc x}}}\big)}{\, \der{\app{(\abs{z}{\oc \oc z})}{\oc y}}}}
            \nonumber
    \\
            &\bangArrSet_F<dB> 
            \Der{\, \oc \der{\app{\der{\der{y \oc x}}}{\, \oc x}}{\, \esub{x}{\der{\app{(\abs{z}{\oc \oc z})}{\oc y}}}}\,} = s'
    \\
            &\bangArrSet_F<d!>
            \der{\app{\der{\der{y \oc x}}}{\, \oc x}}{ \, \esub{x}{\der{\app{(\abs{z}{\oc \oc z})}{\oc y}}}}
                = \cbvToBangAGK{u_1}
             \nonumber
    \\[0.1em]
        \cbvToBangAGK{u_2} 
            &= \Der{ \app{\der{\der{y \oc x}}}{\, \oc x} }{\esub{x}{\oc y}}
%    \\
%            &
            \bangArrSet_F<s!> 
            \Der{\app[\,]{\der{\der{\app[\,]{y}{\oc y}}}}{\oc y}}
                = \cbvToBangAGK{u_3}
                \nonumber
    \end{align}

\item In \CBNSymb one-step simulation the rule name and context are
    independently translated. It is slightly more subtle in \CBVSymb:
    the rule name translates to the corresponding one in \BANGSymb
    without any ambiguity, yet the translation of the context
    $\cbvFCtxt$ depends not only on the initial context $\cbvFCtxt$
    but also on the rule name $\rel$ and the initial term $t$. Two
    distinct situations can emerge:
    \begin{itemize}
    \item $\cbvSymbBeta$-steps require to add a dereliction to the
        translated context: for example, the $\cbvSymbBeta$-redex
        position $\Hole$ in $t = \app{(\abs{x}{x})}{y}$ needs to be
        translated to the redex position $\der{\Hole}$ in
        $\cbvToBangAGK{t} = \der{\app{(\abs{x}{\oc x})}{\oc y}}$.

    \item $\cbvSymbSubs$-steps may need to remove a dereliction from
        the translated context: for instance, the $\cbvSymbSubs$-redex
        position $\app[\,]{\Hole}{y}$ in $t =
        \app[\,]{(\abs{z}{x})\esub{x}{y}}{y}$ is translated to the
        redex position $\der{\app[\,]{\Hole}{\oc y}}$ in
        $\cbvToBangAGK{t} = \der{\app[\,]{(\abs{z}{\oc x})\esub{x}{\oc
        y}}{(\oc y)}}$. 
    	The context translation anticipates the
        super-development used in $\cbvToBangAGK{t}$.
    \end{itemize}

    Note that both situations can be detected by case-analysis on
    $\rel$ and $t$, where the target context translation is a slight
    variation over the original one.
\end{itemize}

While the \CBVSymb embedding $\bkrvCbv{\cdot}$ used in
\cite{BucciarelliKesnerRiosViso20,BucciarelliKesnerRiosViso23} successfully enables the simulation of \CBVSymb
into \BANGSymb, it falls short when it comes to reverse simulation, as
shown in \Cref{CE: CBV Reverse Simulation}. 
Therefore, $\bkrvCbv{\cdot}$ cannot be used to transfer dynamic properties from
\BANGSymb back to \CBVSymb, thus failing in particular to derive
\CBVSymb factorization from \BANGSymb (\Cref{sec:Factorization}).
Our new embedding instead satisfies reverse simulation.

\begin{restatable}[\CBVSymb One-Step Reverse Simulation]{lemma}{RecCbvOneStepReverseSimulation}
    % \LemmaToFromProof{Cbv_Reverse_Simulation_One_Step}%
    \label{lem:Cbv_Reverse_Simulation_One_Step}%
    Let $t \in \setCbvTerms$, $u' \in \setBangTerms$, $\bangFCtxt \in
    \bangFCtxtSet$ and $\rel' \in
    \{\bangSymbBeta, \bangSymbSubs, \bangSymbBang\}$.
    If $u'$ is a \bangSetFNF<d!>, then%
    \begin{equation*}
        \cbvToBangAGK{t} \bangArr_F<R'>\bangArrSet*_F<d!> u'
        % \;\;\text{and}\;\; \bangPredFNF<d!>{u'}
            \ \implies \
        \left\{\begin{array}{lr}
            \exists\; u \in \setCbvTerms,
                &\cbvToBangAGK{u} = u'
        \\
            \exists\; \rel \in \{\cbvSymbBeta, \cbvSymbSubs\},
                &\cbvToBangAGK{\rel} = \rel'
        \\
            \exists\; \cbvFCtxt \in \cbvFCtxtSet,
                % &\bangFCtxt = f_{\rel, t}(\cbvToBangAGK{\cbvFCtxt})
        \end{array}\right\}
         \text{
        such that } t \cbvArr_F<R> u.
    \end{equation*}%
\end{restatable}

\Cref{lem:Cbv_Reverse_Simulation_One_Step} states that any \BANGSymb
diligent step from the image $\cbvToBangAGK{t}$ of a \CBVSymb term $t$
actually simulates a \CBVSymb step from $t$. As expected, the same
subtleties encountered in the \CBVSymb one-step simulation (\Cref{lem:Cbv_Simulation_One_Step}) apply in
this last result, in particular regarding the construction of
$\cbvFCtxt$. In the \CBNSymb case, the absence of administrative steps
renders all sequences from images of \CBNSymb terms diligent, making
stability, normal form preservation and simulations direct
consequences of one-step simulation (\Cref{l:one-step-cnb}) and
reverse simulation (\Cref{l:reverse-simulation-cbn}). This is not the
case for \CBVSymb, due to the presence of administrative steps in the
simulation process. Indeed, when simulating \CBVSymb reduction within
\BANGSymb (\Cref{lem:Cbv_Simulation_One_Step}), administrative steps
are performed as soon as they become available, thus constructing a
diligent sequence. Conversely, projecting a reduction step from
\BANGSymb to \CBVSymb (\Cref{lem:Cbv_Reverse_Simulation_One_Step})
requires a diligent step. However, in the case of sequences,
in contrast to one-steps,  there is
no requirement for administrative steps to be correctly synchronized,
and this may lead to deviations from the embedding's image,
significantly complicating reverse simulation. Fortunately, the
diligence presented in \BANGSymb (\Cref{lem:Implicitation})
resynchronizes administrative steps yielding sequences
that~are~easy~to~project.%

\begin{restatable}{corollary}{RecCbvPropFull}
    \label{cor:cbv prop-full}%
    Let $t, u \in \setCbvTerms$ and $s' \in \setBangTerms$. 
    \begin{enumerate}
    \item \textbf{(Stability)}: 
        if $\cbvToBangAGK{t} \bangArrSet*_F s'$ and $s'$ is a \bangSetFNF<d!>, 
        	then $s' = \cbvToBangAGK{s}$ for some $s \in \setCbvTerms$.
            \label{lem:cbv_prop-full_Stability}%

    \item \textbf{(Normal Forms)}: $t \text{ is a \cbvSetFNF}$ if and
        only if $\cbvToBangAGK{t}$ is a \bangSetFNF.
        \label{lem:cbv_prop-full_NormalForms}%

    \item \textbf{(Simulations)}: $t \cbvArrSet*_F u$ if and only if
        $\cbvToBangAGK{t} \bangArrSet*_F \cbvToBangAGK{u}$. Moreover,
        the number of $\cbvSymbBeta/\cbvSymbSubs$-steps on the left
        matches the number $\bangSymbBeta/\bangSymbSubs$-steps on the
        right.
        \label{lem:cbv_prop-full_Simulations}%
    \end{enumerate}
\end{restatable}

% The reader may wonder whether similar preservation results hold for
% \surface reductions. Given that it is a subrelation of full
% reduction, the previous result already provides stability for
% \surface reduction. However, it does not imply preservation of
% \surface normal forms, and only yields back and forth simulation of
% \surface reduction via \full reduction, which is not exactly what we
% want: $\cbnToBangAGK{t} \bangArrSet*_F \cbnToBangAGK{u}$ if $t
% \cbnArrSet*_S u$ and $t \cbnArrSet*_F u$ if $\cbnToBangAGK{t}
% \bangArrSet*_S \cbnToBangAGK{u}$. So let us come back to analyse the
% situation for the \emph{one-step} simulation and reverse simulation.
% Since \surface contexts are special cases of \full contexts then $t
% \cbnArr_S<R> u$ implies $\cbnToBangAGK{t}
% \bangArr_{\cbnToBangAGK{\cbnSCtxt}\cbnCtxtPlug{\cbnToBangAGK{R}}}
% \cbnToBangAGK{u}$ by \Cref{l:one-step-cnb}. To prove that this
% simulating step is actually a \surface step, we need an additional
% property : that \CBNSymb \surface contexts are translated into
% \BANGSymb \surface contexts (\Cref{lem:Cbn Redex Position
% Stability}.1). A more subtle analysis is required for \surface
% reverse simulation: positions of \BANGSymb \surface \emphit{redexes}
% are always in the image of \CBNSymb \surface contexts:

As in \CBNSymb, we may wonder whether similar preservation
results hold for surface reductions. Such results cannot be
entirely derived out from \Cref{cor:cbv prop-full} alone. Still, as
with \CBNSymb, the \CBVSymb one-step simulation and reverse simulation
properties
(\Cref{lem:Cbv_Simulation_One_Step,lem:Cbv_Reverse_Simulation_One_Step})
already encompass the \surface case. However,
even though \surface redexes positions are  mutually
mapped by the embedding, it does not yet imply surface stability,
preservation of normal forms, and simulations. As previously
explained, diligence is required to deal with administrative steps.
Fortunately, the \surface
fragment admits a diligence process, as illustrated in
\Cref{lem:Implicitation}, which can then be
leveraged to obtain the following results.

\begin{restatable}{corollary}{RecCbvPropSurface}
    \label{cor:cbv prop-surface}
    Let $t, u \in \setCbvTerms$  and $s' \in \setBangTerms$.  
    \begin{enumerate}
    \item \textbf{(Stability)}:
        if $\cbvToBangAGK{t} \bangArrSet*_S s'$ and $s'$ is a 
        	\bangSetSNF<d!>, then $s' = \cbvToBangAGK{s}$ for some $s \in \setCbvTerms$.

    \item \textbf{(Normal Forms)}: $t \text{ is a \cbvSetSNF}$ if and
        only if $\cbvToBangAGK{t}$ is a \bangSetSNF.

    \item \textbf{(Simulations)}: $t \cbvArrSet*_S u$ if and only if
    	$\cbvToBangAGK{t} \bangArrSet*_S \cbvToBangAGK{u}$. Moreover,
    	the number of $\cbvSymbBeta/\cbvSymbSubs$-steps on the left
    	matches the number of $\bangSymbBeta/\bangSymbSubs$-steps on
    	the right.
    \end{enumerate}
\end{restatable}

Stability statements in \CBVSymb %
(\Cref{cor:cbv prop-full}.1 and \ref{cor:cbv prop-surface}.1) require
the reached term $s'$ to be normal for $\bangSymbBang$, otherwise
stability does not hold (\eg, $s'$ in \eqref{ex:cbv-simulation} before
is not in the image of $\cbvToBangAGK{\cdot}$), This is not required
in the \CBNSymb stability statements %
(\Cref{cor:cbn prop-full}.1 and \ref{cor:cbn prop-surface}.1) since
every term to which $\cbnToBangAGK{t}$ reduces is $\derSymb$-free and
so \mbox{normal for $\bangSymbBang$}.

Proving simulation and reverse simulation requires a considerable
effort. However, this initial investment lays the groundwork for
numerous benefits without extra costs. For example,
in~\Cref{sec:Bang_CBN_CBV_Confluence,subsec:Bang Factorization}, we
demonstrate that typically challenging tasks like proving confluence
and factorization can be easily achieved through simulations. This
approach not only unifies the proofs but also minimizes the workload
for future proofs.

\section{Confluence}
\label{sec:Bang_CBN_CBV_Confluence}%

Confluence is a crucial property in the $\lambda$-calculus, ensuring
that expressions consistently produce a single result, regardless of
the chosen reduction path. In this section, we examine confluence
across different reduction relations (surface and full) and within
various calculus frameworks: \CBNSymb, \CBVSymb, and \BANGSymb. We
specifically leverage simulation and reverse simulation properties to
seamlessly project these results from \BANGSymb to \CBNSymb and
\CBVSymb, providing a comprehensive solution across three contexts.

Surface confluence is usually proved by showing that surface reduction
is diamond,  as for example
in~\cite{BucciarelliKesnerRiosViso20,BucciarelliKesnerRiosViso23}.
Full confluence is more complex,  as the full reduction relation is not
diamond, as one can easily notice with the term
$\app{(\abs{x}{y})}{\oc\Omega}$. Alternative
techniques~\cite{Terese03,Oostrom08} can establish full reduction's
confluence, albeit often requiring numerous commutation diagrams and
potentially non-trivial bounding measures.

\begin{restatable}[\BANGSymb Confluence]{theorem}{RecBangConfluence}
    \label{lem:Bang_Confluence}%
    \begin{enumerate}
    \item \textbf{(Surface)} The reduction relation $\bangArrSet_S$ is
        confluent.%
        \label{lem:Bang_Confluence_Surface}%

        \quad Moreover, any two different surface reduction paths to a
        $\bangSCtxtSet$-normal form have the same length and number of
        $\bangSymbBeta$, $\bangSymbSubs$ and $\bangSymbBang$-steps.
    \item \hspace{9pt}\textbf{(Full)}\hspace{9pt} The reduction
        relation $\bangArrSet_F$ is confluent.
        \label{lem:Bang_Confluence_Full}%
    \end{enumerate}
\end{restatable}
\begin{proof}
\textbf{(Surface)}
See~\cite{BucciarelliKesnerRiosViso20,BucciarelliKesnerRiosViso23}
\textbf{(Full)} See~\cite{KAG24BangMeaningfulness}.
\end{proof}

% Programming langages are usually weak, meaning for example that they
% do not compute in the body of functions in the call-by-value setting.
% Thus, confluence of the surface reduction ensures that there is at
% most one result for any actual computations. Full confluence is mostly
% required when studying equational theories of the associated calculus,
% as for example trying to identify meaningless term~\victorTodo.

These proofs are typically highly technical, requiring a significant
amount of time to write and verify, and are prone to errors.
Therefore, it is extremely beneficial to have a method to streamline
them, especially when mecanising proofs. With the robust preservation
of \CBNSymb reductions in the \BANGSymb, we can actually project
\CBNSymb confluences directly from that of \BANGSymb.

\begin{corollary}[\CBNSymb Confluence]
    \label{lem:Cbn_Confluence}%
    \begin{enumerate}
    \item \textbf{(Surface)} The reduction relation $\cbnArrSet_S$ is
        confluent.%
        \label{lem:Cbn_Confluence_Surface}%

        \quad Moreover, any two different surface reduction paths to a
        $\cbnSCtxtSet$-normal form have the same length and same
        number of $\cbnSymbBeta$ and $\cbnSymbSubs$-steps.
    \item \hspace{9pt}\textbf{(Full)}\hspace{9pt} The reduction
        relation $\cbnArrSet_F$ is confluent.
        \label{lem:Cbn_Confluence_Full}%
    \end{enumerate}
\end{corollary}
\begin{proof}
\textbf{(Surface)} See Fig.\ref{fig:Cbn_Confluence_Surface}.
\textbf{(Full)} Following the same reasoning as
Fig.\ref{fig:Cbv_Confluence_Full}.
\end{proof}
% \vspace{-0.8cm}
\begin{figure}[t]
    \begin{minipage}[c]{0.49\textwidth}%
        \hspace{0.5cm}
        \begin{tikzpicture}[scale=0.85]%
            \node (t) at (0, 0)                                             {$t$};%
            \node (u1) at (-2, -1)                                          {$u_1$};%
            \node (u2) at (2, -1)                                           {$u_2$};%
            \draw[->, shorten <=4pt] (t) to (u1);
            \node at ([yshift=6pt]$(t)!0.5!(u1)$)                           {$\scriptstyle\cbnSCtxtSet$};%
            \node at ([yshift=5pt]$(t)!0.85!(u1)$)                          {$\scriptstyle*$};%
            \draw[->, shorten <=4pt] (t) to (u2);
            \node at ([yshift=6pt]$(t)!0.5!(u2)$)                           {$\scriptstyle\cbnSCtxtSet$};%
            \node at ([yshift=5pt]$(t)!0.85!(u2)$)                          {$\scriptstyle*$};%
            \draw [->, squigArrowStyle]  (-1, -0.65) -- (-1, -1.7);
            \node at (-1.8, -1.5)                                            {\scriptsize({\tt Cor.}\ref{cor:cbn prop-surface}.\ref{lem:cbn_prop-surface_Simulations})};%
            \draw [->, squigArrowStyle]  (1, -0.65) -- (1, -1.7);
            \node at (1.8, -1.5)                                             {\scriptsize({\tt Cor.}\ref{cor:cbn prop-surface}.\ref{lem:cbn_prop-surface_Simulations})};%
            \node (b_t) at (0, -1.25)                                       {$\cbnToBangAGK{t}$};%
            \node (b_u1) at (-2,-2.5)                                       {$\cbnToBangAGK{u_1}$};%
            \node (b_u2) at (2,-2.5)                                        {$\cbnToBangAGK{u_2}$};%
            \node (b_s') at (0,-3.75)                                       {$s'$};%
            \node[rotate=90] (eq_s) at (0, -4.375)                          {$=$};%
            \node (b_s) at (0,-5)                                           {$\cbnToBangAGK{s}$};%
            \draw[->, shorten <=4pt] (b_t) to (b_u1);
            \node at ([yshift=-6pt]$(b_t)!.5!(b_u1)$)                       {$\scriptstyle\bangSCtxtSet$};%
            \node at ([yshift=5pt]$(b_t)!.85!(b_u1)$)                       {$\scriptstyle*$};%
            \draw[->, shorten <=4pt] (b_t) to (b_u2);
            \node at ([yshift=-6pt]$(b_t)!.5!(b_u2)$)                       {$\scriptstyle\bangSCtxtSet$};%
            \node at ([yshift=5pt]$(b_t)!.85!(b_u2)$)                       {$\scriptstyle*$};%
            \node at (0, -2.5)                                              {\scriptsize({\tt Thm.}\ref{lem:Bang_Confluence}.\ref{lem:Bang_Confluence_Surface})};%
            \draw[->, shorten <=4pt] (b_u1) to (b_s');
            \node at ([yshift=6pt]$(b_u1)!.5!(b_s')$)                       {$\scriptstyle\bangSCtxtSet$};%
            \node at ([yshift=5pt]$(b_u1)!.85!(b_s')$)                      {$\scriptstyle*$};%
            \draw[->, shorten <=4pt] (b_u2) to (b_s');
            \node at ([yshift=6pt]$(b_u2)!.5!(b_s')$)                       {$\scriptstyle\bangSCtxtSet$};%
            \node at ([yshift=5pt]$(b_u2)!.85!(b_s')$)                      {$\scriptstyle*$};%
            \draw[->, shorten <=4pt, bend right=25] (b_u1) to (b_s);
            \node at ([yshift=-12pt]$(b_u1)!.5!(b_s)$)                      {$\scriptstyle\bangSCtxtSet$};%
            \node at ([yshift=-1pt]$(b_u1)!.85!(b_s)$)                      {$\scriptstyle*$};%
            \draw[->, shorten <=4pt, bend left=25] (b_u2) to (b_s);
            \node at ([yshift=-12pt]$(b_u2)!.5!(b_s)$)                      {$\scriptstyle\bangSCtxtSet$};%
            \node at ([yshift=-1pt]$(b_u2)!.85!(b_s)$)                      {$\scriptstyle*$};%
            \draw [->, squigArrowStyle]  (-1, -4.6) -- (-1, -5.6);
            \node at (-1.8, -4.75)                                          {\scriptsize({\tt Cor.}\ref{cor:cbn prop-surface}.\ref{lem:cbn_prop-surface_Simulations})};%
            \draw [->, squigArrowStyle]  (1, -4.6) -- (1, -5.6);
            \node at (1.8, -4.75)                                           {\scriptsize({\tt Cor.}\ref{cor:cbn prop-surface}.\ref{lem:cbn_prop-surface_Simulations})};%
            \node at (0, -5.35)                                             {\scriptsize({\tt Cor.}\ref{cor:cbn prop-surface}.\ref{lem:cbn_prop-surface_Stability})};%
            \node (u1) at (-2, -5.25)                                       {$u_1$};%
            \node (u2) at (2, -5.25)                                        {$u_2$};%
            \node  (s) at (0, -6.25)                                        {$s$};%
            \draw[->, shorten <=4pt] (u1) to (s);
            \node at ([yshift=-6pt]$(u1)!0.5!(s)$)                          {$\scriptstyle\cbnSCtxtSet$};%
            \node at ([yshift=5pt]$(u1)!0.85!(s)$)                          {$\scriptstyle*$};%
            \draw[->, shorten <=4pt] (u2) to (s);
            \node at ([yshift=-6pt]$(u2)!0.5!(s)$)                          {$\scriptstyle\cbnSCtxtSet$};%
            \node at ([yshift=5pt]$(u2)!0.85!(s)$)                          {$\scriptstyle*$};%
            \node at (0, -6.75) {\phantom{where $\cbvToBangAGK{s}$ is a $\bangFCtxt<\bangSymbBang>$-NF}};%
        \end{tikzpicture}%
        % \vspace{-1.2cm}%
        \caption{Schematic proof of \Cref{lem:Cbn_Confluence}.\ref{lem:Cbn_Confluence_Surface}}%
        \label{fig:Cbn_Confluence_Surface}%
    \end{minipage}
    \hspace{.5cm}
    \begin{minipage}[c]{0.49\textwidth}%
        \hspace{0.5cm}
        \begin{tikzpicture}[scale=0.85]%
            \node (t) at (0, 0)                                             {$t$};%
            \node (u1) at (-2, -1)                                          {$u_1$};%
            \node (u2) at (2, -1)                                           {$u_2$};%
            \draw[->, shorten <=4pt] (t) to (u1);
            \node at ([yshift=6pt]$(t)!0.5!(u1)$)                           {$\scriptstyle\cbvFCtxtSet$};%
            \node at ([yshift=5pt]$(t)!0.85!(u1)$)                          {$\scriptstyle*$};%
            \draw[->, shorten <=4pt] (t) to (u2);
            \node at ([yshift=6pt]$(t)!0.5!(u2)$)                           {$\scriptstyle\cbvFCtxtSet$};%
            \node at ([yshift=5pt]$(t)!0.85!(u2)$)                          {$\scriptstyle*$};%
            \draw [->, squigArrowStyle]  (-1, -0.65) -- (-1, -1.7);
            \node at (-1.8, -1.5)                                            {\scriptsize({\tt Cor.}\ref{cor:cbv prop-full}.\ref{lem:cbv_prop-full_Simulations})};%
            \draw [->, squigArrowStyle]  (1, -0.65) -- (1, -1.7);
            \node at (1.8, -1.5)                                             {\scriptsize({\tt Cor.}\ref{cor:cbv prop-full}.\ref{lem:cbv_prop-full_Simulations})};%
            \node (b_t) at (0, -1.25)                                       {$\cbvToBangAGK{t}$};%
            \node (b_u1) at (-2,-2.5)                                       {$\cbvToBangAGK{u_1}$};%
            \node (b_u2) at (2,-2.5)                                        {$\cbvToBangAGK{u_2}$};%
            \node (b_s') at (0,-3.75)                                       {$s'$};%
            \draw[->, shorten >=2pt] (b_s') to (b_s);
            \node at ([xshift=10pt]$(b_s')!.4!(b_s)$)                      {$\scriptstyle\bangFCtxtSet<\bangSymbBang>$};%
            \node at ([xshift=-4pt]$(b_s')!.6!(b_s)$)                      {$\scriptstyle*$};%
            \node (b_s) at (0,-5)                                           {$\cbvToBangAGK{s}$};%
            \draw[->, shorten <=4pt] (b_t) to (b_u1);
            \node at ([yshift=-6pt]$(b_t)!.5!(b_u1)$)                       {$\scriptstyle\bangFCtxtSet$};%
            \node at ([yshift=5pt]$(b_t)!.85!(b_u1)$)                       {$\scriptstyle*$};%
            \draw[->, shorten <=4pt] (b_t) to (b_u2);
            \node at ([yshift=-6pt]$(b_t)!.5!(b_u2)$)                       {$\scriptstyle\bangFCtxtSet$};%
            \node at ([yshift=5pt]$(b_t)!.85!(b_u2)$)                       {$\scriptstyle*$};%
            \node at (0, -2.5)                                              {\scriptsize({\tt Thm.}\ref{lem:Bang_Confluence}.\ref{lem:Bang_Confluence_Full})};%
            \draw[->, shorten <=4pt] (b_u1) to (b_s');
            \node at ([yshift=6pt]$(b_u1)!.5!(b_s')$)                       {$\scriptstyle\bangFCtxtSet$};%
            \node at ([yshift=5pt]$(b_u1)!.85!(b_s')$)                      {$\scriptstyle*$};%
            \draw[->, shorten <=4pt] (b_u2) to (b_s');
            \node at ([yshift=6pt]$(b_u2)!.5!(b_s')$)                       {$\scriptstyle\bangFCtxtSet$};%
            \node at ([yshift=5pt]$(b_u2)!.85!(b_s')$)                      {$\scriptstyle*$};%
            \draw[->, shorten <=4pt, bend right=25] (b_u1) to (b_s);
            \node at ([yshift=-12pt]$(b_u1)!.5!(b_s)$)                      {$\scriptstyle\bangFCtxtSet$};%
            \node at ([yshift=-1pt]$(b_u1)!.85!(b_s)$)                      {$\scriptstyle*$};%
            \draw[->, shorten <=4pt, bend left=25] (b_u2) to (b_s);
            \node at ([yshift=-12pt]$(b_u2)!.5!(b_s)$)                      {$\scriptstyle\bangFCtxtSet$};%
            \node at ([yshift=-1pt]$(b_u2)!.85!(b_s)$)                      {$\scriptstyle*$};%
            \draw [->, squigArrowStyle]  (-1, -4.6) -- (-1, -5.6);
            \node at (-1.8, -4.75)                                          {\scriptsize({\tt Cor.}\ref{cor:cbv prop-full}.\ref{lem:cbv_prop-full_Simulations})};%
            \draw [->, squigArrowStyle]  (1, -4.6) -- (1, -5.6);
            \node at (1.8, -4.75)                                           {\scriptsize({\tt Cor.}\ref{cor:cbv prop-full}.\ref{lem:cbv_prop-full_Simulations})};%
            \node at (0, -5.35)                                             {\scriptsize({\tt Cor.}\ref{cor:cbv prop-full}.\ref{lem:cbv_prop-full_Stability})};%
            \node (u1) at (-2, -5.25)                                       {$u_1$};%
            \node (u2) at (2, -5.25)                                        {$u_2$};%
            \node  (s) at (0, -6.25)                                        {$s$};%
            \draw[->, shorten <=4pt] (u1) to (s);
            \node at ([yshift=-6pt]$(u1)!0.5!(s)$)                          {$\scriptstyle\cbvFCtxtSet$};%
            \node at ([yshift=5pt]$(u1)!0.85!(s)$)                          {$\scriptstyle*$};%
            \draw[->, shorten <=4pt] (u2) to (s);
            \node at ([yshift=-6pt]$(u2)!0.5!(s)$)                          {$\scriptstyle\cbvFCtxtSet$};%
            \node at ([yshift=5pt]$(u2)!0.85!(s)$)                          {$\scriptstyle*$};%
        %
            % \node at (0, -8) {};%
            \node at (0, -6.75) {where $\cbvToBangAGK{s}$ is a \bangSetFNF<d!>};
        \end{tikzpicture}%
        % \vspace{-1.2cm}
        \caption{Schematic proof of \Cref{lem:Cbv_Confluence}.\ref{lem:Cbv_Confluence_Full}}%
        \label{fig:Cbv_Confluence_Full}%
    \end{minipage}
\end{figure}

The same technic can be used for \CBVSymb, with the additional help of
diligence. Thus, we get the following results for free.

\begin{corollary}[\CBVSymb Confluence]
    \label{lem:Cbv_Confluence}
    \begin{enumerate}
    \item \textbf{(Surface)} The reduction relation $\cbvArrSet_S$ is
        confluent.%
        \label{lem:Cbv_Confluence_Surface}%

        \quad Moreover, any two different surface reduction paths to a
        $\cbvSCtxtSet$-normal form have the same length and same
        number of $\cbvSymbBeta$ and $\cbvSymbSubs$-steps.
    \item \hspace{9pt}\textbf{(Full)}\hspace{9pt} The reduction
        relation $\cbvArrSet_F$ is confluent.
        \label{lem:Cbv_Confluence_Full}%
    \end{enumerate}
\end{corollary}
\begin{proof}
\textbf{(Surface)} Following the same reasoning as
Fig.\ref{fig:Cbn_Confluence_Surface}. \textbf{(Full)} See
Fig.\ref{fig:Cbv_Confluence_Full}.
\end{proof}

% This section offered a clear demonstration of how simulation and
% reverse simulation allow us to project fundamental properties like
% confluence from \BANGSymb to \CBNSymb and \CBVSymb. In the following
% section, we show that this same technique can be extended to address
% more advanced properties, such as factorization.

\newcommand{\nameSectionIV}
    {Factorization}

\section{Factorization}
\label{sec:Factorization}

In $\lambda$-calculi, reduction is a relation, so different reduction
steps are possible starting from the same term. Some steps (\eg head
steps) are more significant than others, and they may occur in the
middle of a {reduction} sequence.
Factorization is the process of disentangling the significant steps
from the ``superfluous'' ones, bringing the former forward and leaving
the start~and~end~terms~unchanged.

This section is devoted to the factorization property for \BANGSymb,
\CBNSymb and \CBVSymb. We start by revisiting an abstract
factorization theorem~\cite{Accattoli12}. We first apply this abstract
method to \BANGSymb, thus obtaining a new result of factorization not
previously appearing in the literature. Then, we use the properties of
simulation and reverse simulation proved in
\Cref{sec:CBN_CBV_Embeddings} to project the factorization result for
\BANGSymb into \CBNSymb and \CBVSymb. Although these two results can
be directly derived from the abstract factorization
theorem~\cite{Accattoli12}, our approach circumvents the numerous
commutation properties required by the abstract approach. Also, it
provides a tangible illustration of how the simulation and reverse
simulation properties discussed in
\Cref{subsec:Call-by-Name_Calculus_and_Embedding,subsec:Call-by-Value_Calculus_and_Embedding}
can be applied in~concrete~cases.

\paragraph{Abstract Factorization.}
\label{subsec:Abstract Factorization}

We recall an abstract factorization method from~\cite{Accattoli12}
that relies on local rewrite conditions given by the notion of
\emph{square factorization system} (SFS). While its original
presentation concerns only two subreductions, we (straightforwardly)
extend the notion of SFS to a \emph{family} of subreductions, as in
\BANGSymb the reduction consists of more than two subreductions.

\begin{definition}
    Let $\mathcal{R} = (R,\, \rightarrow_\rel)$ be an abstract
    rewriting system The family $\left(\arrRExt[\rel_k],
    \arrRInt[\rel_k]\right)_{k \in K}$ of paired reduction relations
    is a \emphasis{square factorization system (SFS)} for
    $\mathcal{R}$ if it covers the reduction relation (\ie
    $\rightarrow_\rel\;= \bigcup_{k \in K} \arrR[\rel_k]$ where
    $\arrR[\rel_k] \,\coloneqq\, \arrRExt[\rel_k] \cup \arrRInt[\rel_k]$) and
    satisfies the following conditions:
    \begin{enumerate}
    \item \textbf{Termination:} %
        $\forall\, k \in K$, $\arrRExt[\rel_k]$ is terminating.

    \item \textbf{Row-swaps:} %
        $\forall\, k \in K$,\; $\arrRInt[\rel_k]\arrRExt[\rel_k]
            \;\subseteq\; \arrpRExt[\rel_{k}]\arrsRInt[\rel_{k}]$.

    \item \textbf{Diagonal-swaps:} %
            $\forall\, k_1, k_2 \in K$, $k_1 \neq k_2$,\;
            $\arrRInt[\rel_{k_1}]\arrRExt[\rel_{k_2}] \;\subseteq\;
            \arrRExt[\rel_{k_2}]\arrsR[\rel_{k_1}]$\!.
        \end{enumerate}%
\end{definition}

The symbol $\circ$ tags \emph{significant} (also called
\emph{external}) steps, while $\bullet$ is used for \emph{\insignificant}
(also called \emph{internal}) ones. 
The commutations required in an
SFS are sufficient to achieve factorization, which consists in
rearranging a reduction sequence to prioritize significant steps
$\arrRExt$ over {\insignificant} steps $\arrRInt$.

\begin{restatable}[\cite{Accattoli12}]{proposition}{RecSFSFactorizable}
    % \LemmaToFromProof{SFS_Factorizable}%
    \label{lem:SFS_Factorizable}%
    Let $\mathcal{R} = (R, \rightarrow_\rel)$ be an abstract rewriting
    system and $(\arrRExt[\rel_k], \arrRInt[\rel_k])_{k \in K}$ be an
    $SFS$ for $\mathcal{R}$. Then the reduction relation factorizes,
    that is, $\rightarrow^*_\rel \ \subseteq \ \rightarrow^*_\circ
    \rightarrow^*_\bullet$ where $\rightarrow_\circ \coloneqq
    \bigcup_{k \in K} \arrRExt[\rel_k]$ and $\rightarrow_\bullet
    \coloneqq \bigcup_{k \in K} \arrRInt[\rel_k]$. 
\end{restatable}
%\begin{proof}
%    Straightforward extension of \cite{Accattoli12}.
%\end{proof}

\paragraph{Factorization in \BANGSymb.}
\label{subsec:Bang Factorization}

In \BANGSymb we claim that \surface reduction is the significant part
of \full reduction, and our goal is to factor it out. To exploit the
abstract method, we first formally identify the \insignificant\
subreduction of \full reduction, called here \emph{internal}, as
reduction under the scope of a $\oc$. 
\emphasis{Internal contexts}
$\bangICtxt \in \bangICtxtSet$ are full contexts $\bangFCtxtSet$ for
which the hole is placed under a bang. Formally,%
\begin{align*}
	\emphasis{(\BANGSymb Internal Contexts)}
	& &
	\bangICtxt &\;\Coloneqq\; \oc\bangFCtxt
	\vsep \bangSCtxt^*\bangCtxtPlug{\bangICtxt}
	&&
	\text{with } \bangSCtxt^* \in \bangSCtxtSet \setminus \{\Hole\}
\end{align*}

\noindent
Clearly, $\bangICtxtSet = \bangFCtxtSet \setminus \bangSCtxtSet$.
As usual, $\bangArrSet_I<R>$ is the closure
of the rewrite rules $\rel \in \{\bangSymbBeta, \bangSymbSubs,
\bangSymbBang\}$ over all contexts $\bangICtxt \in \bangICtxtSet$.
The  \emphasis{\BANGSymb internal reduction} is the relation
$\bangArrSet_I \, \coloneqq \, \bangArrSet_I<dB> \cup \bangArrSet_I<s!> \cup \bangArrSet_I<d!>$. For example, $\app[\,]{(\abs{x}{\oc \Hole})}{y} $ is an internal context while $\Hole$ is not. Thus,
$\app[\,]{(\abs{x}{\oc (z\esub{z}{x})})}{y} \bangArrSet_I<s!>
\app[\,]{(\abs{x}{\oc x})}{y} \not\bangArrSet_I (\oc x)\esub{x}{y}$.
%Having this concept at hand, it is now possible to
We can now show that surface
and internal reductions enjoy the abstract properties of an SFS.

\begin{restatable}{lemma}{BangSFS}
    % \LemmaToFromProof{BangSFS}%
    \label{lem:BangSFS}%
    The family  $(\bangArrSet_S<\rel>, \bangArrSet_I<\rel>)_{\rel \in
        \{\bangSymbBeta, \bangSymbSubs, \bangSymbBang\}}$%
    is an SFS for $(\setBangTerms, \bangArr_F)$.
\end{restatable}

This immediately gives the following novel factorization result for
the Distant Bang Calculus, by applying \Cref{lem:SFS_Factorizable} and \Cref{lem:BangSFS}.

\begin{corollary}[\BANGSymb Factorization]
    \label{lem:Bang_Factorization}%
    We have that $\bangArrSet*_F \;=\; \bangArrSet*_S\bangArrSet*_I$.
\end{corollary}

\begin{example} Take $t=(\app{x}{y})\esub{y}{\oc (\app{I}{ \oc
  \oc (\app{I}{ w})})}$ where $I = \abs{z}{z}$. 
Then, the
  factorization of the first sequence starting at $t$ below,
  is given by the second one:
\begin{gather*}
        t
    \bangArrSet_F
        (\app{x}{y})\esub{y}{\oc(z\esub{z}{\oc \oc (\app{I}{w})})}
    \bangArrSet_S
       \app{x}{(z\esub{z}{\oc \oc (\app{I}{w})})}
    \bangArrSet_F
       \app{x}{(z\esub{z}{\oc \oc (z\esub{z}{w})})}
    \bangArrSet_S
       \app{x}{\oc(z\esub{z}{w})}
\\
       t
    \;\;\bangArrSet_S\;\;
       \app{x}{(\app{I}{\oc\oc(\app{I}{w})})}
    \;\;\bangArrSet_S\;\;
       \app{x}{(z\esub{z}{\oc\oc(\app{I}{w})})}
    \;\;\bangArrSet_S\;\;
       \app{x}{\oc(\app{I}{w})}
    \;\;{\bangArrSet_I}\;\;
       \app{x}{\oc(z\esub{z}{w})}
\end{gather*} 
\end{example}

\paragraph{Factorizations in \CBNSymb and \CBVSymb.}
\label{subsec:CBN-CBV-Factorization}

\begin{figure}[t]
    \begin{minipage}[c]{0.49\textwidth}%
        \begin{tikzpicture}
            \node (0_t) at (-2.5, 0)      {$t$};
            \node (0_u) at (2.5, 0)       {$u$};

            \draw[->] (0_t) to ([xshift=-0.5cm]0_u);
            \node at ([xshift=-0.4cm, yshift=0.1cm]0_u) {*};
            \node at ([xshift=-0.35cm, yshift=-0.15cm]0_u) {$\scriptstyle\cbnFCtxtSet$};
            \node[rotate=-90] at (-0.65, -0.3)               {$\rightsquigarrow$};
            \node at (0, -0.3) {\scriptsize({\tt Cor.}\ref{cor:cbn prop-full})};

            \node (1_t) at (-2.5, -0.6)      {$\cbnToBangAGK{t}$};
            \node (1_u) at (2.5, -0.6)       {$\cbnToBangAGK{u}$};

            \draw[->] (1_t) to ([xshift=-0.6cm]1_u);
            \node at ([xshift=-0.5cm, yshift=0.1cm]1_u) {*};
            \node at ([xshift=-0.5cm, yshift=-0.15cm]1_u) {$\scriptstyle\bangFCtxtSet$};
            \node[rotate=-90] at (-0.65, -0.9)               {$\rightsquigarrow$};
            \node at (0, -0.9) {\scriptsize({\tt Cor.}\ref{lem:Bang_Factorization})};

            \node (2_t) at (-2.5, -1.2)      {$\cbnToBangAGK{t}$};
            \node (2_s) at (0, -1.2)         {$s'$};
            \node (2_u) at (2.5, -1.2)       {$\cbnToBangAGK{u}$};

            \draw[->] (2_t) to ([xshift=-0.4cm]2_s);
            \node at ([xshift=-0.3cm, yshift=0.1cm]2_s) {*};
            \node at ([xshift=-0.3cm, yshift=-0.15cm]2_s) {$\scriptstyle\bangSCtxtSet$};
            \draw[->] (2_s) to ([xshift=-0.6cm]2_u);
            \node at ([xshift=-0.5cm, yshift=0.1cm]2_u) {*};
            \node at ([xshift=-0.5cm, yshift=-0.15cm]2_u) {$\scriptstyle\bangICtxtSet$};
            \node[rotate=-90] at (-1.85, -1.5)               {$\rightsquigarrow$};
            \node at (-1.2, -1.5) {\scriptsize({\tt Cor.}\ref{cor:cbn prop-surface})};
            \node[rotate=-90] at (0.55, -1.5)               {$\rightsquigarrow$};
            \node at (1.2, -1.5) {\scriptsize({\tt Cor.}\ref{cor:cbn prop-internal})};

            % \node at (4.5, -2) {$(\cbnToBangAGK{s} = s')$};

            \node (3_t) at (-2.5, -1.8)      {$t$};
            \node (3_s) at (0, -1.8)         {$s$};
            \node (3_u) at (2.5, -1.8)       {$u$};

            \draw[->] (3_t) to ([xshift=-0.4cm]3_s);
            \node at ([xshift=-0.3cm, yshift=0.1cm]3_s) {*};
            \node at ([xshift=-0.25cm, yshift=-0.15cm]3_s) {$\scriptstyle\cbnSCtxtSet$};
            \draw[->] (3_s) to ([xshift=-0.4cm]3_u);
            \node at ([xshift=-0.3cm, yshift=0.1cm]3_u) {*};
            \node at ([xshift=-0.25cm, yshift=-0.15cm]3_u) {$\scriptstyle\cbnICtxtSet$};
        \end{tikzpicture}
        \vspace{-0.3cm}
        \caption{\CBNSymb factorization (with $s' \!=\! \cbnToBangAGK{s}$)}
        \label{fig:proof cbn factorization}
    \end{minipage}
\hspace{.5cm}
    \begin{minipage}[c]{0.49\textwidth}%
        \begin{tikzpicture}
            \node (0_t) at (-2.5, 0)      {$t$};
            \node (0_u) at (2.5, 0)       {$u$};

            \draw[->] (0_t) to ([xshift=-0.5cm]0_u);
            \node at ([xshift=-0.4cm, yshift=0.1cm]0_u) {*};
            \node at ([xshift=-0.35cm, yshift=-0.15cm]0_u) {$\scriptstyle\cbvFCtxtSet$};
            \node[rotate=-90] at (-0.65, -0.3)               {$\rightsquigarrow$};
            \node at (0, -0.3) {\scriptsize({\tt Cor.}\ref{cor:cbv prop-full})};

            \node (1_t) at (-2.5, -0.6)      {$\cbvToBangAGK{t}$};
            \node (1_u) at (2.5, -0.6)       {$\cbvToBangAGK{u}$};

            \draw[->] (1_t) to ([xshift=-0.6cm]1_u);
            \node at ([xshift=-0.5cm, yshift=0.1cm]1_u) {*};
            \node at ([xshift=-0.5cm, yshift=-0.15cm]1_u) {$\scriptstyle\bangFCtxtSet$};
            \node[rotate=-90] at (-0.65, -0.9)               {$\rightsquigarrow$};
            \node at (0, -0.9) {\scriptsize({\tt Cor.}\ref{lem:Bang_Factorization})};

            \node (2_t) at (-2.5, -1.2)      {$\cbvToBangAGK{t}$};
            \node (2_s) at (0, -1.2)         {$s'$};
            \node (2_u) at (2.5, -1.2)       {$\cbvToBangAGK{u}$};

            \draw[->] (2_t) to ([xshift=-0.4cm]2_s);
            \node at ([xshift=-0.3cm, yshift=0.1cm]2_s) {*};
            \node at ([xshift=-0.3cm, yshift=-0.15cm]2_s) {$\scriptstyle\bangSCtxtSet$};
            \draw[->] (2_s) to ([xshift=-0.6cm]2_u);
            \node at ([xshift=-0.5cm, yshift=0.1cm]2_u) {*};
            \node at ([xshift=-0.5cm, yshift=-0.15cm]2_u) {$\scriptstyle\bangICtxtSet$};
            \node[rotate=-90] at (-1.85, -1.5)               {$\rightsquigarrow$};
            \node at (-1.2, -1.5) {\scriptsize({\tt Cor.}\ref{cor:cbv prop-surface})};
            \node[rotate=-90] at (0.55, -1.5)               {$\rightsquigarrow$};
            \node at (1.2, -1.5) {\scriptsize({\tt Cor.}\ref{cor:cbv prop-internal})};

            % \node at (4.5, -2) {$(\cbvToBangAGK{s} = s')$};

            \node (3_t) at (-2.5, -1.8)      {$t$};
            \node (3_s) at (0, -1.8)         {$s$};
            \node (3_u) at (2.5, -1.8)       {$u$};

            \draw[->] (3_t) to ([xshift=-0.4cm]3_s);
            \node at ([xshift=-0.3cm, yshift=0.1cm]3_s) {*};
            \node at ([xshift=-0.25cm, yshift=-0.15cm]3_s) {$\scriptstyle\cbvSCtxtSet$};
            \draw[->] (3_s) to ([xshift=-0.4cm]3_u);
            \node at ([xshift=-0.3cm, yshift=0.1cm]3_u) {*};
            \node at ([xshift=-0.25cm, yshift=-0.15cm]3_u) {$\scriptstyle\cbvICtxtSet$};
        \end{tikzpicture}
        \vspace{-0.3cm}
        \caption{\CBVSymb factorization (with $s' \!=\! \cbvToBangAGK{s}$)}
        \label{fig:proof cbv factorization}
    \end{minipage}
\end{figure}

To achieve factorization in \CBNSymb and \CBVSymb via the abstract
method, we need to establish the existence of an $SFS$ in each case.
This requires validating multiple commutations. We bypass these
lengthy proofs by adopting a simpler projection approach
from~\BANGSymb.

As in \BANGSymb, we claim that \surface reduction is the significant
part of \full reduction in \CBNSymb/\CBVSymb, and we consequently
identify the \insignificant\ subreduction, called here
\emph{internal}. The \CBNSymb (resp. \CBVSymb) internal contexts
$\cbnICtxt \in \cbnICtxtSet$ (resp. $\cbvICtxt \in \cbvICtxtSet$) are
\full contexts whose hole is in an argument (\resp under a $\lambda$).
Formally,
\begin{align*}
    \emphasis{(\CBNSymb Internal Contexts)}
& &
    \cbnICtxt &\;\Coloneqq\; \app[\,]{t}{\cbnFCtxt}
        \vsep t\esub{x}{\cbnFCtxt}
        \vsep \cbnSCtxt^*\cbnCtxtPlug{\cbnICtxt}
& &
    \text{with }
    \cbnSCtxt^* \in \cbnSCtxtSet \setminus \{\Hole\} \\
    \emphasis{(\CBVSymb Internal Contexts)}
& &
		\cbvICtxt&\;\Coloneqq\; \abs{x}{\cbvFCtxt}
		\vsep \cbvSCtxt^*\cbvCtxtPlug{\cbvICtxt}
& &
		\text{with }
		\cbvSCtxt^* \in \cbvSCtxtSet \setminus \{\Hole\}
\end{align*}	

\noindent
The \emphasis{\CBNSymb (resp. \CBVSymb) internal reduction}
$\cbnArrSet_I$ (\resp  $\cbvArrSet_I$) is the closure over all
internal contexts $\cbnICtxt \in \cbnICtxtSet$ (resp. $\cbvICtxt \in
\cbvICtxtSet$) of the rewrite rules \cbnSymbBeta and \cbnSymbSubs\
(\resp \cbnSymbBeta and \cbvSymbSubs). For example,
$\app{(\abs{x}{x})}{\Hole}$ is a \CBNSymb\ internal context, while
$\Hole$ is not, thus $\app{(\abs{x}{x})}{(\app{(\abs{y}{z})}{t})}
\cbnArrSet_I \app{(\abs{x}{x})}{z} \not\cbnArrSet_I z$. And
$\app{(\abs{x}{\Hole})}{z}$ is a \CBVSymb internal context while
$\Hole$ is not, thus $\app{(\abs{x}{\app{(\abs{y}{y})}{z}})}{z}
\cbvArrSet_I \app{(\abs{x}{y\esub{y}{z}})}{z} \cbvArrSet_I
\app{(\abs{x}{z})}{z} \not\cbvArrSet_I z\esub{x}{z}$.

As in the surface case, the one-step simulation and reverse simulation
(\Cref{l:one-step-cnb,l:reverse-simulation-cbn} for \CBNSymb,
\Cref{lem:Cbv_Simulation_One_Step,lem:Cbv_Reverse_Simulation_One_Step}
for \CBVSymb) can be specialized to the internal case. This allows us
to show in particular the following property. 

\begin{restatable}{corollary}{RecCbnPropInternal}
    \label{cor:cbn prop-internal}\label{cor:cbv prop-internal}
    Let $t, u \in \setCbnTerms$  and $s' \in \setBangTerms$.  
    \begin{itemize}
    \item \textbf{(Stability)}: if $\cbnToBangAGK{t} \bangArrSet*_I
        s'$ (\resp $\cbvToBangAGK{t} \bangArrSet*_I s'$ and
        {$s'$ is a \bangSetINF<d!>}) then there is $s \in \setCbnTerms$
        such that $\cbnToBangAGK{s} = s'$ (\resp $\cbvToBangAGK{s} =
        s'$).

    \item \textbf{(Normal Forms)}: $t \text{ is a \cbnSetINF}$ (\resp
        	\cbvSetINF) iff $\cbnToBangAGK{t}$ (\resp
        	$\cbvToBangAGK{t}$) is a \bangSetINF.

    \item \textbf{(Simulations)}: $t \cbnArrSet*_I u$ (\resp $t
        \cbvArrSet*_I u$) iff $\cbnToBangAGK{t} \bangArrSet*_I
        \cbnToBangAGK{u}$ (\resp $\cbvToBangAGK{t} \bangArrSet*_I
        \cbvToBangAGK{u}$). Moreover, the number of
        $\cbnSymbBeta/\cbnSymbSubs$-steps (\resp
        $\cbvSymbBeta/\cbvSymbSubs$-steps) on the left matches the
        number of $\bangSymbBeta/\bangSymbSubs$-steps on the right.
    \end{itemize}
\end{restatable}

%%%%%%%% Proposition B;
% \begin{restatable}{corollary}{RecCbnPropInternal}
%     \label{cor:cbn prop-internal}
%     Let $t, u \in \setCbnTerms$ and $s' \in \setBangTerms$. For any
%     $(\mathbb{I}_{\tt K}, \cdot^{\tt k}) \in \{(\cbnICtxtSet,
%     \cbnToBangAGK{\cdot}), (\cbvICtxtSet, \cbvToBangAGK{\cdot})\}$:
%     \begin{itemize}
%     \item \textbf{(Stability)}: if $t^{\tt k}
%         	\rightarrow^*_{\mathbb{I}_{\tt K}} s'$ and ${\tt
%         	NF}_{\mathbb{I}_{\tt
%         	K}\bangCtxtPlug{\bangSymbBang}}\!(s')$ then there is $s
%         	\in \setCbnTerms$ such that $s^{\tt k} = s'$.

%     \item \textbf{(Normal Forms)}: ${\tt NF}_{\mathbb{I}_{\tt
%         K}}\!(t)$ if and only if ${\tt NF}_{\mathbb{I}}(t^{\tt k})$.

%     \item \textbf{(Simulations)}: $t \rightarrow^*_{\mathbb{I}_{\tt
%         K}} u$ if and only if $t^{\tt e} \bangArrSet*_I u^{\tt e}$.  
%         Moreover, the number of $\bangSymbBeta^{\tt
%         -e}/\bangSymbSubs^{\tt -e}$ on the left matches the number of
%         $\bangSymbBeta/\bangSymbSubs$-steps on the right.
%     \end{itemize}
% \end{restatable}

\noindent
Via \Cref{cor:cbn prop-surface,cor:cbv prop-surface,cor:cbn
prop-internal}, we can project factorization from \BANGSymb
back~to~\CBNSymb/\CBVSymb.

\begin{theorem}[\CBNSymb/\CBVSymb Factorizations]
     $\cbnArrSet*_F \,=\, \cbnArrSet*_S\cbnArrSet*_I$ and
     $\cbvArrSet*_F \,= \, \cbvArrSet*_S\cbvArrSet*_I$.
\end{theorem}
\begin{proof}
    The proof for \CBNSymb is depicted in \Cref{fig:proof cbn
    factorization}. In particular, since $\cbnToBangAGK{t}
    \bangArrSet*_S s'$, one deduces using \Cref{cor:cbn prop-surface}
    that there exists $s \in \setCbnTerms$ such that $\cbnToBangAGK{s}
    = s'$.

    The proof for  \CBVSymb is depicted in \Cref{fig:proof cbv
    factorization}. In particular, by construction $\cbvToBangAGK{u}$
    is a \bangSetFNF<d!> and by induction on the length of $s'
    \bangArr*_I \cbvToBangAGK{u}$, one has that $s'$ is a
    \bangSetSNF<d!>. Using \Cref{cor:cbv prop-surface}, one deduces
    that there exists $s \in \setCbvTerms$ such that $\cbvToBangAGK{s}
    = s'$.
   \qed
\end{proof}

\noindent

% {Interestingly, the
% \CBNSymb and \CBVSymb factorizations obtained with our method bring
% forward two well-known subreductions in the theory of CBN and CBV: (a
% non-deterministic but diamond variant of) CBN \emph{head} reduction
% and CBV \emph{weak} reduction (not evaluating under abstraction). So,
% our results are~not~\emph{ad-hoc}!}

%% Ovalbox autour de la figure

\section{Conclusion and Related Work}
\label{s:conclusion}

Our first contribution is to revisit and extend several properties
concerning the encoding of \CBNSymb into \BANGSymb. The second
contribution, 
more significant, consists in introducing a new embedding from
\CBVSymb to \BANGSymb, which is conservative with respect to previous
results in the
literature~\cite{BucciarelliKesnerRiosViso20,BucciarelliKesnerRiosViso23},
but also (and this is a novelty)  allows us to establish the essential
reverse simulation property, achieved through the non-trivial concept
of diligent sequence. We illustrate the strength of our methodology by
means of an example, namely factorization. For that, we first prove a
factorization theorem for \BANGSymb, another major contribution of the
paper, and we then deduce factorization for \CBNSymb and \CBVSymb by
projecting that for \BANGSymb.

%<<<<<<< HEAD
%In addition to the tangible contributions presented in this
%  paper, we believe our methodology  enhances the understanding of
%  the semantic aspects of CBV, especially concerning untyped and typed
%  approximants. This remains a topic that, while gradually gaining
%  attention in the
%  literature~\cite{ManzonettoPaganiRonchi19,KerinecManzonettoPagani20,ArrialGuerrieriKesner23},
%  is yet to be thoroughly explored. Our novel CBV
% embedding would also suggest a logical
%  counterpart (a new translation of intuitionistic logic into
%    linear logic), which remains to be investigated. Last but not
%  least, reverse simulation constitutes also a key tool in
%  transferring the confluence property from \BANGSymb to \CBNSymb and
%  \CBVSymb, although this aspect was not included in the paper due to
%  space constraints.  More importantly, we aim to further leverage our
%  technique to explore other crucial dynamic properties of \CBNSymb
%  and \CBVSymb, such as standardization, normalization, solvability,
%  etc. Specifically, we believe that the unified approach adopted in
%  this paper can bring exciting ideas for transferring solvability and
%  genericity between \CBVSymb and \BANGSymb, in both directions. 
%=======

In \cite{FaggianGuerrieri21}, factorization
for the (non-distant) Bang Calculus has been proved and from that,
factorizations results for standard (non distant) CBN
$\lambda$-calculus and Plotkin's original CBV $\lambda$-calculus has
been deduced. But the {(non-distant)}
Bang Calculus and Plotkin's CBV are \emph{not
adequate}, in the sense explained in \Cref{sec:intro}, thus decreasing
the significance of those preliminary results. 

When taking \emph{adequate} versions of the Bang Calculus, by adding
ES and distance, or $\sigma$-reduction
\cite{EhrhardGuerrieri16,GuerrieriOlimpieri21}, the CBV encodings in
the
literature~\cite{EhrhardGuerrieri16,GuerrieriManzonetto19,BucciarelliKesnerRiosViso20,FaggianGuerrieri21,BucciarelliKesnerRiosViso23}
fail to enjoy reverse simulation, thus preventing one deducing dynamic
properties from
the Bang Calculus into CBV. Other CBN and CBV encodings into a
unifying framework appear in~\cite{BakelTyeWu23}, but there is no
reverse simulation property, so that no concrete application of the
proposed encoding to export properties into CBN and CBV. The same
occurs in~\cite{EspiritoPintoUustalu19}. The only exceptions
are~\cite{ArrialGuerrieriKesner23,KesnerViso22} ---where only
\emph{static} properties are obtained---,
and~\cite{GuerrieriManzonetto19,FaggianGuerrieri21} ---where the
Bang and CBV calculi are not adequate in the sense explained in
\Cref{sec:intro}.

Sabry and Wadler~\cite{SabryWadler97} showed that simulation and
reverse simulation between two calculi are for free when their back
and forth translations give rise to an adjoint. One of the
difficulties to achieve our results is that our CBN and CBV
embeddings, as well as the ones used
in~\cite{EhrhardGuerrieri16,GuerrieriManzonetto19,BucciarelliKesnerRiosViso20,FaggianGuerrieri21,BucciarelliKesnerRiosViso23},
do not form an adjoint. This is basically due to the fact that a
CBN/CBV term can be decorated by $\oc$ and $\derSymb$ so as that
administrative steps performed in the (Distant) Bang Calculus do not
correspond to anything \mbox{in CBN or CBV}. Our contribution is
precisely to achieve simulation and reverse simulation without the
need for any  adjoint.

As discussed at the end of~\Cref{sec:CBN_CBV_Embeddings}, proving simulation and reverse simulation requires a
considerable initial effort. However, this investment lays the
groundwork for numerous benefits without extra costs, as we showed
in \Cref{sec:Bang_CBN_CBV_Confluence,sec:Factorization}. 

In addition to the tangible contributions presented in this paper, we
believe our methodology enhances the understanding of the semantic
aspects of CBV, especially concerning untyped and typed approximants.
This remains a topic that, while gradually gaining attention in the
literature~\cite{ManzonettoPaganiRonchi19,KerinecManzonettoPagani20,ArrialGuerrieriKesner23},
is yet to be thoroughly explored. Our novel CBV embedding would also
suggest a logical counterpart (a new encoding  of intuitionistic logic
into linear logic), which remains to be investigated. 
Moreover, we aim to further leverage our technique to explore other
crucial dynamic properties of \CBNSymb and \CBVSymb, such as
standardization, normalization, genericity as well as some specific
deterministic strategies.

\bibliography{main}

\ifbool{versionLongue}{
\booltrue{inAppendix}
\clearpage
\appendix
% \setAppendixSymbols

\section{Appendix: Proofs of \Cref{sec:Distant_Bang_Calculus}}

\label{l:proofs-bang}
\subsection{Basic Properties}

We introduce the following measure, which counts the number of
constructors contained within a term.

\begin{definition}[Term Full Size]
    \begin{equation*}
        \begin{array}{rcl c rcl}
            \bangTermFSize{x}
                &:=&
                    0
        &\hspace{2cm}&
            \bangTermFSize{t\esub{x}{u}}
                &:=&
                    1 + \bangTermFSize{t} + \bangTermFSize{u}
        \\
            \bangTermFSize{\app{t}{u}}
                &:=&
                    1 + \bangTermFSize{t} + \bangTermFSize{u}
        &&
            \bangTermFSize{\der{t}}
                &:=&
                    1 + \bangTermFSize{t}
        \\
            \bangTermFSize{\abs{x}{t}}
                &:=&
                    1 + \bangTermFSize{t}
        &&
            \bangTermFSize{\oc t}
                &:=&
                    1 + \bangTermFSize{t}
        \end{array}
    \end{equation*}
\end{definition}

It is extended to all contexts by taking $\bangTermFSize{\Hole} = 0$
and show that full term size is compatible with contexts.

\begin{lemma}
	\label{lem: |F<t>|_F = |F|_F + |t|_F}%
	Let $t \in \setBangTerms$, then for any context $\bangFCtxt$,
	$\bangTermFSize{\bangFCtxt<t>} = \bangTermFSize{\bangFCtxt} +
	\bangTermFSize{t}$.
\end{lemma}
% \hideProof{
    \begin{proof}
By induction on $\bangFCtxt$. Cases:
\begin{itemize}
\item[\bltI] $\bangFCtxt = \Hole$: Then $\bangFCtxt<t> = t$ thus
$\bangTermFSize{\bangFCtxt<t>} = \bangTermFSize{t}= 0 +
\bangTermFSize{t} = \bangTermFSize{\bangFCtxt} + \bangTermFSize{t}$.

\item[\bltI] $\bangFCtxt = \app[\,]{\bangFCtxt'}{u}$. Then:
\begin{equation*}
    \begin{array}{rcl}
        \bangTermFSize{\bangFCtxt<t>}
        &=& \bangTermFSize{\app{\bangFCtxt'<t>}{u}}
        \;=\; 1 + \bangTermFSize{\bangFCtxt'<t>} + \bangTermFSize{u}
        \\
        &\eqih& 1 + \bangTermFSize{\bangFCtxt'} + \bangTermFSize{t} + \bangTermFSize{u}
        \;=\; \bangTermFSize{\app{\bangFCtxt'}{u}} + \bangTermFSize{t}
        \;=\; \bangTermFSize{\bangFCtxt} + \bangTermFSize{t}
    \end{array}
\end{equation*}

\item[\bltI] $\bangFCtxt = \app[\,]{u}{\bangFCtxt'}$. Then:
\begin{equation*}
    \begin{array}{rcl}
        \bangTermFSize{\bangFCtxt<t>}
        &=& \bangTermFSize{\app[\,]{u}{\bangFCtxt'<t>}}
        \;=\; 1 + \bangTermFSize{u} + \bangTermFSize{\bangFCtxt'<t>}
        \\
        &\eqih& 1 + \bangTermFSize{u} + \bangTermFSize{\bangFCtxt'} + \bangTermFSize{t}
        \;=\; \bangTermFSize{\app[\,]{u}{\bangFCtxt'}} + \bangTermFSize{t}
        \;=\; \bangTermFSize{\bangFCtxt} + \bangTermFSize{t}
    \end{array}
\end{equation*}

\item[\bltI] $\bangFCtxt = \bangFCtxt'\esub{x}{u}$. Then:
\begin{equation*}
    \begin{array}{rcl}
        \bangTermFSize{\bangFCtxt<t>}
        &=& \bangTermFSize{\bangFCtxt'<t>\esub{x}{u}}
        \;=\; 1 + \bangTermFSize{\bangFCtxt'<t>} + \bangTermFSize{u}
        \\
        &\eqih& 1 + \bangTermFSize{\bangFCtxt'} + \bangTermFSize{t} + \bangTermFSize{u}
        \;=\; \bangTermFSize{\bangFCtxt'\esub{x}{u}} + \bangTermFSize{t}
        \;=\; \bangTermFSize{\bangFCtxt} + \bangTermFSize{t}
    \end{array}
\end{equation*}

\item[\bltI] $\bangFCtxt = u\esub{x}{\bangFCtxt'}$. Then:
\begin{equation*}
    \begin{array}{rcl}
        \bangTermFSize{\bangFCtxt<t>}
        &=& \bangTermFSize{u\esub{x}{\bangFCtxt'<t>}}
        \;=\; 1 + \bangTermFSize{u} + \bangTermFSize{\bangFCtxt'<t>}
        \\
        &\eqih& 1 + \bangTermFSize{u} + \bangTermFSize{\bangFCtxt'} + \bangTermFSize{t}
        \;=\; \bangTermFSize{u\esub{x}{\bangFCtxt'}} + \bangTermFSize{t}
        \;=\; \bangTermFSize{\bangFCtxt} + \bangTermFSize{t}
    \end{array}
\end{equation*}

\item[\bltI] $\bangFCtxt = \abs{x}{\bangFCtxt'}$. Then:
\begin{equation*}
    \begin{array}{rcl}
        \bangTermFSize{\bangFCtxt<t>}
        &=& \bangTermFSize{\abs{x}{\bangFCtxt'<t>}}
        \;=\; 1 + \bangTermFSize{\bangFCtxt'<t>}
        \\
        &\eqih& 1 + \bangTermFSize{\bangFCtxt'} + \bangTermFSize{t}
        \;=\; \bangTermFSize{\abs{x}{\bangFCtxt'}} + \bangTermFSize{t}
        \;=\; \bangTermFSize{\bangFCtxt} + \bangTermFSize{t}
    \end{array}
\end{equation*}

\item[\bltI] $\bangFCtxt = \der{\bangFCtxt'}$. Then:
\begin{equation*}
    \begin{array}{rcl}
        \bangTermFSize{\bangFCtxt<t>}
        &=& \bangTermFSize{\der{\bangFCtxt'<t>}}
        \;=\; 1 + \bangTermFSize{\bangFCtxt'<t>}
        \\
        &\eqih& 1 + \bangTermFSize{\bangFCtxt'} + \bangTermFSize{t}
        \;=\; \bangTermFSize{\der{\bangFCtxt'}} + \bangTermFSize{t}
        \;=\; \bangTermFSize{\bangFCtxt} + \bangTermFSize{t}
    \end{array}
\end{equation*}

\item[\bltI] $\bangFCtxt = \oc\bangFCtxt'$. Then:
\begin{equation*}
    \begin{array}{rcl}
        \bangTermFSize{\bangFCtxt<t>}
        &=& \bangTermFSize{\oc \bangFCtxt'<t>}
        \;=\; 1 + \bangTermFSize{\bangFCtxt'<t>}
        \\
        &\eqih& 1 + \bangTermFSize{\bangFCtxt'} + \bangTermFSize{t}
        \;=\; \bangTermFSize{\oc \bangFCtxt'} + \bangTermFSize{t}
        \;=\; \bangTermFSize{\bangFCtxt} + \bangTermFSize{t}
    \end{array}
\end{equation*}
\qed
\end{itemize}
\end{proof}
% }

We now use this measure to prove the termination of reductions
$\bangArrSet_F<d!>, \bangArrSet_S<d!>$.

\begin{restatable}{lemma}{RecFullDistantBangSN}
	\label{lemma:RecFullDistantBangSN}
	% \LemmaToFromProof{->F<d!>SN}%
	\label{lem:->F<d!>SN}
	Reductions $\bangArrSet_F<d!>$ and $\bangArrSet_S<d!>$ are
	terminating.
\end{restatable}
%\RecFullDistantBangSN*
% \hideProof{
    \begin{proof} \label{prf:->F<d!>SN}%
To prove that $\bangArrSet_F<d!>$ is terminating, it is enough to show
that $\bangTermFSize{t} > \bangTermFSize{u}$ for any $t, u \in
\setBangTerms$ such that $t \bangArrSet_F<d!> u$. By definition, there
are $t', u' \in \setBangTerms$ and a contexts $\bangFCtxt$ such that
$t = \bangFCtxt<t'>$, $u = \bangFCtxt<u'>$ with $t' \bangMapstoBang
u'$. So, $t' = \der{\bangLCtxt<\oc s>}$ and $u' = \bangLCtxt<s>$ for
some $s \in \setBangTerms$ and list context $\bangLCtxt$. Using
\Cref{lem: |F<t>|_F = |F|_F + |t|_F}, one deduces that
$\bangTermFSize{t'} = \bangTermFSize{\der{\bangLCtxt<\oc s>}} = 2 +
\bangTermFSize{s} +  \bangTermFSize{\bangLCtxt} > \bangTermFSize{s} +
\bangTermFSize{\bangLCtxt} = \bangTermFSize{\bangLCtxt<s>} =
\bangTermFSize{u'}$, and using \Cref{lem: |F<t>|_F = |F|_F + |t|_F}
again we conclude that $\bangTermFSize{t} =
\bangTermFSize{\bangFCtxt<t'>} > \bangTermFSize{\bangFCtxt<u'>} =
\bangTermFSize{u}$.

As $\bangArrSet_S<d!> \,\subseteq\, \bangArrSet_F<d!>$, termination of
$\bangArrSet_F<d!>$ implies termination of $\bangArrSet_S<d!>$.
\qed
\end{proof}
% }

\begin{lemma}
		\label{lemma:RecFullDistantBetaSN}
		Reductions $\bangArrSet_F<dB>$ and $\bangArrSet_S<dB>$ are
		terminating.
\end{lemma}
% \hideProof{
\begin{proof}
To prove that $\bangArrSet_F<dB>$ is terminating, it is enough to show
that $\bangTermFSize{t} > \bangTermFSize{u}$ for any $t, u \in
\setBangTerms$ such that $t \bangArrSet_F<dB> u$. By definition, there
are $t', u' \in \setBangTerms$ and a context $\bangFCtxt$ such that $t
= \bangFCtxt<t'>$, $u = \bangFCtxt<u'>$ with $t' \bangMapstoBeta u'$.
So, $t' = \app{\bangLCtxt<\abs{x}{s}>}{r}$ and $u' =
\bangLCtxt<s\esub{x}{r}>$ for some $s, r \in \setBangTerms$ and list
context $\bangLCtxt$. By \Cref{lem: |F<t>|_F = |F|_F + |t|_F}, one has
$\bangTermFSize{t'} = \bangTermFSize{\app{\bangLCtxt<\abs{x}{s}>}{r}}
= 2 + \bangTermFSize{s} + \bangTermFSize{r} +
\bangTermFSize{\bangLCtxt} > 1 + \bangTermFSize{s} + \bangTermFSize{r}
+ \bangTermFSize{\bangLCtxt} =
\bangTermFSize{\bangLCtxt<s\esub{x}{r}>} = \bangTermFSize{u'}$, and
using \Cref{lem: |F<t>|_F = |F|_F + |t|_F} again we conclude that
$\bangTermFSize{t} = \bangTermFSize{\bangFCtxt<t'>} >
\bangTermFSize{\bangFCtxt<u'>} = \bangTermFSize{u}$.

As $\bangArrSet_S<dB> \,\subseteq\, \bangArrSet_F<dB>$, termination of
$\bangArrSet_F<dB>$ implies termination of $\bangArrSet_S<dB>$.
\qed
\end{proof}
% }

\subsection{Abstract Diligence for \BANGSymb}
\label{a:abstract-diligence}

In this subsection we introduce an abstract proof of diligence
process. The proof is parametrized using a family $\bangECtxtSet$ of
\BANGSymb contexts.

We first introduce a parametric definition of administrive diligent
reduction.

\begin{definition}[Administrative Diligent Reduction]%
  \label{def:abstract-diligence}
  The \emph{administrative diligent reduction} $\bangArrSet_{Eai}$
    associated with a family $\bangECtxtSet \subseteq \bangFCtxtSet$
    of contexts is a subset of the reduction $\bangArrSet_E$ obtained
    by restricting $\bangSymbBeta$ and $\bangSymbSubs$-steps to
    $\bangECtxtSet<\bangSymbBang>$-normal forms. More precisely, it is
    defined as follows:
    \begin{equation*}
        \begin{array}{rcl}
            \bangArrSet_{Eai}
                &\coloneqq& (\bangArrSet_E<dB> \cap\; \text{\bangSetENF<d!>} \times \setBangTerms)
                    \;\;\cup\;\; (\bangArrSet_E<s!> \cap\; \text{\bangSetENF<d!>} \times \setBangTerms)
                    \;\;\cup\; \bangArrSet_E<d!>
        \end{array}
    \end{equation*}

    We say that the reduction $\bangArrSet_E$ \emph{admits a diligence
    process} if $t \bangArrSet*_{Eai} u$ whenever $t \bangArrSet*_E u$
    with $u$ a \bangSetENF<d!>.
\end{definition}

\begin{lemma}[$\bangECtxtSet<\bangSymbBang>$-Head Steps Extraction]
\label{lem:t ->*|Eai u => t ->*|E<d!> u or (s |E<d!>-NF and t ->*|E<d!> s and s ->*|E u)}%
Let $t \bangArrSet*_{Eai} u$, then $t \bangArrSet*_E<d!> u$
or there is a \bangSetENF<d!> $s \in \setBangTerms$ such that
$t \bangArrSet*_E<d!> s$ and $s \bangArrSet*_{Eai} u$.
\end{lemma}
% \hideProof{
    \begin{proof}
By induction on the length of the reduction sequence $t
\bangArrSet*_{Eai} u$. Cases:
\begin{itemize}
\item[\bltI] $t \bangArrSet*^0_{Eai} u$: Then $t = u$ thus in
    particular $t \bangArrSet*_E<d!> u$ concluding this case.

\item[\bltI] $t \bangArrSet_{Eai} s_1 \bangArrSet*^n_{Eai} u$: We
    distinguish $2$ cases:
    \begin{itemize}
    \item[\bltII] $t$ is a \bangSetENF<d!>: Thus taking $s = t$
        concludes this case since by reflexivity $t \bangArrSet*_E<d!>
        s$ and by hypothesis $s \bangArrSet*_{Eai} u$.
    \item[\bltII] $t$ is not a \bangSetENF<d!>: Then necessarily $t
        \bangArrSet_E<d!> s_1 \bangArrSet_{Eai}^n u$. By \ih on $s_1
        \bangArrSet*^n_{Eai} u$, two cases can be distinguished:
        \begin{itemize}
        \item[\bltIII] $s_1 \bangArrSet*_E<d!> u$: Then $t
            \bangArr*_E<d!> u$ concluding this case.

        \item[\bltIII] $s_1 \bangArrSet*_E<d!> s$ and $s
            \bangArrSet*_{Eai} u$ for some \bangSetENF<d!> $s \in
            \setBangTerms$: Then $t \bangArrSet*_E<d!> s$ and $s
            \bangArrSet*_{Eai} u$ concluding this case.
            \qed
        \end{itemize}
    \end{itemize}
\end{itemize}
\end{proof}
% }

\begin{lemma}
    \label{lem:R=dB/s! square ->|E<d!>/->|E<R> => square ->*E<d!>/->E<R>}%
Let $\rel \in \{\bangSymbBeta, \bangSymbSubs\}$. Suppose that for any
$t, u_1, u_2 \in \setBangTerms$ such that
$t \bangArrSet_E<\rel> u_1$ and $t \bangArrSet_E<d!> u_2$, there is $s \in \setBangTerms$
making the diagram below commute.
\begin{equation*}
    \begin{array}{ccc}
        t                               &\bangLongArrSet{1cm}_E<R>  &u_1
    \\
        \bangLongDownArrBangESet{1cm}   &                           &\bangLongDownArrsBangESet{1cm}
    \\
        u_2                             &\bangLongArrSet{1cm}_E<R>  &s
    \end{array}
\end{equation*}
Then, for any $t, u_1, u_2 \in \setBangTerms$ such that
$t \bangArrSet_E<\rel> u_1$ and $t \bangArrSet*_E<d!> u_2$, there exists $s \in
\setBangTerms$ such that the following diagram also commutes.
\begin{equation*}
    \begin{array}{ccc}
        t                               &\bangLongArrSet{1cm}_E<R>  &u_1
    \\
        \bangLongDownArrsBangESet{1cm}  &                           &\bangLongDownArrsBangESet{1cm}
    \\
        u_2                             &\bangLongArrSet{1cm}_E<R>  &s
    \end{array}
\end{equation*}
\end{lemma}
% \hideProof{
    \begin{proof}
By induction on the length of the derivation $t \bangArrSet*_E<d!>
u_2$. Cases:
\begin{itemize}
\item[\bltI] $t = u_2$: One trivially concludes by taking $s = u_1$.

\item[\bltI] $t \bangArrSet_E<d!> u'_2 \bangArrSet*_E<d!> u_2$: We
consider the diagram below. By hypothesis, there is $u_1' \in
\setBangTerms$ such that the top square commutes. By \ih, there is $s
\in \setBangTerms$ such that the bottom square commutes.
\begin{equation*}
    \begin{array}{ccc}
        t                               &\bangLongArrSet{1cm}_E<R>  &u_1
    \\
        \bangLongDownArrBangESet{1cm}   &                           &\bangLongDownArrBangESet{1cm}
    \\
        u'_2                            &\bangLongArrSet{1cm}_E<R>  &u'_1
    \\
        \bangLongDownArrsBangESet{1cm}  &                           &\bangLongDownArrsBangESet{1cm}
    \\
        u_2                             &\bangLongArrSet{1cm}_E<R>  &s
    \end{array}
\end{equation*}
\qed
\end{itemize}
\end{proof}
% }

\begin{lemma}[One Step Abstract Diligence]
    \label{lem: One Step Abstract Implicitation}
    Let $\mathbb{E} \subseteq \bangFCtxtSet$ be a family of contexts
    such that:
    \begin{itemize}
    \item[\bltI] the reduction relation $\bangArrSet_E<d!>$
        is confluent;
    \item[\bltI] for any $\rel \in \{\bangSymbBeta, \bangSymbSubs\}$ and any
        $t, u_1, u_2 \in \setBangTerms$ such that
        $t \bangArrSet_E<\rel> u_1$ and $t \bangArrSet_E<d!> u_2$, there exists $s \in \setBangTerms$
        making the following diagram commutes. 
        \begin{equation*}
            \begin{array}{ccc}
                t                               &\bangLongArrSet{1cm}_E<R>  &u_1
            \\
                \bangLongDownArrBangESet{1cm}   &                           &\bangLongDownArrsBangESet{1cm}
            \\
                u_2                             &\bangLongArrSet{1cm}_E<R>  &s
            \end{array}
        \end{equation*}
    \end{itemize}
    If $t, u \in \setBangTerms$ where $u$ is a \bangSetENF<d!> and
    such that $t \bangArrSet_E \bangArrSet*_E<d!> u$, then $t
    \bangArrSet*_{Eai} u$.
\end{lemma}
% \hideProof{
    \begin{proof}
Let $t, u, s \in \setBangTerms$ with $t \bangArrSet_E s
\bangArrSet*_E<d!> u$ where $u$ is a \bangSetENF<d!>. Two cases:
\begin{itemize}
\item[\bltI] $t \bangArrSet_E<R> s$ with $\rel \in \{\bangSymbBeta,
    \bangSymbSubs\}$. Consider the diagram below. Since
    $\bangArrSet_E<d!>$ terminates
    (\Cref{lemma:RecFullDistantBangSN}), there is a \bangSetENF<d!>
    $t' \in \setBangTerms$ such that $t \bangArrSet*_F<d!> t'$. Using
    \Cref{lem:R=dB/s! square ->|E<d!>/->|E<R> => square
    ->*E<d!>/->E<R>} there is $s' \in \setBangTerms$ such that the
    left square commutes. By confluence of $\bangArrSet_E<d!>$ there
    is $u' \in \setBangTerms$ such that the right square commutes.
    \begin{equation*}
        \begin{array}{clclc}
            t                               &\bangLongArrSet{1cm}_E<\rel>   & s                                 &\bangLongArrSet{1cm}*_E<d!>    & u
        \\
            \bangLongDownArrsBangESet{1cm}  &                               & \bangLongDownArrsBangESet{1cm}    &                               & \bangLongDownArrsBangESet{1cm}
        \\
            t'                              & \bangLongArrSet{1cm}_E<\rel>  & s'                                &\bangLongArrSet{1cm}*_E<d!>    & u'
        \end{array}
    \end{equation*}
    As $u$ is \bangSetENF<d!>, necessarily $u = u'$. 
    As  $\bangArrSet*_E<d!> \subseteq \bangArr*_{Eai}$, one has $t
    \bangArrSet*_{Eai} t'$ and $s' \bangArrSet*_{Eai} u$. 
    Since $t'$ is 
    \bangSetENF<d!>, then $t' \bangArrSet_{Eai} s'$ and so $t
    \bangArrSet*_{Eai} u$ by transitivity.

\item[\bltI] $t \bangArrSet*_E<d!> u$: Then $t
    \bangArrSet*_E<d!> u$, and so $t \bangArrSet*_{Eai} u$, since
    $\bangArrSet*_E<d!> \subseteq \bangArrSet*_{Eai}$.
\qed
\end{itemize}
\end{proof}
% }

\begin{lemma}%
	\label{lem: Abstract Implicitation}
	Let $\mathbb{E}$ be a family of contexts such that:
	\begin{itemize}
		\item[\bltI] The reduction relation $\bangArrSet_E<d!>$ is confluent.

		\item[\bltI] For any $\rel \in \{\bangSymbBeta, \bangSymbSubs\}$ and any
		$t, u_1, u_2 \in \setBangTerms$ such that
		$t \bangArrSet_E<\rel> u_1$ and $t \bangArrSet_E<d!> u_2$, there exists $s \in \setBangTerms$
		making the following diagram commutes: 
		\begin{equation*}
			\begin{array}{ccc}
				t                               &\bangLongArrSet{1cm}_E<R>  &u_1
				\\
				\bangLongDownArrBangESet{1cm}   &                           &\bangLongDownArrsBangESet{1cm}
				\\
				u_2                             &\bangLongArrSet{1cm}_E<R>  &s
			\end{array}
		\end{equation*}
	\end{itemize}
	Then the reduction relation $\bangArrSet_E$ admits a diligence
	process.
\end{lemma}
% \hideProof{
	\begin{proof}
Let $t, u \in \setBangTerms$ and $t \bangArrSet*_E u$ such that $u$ is
a \bangSetENF<d!>. We show that $t \bangArrSet*_{Eai} u$.  
We proceed by induction on the length of the derivation $t
\bangArrSet*_E u$:
\begin{itemize}
\item[\bltI] $t = u$: Then by reflexivity $t \bangArrSet*_{Eai} u$.

\item[\bltI] $t \bangArrSet_E s \bangArrSet^n_E u$ for some $s \in
    \setBangTerms$: By \ih on $s \bangArrSet^n_E u$, one obtains $s
    \bangArrSet*_{Eai} u$. Using \Cref{lem:t ->*|Eai u => t ->*|E<d!>
    u or (s |E<d!>-NF and t ->*|E<d!> s and s ->*|E u)}, two cases can
    be distinguished:
    \begin{itemize}
    \item[\bltII] $s \bangArrSet*_E<d!> u$: Then $t \bangArrSet_E
        \bangArrSet*_E<d!> u$ and hence, by \Cref{lem: One Step
        Abstract Implicitation}, $t \bangArrSet*_{Eai} u$.

    \item[\bltII] $s \bangArrSet*_E<d!> s' \bangArrSet*_{Eai} u$ for
        some \bangSetENF<d!> $s' \in \setBangTerms$: Then $t
        \bangArrSet_E\bangArrSet*_E<d!> s'$ and so $t
        \bangArrSet*_{Eai} s'$ by \Cref{lem: One Step Abstract
        Implicitation}, hence $t \bangArrSet*_{Eai} u$ by
        transitivity.
        \qed
    \end{itemize}
\end{itemize}
\end{proof}
% }

\RecImplicitation*
\begin{proof} \label{prf:Implicitation} ~
    \begin{itemize}
    \item  \textbf{(Surface)} \;%
        See \Cref{prfSec:SurfaceImplicitation}, in particular
        \Cref{lem:SurfaceReorchestration}.
    \item \hspace{0.2cm} \textbf{(Full)} \;\hspace{0.3cm}%
        See \Cref{prfSec:FullImplicitation}, in particular
        \Cref{lem:FullReorchestration}.
        \qed
    \end{itemize}
\end{proof}

\subsection{Surface Diligence for \BANGSymb}
\label{prfSec:SurfaceImplicitation}

\begin{lemma}
	\label{lem:R=dB/s!andt->S<R>u1andt->S<d!>u2=>u1->S<d!>sandu2->S<R>s}
	Let $\rel \in \{\bangSymbBeta, \bangSymbSubs\}$ and $t, u_1, u_2
    \in \setBangTerms$ such that $t \bangArrSet_S<\rel> u_1$ and $t
    \bangArrSet_S<d!> u_2$. Then the diagram below commutes, for some
    $s \in \setBangTerms$.
	\begin{equation*}
		\begin{array}{ccc}
			t                               &\bangLongArrSet{1cm}_S<R>  &u_1
        \\
			\bangLongDownArrBangSSet{1cm}   &                           &\bangLongDownArrBangSSet{1cm}
        \\
			u_2                             &\bangLongArrSet{1cm}_S<R>  &s
		\end{array}
	\end{equation*}
\end{lemma}

\begin{proof}
	See \cite[Lemma 2.4]{BucciarelliKesnerRiosViso23}.
	\qed
\end{proof}

\begin{lemma}[Diamond of $\bangArrSet_S<d!>$]
	\label{lem:->S<d!> Local Confluence}
	Let $t, u_1, u_2 \in \setBangTerms$ with $t \bangArrSet_S<d!>
	u_1$, $t \bangArrSet_S<d!> u_2$ and $u_1 \neq u_2$. Then there is
	$s \in \setBangTerms$ such that the diagram below commutes.
	\begin{equation*}
		\begin{array}{ccc}
			t                               &\bangLongArrSet{1cm}_S<d!> &u_1
        \\
			\bangLongDownArrBangSSet{1cm}   &                           &\bangLongDownArrBangSSet{1cm}
        \\
			u_2                             &\bangLongArrSet{1cm}_S<d!> &s
		\end{array}
	\end{equation*}
\end{lemma}
\begin{proof}
    See \cite[Lemma 2.4]{BucciarelliKesnerRiosViso23}, taking $p_1 =
    p_2 = \bangSymbBang$.
	\qed
\end{proof}

\begin{corollary}[Confluence of $\bangArrSet_S<d!>$]
    \label{lem:t->*S<d!>u1andt->*S<d!>u2=>u1->*S<d!>sandu2->*S<d!>s}
    Let $t, u_1, u_2 \in \setBangTerms$ with $t \bangArrSet*_S<d!>
    u_1$ and $t \bangArrSet*_S<d!> u_2$. Then there exists $s \in
    \setBangTerms$ such that the diagram below commutes.
    \begin{equation*}
        \begin{array}{ccc}
            t                               &\bangLongArrSet{1cm}*_S<d!>    &u_1
        \\
            \bangLongDownArrsBangSSet{1cm}  &                               &\bangLongDownArrsBangSSet{1cm}
        \\
            u_2                             &\bangLongArrSet{1cm}*_S<d!>    &s
        \end{array}
    \end{equation*}
\end{corollary}
\begin{proof} %\deliaLu
    Immediate consequence of the diamond for $\bangArrSet_S<d!>$
    (\Cref{lem:->S<d!> Local Confluence}).
    \qed
\end{proof}

\begin{restatable}[\Surface Diligence]{lemma}{RecSurfaceReorchestration}
    \label{lem:SurfaceReorchestration}%
    Let $t, u \in \setBangTerms$ such that $t \bangArrSet*_S u$. If
    $u$ is a \bangSetSNF<d!>, then $t \bangArrSet*_{Sai} u$.
\end{restatable}
\begin{proof} \label{prf:SurfaceReorchestration}%
    Just apply \Cref{lem: Abstract Implicitation} to $\bangArrSet_S$,
	since its hypotheses are fulfilled
	(\Cref{lem:t->*S<d!>u1andt->*S<d!>u2=>u1->*S<d!>sandu2->*S<d!>s,lem:R=dB/s!andt->S<R>u1andt->S<d!>u2=>u1->S<d!>sandu2->S<R>s}).
	\qed
\end{proof}

\subsection{Full Diligence for \BANGSymb}
\label{prfSec:FullImplicitation}

\begin{definition}
    \label{d:context-reduction}%
\emph{\Full reduction} $\bangArrSet_F<R>$ is extended to contexts
as expected:
\begin{itemize}
\item either the redex occurs in a \emph{subterm} of some context,  
\item or the redex occurs in a   \emph{subcontext}, where, in particular,
  $\bangSymbSubs$-redexes are only defined
as   
$t\esub{x}{(\oc u)\esub{x_1}{u_1} \ldots \esub{x_i}{\bangFCtxt_i}
\ldots \esub{x_n}{u_n}}$, but not as $t\esub{x}{\bangLCtxt<\oc
\bangFCtxt>}$ or % otherwise the hole would be duplicated
 $\bangFCtxt\esub{x}{\bangLCtxt<\oc u>}$.  
\end{itemize}
\end{definition}
%}

  Examples of the first case are $\Hole\esub{x}{\app{(\abs{w}{w})}{z}}
  \bangArrSet_F<dB> \Hole \esub{x}{w\esub{w}{z}}$, $\Hole \, \der{\oc
    t} \bangArrSet_F<d!> \Hole \, t$ and $\oc (x\esub{x}{\oc
      y}) \Hole\bangArrSet_F<s!> \oc (y) \Hole$, while examples of the
  second case are $(\abs{x}{x}){\Hole}\bangArrSet_F<dB>
  x\esub{x}{\Hole}$ and  $\der{\oc \Hole} \bangArrSet_F<d!> \Hole$ but $\Hole\esub{x}{z}$ 
  but  $x\esub{x}{\oc \Hole}$
  are  $\bangArrSet_F<R>$-irreducible.

\begin{lemma}
    \label{l:stability}
    If $t \bangArrSet_F<R> t'$ and $\bangFCtxt \bangArrSet_F<\rel>
    \bangFCtxt'$, then 
    \begin{enumerate}
        \item $t\isub{x}{u} \bangArrSet_F<R> t'\isub{x}{u}$;
        \item $u\isub{x}{t} \bangArrSet*_F<R> u\isub{x}{t'}$;
    \item $\bangFCtxt<t>  \bangArrSet_F<R> \bangFCtxt<t'>$;
    \item $\bangFCtxt<t>  \bangArrSet_F<R> \bangFCtxt'<t>$. 
%    \giulio{[G: As soon as $\Hole\esub{x}{t} \bangArrSet_F<s!> \Hole$, this point becomes false. The point is that $\Hole\isub{x}{t} = \Hole$ but in general $s\isub{x}{t} \neq s$. If we want to extend $\bangArrSet_F<\rel>$ for contexts to encompass the examples in my comment to \Cref{d:context-reduction}, we should be careful.]}
        \end{enumerate}
\end{lemma}
\begin{proof}
    \begin{enumerate}
\item By straightforward induction on the definition of $t
    \bangArrSet_F<R> t'$.
\item By straightforward induction on $u$.
\item Since $t \bangArrSet_F<R> t'$, there is a full context
    $\bangFCtxt'$ and $r,r' \in \setBangTerms$ such that $t =
    \bangFCtxt'<r> \bangArrSet_F<R> \bangFCtxt'<r'> = t'$ with $r
    \mapsto_\rel r'$. As $\bangFCtxt<\bangFCtxt'>$ is a full context,
    $\bangFCtxt<t> = \bangFCtxt<\bangFCtxt'<r>> \bangArrSet_F<R>
    \bangFCtxt<\bangFCtxt'<r'>> = \bangFCtxt<t'>$.
\item By straightforward induction on $\bangFCtxt \bangArrSet_F<\rel>
    \bangFCtxt'$. \qed
\end{enumerate}

\end{proof}

\begin{lemma}
    \label{lem:R=dB/s!andt->F<R>u1andt->F<d!>u2=>u1->*F<d!>sandu2->F<R>s}
    Let $\rel \in \{\bangSymbBeta, \bangSymbSubs\}$ and $t, u_1, u_2
    \in \setBangTerms$ such that $t \bangArrSet_F<d!> u_1$ and $t
    \bangArrSet_F<\rel> u_2$. Then there exists $s \in \setBangTerms$
    such that the diagram below commutes.
    \begin{equation*}
        \begin{array}{ccc}
            t                               &\bangLongArrSet{1cm}_F<R>  &u_2
        \\
            \bangLongDownArrBangFSet{1cm}   &                           &\bangLongDownArrsBangFSet{1cm}
        \\
            u_1                             &\bangLongArrSet{1cm}_F<R>  &s
        \end{array}
    \end{equation*}
\end{lemma}
% \hideProof{
    \begin{proof}
We analyze different (potentially overlapping) situations that at the
end cover all  possible cases. To close the diagrams we
use~\Cref{l:stability}.

  \begin{enumerate}
  \item $t =    \bangFCtxt<\der{\bangLCtxt<\oc s_1>}>$ and $u_1 =
    \bangFCtxt<\bangLCtxt<s_1>>$  and the step $t \bangArrSet_F<R>
    u_2$ occurs inside $\bangFCtxt$, $\bangLCtxt$, or $s_1$. We
    analyze all the possible cases. 
    \begin{itemize}
    \item[\bltII]  If $\bangFCtxt \bangArrSet_F<R> \bangFCtxt'$, then:
      \begin{equation*}
            \begin{array}{ccc}
                \bangFCtxt<\der{\bangLCtxt<\oc s_1>}>     &\bangLongArrSet{1cm}_F<R> &\bangFCtxt'<\der{\bangLCtxt<\oc s_1>}> 
            \\[0.2cm]
                \bangLongDownArrBangFSet{1cm}                       &                           &\bangLongDownArrBangFSet{1cm}
            \\[0.2cm]
                \bangFCtxt<\bangLCtxt<s_1>>    &\bangLongArrSet{1cm}_F<R> &\bangFCtxt'<\bangLCtxt<s_1>>
            \end{array}
        \end{equation*}

      \item[\bltII]  If $\bangLCtxt \bangArrSet_F<R> \bangLCtxt'$,
      then:
      \begin{equation*}
            \begin{array}{ccc}
                \bangFCtxt<\der{\bangLCtxt<\oc s_1>}>     &\bangLongArrSet{1cm}_F<R> &\bangFCtxt<\der{\bangLCtxt'<\oc s_1>}> 
            \\[0.2cm]
                \bangLongDownArrBangFSet{1cm}                       &                           &\bangLongDownArrBangFSet{1cm}
            \\[0.2cm]
                \bangFCtxt<\bangLCtxt<s_1>>    &\bangLongArrSet{1cm}_F<R> &\bangFCtxt<\bangLCtxt'<s_1>>
            \end{array}
        \end{equation*}

    \item[\bltII]  If $s_1 \bangArrSet_F<R> s'_1$, then:
\begin{equation*}
            \begin{array}{ccc}
                \bangFCtxt<\der{\bangLCtxt<\oc s_1>}>     &\bangLongArrSet{1cm}_F<R> &\bangFCtxt<\der{\bangLCtxt<\oc s'_1>}> 
            \\[0.2cm]
                \bangLongDownArrBangFSet{1cm}                       &                           &\bangLongDownArrBangFSet{1cm}
            \\[0.2cm]
                \bangFCtxt<\bangLCtxt<s_1>>    &\bangLongArrSet{1cm}_F<R> &\bangFCtxt<\bangLCtxt<s'_1>>
            \end{array}
          \end{equation*}
        \end{itemize}

\item    $t = \bangFCtxt<\app{\bangLCtxt<\abs{x}{s_1}>}{s_2}>$ and
    $u_2 = \bangFCtxt<\bangLCtxt<s_1\esub{x}{s_2}>>$  and the step $t
    \bangArrSet_F<d!> u_1$ occurs inside $\bangFCtxt$, $\bangLCtxt$,
    $s_1$, or $s_2$. We analyze all the possible cases.
    \begin{itemize}
    \item[\bltII] If $\bangFCtxt \bangArrSet_F<d!> \bangFCtxt'$, then:
        \begin{equation*}
            \begin{array}{ccc}
                \bangFCtxt<\app{\bangLCtxt<\abs{x}{s_1}>}{s_2}>     &\bangLongArrSet{1cm}_F<dB> &\bangFCtxt<\bangLCtxt<s_1\esub{x}{s_2}>> 
            \\[0.2cm]
                \bangLongDownArrBangFSet{1cm}                       &                           &\bangLongDownArrBangFSet{1cm}
            \\[0.2cm]
                \bangFCtxt'<\app{\bangLCtxt<\abs{x}{s_1}>}{s_2}>    &\bangLongArrSet{1cm}_F<dB> &\bangFCtxt'<\bangLCtxt<s_1\esub{x}{s_2}>>
            \end{array}
        \end{equation*}

    \item[\bltII] If $\bangLCtxt \bangArrSet_F<d!> \bangLCtxt'$, then:
        \begin{equation*}
            \hspace{-1.5cm}
            \begin{array}{ccc}
                \bangFCtxt<\app{\bangLCtxt<\abs{x}{s_1}>}{s_2}>     &\bangLongArrSet{1cm}_F<dB> &\bangFCtxt<\bangLCtxt<s_1\esub{x}{s_2}>>
            \\[0.2cm]
                \bangLongDownArrBangFSet{1cm}                       &                           &\bangLongDownArrBangFSet{1cm}
            \\[0.2cm]
                \bangFCtxt<\app{\bangLCtxt'<\abs{x}{s_1}>}{s_2}>    &\bangLongArrSet{1cm}_F<dB> &\bangFCtxt<\bangLCtxt'<s_1\esub{x}{s_2}>>
            \end{array}
        \end{equation*}

    \item[\bltII] If $s_1 \bangArrSet_F<d!> s'_1$, then:
        \begin{equation*}
            \begin{array}{ccc}
                \bangFCtxt<\app{\bangLCtxt<\abs{x}{s_1}>}{s_2}>     &\bangLongArrSet{1cm}_F<dB> &\bangFCtxt<\bangLCtxt<s_1\esub{x}{s_2}>>
            \\[0.2cm]
                \bangLongDownArrBangFSet{1cm}                       &                           &\bangLongDownArrBangFSet{1cm}
            \\[0.2cm]
                \bangFCtxt<\app{\bangLCtxt<\abs{x}{s'_1}>}{s_2}>    &\bangLongArrSet{1cm}_F<dB> &\bangFCtxt<\bangLCtxt<s'_1\esub{x}{s_2}>>
            \end{array}
        \end{equation*}

    \item[\bltII] If $s_2 \bangArrSet_F<d!> s'_2$, then:
        \begin{equation*}
            \begin{array}{ccc}
                \bangFCtxt<\app{\bangLCtxt<\abs{x}{s_1}>}{s_2}>     &\bangLongArrSet{1cm}_F<dB> &\bangFCtxt<\bangLCtxt<s_1\esub{x}{s_2}>>
            \\[0.2cm]
                \bangLongDownArrBangFSet{1cm}                       &                           &\bangLongDownArrBangFSet{1cm}
            \\[0.2cm]
                \bangFCtxt<\app{\bangLCtxt<\abs{x}{s_1}>}{s'_2}>    &\bangLongArrSet{1cm}_F<dB> &\bangFCtxt<\bangLCtxt<s_1\esub{x}{s'_2}>> 
            \end{array}
        \end{equation*}
\end{itemize}
    \item $t = \bangFCtxt<s_1\esub{x}{\bangLCtxt<\oc s_2>}>$ and $u_2
    = \bangFCtxt<\bangLCtxt<s_1\isub{x}{s_2}>>$  and the step $t
    \bangArrSet_F<d!> u_1$ occurs inside $\bangFCtxt$, $\bangLCtxt$,
    $s_1$ or $s_2$. We analyze all the possible cases.
    \begin{itemize}
    \item[\bltII] If $\bangFCtxt \bangArrSet_F<d!> \bangFCtxt'$, then:
        \begin{equation*}
            \hspace{-1cm}
            \begin{array}{ccc}
                \bangFCtxt<s_1\esub{x}{\bangLCtxt<\oc s_2>}>    &\bangLongArrSet{1cm}_F<s!> &\bangFCtxt<\bangLCtxt<s_1\isub{x}{s_2}>>
            \\[0.2cm]
                \bangLongDownArrBangFSet{1cm}                   &                           &\bangLongDownArrBangFSet{1cm}
            \\[0.2cm]
                \bangFCtxt'<s_1\esub{x}{\bangLCtxt<\oc s_2>}>   &\bangLongArrSet{1cm}_F<s!> &\bangFCtxt'<\bangLCtxt<s_1\isub{x}{s_2}>>
            \end{array}
        \end{equation*}

    \item[\bltII] If $s_1 \bangArrSet_F<d!> s'_1$, then:  
        \begin{equation*}
            \hspace{-1cm}
            \begin{array}{ccc}
                \bangFCtxt<s_1\esub{x}{\bangLCtxt<\oc s_2>}>    &\bangLongArrSet{1cm}_F<s!> &\bangFCtxt<\bangLCtxt<s_1\isub{x}{s_2}>>
            \\[0.2cm]
                \bangLongDownArrBangFSet{1cm}                   &                           &\bangLongDownArrBangFSet{1cm}
            \\[0.2cm]
                \bangFCtxt<s'_1\esub{x}{\bangLCtxt<\oc s_2>}>   &\bangLongArrSet{1cm}_F<s!> &\bangFCtxt<\bangLCtxt<s'_1\isub{x}{s_2}>>
            \end{array}
        \end{equation*}

    \item[\bltII] If $\bangLCtxt \bangArrSet_F<d!> \bangLCtxt'$, then:
        \begin{equation*}
            \hspace{-1cm}
            \begin{array}{ccc}
                \bangFCtxt<s_1\esub{x}{\bangLCtxt<\oc s_2>}>    &\bangLongArrSet{1cm}_F<s!> &\bangFCtxt<\bangLCtxt<s_1\isub{x}{s_2}>>
            \\[0.2cm]
                \bangLongDownArrBangFSet{1cm}                   &                           &\bangLongDownArrBangFSet{1cm}
            \\[0.2cm]
                \bangFCtxt<s_1\esub{x}{\bangLCtxt'<\oc s_2>}>   &\bangLongArrSet{1cm}_F<s!> &\bangFCtxt<\bangLCtxt'<s_1\isub{x}{s_2}>>
            \end{array}
        \end{equation*}

    \item[\bltII] If $s_2 \bangArrSet_F<d!> s'_2$, then:
        \begin{equation*}
            \hspace{-1cm}
            \begin{array}{ccc}
                \bangFCtxt<s_1\esub{x}{\bangLCtxt<\oc s_2>}>    &\bangLongArrSet{1cm}_F<s!> &\bangFCtxt<\bangLCtxt<s_1\isub{x}{s_2}>>
            \\[0.2cm]
                \bangLongDownArrBangFSet{1cm}                   &                           &\bangLongDownArrsBangFSet{1cm}
            \\[0.2cm]
                \bangFCtxt<s_1\esub{x}{\bangLCtxt<\oc s'_2>}>   &\bangLongArrSet{1cm}_F<s!> &\bangFCtxt<\bangLCtxt<s_1\isub{x}{s'_2}>>
            \end{array}
        \end{equation*}
    \qed
    \end{itemize}
\end{enumerate}
\end{proof}
% }

\begin{lemma}[Diamond of $\bangArrSet_F<d!>$]
    \label{lem:Local Confluence ->F<d!>}
    Let $t, u_1, u_2 \in \setBangTerms$.
    If $t \bangArrSet_F<d!>
    u_1$, $t \bangArrSet_F<d!> u_2$ and $u_1 \neq u_2$, then there is
    $s \in \setBangTerms$ such that the diagram below commutes.
    \begin{equation*}
        \begin{array}{ccc}
            t                               &\bangLongArrSet{1cm}_F<d!> &u_2
        \\
            \bangLongDownArrBangFSet{1cm}   &                           &\bangLongDownArrBangFSet{1cm}
        \\
            u_1                             &\bangLongArrSet{1cm}_F<d!> &s
        \end{array}
    \end{equation*}
\end{lemma}
% \hideProof{
    \begin{proof}
As $t \bangArrSet_F<d!> u_2$, then $t = \bangFCtxt<\der{\bangLCtxt<\oc
s_1>}>$ and $u_2 = \bangFCtxt<\bangLCtxt<s_1>>$. Different cases have
to be distinguished, since the reduction step $t \bangArrSet_{F}<d!>
u_1$ may occur inside $\bangFCtxt$, $\bangLCtxt$, or $s_1$; in all
these cases we use \Cref{l:stability} to conclude:
\begin{itemize}
\item[\bltI] If $\bangFCtxt \bangArrSet_{F}<d!> \bangFCtxt'$, then:
    \begin{equation*}
        \hspace{-1cm}
        \begin{array}{ccc}
            \bangFCtxt<\der{\bangLCtxt<\oc s_1>}>
                &\bangLongArrSet{1cm}_F<d!>&
            \bangFCtxt<\bangLCtxt<s_1>>
        \\[0.2cm]
            \bangLongDownArrBangFSet{1cm}
                &  &
            \bangLongDownArrBangFSet{1cm}
        \\[0.2cm]
        \bangFCtxt'<\der{\bangLCtxt<\oc s_1>}>
        & \bangLongArrSet{1cm}_F<d!>&
       \bangFCtxt'<\bangLCtxt<s_1>>
        \end{array}
    \end{equation*}

\item[\bltI] If $\bangLCtxt \bangArrSet_{F}<d!> \bangLCtxt'$, then:
    \begin{equation*}
        \hspace{-1cm}
        \begin{array}{ccc}
            \bangFCtxt<\der{\bangLCtxt<\oc s_1>}>
                & \bangLongArrSet{1cm}_F<d!>&
            \bangFCtxt<\bangLCtxt<s_1>>
        \\[0.2cm]
            \bangLongDownArrBangFSet{1cm}
                &  &
            \bangLongDownArrBangFSet{1cm}
        \\[0.2cm]
        \bangFCtxt<\der{\bangLCtxt'<\oc s_1>}>
        &\bangLongArrSet{1cm}_F<d!>&
       \bangFCtxt<\bangLCtxt'<s_1>>
        \end{array}
    \end{equation*}

\item[\bltI]  If $s_1 \bangArrSet_{F}<d!> s'_1$, then:
    \begin{equation*}
        \hspace{-1cm}
        \begin{array}{ccc}
            \bangFCtxt<\der{\bangLCtxt<\oc s_1>}>
                &\bangLongArrSet{1cm}_F<d!>&
            \bangFCtxt<\bangLCtxt<s_1>>
        \\[0.2cm]
            \bangLongDownArrBangFSet{1cm}
                &  &
            \bangLongDownArrBangFSet{1cm}
        \\[0.2cm]
        \bangFCtxt<\der{\bangLCtxt<\oc s'_1>}>
        &\bangLongArrSet{1cm}_F<d!>&
       \bangFCtxt<\bangLCtxt<s'_1>>
        \end{array}
    \end{equation*}
\end{itemize}
\end{proof}
% }

\begin{corollary}[Confluence of $\bangArrSet_F<d!>$]
    \label{lem:Confluence_Admin_Full}%
    \label{lem:t->*F<d!>u1andt->*F<d!>u2=>u1->*F<d!>sandu2->*F<d!>s}%
    Let $t, u_1, u_2 \in \setBangTerms$ such that $t \bangArrSet*_F<d!>
    u_1$ and $t \bangArrSet*_F<d!> u_2$. 
    Then there is $s \in
    \setBangTerms$ such that the diagram below commutes.
    \begin{equation*}
        \begin{array}{ccc}
            t                               &\bangLongArrSet{1cm}*_F<d!>    &u_1
        \\
            \bangLongDownArrsBangFSet{1cm}  &                               &\bangLongDownArrsBangFSet{1cm}
        \\
            u_2                             &\bangLongArrSet{1cm}*_F<d!>    &s
        \end{array}
    \end{equation*}
\end{corollary}
\begin{proof} 
    Immediate consequence of the diamond for $\bangArrSet_F<d!>$
    (\Cref{lem:Local Confluence ->F<d!>}).
    \qed
\end{proof}

\begin{restatable}[\Full Diligence]{lemma}{RecFullReorchestration}
    \label{lem:FullReorchestration}%
    Let $t, u \in \setBangTerms$ such that $t \bangArrSet*_F u$. If
    $u$ is a \bangSetFNF<d!>, then $t\bangArrSet*_{Fai} u$.
\end{restatable}
\begin{proof} \label{prf:FullReorchestration}%
    Just apply \Cref{lem: Abstract Implicitation} to $\bangArrSet_F$,
	since its hypotheses are fulfilled
	(\Cref{lem:t->*F<d!>u1andt->*F<d!>u2=>u1->*F<d!>sandu2->*F<d!>s,lem:R=dB/s!andt->F<R>u1andt->F<d!>u2=>u1->*F<d!>sandu2->F<R>s}).
    \qed
\end{proof}

\section{Appendix: Proofs of \Cref{sec:CBN_CBV_Embeddings}}

\subsection{The Call-by-Name Calculus \CBNSymb}

\begin{lemma}
    \label{lem:cbn(F<t>) = cbn(F)<cbn(t)>}
    Let $\cbnFCtxt \in \cbnFCtxtSet$ and $t \in \setCbnTerms$, then
    $\cbnToBangAGK{(\cbnFCtxt<t>)} =
    \cbnToBangAGK{\cbnFCtxt}\cbnCtxtPlug{\cbnToBangAGK{t}}$.
\end{lemma}
\begin{proof}
    By straightforward induction on $\cbnFCtxt$.
    \qed
\end{proof}

\RecCbnOneStepSimulation*
% \hideProof{
    \begin{proof}
    By definition, there exist $\cbnFCtxt \in \cbnFCtxtSet$ and $t',
    u' \in \setCbnTerms$ such that $t = \cbnFCtxt<t'>$, $u =
    \cbnFCtxt<u'>$ and $t' \mapstoR u'$ for some $\rel \in
    \{\cbnSymbBeta, \cbnSymbSubs\}$. We distinguish two cases on
    $\rel$:
    \begin{itemize}
    \item[\bltI] $\rel = \cbnSymbBeta$: Then there are $s_1, s_2 \in
        \setCbnTerms$ and $\cbnLCtxt \in \cbnLCtxtSet$ such that $t' =
        \app{\cbnLCtxt<\abs{x}{s_1}>}{s_2}$ and $s' =
        \cbnLCtxt<s_1\esub{x}{s_2}>$. By \Cref{lem:cbn(F<t>) =
        cbn(F)<cbn(t)>}, $\cbnToBangAGK{{t'}} =
        \app{\cbnToBangAGK{\cbnLCtxt}\cbnCtxtPlug{\abs{x}{\cbnToBangAGK{s_1}}}}{\oc\cbnToBangAGK{s_2}}
        \mapstoR[\cbnToBangAGK{\cbnSymbBeta}]
        \cbnToBangAGK{\cbnLCtxt}\cbnCtxtPlug{\cbnToBangAGK{s_1}\esub{x}{\oc\cbnToBangAGK{s_2}}}
        = \cbnToBangAGK{{u'}}$ so that using \Cref{lem:cbn(F<t>) =
        cbn(F)<cbn(t)>} again, $\cbnToBangAGK{t} =
        \cbnToBangAGK{\cbnFCtxt}<\cbnToBangAGK{{t'}}>
        \bangArr_{\cbnToBangAGK{\cbnFCtxt}\bangCtxtPlug{\cbnToBangAGK{\cbnSymbBeta}}}
        \cbnToBangAGK{\cbnFCtxt}<\cbnToBangAGK{{u'}}> =
        \cbnToBangAGK{u}$.

    \item[\bltI] $\rel = \cbnSymbSubs$:  Then there are $s_1, s_2 \in
        \setCbnTerms$ such that $t' = s_1\esub{x}{s_2}$ and $s' =
        s_1\isub{x}{s_2}$. By induction on $s_1$, one has that
        $\cbnToBangAGK{{u'}} =
        \cbnToBangAGK{s_1}\isub{x}{\cbnToBangAGK{s_2}}$ thus
        $\cbnToBangAGK{{t'}} =
        \cbnToBangAGK{s_1}\esub{x}{\oc\cbnToBangAGK{s_2}}
        \mapstoR[\cbnToBangAGK{\cbnSymbSubs}]
        \cbnToBangAGK{s_1}\isub{x}{\cbnToBangAGK{s_2}} =
        \cbnToBangAGK{{u'}}$ so that $\cbnToBangAGK{t} =
        \cbnToBangAGK{\cbnFCtxt}<\cbnToBangAGK{{t'}}>
        \bangArr_{\cbnToBangAGK{\cbnFCtxt}\bangCtxtPlug{\cbnToBangAGK{\cbnSymbSubs}}}
        \cbnToBangAGK{\cbnFCtxt}<\cbnToBangAGK{{u'}}> =
        \cbnToBangAGK{u}$ by  \Cref{lem:cbn(F<t>) = cbn(F)<cbn(t)>}.
        \qed
    \end{itemize}
\end{proof}
% }

\begin{lemma}
    \label{lem: cbn(t{x:=u}) = cbn(t){x:=cbn(u)}} Let $t, u \in
    \setCbnTerms$, then $\cbnToBangAGK{(t\isub{x}{u})} =
    \cbnToBangAGK{t}\isub{x}{\cbnToBangAGK{u}}$.
\end{lemma}
\begin{proof}
    By induction on $t \in \setCbnTerms$.
\end{proof}

\RecCbnOneStepReverseSimulation*
% \hideProof{
    \begin{proof}
Let $t \in \setCbnTerms$, $s'_t, s'_u \in \setBangTerms$ and $\rel'
\in \{\bangSymbBeta, \bangSymbSubs, \bangSymbBang\}$. Let us show the
following property by induction on $\bangFCtxt \in \bangFCtxtSet$:

<< For all $\cbnFCtxt \in \cbnFCtxtSet$ such that $\cbnToBangAGK{t} =
\bangFCtxt<s'_t>$ and $s'_t \mapstoR[\rel'] s'_u$, there exist $s_t,
s_u \in \setCbnTerms$, $\rel \in \{\cbnSymbBeta, \cbnSymbSubs\}$ and
$\cbnFCtxt \in \cbnFCtxtSet$ such that $\cbnToBangAGK{s_t} = s'_t$,
$\cbnToBangAGK{s_u} = s'_u$, $s_t \mapstoR s_u$, $\cbnToBangAGK{\rel}
= \rel'$ and $\cbnToBangAGK{\cbnFCtxt} = \bangFCtxt$. >>

\begin{itemize}
\item[\bltI] $\bangFCtxt = \Hole$: Let $\cbnFCtxt = \Hole$, then
$\cbnToBangAGK{\cbnFCtxt} = \bangFCtxt$. We distinguish three cases on
$\rel'$:
    \begin{itemize}
    \item[\bltII] $\rel = \bangSymbBeta$: Then $s'_t =
        \app{\bangLCtxt<\abs{x}{s'_1}>}{s'_2}$ and $s'_u =
        \bangLCtxt<s'_1\esub{x}{s'_2}>$ for some $\bangLCtxt \in
        \bangLCtxtSet$ and $s'_1, s'_2 \in \setBangTerms$. Since
        $\cbnToBangAGK{t} = s'_t$, then there exists $\cbnLCtxt \in
        \cbnLCtxtSet$ and $s_1, s_2 \in \setCbnTerms$ such that
        $\cbnToBangAGK{\cbnLCtxt} = \bangLCtxt$, $\cbnToBangAGK{s_1} =
        s'_1$ and $\oc\cbnToBangAGK{s_2} = s'_2$. Taking $s_t =
        \app{\cbnLCtxt<\abs{x}{s_1}>}{s_2}$, $s_u =
        \cbnLCtxt<s_1\esub{x}{s_2}>$ and $\rel = \cbnSymbBeta$
        concludes this case since using \Cref{lem:cbn(F<t>) =
        cbn(F)<cbn(t)>}:
        \begin{equation*}
            \begin{array}{c}
                \cbnToBangAGK{s_t}
                \;=\; \app{\cbnToBangAGK{\cbnLCtxt}<\abs{x}{\cbnToBangAGK{s_1}}>}{\oc\cbnToBangAGK{s_2}}
                \;=\; \app{\bangLCtxt<\abs{x}{s'_1}>}{s'_2}
                \;=\; s'_t
        \\
                \cbnToBangAGK{s_u}
                \;=\; \cbnToBangAGK{\cbnLCtxt}<\cbnToBangAGK{s_1}\esub{x}{\oc\cbnToBangAGK{s_2}}>
                \;=\; \bangLCtxt<s'_1\esub{x}{s'_2}>
                \;=\; s'_u
        \\
                s_t
                \;=\; \app{\cbnLCtxt<\abs{x}{s_1}>}{s_2}
                \;\mapstoR\; \cbnLCtxt<s_1\esub{x}{s_2}>
                \;=\; s_u
        \\
                \cbnToBangAGK{\rel}
                \;=\; \cbnToBangAGK{\cbnSymbBeta}
                \;=\; \bangSymbBeta
                \;=\; \rel'
            \end{array}
        \end{equation*}

    \item[\bltII] $\rel = \bangSymbSubs$: Then $s'_t =
        s'_1\esub{x}{s'_2}$ and $s'_u = s'_1\isub{x}{s'_2}$ for some
        $s'_1, s'_2 \in \setBangTerms$. Since $\cbnToBangAGK{t} =
        s'_t$, then there exists $s_1, s_2 \in \setCbnTerms$ such that
        $\cbnToBangAGK{s_1} = s'_1$ and $\oc\cbnToBangAGK{s_2} =
        s'_2$. Taking $s_t = s_1\esub{x}{s_2}$, $s_u =
        s_1\isub{x}{s_2}$ and $\rel = \cbnSymbSubs$ concludes this
        case since:
        \begin{equation*}
            \begin{array}{c}
                \cbnToBangAGK{s_t}
                \;=\; \cbnToBangAGK{s_1}\esub{x}{\oc\cbnToBangAGK{s_2}}
                \;=\; s'_1\esub{x}{s'_2}
                \;=\; s'_t
        \\
                \cbnToBangAGK{s_u}
                \;\stackon{=}{{\scriptsize Lem.\ref{lem: cbn(t{x:=u}) = cbn(t){x:=cbn(u)}}}}\; \cbnToBangAGK{s_1}\isub{x}{\oc\cbnToBangAGK{s_2}}
                \;=\; s'_1\isub{x}{s'_2}
                \;=\; s'_u
        \\
                s_t
                \;=\; s_1\esub{x}{s_2}
                \;\mapstoR\; s_1\isub{x}{s_2}
                \;=\; s_u
        \\
                \cbnToBangAGK{\rel}
                \;=\; \cbnToBangAGK{\cbnSymbSubs}
                \;=\; \bangSymbSubs
                \;=\; \rel'
            \end{array}
        \end{equation*}

    \item[\bltII] $\rel = \bangSymbBang$: Then $\cbnToBangAGK{t} =
        s'_t = \der{\bangLCtxt<\oc s>}$ for some $\bangLCtxt \in
        \bangLCtxtSet$ and $s \in \setBangTerms$ which is impossible
        by definition of the embedding.
    \end{itemize}

\item[\bltI] $\bangFCtxt = \abs{x}{\bangFCtxt'}$: Then
    $\cbnToBangAGK{t} = \abs{x}{\bangFCtxt'<s'_t>}$ thus necessarily
    $t = \abs{x}{t'}$ for some $t' \in \setCbnTerms$ such that
    $\cbnToBangAGK{{t'}} = \bangFCtxt'<s_t>$. By induction on
    $\bangFCtxt'$, there exist $s_t, s_u \in \setCbnTerms$, $\rel \in
    \{\cbnSymbBeta, \cbnSymbSubs\}$ and $\cbnFCtxt' \in \cbnFCtxtSet$
    such that $\cbnToBangAGK{s_t} = s'_t$, $\cbnToBangAGK{s_u} =
    s'_u$, $s_t \mapstoR s_u$, $\cbnToBangAGK{\rel} = \rel'$ and
    $\cbnToBangAGK{{\cbnFCtxt'}} = \bangFCtxt'$. Taking $\cbnFCtxt =
    \abs{x}{\cbnFCtxt'}$ concludes this case since
    $\cbnToBangAGK{\cbnFCtxt} = \abs{x}{\cbnToBangAGK{{\cbnFCtxt'}}} =
    \abs{x}{\bangFCtxt'} = \bangFCtxt$.

\item[\bltI] $\bangFCtxt = \app{\bangFCtxt'}{s'}$: Then
    $\cbnToBangAGK{t} = \app{\bangFCtxt'<s'_t>}{s'}$ thus necessarily
    $t = \app{t'}{s}$ for some $t', s \in \setCbnTerms$ such that
    $\cbnToBangAGK{t'} = \bangFCtxt'<s'_t>$ and $\oc\cbnToBangAGK{s} =
    s'$. By induction on $\bangFCtxt'$, there exist $s_t, s_u \in
    \setCbnTerms$, $\rel \in \{\cbnSymbBeta, \cbnSymbSubs\}$ and
    $\cbnFCtxt' \in \cbnFCtxtSet$ such that $\cbnToBangAGK{s_t} =
    s'_t$, $\cbnToBangAGK{s_u} = s'_u$, $s_t \mapstoR s_u$,
    $\cbnToBangAGK{\rel} = \rel'$ and $\cbnToBangAGK{{\cbnFCtxt'}} =
    \bangFCtxt'$. Taking $\cbnFCtxt = \app[\,]{\cbnFCtxt'}{s}$
    concludes this case since $\cbnToBangAGK{\cbnFCtxt} =
    \app{\cbnToBangAGK{{\cbnFCtxt'}}}{\oc\cbnToBangAGK{s}} =
    \app{\bangFCtxt'}{s'} = \bangFCtxt$.

\item[\bltI] $\bangFCtxt = \app[\,]{s'}{\bangFCtxt'}$: Then
    $\cbnToBangAGK{t} = \app[\,]{s'}{\bangFCtxt'<s'_t>}$ thus
    necessarily $t = \app[\,]{s}{t'}$ for some $s, t' \in
    \setCbnTerms$ such that $\cbnToBangAGK{s} = s'$ and
    $\oc\cbnToBangAGK{t'} = \bangFCtxt'<s'_t>$. Therefore
    $\bangFCtxt'' = \oc\bangFCtxt''$ for some $\bangFCtxt'' \in
    \bangFCtxtSet$ such that $\cbnToBangAGK{t'} = \bangFCtxt''<s'_t>$.
    By induction on $\bangFCtxt''$, there exist $s_t, s_u \in
    \setCbnTerms$, $\rel \in \{\cbnSymbBeta, \cbnSymbSubs\}$ and
    $\cbnFCtxt' \in \cbnFCtxtSet$ such that $\cbnToBangAGK{s_t} =
    s'_t$, $\cbnToBangAGK{s_u} = s'_u$, $s_t \mapstoR s_u$,
    $\cbnToBangAGK{\rel} = \rel'$ and $\cbnToBangAGK{{\cbnFCtxt'}} =
    \bangFCtxt''$. Taking $\cbnFCtxt = \app[\,]{s}{\cbnFCtxt'}$
    concludes this case since $\cbnToBangAGK{\cbnFCtxt} =
    \app{\cbnToBangAGK{s}}{\oc\cbnToBangAGK{{\cbnFCtxt'}}} =
    \app[\,]{s'}{\oc\bangFCtxt''} = \app[\,]{s'}{\bangFCtxt'} =
    \bangFCtxt$.

\item[\bltI] $\bangFCtxt = \bangFCtxt'\esub{x}{s'}$: Then
    $\cbnToBangAGK{t} = \bangFCtxt'<s'_t>\esub{x}{s'}$ thus
    necessarily $t = t'\esub{x}{s}$ for some $t', s \in \setCbnTerms$
    such that $\cbnToBangAGK{t'} = \bangFCtxt'<s'_t>$ and
    $\oc\cbnToBangAGK{s} = s'$. By induction on $\bangFCtxt'$, there
    exist $s_t, s_u \in \setCbnTerms$, $\rel \in \{\cbnSymbBeta,
    \cbnSymbSubs\}$ and $\cbnFCtxt' \in \cbnFCtxtSet$ such that
    $\cbnToBangAGK{s_t} = s'_t$, $\cbnToBangAGK{s_u} = s'_u$, $s_t
    \mapstoR s_u$, $\cbnToBangAGK{\rel} = \rel'$ and
    $\cbnToBangAGK{{\cbnFCtxt'}} = \bangFCtxt'$. Taking $\cbnFCtxt =
    \cbnFCtxt'\esub{x}{s}$ concludes this case since
    $\cbnToBangAGK{\cbnFCtxt} =
    \cbnToBangAGK{{\cbnFCtxt'}}\esub{x}{\oc\cbnToBangAGK{s}} =
    \bangFCtxt'\esub{x}{s'} = \bangFCtxt$.

\item[\bltI] $\bangFCtxt = s'\esub{x}{\bangFCtxt'}$: Then
    $\cbnToBangAGK{t} = s'\esub{x}{\bangFCtxt'<s'_t>}$ thus
    necessarily $t = s\esub{x}{t'}$ for some $s, t' \in \setCbnTerms$
    such that $\cbnToBangAGK{s} = s'$ and $\oc\cbnToBangAGK{t'} =
    \bangFCtxt'<s'_t>$. Therefore $\bangFCtxt'' = \oc\bangFCtxt''$ for
    some $\bangFCtxt'' \in \bangFCtxtSet$ such that $\cbnToBangAGK{t'}
    = \bangFCtxt''<s'_t>$. By induction on $\bangFCtxt''$, there exist
    $s_t, s_u \in \setCbnTerms$, $\rel \in \{\cbnSymbBeta,
    \cbnSymbSubs\}$ and $\cbnFCtxt' \in \cbnFCtxtSet$ such that
    $\cbnToBangAGK{s_t} = s'_t$, $\cbnToBangAGK{s_u} = s'_u$, $s_t
    \mapstoR s_u$, $\cbnToBangAGK{\rel} = \rel'$ and
    $\cbnToBangAGK{{\cbnFCtxt'}} = \bangFCtxt''$. Taking $\cbnFCtxt =
    s\esub{x}{\cbnFCtxt'}$ concludes this case since
    $\cbnToBangAGK{\cbnFCtxt} =
    \cbnToBangAGK{s}\esub{x}{\oc\cbnToBangAGK{{\cbnFCtxt'}}} =
    s'\esub{x}{\oc\bangFCtxt''} = s'\esub{x}{\bangFCtxt'} =
    \bangFCtxt$.

\item[\bltI] $\bangFCtxt = \der{\bangFCtxt'}$: Then $\cbnToBangAGK{t}
    = \der{\bangFCtxt'<s'_t>}$ which is impossible by definition of
    $\cbnToBangAGK{\cdot}$.

\item[\bltI] $\bangFCtxt = \oc\bangFCtxt'$: Then $\cbnToBangAGK{t} =
    \oc\bangFCtxt'<s'_t>$ which is impossible by definition of
    $\cbnToBangAGK{\cdot}$.
\end{itemize}
\end{proof}
% }

\begin{lemma}
    \label{lem: cbn preservation restricted}%
    Let $\cbnECtxtSet \subseteq \cbnFCtxtSet$, $\bangECtxtSet
    \subseteq \bangFCtxtSet$ be two families of contexts such that:
    \begin{enumerate}
    \item Let $\cbnECtxt \in \cbnECtxtSet$, then
        $\cbnToBangAGK{\cbnECtxt} \in \bangECtxtSet$.
    \item Let $\cbnFCtxt \in \cbnFCtxtSet$, if there exists
        $\bangECtxt \in \bangECtxtSet$ such that
        $\cbnToBangAGK{\cbnFCtxt} = \bangECtxt$, then $\cbnFCtxt \in
        \cbnECtxtSet$.
    \end{enumerate}
    Then, the following properties hold:
    \begin{itemize}
    \item[\bltI] \textbf{(Normal Forms)}: $\forall\, t \in
        \setCbnTerms,$ \quad%
        \begin{equation*}
            t \text{ is a \cbnSetENF}
                \quad\Leftrightarrow\quad
            \cbnToBangAGK{t} \text{ is a \bangSetENF}
        \end{equation*}

    \item[\bltI] \textbf{(Stability)}: $\forall\, t, \in
        \setCbnTerms,\; \forall\, u \in \setBangTerms,$
        \begin{equation*}
            \cbnToBangAGK{t} \bangArrSet*_E u'
                \quad \Rightarrow \quad
            \exists\, u \in \setCbnTerms, \; \cbnToBangAGK{u} = u'
        \end{equation*}

    \item[\bltI] \textbf{(Simulation and Reverse Simulation)}:
        $\forall\, t, u \in \setCbnTerms,$
        \begin{equation*}
            t \cbnArrSet*_E u
                \quad \Leftrightarrow \quad
            \cbnToBangAGK{t} \bangArrSet*_E \cbnToBangAGK{u}
        \end{equation*}
        Moreover, the number of $\cbnSymbBeta/\cbnSymbSubs$-steps
        matches the number $\bangSymbBeta/\bangSymbSubs$-steps.
    \end{itemize}
\end{lemma}
% \hideProof{
    \begin{proof}
Let us first show that the following two properties hold:
\begin{itemize}
\item[\bltI] \textbf{(One Step Simulation)}: $\forall\, t, u \in
    \setCbnTerms,$
    \begin{equation*}
        t \cbnArrSet_E u
            \quad \Rightarrow \quad
        \cbnToBangAGK{t} \bangArrSet_E \cbnToBangAGK{u}
    \end{equation*}

\item[\bltI] \textbf{(One Step Reverse Simulation)}: $\forall\, t, \in
    \setCbnTerms,\; \forall\, u' \in \setBangTerms,$
    \begin{equation*}
        \cbnToBangAGK{t} \bangArrSet_E u'
            \quad \Rightarrow \quad
        \exists\, u \in \setCbnTerms, \; t \bangArrSet_E u
        \text{ and } \cbnToBangAGK{u} = u'
    \end{equation*}
\end{itemize}
In particular, $\cbnSymbBeta$-steps (\resp $\cbnSymbSubs$-steps) are
simulated by $\bangSymbBeta$-steps (\resp $\bangSymbSubs$-steps).

~

We distinguish the two cases:
\begin{itemize}
\item[\bltI] \textbf{(One Step Simulation)}: Let $t, u \in
    \setCbnTerms$ such that $t \cbnArrSet_E u$. Thus $t \cbnArr_E<R>
    u$ for some $\cbnECtxt \in \cbnECtxtSet$ and $\rel \in
    \{\cbnSymbBeta, \cbnSymbSubs\}$. Since $\cbnECtxtSet \subseteq
    \cbnFCtxtSet$, then $\cbnToBangAGK{t}
    \bangArr_{\cbnToBangAGK{\cbnECtxt}\cbnCtxtPlug{\cbnToBangAGK{\rel}}}
    \cbnToBangAGK{u}$ using one step \full simulation
    (\Cref{lem:Cbn_Simulation_One_Step}). Since $\cbnECtxt \in
    \cbnECtxtSet$, then by hypothesis $\cbnToBangAGK{\cbnECtxt} \in
    \bangECtxt$ and one finally concludes that $\cbnToBangAGK{t}
    \bangArrSet_E \cbnToBangAGK{u}$.

\item[\bltI] \textbf{(One Step Reverse Simulation)}: Let $t \in
    \setCbnTerms$ and $u' \in \setBangTerms$ such that
    $\cbnToBangAGK{t} \bangArrSet_E u'$. Then $\cbnToBangAGK{t}
    \bangArr_E<R'> u'$ for some $\bangECtxt \in \bangECtxtSet$ and
    $\rel' \in \{\bangSymbBeta, \bangSymbSubs\}$. Since $\bangECtxtSet
    \subseteq \bangFCtxtSet$, then using one step \full reverse
    simulation (\Cref{l:reverse-simulation-cbn}) there exists $u \in
    \setCbnTerms$ such that $\cbnToBangAGK{u} = u'$, $\cbnFCtxt \in
    \cbnFCtxtSet$ with  $\cbnToBangAGK{\cbnFCtxt} = \bangECtxt$, $\rel
    \in \{\cbnSymbBeta, \cbnSymbSubs\}$, $\cbnToBangAGK{\rel} = \rel'$
    and $t \cbnArr_E<R> u$. Since $\cbnToBangAGK{\cbnFCtxt} \in
    \bangECtxtSet$, then by hypothesis $\cbnFCtxt \in \cbnECtxtSet$ so
    that $t \cbnArrSet_E u$.
\end{itemize}
Using these, let us now prove the three expected properties:
\begin{itemize}
\item[\bltI] \textbf{(Normal Forms)}: Direct consequence of simulation
    and reverse simulation.

\item[\bltI] \textbf{(Stability)}: Let $t \in \setCbnTerms$ and $u'
    \in \setBangTerms$ such that $\cbnToBangAGK{t} \bangArrSet*_E u'$.
    Let us proceed by induction on the length of the derivation
    $\cbnToBangAGK{t} \bangArrSet*_E u'$. Cases:
    \begin{itemize}
    \item[\bltII] $\cbnToBangAGK{t} \bangArrSet^0_E u'$: Then $u' =
        \cbnToBangAGK{t}$ and taking $u = t$ trivially concludes this
        case.
    \item[\bltII] $\cbnToBangAGK{t} \bangArrSet_E s' \bangArrSet*_E
        u'$: Using one step reverse simulation, there exists $s \in
        \setCbnTerms$ such that $\cbnToBangAGK{s} = s'$ and thus by
        \ih on $\cbnToBangAGK{s} \bangArrSet*_E u'$, one concludes
        that there exists $u \in \setCbnTerms$ such that
        $\cbnToBangAGK{u} = u'$.
    \end{itemize}

\item[\bltI] \textbf{(Simulation)}: Let $t, u \in \setCbnTerms$ such
    that $t \cbnArrSet*_E u$. We proceed by induction on the length of
    the derivation $t   \cbnArrSet*_E u$. Cases:
    \begin{itemize}
    \item[\bltII] $t \cbnArrSet^0_E u$: Then $t = u$ thus
        $\cbnToBangAGK{t} = \cbnToBangAGK{u}$ and by reflexivity
        $\cbnToBangAGK{t} \bangArrSet*_E \cbnToBangAGK{u}$.
    \item[\bltII] $t \cbnArrSet_E s \cbnArrSet*_E u$: On one hand
        $\cbnToBangAGK{t} \cbnArrSet_E \cbnToBangAGK{s}$ using one
        step simulation. On the other hand, by \ih on $s \cbnArrSet*_E
        u$, one has that $\cbnToBangAGK{s} \bangArrSet*_E
        \cbnToBangAGK{u}$. Finally, one concludes by transitivity that
        $\cbnToBangAGK{t} \bangArrSet*_E \cbnToBangAGK{u}$.
    \end{itemize}

\item[\bltI] \textbf{(Reverse Simulation)}: Let $t, u \in
    \setCbnTerms$ such that $\cbnToBangAGK{t} \bangArrSet*_E
    \cbnToBangAGK{u}$. We proceed by induction on the length of the
    derivation $\cbnToBangAGK{t} \bangArrSet*_E \cbnToBangAGK{u}$.
    Cases:
    \begin{itemize}
    \item[\bltII] $\cbnToBangAGK{t} \bangArrSet^0_E \cbnToBangAGK{u}$:
        Then $\cbnToBangAGK{t} = \cbnToBangAGK{u}$ thus by injectivity
        of $\cbnToBangAGK{(\cdot)}$ $t = u$ and by reflexivity $t
        \cbnArrSet*_E u$.
    \item[\bltII] $\cbnToBangAGK{t} \bangArrSet_E s' \bangArrSet*_E
        \cbnToBangAGK{u}$: On one hand, using one step reverse
        simulation, there exists $s \in \setCbnTerms$ such that $t
        \cbnArrSet_E s$ with $\cbnToBangAGK{s} = s'$. On the other
        hand, by \ih on $\cbnToBangAGK{s} \bangArrSet*_E
        \cbnToBangAGK{u}$, one has that $s \cbnArrSet*_E u$. Finally,
        one concludes by transitivity that $t \cbnArrSet*_E u$.
        \qed
    \end{itemize}
\end{itemize}
\end{proof}
% }

\begin{lemma}
    \label{lem:cbn_full_context_stability}
    Let $t \in \setCbnTerms$, then:%
    \begin{enumerate}
    \item \textbf{(\CBNSymb $\rightarrow$ \BANGSymb)} %
        Let $\cbnFCtxt \in \cbnFCtxtSet$, then
        $\cbnToBangAGK{\cbnFCtxt} \in \bangFCtxtSet$.
    \item \textbf{(\BANGSymb $\rightarrow$ \CBNSymb)} %
        Let $\bangFCtxt \in \bangFCtxtSet$ and $\cbnFCtxt \in
        \cbnFCtxtSet$ such that $\cbnToBangAGK{\cbnFCtxt} =
        \bangFCtxt$, then $\cbnFCtxt \in \cbnFCtxtSet$.
    \end{enumerate}
\end{lemma}
\begin{proof} ~
    \begin{enumerate}
    \item \textbf{(\CBNSymb $\rightarrow$ \BANGSymb)} By definition of
        $\cbnToBangAGK{\cdot}$.

    \item \textbf{(\BANGSymb $\rightarrow$ \CBNSymb)} By hypothesis.
    \end{enumerate}
\end{proof}

\RecCbnPropFull*
\begin{proof}
    Using \Cref{lem: cbn preservation restricted} and
    \Cref{lem:cbn_full_context_stability}.
\end{proof}

\RecCBNCtxtStability*
% \hideProof{
    \begin{proof}
\begin{enumerate}
\item \textbf{(\CBNSymb $\rightarrow$ \BANGSymb)} By induction on
    $\cbnSCtxt \in \cbnSCtxtSet$:
    \begin{itemize}
    \item[\bltI] $\cbnSCtxt = \Hole$: Then $\cbnToBangAGK{\cbnSCtxt} =
        \cbnToBangAGK{\Hole} = \Hole \in \bangSCtxtSet$.

    \item[\bltI] $\cbnSCtxt = \abs{x}{\cbnSCtxt'}$: By \ih,
        $\cbnToBangAGK{{\cbnSCtxt'}} \in \bangSCtxtSet$ thus
        $\cbnToBangAGK{\cbnSCtxt} =
        \abs{x}{\cbnToBangAGK{{\cbnSCtxt'}}} \in
        \abs{x}{\bangSCtxtSet} \subseteq \bangSCtxtSet$.

    \item[\bltI] $\cbnSCtxt = \app[\,]{\cbnSCtxt'}{s}$: By \ih on
        $\cbnSCtxt'$, one has that $\cbnToBangAGK{{\cbnSCtxt'}} \in
        \bangSCtxtSet$ which concludes this case since
        $\cbnToBangAGK{\cbnSCtxt} =
        \app[\,]{\cbnToBangAGK{{\cbnSCtxt'}}}{\oc \cbnToBangAGK{s}}
        \in \app[\,]{\bangSCtxtSet}{\oc \cbnToBangAGK{s}} \subseteq
        \bangSCtxtSet$.

    \item[\bltI] $\cbnSCtxt = \cbnSCtxt'\esub{x}{s}$: By \ih on
        $\cbnSCtxt'$, one has that $\cbnToBangAGK{{\cbnSCtxt'}} \in
        \bangSCtxtSet$ which concludes this case since
        $\cbnToBangAGK{\cbnSCtxt} =
        \cbnToBangAGK{{\cbnSCtxt'}}\esub{x}{\oc \cbnToBangAGK{s}} \in
        \bangSCtxtSet\esub{x}{\oc \cbnToBangAGK{s}} \subseteq
        \bangSCtxtSet$.
    \end{itemize}

\item \textbf{(\BANGSymb $\rightarrow$ \CBNSymb)} %
    Let $\bangSCtxt \in \bangSCtxtSet$ and $\cbnFCtxt \in
    \cbnFCtxtSet$ such that $\cbnToBangAGK{\cbnFCtxt} = \bangSCtxt$.
    By induction on $\bangSCtxt \in \bangSCtxtSet$:
    \begin{itemize}
    \item[\bltI] $\bangSCtxt = \Hole$: Then necessarily $\cbnFCtxt =
        \Hole$ thus $\cbnFCtxt \in \cbnSCtxtSet$.

    \item[\bltI] $\bangSCtxt = \abs{x}{\bangSCtxt'}$: Then necessarily
        $\cbnFCtxt = \abs{x}{\cbnFCtxt'}$ for some $\cbnFCtxt' \in
        \cbnFCtxtSet$ such that $\cbnToBangAGK{{\cbnFCtxt'}} =
        \bangSCtxt'$. By \ih on $\bangSCtxt'$, one has that
        $\cbnFCtxt' \in \cbnSCtxtSet$ so that $\cbnFCtxt =
        \abs{x}{\cbnFCtxt'} \in \abs{x}{\cbnSCtxtSet} \subseteq
        \cbnSCtxtSet$.

    \item[\bltI] $\bangSCtxt = \app[\,]{\bangSCtxt'}{s'}$: Then
        necessarily $\cbnFCtxt = \app[\,]{\cbnFCtxt'}{s}$ for some
        $\cbnFCtxt' \in \cbnFCtxtSet$ and $s \in \setCbnTerms$ such
        that $\cbnToBangAGK{{\cbnFCtxt'}} = \bangSCtxt'$ and
        $\oc\cbnToBangAGK{s} = s'$. By \ih on $\bangSCtxt'$, one has
        that $\cbnFCtxt' \in \cbnSCtxtSet$ so that $\cbnFCtxt =
        \app[\,]{\cbnFCtxt'}{s} \in \app[\,]{\cbnSCtxtSet}{s}
        \subseteq \cbnSCtxtSet$.

    \item[\bltI] $\bangSCtxt = \app[\,]{s'}{\bangSCtxt}$: Then
        necessarily $\cbnFCtxt = \app[\,]{s}{\cbnFCtxt'}$ for some $s
        \in \setCbnTerms$ and $\cbnFCtxt' \in \cbnFCtxtSet$ such that
        $\cbnToBangAGK{s} = s'$ and $\oc\cbnToBangAGK{{\cbnFCtxt'}} =
        \bangSCtxt$ which is impossible since no $\bangSCtxt \in
        \bangSCtxtSet$ can have a hole under a bang.

    \item[\bltI] $\bangSCtxt = \bangSCtxt'\esub{x}{s'}$: Then
        necessarily $\cbnFCtxt = \cbnFCtxt'\esub{x}{s}$ for some
        $\cbnFCtxt' \in \cbnFCtxtSet$ and $s \in \setCbnTerms$ such
        that $\cbnToBangAGK{{\cbnFCtxt'}} = \bangSCtxt'$ and
        $\oc\cbnToBangAGK{s} = s'$. By \ih on $\bangSCtxt'$, one has
        that $\cbnFCtxt' \in \cbnSCtxtSet$ so that $\cbnFCtxt =
        \cbnFCtxt'\esub{x}{s} \in \cbnSCtxtSet\esub{x}{s} \subseteq
        \cbnSCtxtSet$.

    \item[\bltI] $\bangSCtxt = s'\esub{x}{\bangSCtxt'}$: Then
        necessarily $\cbnFCtxt = s\esub{x}{\cbnFCtxt'}$ for some $s
        \in \setCbnTerms$ and $\cbnFCtxt' \in \cbnFCtxtSet$ such that
        $\cbnToBangAGK{s} = s'$ and $\oc\cbnToBangAGK{{\cbnFCtxt'}} =
        \bangSCtxt$ which is impossible since no $\bangSCtxt \in
        \bangSCtxtSet$ can have a hole under a bang.

    \item[\bltI] $\bangSCtxt = \der{\bangSCtxt'}$: Impossible since
        one cannot have $\cbnToBangAGK{\cbnFCtxt} = \der{\bangSCtxt}$
        by definition of the \CBNSymb embedding on contexts.
    \qed
    \end{itemize}
\end{enumerate}
\end{proof}
% }

\RecCbnPropSurface*
\begin{proof} 
    Immediate consequence of \Cref{lem: cbn preservation
    restricted,lem:Cbn Redex Position Stability}.
	\qed
\end{proof}

\subsection{The Call-by-Value Calculus \CBVSymb}

We use here an equivalent presentation of the \emphasis{\CBVSymb
embedding} (\Cref{def:cbvAGK_Embedding}) to improve readability within
the proofs:
	\begin{equation*}
	\begin{array}{rcl}
		\cbvToBangAGK{x}
		&\;\;\coloneqq\;\;& 
		\oc x
		\\
		\cbvToBangAGK{(\abs{x}{t})}
		&\;\;\coloneqq\;\;& 
		\oc\abs{x}{\oc\cbvToBangAGK{t}}
		\\
		\cbvToBangAGK{(\app{t}{u})}
		&\;\;\coloneqq\;\;&
		\left\{\begin{array}{lcl}
			\der{\app[\,]{\bangStrip{\cbvToBangAGK{t}}}{\cbvToBangAGK{u}}}
			&\hspace{0.25cm}&\text{if } \bangBangPred{\cbvToBangAGK{t}}
			\\
			\der{\app[\,]{\der{\cbvToBangAGK{t}}}{\cbvToBangAGK{u}}}
			&\hspace{0.25cm}&\text{otherwise}
		\end{array}\right.
		\\
		\cbvToBangAGK{(t\esub{x}{u})}
		&\;\;\coloneqq\;\;& 
		\cbvToBangAGK{t}\esub{x}{\cbvToBangAGK{u}}.
	\end{array}
\end{equation*}
where the predicate $\bangBangPred{t}$ holds if $t = \bangLCtxt<\oc
u>$, and $\bangStrip{t} \coloneqq \left\{\begin{array}{ll}
\bangLCtxt<u> &\text{if } t = \bangLCtxt<\oc u> \\
	t &\text{otherwise.} \end{array}\right.$

In order to properly extend the \CBVSymb embedding to contexts, we
introduce a refined "super-developped" embedding. We first extend the
$\bangBangPred{\cdot}$ predicate and $\bangStrip{\cdot}$ function.

\begin{definition}
    Let $\bangFCtxt \in \bangFCtxtSet$.

    The context predicate $\bangBangPred{\bangFCtxt}$ holds if
    $\bangFCtxt = \bangLCtxt<\oc\bangFCtxt'>$ or $\bangFCtxt =
    \bangLCtxt_1<\bangLCtxt_2<\oc t>\esub{x}{\bangFCtxt'}>$.

    The context function $\bangStrip{\bangFCtxt}$ is defined as
    follows:
    \begin{equation*}
        \bangStrip{\bangFCtxt} := 
        \left\{\begin{array}{ll}
            \bangLCtxt<\bangFCtxt'>
            &\text{ if } \bangFCtxt = \bangLCtxt<\oc\bangFCtxt'>
        \\
            \bangLCtxt_1<\bangLCtxt_2<t>\esub{x}{\bangFCtxt'}>
            &\text{ if } \bangFCtxt = \bangLCtxt_1<\bangLCtxt_2<\oc t>\esub{x}{\bangFCtxt'}>
        \end{array}\right.
    \end{equation*}
\end{definition}

Let us also introduce the $\cbvFuncPred{\cdot}$ context predicate
which captures contexts with holes placed directly to the left of an
application.

\begin{definition}
    Let $\cbvFCtxt \in \cbvFCtxtSet$. The functional predicate
    $\cbvFuncPred{\cbvFCtxt}$ holds if $\cbvFCtxt =
    \cbvFCtxt'<\app[\,]{\cbvLCtxt}{t}>$ for some $\cbvFCtxt' \in
    \cbvFCtxtSet$, $\cbvLCtxt \in \cbvLCtxtSet$ and $t \in
    \setCbvTerms$. Similarly for $\bangFCtxt \in \bangFCtxtSet$,
    $\bangFuncPred{\bangFCtxt}$ holds if $\bangFCtxt =
    \bangFCtxt'<\app[\,]{\bangLCtxt}{t}>$ for some $\bangFCtxt' \in
    \bangFCtxtSet$, $\bangLCtxt \in \bangLCtxtSet$ and $t \in
    \setBangTerms$.
\end{definition}

This predicates is used to detecte places where super-development must
be used. We now introduce the \CBVSymb super-developped embedding for
contexts.

\begin{definition}
	\label{def:cbvAGK_Embedding_Superdev}
    The super-developped \emphasis{\CBVSymb context embedding}
	$\cbvToBangAGK[*]{\cdot} \colon \setCbvTerms \to \setBangTerms$ is
	defined on contexts as follows:
	\begin{equation*}
		\begin{array}{rcl}
            \cbvToBangAGK[*]{\Hole}
                &\coloneqq& \Hole
        \\
            \cbvToBangAGK[*]{(\abs{x}{\cbvFCtxt})}
                &\coloneqq& \oc\abs{x}{\oc\cbvToBangAGK[*]{\cbvFCtxt}}
        \\
            \cbvToBangAGK[*]{(\app[\,]{\cbvFCtxt}{t})}
                &\coloneqq&
                \left\{\begin{array}{lcl}
                    \der{\app[\,]{\cbvToBangAGK{\cbvFCtxt}}{\cbvToBangAGK{u}}}
                    &\hspace{0.25cm}&\text{if } \cbvFCtxt \in \cbvLCtxtSet
                \\
                    \der{\app[\,]{\bangStrip{\cbvToBangAGK[*]{\cbvFCtxt}}}{\cbvToBangAGK{u}}}
                    &\hspace{0.25cm}&\text{if } \bangBangPred{\cbvToBangAGK{\cbvFCtxt}}
                \\
                    \der{\app[\,]{\der{\cbvToBangAGK[*]{\cbvFCtxt}}}{\cbvToBangAGK{u}}}
                    &\hspace{0.25cm}&\text{otherwise}
                \end{array}\right.
        \\
            \cbvToBangAGK[*]{(\app[\,]{t}{\cbvFCtxt})}
                &\coloneqq&
                \left\{\begin{array}{lcl}
                    \der{\app[\,]{\bangStrip{\cbvToBangAGK{t}}}{\cbvToBangAGK[*]{\cbvFCtxt}}}
                    &\hspace{0.25cm}&\text{if } \bangBangPred{\cbvToBangAGK{t}}
                \\
                    \der{\app[\,]{\der{\cbvToBangAGK{t}}}{\cbvToBangAGK[*]{\cbvFCtxt}}}
                    &\hspace{0.25cm}&\text{otherwise}
                \end{array}\right.
        \\
            \cbvToBangAGK[*]{(\cbvFCtxt\esub{x}{t})}
                &\coloneqq& \cbvToBangAGK[*]{\cbvFCtxt}\esub{x}{\cbvToBangAGK{t}}
        \\
            \cbvToBangAGK[*]{(t\esub{x}{\cbvFCtxt})}
                &\coloneqq& \cbvToBangAGK{t}\esub{x}{\cbvToBangAGK[*]{\cbvFCtxt}}
		\end{array}
	\end{equation*}
\end{definition}

For example, $\cbvToBangAGK[*]{(\app[\,]{\app[\,]{\Hole}{y}}{z})} =
\der{\app[\,]{\der{\der{\app[\,]{\Hole}{\oc y}}}}{\oc z}}$ whereas
$\cbvToBangAGK{(\app[\,]{\app[\,]{\Hole}{y}}{z})} =
\der{\app[\,]{\der{\der{\app[\,]{\der{\Hole}}{\oc y}}}}{\oc z}}$. This
allows to work with embedding in a contextual manner.

\begin{lemma}
    \label{lem:cbvAGK_Contextual_Translation}%
    \label{lem: cbv embedding with contexts}
    Let $\cbvFCtxt \in \cbvFCtxtSet$ and $t \in \setCbvTerms$. Then:
    \begin{equation*}
        \cbvToBangAGK{(\cbvFCtxt<t>)} = 
        \left\{\begin{array}{ll}
            \cbvToBangAGK[*]{\cbvFCtxt}\bangCtxtPlug{\bangStrip{\cbvToBangAGK{t}}}
            &\text{ if } \cbvFuncPred{\cbvFCtxt}
            \text{ and } \bangBangPred{\cbvToBangAGK{t}}
        \\
            \cbvToBangAGK{\cbvFCtxt}\bangCtxtPlug{\cbvToBangAGK{t}}
            &\text{ otherwise}
        \end{array}\right.
    \end{equation*}
\end{lemma}

\begin{proof}
    By induction on $\cbvFCtxt$.
\end{proof}

\begin{definition}
    Let $\cbvFCtxt \in \cbvFCtxtSet$, then $\cbvRmDst{\cbvFCtxt}$ is
    defined as follows:
    \begin{equation*}
        \begin{array}{rcl}
            \cbvRmDst{\Hole} &\coloneqq& \Hole
        \\
            \cbvRmDst{\abs{x}{\cbvFCtxt}}
                &\coloneqq& \abs{x}{\cbvRmDst{\cbvFCtxt}}
        \\[0.3cm]
            \cbvRmDst{\app[\,]{\cbvFCtxt}{t}}
                &\coloneqq& \app[\,]{\cbvRmDst{\cbvFCtxt}}{t}
        \\[0.3cm]
            \cbvRmDst{\app[\,]{t}{\cbvFCtxt}}
                &\coloneqq& \app[\,]{t}{\cbvRmDst{\cbvFCtxt}}
        \\[0.3cm]
            \cbvRmDst{\cbvFCtxt\esub{x}{t}}
                &\coloneqq&
                \left\{\begin{array}{ll}
                    \Hole
                    &\text{ if } \cbvRmDst{\cbvFCtxt} = \Hole
                \\
                    \cbvRmDst{\cbvFCtxt}\esub{x}{t}
                    &\text{ otherwise}
                \end{array}\right.
        \\[0.3cm]
            \cbvRmDst{t\esub{x}{\cbvFCtxt}}
                &\coloneqq& t\esub{x}{\cbvRmDst{\cbvFCtxt}}
        \end{array}
    \end{equation*}
\end{definition}

\begin{lemma}
    \label{lem:cbvAGK_BangPred_on_Isub}%
    Let $t \in \setCbvTerms$ and $v \in \setCbvValues$. Then
    $\bangBangPred{\cbvToBangAGK{t}}$ if and only if
    $\bangBangPred{\cbvToBangAGK{(t\isub{x}{u})}}$.
\end{lemma}
\begin{proof}
By induction on $t \in \setCbvTerms$:
\begin{itemize}
\item[\bltI] $t = y$: Then $\cbvToBangAGK{t} = \oc y$ thus
    $\bangBangPred{\cbvToBangAGK{t}}$. We distinguish two cases:
    \begin{itemize}
    \item[\bltII] $x = y$: Then $t\isub{x}{v} = v$ thus
        $\bangBangPred{\cbvToBangAGK{t\isub{x}{v}}}$.

    \item[\bltII] $x \neq y$: Then $t\isub{x}{v} = y$ thus
        $\cbvToBangAGK{(t\isub{x}{v})} = \oc y$ and therefore
        $\bangBangPred{\cbvToBangAGK{(t\isub{x}{v})}}$.
    \end{itemize}

\item[\bltI] $t = \abs{y}{t'}$: Then $\cbvToBangAGK{t} =
    \oc\abs{y}{\cbvToBangAGK{t'}}$ thus
    $\bangBangPred{\cbvToBangAGK{t}}$. Moreover, $t\isub{x}{v} =
    \abs{y}{(t'\isub{x}{v})}$ thus $\cbvToBangAGK{(t\isub{x}{v})} =
    \oc\abs{y}{\oc\cbvToBangAGK{(t'\isub{x}{v})}}$ and therefore
    $\bangBangPred{\cbvToBangAGK{(t\isub{x}{v})}}$.

\item[\bltI] $t = \app{t_1}{t_2}$: We distinguish two cases:
    \begin{itemize}
    \item[\bltII] $\bangBangPred{\cbvToBangAGK{t_1}}$: Then
        $\cbvToBangAGK{t} =
        \der{\app[\,]{\bangStrip{\cbvToBangAGK{t_1}}}{\cbvToBangAGK{t_2}}}$
        thus $\neg\bangBangPred{\cbvToBangAGK{t}}$ and by \ih on
        $t_1$, $\bangBangPred{(t_1\isub{x}{v})}$. Moreover,
        $t\isub{x}{v} = \app{(t_1\isub{x}{v})}{(t_2\isub{x}{v})}$ thus
        $\cbvToBangAGK{t} =
        \der{\app{\bangStrip{\cbvToBangAGK{(t_1\isub{x}{v})}}}{\cbvToBangAGK{(t_2\isub{x}{v})}}}$
        and therefore
        $\neg\bangBangPred{\cbvToBangAGK{(t\isub{x}{v})}}$.

    \item[\bltII] $\neg\bangBangPred{\cbvToBangAGK{t_1}}$: Then
        $\cbvToBangAGK{t} =
        \der{\app[\,]{\der{\cbvToBangAGK{t_1}}}{\cbvToBangAGK{t_2}}}$
        thus $\neg\bangBangPred{\cbvToBangAGK{t}}$ and by \ih on
        $t_1$, $\neg\bangBangPred{(t_1\isub{x}{v})}$. Moreover,
        $t\isub{x}{v} = \app{(t_1\isub{x}{v})}{(t_2\isub{x}{v})}$ thus
        $\cbvToBangAGK{t} =
        \der{\app{\der{\cbvToBangAGK{(t_1\isub{x}{v})}}}{\cbvToBangAGK{(t_2\isub{x}{v})}}}$
        and therefore
        $\neg\bangBangPred{\cbvToBangAGK{(t\isub{x}{v})}}$.
    \end{itemize}

\item[\bltI] $t = t_1\esub{y}{t_2}$: Then $\cbvToBangAGK{t} =
    \cbvToBangAGK{t_1}\esub{y}{\cbvToBangAGK{t_2}}$. We distinguish
    two cases:
    \begin{itemize}
    \item[\bltII] $\bangBangPred{\cbvToBangAGK{t_1}}$: By \ih on
        $t_1$, one has that
        $\bangBangPred{\cbvToBangAGK{(t_1\isub{x}{v})}}$ thus
        $\bangBangPred{\cbvToBangAGK{(t\isub{x}{v})}}$.

    \item[\bltII] $\neg\bangBangPred{\cbvToBangAGK{t_1}}$: By \ih on
        $t_1$, one has that
        $\neg\bangBangPred{\cbvToBangAGK{(t_1\isub{x}{v})}}$ thus
        $\neg\bangBangPred{\cbvToBangAGK{(t\isub{x}{v})}}$.
    \end{itemize}
\end{itemize}
\end{proof}

\begin{lemma}
    \label{lem:bang_Strip_on_Isub}%
    Let $t, u \in \setBangTerms$ such that $\bangBangPred{t}$, then
    $\bangStrip{t\isub{x}{u}} = \bangStrip{t}\isub{x}{u}$.
\end{lemma}
\begin{proof}
By induction on $t \in \setBangTerms$:
\begin{itemize}
\item[\bltI] $t = \oc t'$: Then $\bangStrip{t\isub{x}{u}} =
    \bangStrip{\oc (t'\isub{x}{u})} = t'\isub{x}{u} = \bangStrip{\oc
    t'}\isub{x}{u} = \bangStrip{t}\isub{x}{u}$.

\item[\bltI] $t = t_1\esub{y}{t_2}$: Then necessarily
    $\bangBangPred{t_1}$ and by \ih on $t_1$, one has that
    $\bangStrip{t_1\isub{x}{u}} = \bangStrip{t_1}\isub{x}{u}$ thus:
    \begin{equation*}
        \begin{array}{rcl}
            \bangStrip{t\isub{x}{u}}
                &=& \bangStrip{(t_1\isub{x}{u})\esub{y}{t_2\isub{x}{u}}}
            \\
                &=& \bangStrip{t_1\isub{x}{u}}\esub{y}{t_2\isub{x}{u}}
            \\
                &=& (\bangStrip{t_1}\isub{x}{u})\esub{y}{t_2\isub{x}{u}}
            \\
                &=& (\bangStrip{t_1}\esub{y}{t_2})\isub{x}{u}
            \\
                &=& \bangStrip{t}\isub{x}{u}
        \end{array}
    \end{equation*}

\item[\bltI] Otherwise: Impossible since it contradicts the hypothesis
    $\bangBangPred{t}$.
\end{itemize}
\end{proof}

\begin{lemma}
    \label{lem:cbvAGK_Translation_on_Isub}%
    Let $t \in \setCbvTerms$ and $v \in \setCbvValues$. Then
    $\cbvToBangAGK{(t\isub{x}{v})} =
    \cbvToBangAGK{t}\isub{x}{\bangStrip{\cbvToBangAGK{v}}}$.
\end{lemma}
\begin{proof}
By induction on $t \in \setCbvTerms$:
\begin{itemize}
\item[\bltI] $t = y$: Then $\cbvToBangAGK{t} = \oc y$ and we
    distinguish two cases:
    \begin{itemize}
    \item[\bltII] $x = y$: Then $t\isub{x}{v} = v$ and we distinguish
        two cases:
        \begin{itemize}
        \item[\bltIII] $v = z$: Then $\cbvToBangAGK{(t\isub{x}{v})} =
            \cbvToBangAGK{z} = \oc z$ and
            $\cbvToBangAGK{t}\isub{x}{\bangStrip{\cbvToBangAGK{v}}} =
            (\oc x)\isub{x}{\bangStrip{\oc z}} = (\oc x)\isub{x}{z} =
            \oc z$.

        \item[\bltIII] $v = \abs{z}{t'}$: Then
            $\cbvToBangAGK{(t\isub{x}{v})} =
            \cbvToBangAGK{(\abs{z}{t'})} =
            \oc\abs{z}{\oc\cbvToBangAGK{t'}}$ and
            $\cbvToBangAGK{t}\isub{x}{\bangStrip{\cbvToBangAGK{v}}} =
            (\oc
            x)\isub{x}{\bangStrip{\oc\abs{z}{\oc\cbvToBangAGK{t'}}}} =
            (\oc x)\isub{x}{\abs{z}{\oc\cbvToBangAGK{t'}}} =
            \oc\abs{z}{\oc\cbvToBangAGK{t'}}$.
        \end{itemize}

    \item[\bltII] $x \neq y$: Then $t\isub{x}{v} = y$ thus
        $\cbvToBangAGK{(t\isub{x}{v})} = \oc y$ and
        $\cbvToBangAGK{t}\isub{x}{\bangStrip{\cbvToBangAGK{v}}} = (\oc
        y)\isub{x}{\bangStrip{\cbvToBangAGK{v}}} = \oc y$.
    \end{itemize}

\item[\bltI] $t = \abs{y}{t'}$: By \ih on $t'$, one has that
    $\cbvToBangAGK{(t'\isub{x}{v})} =
    \cbvToBangAGK{t'}\isub{x}{\bangStrip{\cbvToBangAGK{v}}}$ thus
    \begin{equation*}
        \begin{array}{rcl}
            \cbvToBangAGK{(t\isub{x}{v})}
                &=& \cbvToBangAGK{\abs{y}{(t'\isub{x}{v})}}
            \\
                &=& \oc\abs{y}{\oc\cbvToBangAGK{(t'\isub{x}{v})}}
            \\
                &=& \oc\abs{y}{(\oc\cbvToBangAGK{t'}\isub{x}{\bangStrip{\cbvToBangAGK{v}}})}
            \\
                &=& (\oc\abs{y}{\oc\cbvToBangAGK{t'}})\isub{x}{\bangStrip{\cbvToBangAGK{v}}}
            \\
                &=& \cbvToBangAGK{t}\isub{x}{\bangStrip{\cbvToBangAGK{v}}}
        \end{array}
    \end{equation*}

\item[\bltI] $t = \app{t_1}{t_2}$: By \ih on $t_1$ and $t_2$, one has
    that $\cbvToBangAGK{(t_1\isub{x}{v})} =
    \cbvToBangAGK{t_1}\isub{x}{\bangStrip{\cbvToBangAGK{v}}}$ and
    $\cbvToBangAGK{(t_2\isub{x}{v})} =
    \cbvToBangAGK{t_2}\isub{x}{\bangStrip{\cbvToBangAGK{v}}}$. We
    distinguish two cases:
    \begin{itemize}
    \item[\bltII] $\bangBangPred{\cbvToBangAGK{t_1}}$: Then
        $\bangBangPred{\cbvToBangAGK{(t_1\isub{x}{v})}}$ using
        \Cref{lem:cbvAGK_BangPred_on_Isub}, thus:
        \begin{equation*}
            \begin{array}{rcl}
                \cbvToBangAGK{(t\isub{x}{v})}
                    &=& \cbvToBangAGK{(\app{(t_1\isub{x}{v})}{(t_2\isub{x}{v})})}
                \\
                    &=& \der{\app{\bangStrip{\cbvToBangAGK{(t_1\isub{x}{v})}}}{\cbvToBangAGK{(t_2\isub{x}{v})}}}
                \\
                    &=& \der{\app{\bangStrip{\cbvToBangAGK{t_1}\isub{x}{\bangStrip{\cbvToBangAGK{v}}}}}{(\cbvToBangAGK{t_2}\isub{x}{\bangStrip{\cbvToBangAGK{v}}})}} \quad \text{(\Cref{lem:bang_Strip_on_Isub})}
                \\
                    &=& (\der{\app[\,]{\bangStrip{\cbvToBangAGK{t_1}}}{\cbvToBangAGK{t_2}}})\isub{x}{\bangStrip{\cbvToBangAGK{v}}}
                \\
                    &=& \cbvToBangAGK{t}\isub{x}{\bangStrip{\cbvToBangAGK{v}}}
            \end{array}
        \end{equation*}

    \item[\bltII] $\neg\bangBangPred{\cbvToBangAGK{t_1}}$: Then
        $\neg\bangBangPred{\cbvToBangAGK{(t_1\isub{x}{v})}}$ using
        \Cref{lem:cbvAGK_BangPred_on_Isub}, thus:
        \begin{equation*}
            \begin{array}{rcl}
                \cbvToBangAGK{(t\isub{x}{v})}
                    &=& \cbvToBangAGK{(\app{(t_1\isub{x}{v})}{(t_2\isub{x}{v})})}
                \\
                    &=& \der{\app{\der{\cbvToBangAGK{(t_1\isub{x}{v})}}}{\cbvToBangAGK{(t_2\isub{x}{v})}}}
                \\
                    &=& \der{\app{\der{\cbvToBangAGK{t_1}\isub{x}{\bangStrip{\cbvToBangAGK{v}}}}}{(\cbvToBangAGK{t_2}\isub{x}{\bangStrip{\cbvToBangAGK{v}}})}}
                \\
                    &=& (\der{\app[\,]{\der{\cbvToBangAGK{t_1}}}{\cbvToBangAGK{t_2}}})\isub{x}{\bangStrip{\cbvToBangAGK{v}}}
                \\
                    &=& \cbvToBangAGK{t}\isub{x}{\bangStrip{\cbvToBangAGK{v}}}
            \end{array}
        \end{equation*}
    \end{itemize}

\item[\bltI] $t = t_1\esub{y}{t_2}$: By \ih on $t_1$ and $t_2$, one
    has that $\cbvToBangAGK{(t_1\isub{x}{v})} =
    \cbvToBangAGK{t_1}\isub{x}{\bangStrip{\cbvToBangAGK{v}}}$ and
    $\cbvToBangAGK{(t_2\isub{x}{v})} =
    \cbvToBangAGK{t_2}\isub{x}{\bangStrip{\cbvToBangAGK{v}}}$, thus
    \begin{equation*}
        \begin{array}{rcl}
            \cbvToBangAGK{(t\isub{x}{v})}
                &=& \cbvToBangAGK{((t_1\isub{x}{v})\esub{y}{(t_2\isub{x}{v})})}
            \\
                &=& \cbvToBangAGK{(t_1\isub{x}{v})}\esub{y}{\cbvToBangAGK{(t_2\isub{x}{v})}}
            \\
                &=& \cbvToBangAGK{t_1}\isub{x}{\bangStrip{\cbvToBangAGK{v}}}\esub{y}{\cbvToBangAGK{t_2}\isub{x}{\bangStrip{\cbvToBangAGK{v}}}}
            \\
                &=& (\cbvToBangAGK{t_1}\esub{y}{\cbvToBangAGK{t_2}})\isub{x}{\bangStrip{\cbvToBangAGK{v}}}
            \\
                &=& \cbvToBangAGK{t}\isub{x}{\bangStrip{\cbvToBangAGK{v}}}
        \end{array}
    \end{equation*}
\end{itemize}
\end{proof}

We rephrase the \CBVSymb one-step simulation lemma by specializing the
resulting \CBVSymb as follows:

\begin{lemma}[\CBVSymb One-Step Simulation]
    \label{lem:cbv detailed one step simulation}
    Let $\rel \in \{\cbvSymbBeta, \cbvSymbSubs\}$. Let $t, u \in
    \setCbvTerms$ such that $\cbvFCtxt<t> \cbvArr_F<R> \cbvFCtxt<u>$.
    Then:
    \begin{itemize}
    \item[\bltI] When $\rel = \cbvSymbBeta$:
        \begin{itemize}
        \item[\bltII] If $\cbvFuncPred{\cbvFCtxt}$ and
            $\bangBangPred{\cbvToBangAGK{u}}$: \quad%
            $\cbvToBangAGK{(\cbvFCtxt<t>)}
            \bangArr_{\bangFCtxt_1\bangCtxtPlug{\cbvToBangAGK{\rel}}}
            \bangArr_{F_2<\bangSymbBang>}
            \bangArr_{F_3<\bangSymbBang>}
            \cbvToBangAGK{(\cbvFCtxt<u>)}$\\
            ~\quad with $\bangFCtxt_1 =
            \cbvToBangAGK{\cbvFCtxt}\bangCtxtPlug{\der{\Hole}}$,
            $\bangFCtxt_2 = \cbvToBangAGK{\cbvFCtxt}$ and
            $\bangFCtxt_3 = \cbvToBangAGK[*]{\cbvRmDst{\cbvFCtxt}}$.

        \item[\bltII] Otherwise: \quad%
            $\cbvToBangAGK{(\cbvFCtxt<t>)}
            \bangArr_{\bangFCtxt_1\bangCtxtPlug{\cbvToBangAGK{\rel}}}
            \bangArr_{F_2<\bangSymbBang>}
            \cbvToBangAGK{(\cbvFCtxt<u>)}$\\
            ~\quad with $\bangFCtxt_1 =
            \cbvToBangAGK{\cbvFCtxt}\bangCtxtPlug{\der{\Hole}}$ and
            $\bangFCtxt_2 = \cbvToBangAGK{\cbvFCtxt}$.
        \end{itemize}
    \item[\bltI] When $\rel = \cbvSymbSubs$:
        \begin{itemize}
        \item[\bltII] $\cbvFuncPred{\cbvFCtxt}$ and
            $\bangBangPred{\cbvToBangAGK{u}}$: \quad%
            $\cbvToBangAGK{(\cbvFCtxt<t>)}
            \bangArr_{\bangFCtxt\bangCtxtPlug{\cbvToBangAGK{\rel}}}
            \cbvToBangAGK{(\cbvFCtxt<u>)}$ with $\bangFCtxt =
            \cbvToBangAGK[*]{\cbvFCtxt}$.

        \item[\bltII] Otherwise: \quad%
            $\cbvToBangAGK{(\cbvFCtxt<t>)}
            \bangArr_{F<\cbvToBangAGK{\rel}>}
            \cbvToBangAGK{(\cbvFCtxt<u>)}$ with $\bangFCtxt =
            \cbvToBangAGK{\cbvFCtxt}$.
        \end{itemize}
    \end{itemize}
\end{lemma}
% \hideProof{
    \begin{proof}
Let $\rel \in \{\cbvSymbBeta, \cbvSymbSubs\}$, $\cbvFCtxt \in
\cbvFCtxtSet$ and $t, u \in \setCbvTerms$ such that $\cbvFCtxt<t>
\cbvArr_F<R> \cbvFCtxt<u>$. We distinguish two cases:
\begin{itemize}
\item[\bltI] $\rel = \cbvSymbBeta$: Then $t =
    \app{\cbvLCtxt<\abs{x}{s_1}>}{s_2}$ and $u =
    \cbvLCtxt<s_1\esub{x}{s_2}>$ for some $\cbvLCtxt \in \cbvLCtxtSet$
    and $s_1, s_2 \in \setCbvTerms$ and thus $\cbvToBangAGK{t} =
    \der{\app{\cbvToBangAGK{\cbvLCtxt}<\abs{x}{\oc\cbvToBangAGK{s_1}}>}{\cbvToBangAGK{s_2}}}$
    and $\cbvToBangAGK{u} =
    \cbvToBangAGK{\cbvLCtxt}<\cbvToBangAGK{s_1}\esub{x}{\cbvToBangAGK{s_2}}>$.
    We distinguish two cases:
    \begin{itemize}
    \item[\bltII] $\cbvFuncPred{\cbvFCtxt}$ and
        $\bangBangPred{\cbvToBangAGK{u}}$: Then $\cbvFCtxt =
        \cbvFCtxt'<\app[\,]{\cbvLCtxt'}{s}>$ for some $\cbvFCtxt' \in
        \cbvFCtxtSet$, $\cbvLCtxt' \in \cbvLCtxtSet$ and $s \in
        \setCbvTerms$ and thus:
        \begin{equation*}
            \begin{array}{rcl}
                \cbvToBangAGK{(\cbvFCtxt<t>)}
                &=& \cbvToBangAGK{(\cbvFCtxt'<\app{\cbvLCtxt'<t>}{s}>)}
            \\
                &\stackon{=}{{\scriptsize Lem.\ref{lem: cbv embedding with contexts}}}&
                \cbvToBangAGK{{\cbvFCtxt'}}\bangCtxtPlug{\cbvToBangAGK{(\app{\cbvLCtxt'<t>}{s})}}
            \\
                &=& \cbvToBangAGK{{\cbvFCtxt'}}\bangCtxtPlug{\der{\app{\der{\cbvToBangAGK{{\cbvLCtxt'}}\bangCtxtPlug{\cbvToBangAGK{t}}}}{\cbvToBangAGK{s}}}}
            \\
                &=& \cbvToBangAGK{{\cbvFCtxt'}}\bangCtxtPlug{\der{\app{\der{\cbvToBangAGK{{\cbvLCtxt'}}\bangCtxtPlug{\der{\app{\cbvToBangAGK{\cbvLCtxt}<\abs{x}{\oc\cbvToBangAGK{s_1}}>}{\cbvToBangAGK{s_2}}}}}}{\cbvToBangAGK{s}}}}
            \\
                &\bangArr_{\bangFCtxt_1<\bangSymbBeta>}&
                \cbvToBangAGK{{\cbvFCtxt'}}\bangCtxtPlug{\der{\app{\der{\cbvToBangAGK{{\cbvLCtxt'}}\bangCtxtPlug{\der{\cbvToBangAGK{\cbvLCtxt}<(\oc\cbvToBangAGK{s_1})\esub{x}{\cbvToBangAGK{s_2}}>}}}}{\cbvToBangAGK{s}}}}
            \\
                &\bangArr_{\bangFCtxt_2<\bangSymbBang>}&
                \cbvToBangAGK{{\cbvFCtxt'}}\bangCtxtPlug{\der{\app{\der{\cbvToBangAGK{{\cbvLCtxt'}}\bangCtxtPlug{\cbvToBangAGK{\cbvLCtxt}<\cbvToBangAGK{s_1}\esub{x}{\cbvToBangAGK{s_2}}>}}}{\cbvToBangAGK{s}}}}
            \\
                &\bangArr_{\bangFCtxt_3<\bangSymbBang>}&
                \cbvToBangAGK{{\cbvFCtxt'}}\bangCtxtPlug{\der{\app{\cbvToBangAGK{{\cbvLCtxt'}}\bangCtxtPlug{\bangStrip{\cbvToBangAGK{\cbvLCtxt}<\cbvToBangAGK{s_1}\esub{x}{\cbvToBangAGK{s_2}}>}}}{\cbvToBangAGK{s}}}}
            \\
                &=& \cbvToBangAGK{{\cbvFCtxt'}}\bangCtxtPlug{\der{\app{\cbvToBangAGK{{\cbvLCtxt'}}\bangCtxtPlug{\bangStrip{\cbvToBangAGK{u}}}}{\cbvToBangAGK{s}}}}
            \\
                &=& \cbvToBangAGK[*]{\cbvFCtxt}<\bangStrip{\cbvToBangAGK{u}}>
            \\
                &\stackon{=}{{\scriptsize Lem.\ref{lem: cbv embedding with contexts}}}&
                \cbvToBangAGK{(\cbvFCtxt<u>)}
            \end{array}
        \end{equation*}
        where:
        \begin{equation*}
            \begin{array}{l}
                \bangFCtxt_1
                = \cbvToBangAGK{{\cbvFCtxt'}}\bangCtxtPlug{\der{\app{\der{\cbvToBangAGK{{\cbvLCtxt'}}\bangCtxtPlug{\der{\Hole}}}}{\cbvToBangAGK{s}}}}
                = \cbvToBangAGK{\cbvFCtxt}<\der{\Hole}>
            \\
                \bangFCtxt_2
                = \cbvToBangAGK{{\cbvFCtxt'}}\bangCtxtPlug{\der{\app{\der{\cbvToBangAGK{{\cbvLCtxt'}}}}{\cbvToBangAGK{s}}}}
                = \cbvToBangAGK{\cbvFCtxt}
            \\
                \bangFCtxt_3
                = \cbvToBangAGK{{\cbvFCtxt'}}\bangCtxtPlug{\der{\app[\,]{\Hole}{\cbvToBangAGK{s}}}}
                = \cbvToBangAGK{{\cbvFCtxt'}}\bangCtxtPlug{\cbvToBangAGK[*]{(\app[\,]{\Hole}{s})}}
                = \cbvToBangAGK{{\cbvFCtxt'}}\bangCtxtPlug{\cbvToBangAGK[*]{\cbvRmDst{\app[\,]{\cbvLCtxt'}{s}}}} = \cbvToBangAGK[*]{\cbvRmDst{\cbvFCtxt}}
            \end{array}
        \end{equation*}

    \item[\bltII] Otherwise:
        \begin{equation*}
            \begin{array}{rcl}
                \cbvToBangAGK{(\cbvFCtxt<t>)}
                &\stackon{=}{{\scriptsize Lem.\ref{lem: cbv embedding with contexts}}}&
                \cbvToBangAGK{\cbvFCtxt}<\cbvToBangAGK{t}>
            \\
                &=& \cbvToBangAGK{\cbvFCtxt}<\der{\app{\cbvToBangAGK{\cbvLCtxt}<\abs{x}{\oc\cbvToBangAGK{s_1}}>}{\cbvToBangAGK{s_2}}}>
            \\
                &\bangArr_{\bangFCtxt_1<\bangSymbBeta>}&
                \cbvToBangAGK{\cbvFCtxt}<\der{\cbvToBangAGK{\cbvLCtxt}<(\oc\cbvToBangAGK{s_1})\esub{x}{\cbvToBangAGK{s_2}}>}>
            \\
                &\bangArr_{\bangFCtxt_2<\bangSymbBang>}&
                \cbvToBangAGK{\cbvFCtxt}<\cbvToBangAGK{\cbvLCtxt}<\cbvToBangAGK{s_1}\esub{x}{\cbvToBangAGK{s_2}}>>
            \\
                &=&
                \cbvToBangAGK{\cbvFCtxt}<\cbvToBangAGK{u}>
            \\
                &\stackon{=}{{\scriptsize Lem.\ref{lem: cbv embedding with contexts}}}&
                \cbvToBangAGK{(\cbvFCtxt<u>)}
            \end{array}
        \end{equation*}
        where $\bangFCtxt_1 = \cbvToBangAGK{\cbvFCtxt}<\der{\Hole}>$
        and $\bangFCtxt_2 = \cbvToBangAGK{\cbvFCtxt}$.
    \end{itemize}

\item[\bltI] $\rel = \cbvSymbSubs$: Then $t =
    s_1\esub{x}{\cbvLCtxt<v>}$ and $u = \cbvLCtxt<s_1\isub{x}{v}>$ for
    some value $v$, thus $\cbvToBangAGK{t} =
    \cbvToBangAGK{s_1}\esub{x}{\cbvToBangAGK{\cbvLCtxt}<\cbvToBangAGK{v}>}$.
    Since $v$ is value then $\cbvToBangAGK{v} = \oc s'_2$ for some
    $s'_2 \in \setBangTerms$ and by induction on $s_1$, one has that
    $\cbvToBangAGK{u} = \cbvToBangAGK{s_1}\isub{x}{s'_2}$. We
    distinguish two cases:
    \begin{itemize}
    \item[\bltII] $\cbvFuncPred{\cbvFCtxt}$ and
        $\bangBangPred{\cbvToBangAGK{u}}$: Since $\cbvToBangAGK{u} =
        \cbvToBangAGK{s_1}\isub{x}{s'_2}$ then, by induction on $s_1$,
        necessarily $\bangBangPred{\cbvToBangAGK{s_1}}$ so that
        $\bangBangPred{\cbvToBangAGK{t}}$. Finally, since
        $\cbvToBangAGK{s_1}\esub{x}{\cbvToBangAGK{\cbvLCtxt}<\oc
        s'_2>} \mapstoR[\bangSymbSubs]
        \cbvToBangAGK{\cbvLCtxt}<\cbvToBangAGK{s_1}\isub{x}{ s'_2}>$,
        then by induction on $s_1$ and $\cbvLCtxt$, one deduces that
        $\bangStrip{\cbvToBangAGK{s_1}\esub{x}{\cbvToBangAGK{\cbvLCtxt}<\oc
        s'_2>}} \mapstoR[\bangSymbSubs]
        \bangStrip{\cbvToBangAGK{\cbvLCtxt}<\cbvToBangAGK{s_1}\isub{x}{
        s'_2}>}$ thus:
        \begin{equation*}
            \begin{array}{rcl}
                \cbvToBangAGK{(\cbvFCtxt<\bangStrip{t}>)}
                &\stackon{=}{{\scriptsize Lem.\ref{lem: cbv embedding with contexts}}}&
                \cbvToBangAGK[*]{\cbvFCtxt}<\bangStrip{\cbvToBangAGK{t}}>
            \\
                &=& \cbvToBangAGK[*]{\cbvFCtxt}<\bangStrip{\cbvToBangAGK{s_1}\esub{x}{\cbvToBangAGK{\cbvLCtxt}<\oc s'_2>}}>
            \\
                &\bangArr_{\cbvToBangAGK{\cbvFCtxt}<\bangSymbSubs>}&
                \cbvToBangAGK[*]{\cbvFCtxt}<\bangStrip{\cbvToBangAGK{\cbvLCtxt}<\cbvToBangAGK{s_1}\isub{x}{s'_2}>}>
            \\
                &=&
                \cbvToBangAGK[*]{\cbvFCtxt}<\bangStrip{\cbvToBangAGK{u}}>
            \\
                &\stackon{=}{{\scriptsize Lem.\ref{lem: cbv embedding with contexts}}}&
                \cbvToBangAGK{(\cbvFCtxt<u>)}
            \end{array}
        \end{equation*}

    \item[\bltII] Otherwise:
        \begin{equation*}
            \begin{array}{rcl}
                \cbvToBangAGK{(\cbvFCtxt<t>)}
                &\stackon{=}{{\scriptsize Lem.\ref{lem: cbv embedding with contexts}}}&
                \cbvToBangAGK{\cbvFCtxt}<\cbvToBangAGK{t}>
            \\
                &=& \cbvToBangAGK{\cbvFCtxt}<\cbvToBangAGK{s_1}\esub{x}{\cbvToBangAGK{\cbvLCtxt}<\oc s'_2>}>
            \\
                &\bangArr_{\cbvToBangAGK{\cbvFCtxt}<\bangSymbSubs>}&
                \cbvToBangAGK{\cbvFCtxt}<\cbvToBangAGK{\cbvLCtxt}<\cbvToBangAGK{s_1}\isub{x}{s'_2}>>
            \\
                &=&
                \cbvToBangAGK{\cbvFCtxt}<\cbvToBangAGK{u}>
            \\
                &\stackon{=}{{\scriptsize Lem.\ref{lem: cbv embedding with contexts}}}&
                \cbvToBangAGK{(\cbvFCtxt<u>)}
            \end{array}
        \end{equation*}
    \end{itemize}
\end{itemize}
\end{proof}
% }

\begin{lemma}
    \label{lem:cbvAGK_dB_redex_Reverse_Translated}%
    Let $\cbvToBangAGK{t} =
    \bangFCtxt<\app{\bangLCtxt<\abs{x}{s'_1}>}{s'_2}>$ for some
    $\bangFCtxt \in \bangFCtxtSet$, $\bangLCtxt \in \bangLCtxtSet$, $t
    \in \setCbvTerms$, $s_1, s_2 \in \setBangTerms$. Then there exist
    $\cbvFCtxt \in \cbvFCtxtSet$, $\cbvLCtxt \in \cbvLCtxtSet$ and
    $s_1, s_2 \in \setCbvTerms$ such that $t =
    \cbvFCtxt<\app{\cbvLCtxt<\abs{x}{s_1}>}{s_2}>$, $\bangFCtxt =
    \cbvToBangAGK{\cbvFCtxt}\cbvCtxtPlug{\der{\Hole}}$, $\bangLCtxt =
    \cbvToBangAGK{\cbvLCtxt}$, $s'_1 = \oc\cbvToBangAGK{s_1}$ and
    $s'_2 = \cbvToBangAGK{s_2}$.
\end{lemma}
\begin{proof}
We strengthen the induction hypothesis as follows:
\begin{enumerate}
\item Let $\cbvToBangAGK{t} =
    \bangFCtxt<\app{\bangLCtxt<\abs{x}{s'_1}>}{s'_2}>$ for some
    $\bangFCtxt \in \bangFCtxtSet$, $\bangLCtxt \in \bangLCtxtSet$, $t
    \in \setCbvTerms$, $s_1, s_2 \in \setBangTerms$. Then there exist
    $\cbvFCtxt \in \cbvFCtxtSet$, $\cbvLCtxt \in \cbvLCtxtSet$ and
    $s_1, s_2 \in \setCbvTerms$ such that $t =
    \cbvFCtxt<\app{\cbvLCtxt<\abs{x}{s_1}>}{s_2}>$, $\bangFCtxt =
    \cbvToBangAGK{\cbvFCtxt}\cbvCtxtPlug{\der{\Hole}}$, $\bangLCtxt =
    \cbvToBangAGK{\cbvLCtxt}$, $s'_1 = \oc\cbvToBangAGK{s_1}$ and
    $s'_2 = \cbvToBangAGK{s_2}$.

\item Let $\bangBangPred{\cbvToBangAGK{t}}$ and
    $\bangStrip{\cbvToBangAGK{t}} =
    \bangFCtxt<\app{\bangLCtxt<\abs{x}{s'_1}>}{s'_2}>$ for some
    $\bangFCtxt \in \bangFCtxtSet$, $\bangLCtxt \in \bangLCtxtSet$, $t
    \in \setCbvTerms$, $s_1, s_2 \in \setBangTerms$. Then there exist
    $\cbvFCtxt \in \cbvFCtxtSet$, $\cbvLCtxt \in \cbvLCtxtSet$ and
    $s_1, s_2 \in \setCbvTerms$ such that $t =
    \cbvFCtxt<\app{\cbvLCtxt<\abs{x}{s_1}>}{s_2}>$,
    $\bangBangPred{\cbvToBangAGK{\cbvFCtxt}}$, $\bangFCtxt =
    \bangStrip{\cbvToBangAGK{\cbvFCtxt}\cbvCtxtPlug{\der{\Hole}}}$,
    $\bangLCtxt = \cbvToBangAGK{\cbvLCtxt}$, $s'_1 =
    \oc\cbvToBangAGK{s_1}$ and $s'_2 = \cbvToBangAGK{s_2}$.
\end{enumerate}

\noindent $1.$ By induction on $t$:
\begin{itemize}
\item[\bltI] $t = y$: Then $\cbvToBangAGK{t} = \oc y$ which
    contradicts the hypothesis.

\item[\bltI] $t = \abs{y}{t'}$: Then $\cbvToBangAGK{t} =
    \oc\abs{x}{\oc\cbvToBangAGK{t'}}$ and since $\cbvToBangAGK{t} =
    \bangFCtxt<\app{\bangLCtxt<\abs{x}{s'_1}>}{s'_2}>$ then
    necessarily $\bangFCtxt = \abs{y}{\bangFCtxt'}$ for some
    $\bangFCtxt' \in \bangFCtxtSet$ such that $\cbvToBangAGK{t'} =
    \bangFCtxt'<\app{\bangLCtxt<\abs{x}{s'_1}>}{s'_2}>$. By \ih on
    $t'$, there exist $\cbvFCtxt' \in \cbvFCtxtSet$, $\cbvLCtxt \in
    \cbvLCtxtSet$ and $s_1, s_2 \in \setCbvTerms$ such that $t' =
    \cbvFCtxt'<\app{\cbvLCtxt<\abs{x}{s_1}>}{s_2}>$, $\bangFCtxt =
    \cbvToBangAGK{{\cbvFCtxt'}}\cbvCtxtPlug{\der{\Hole}}$, $\bangLCtxt
    = \cbvToBangAGK{\cbvLCtxt}$, $s'_1 = \oc\cbvToBangAGK{s_1}$ and
    $s'_2 = \cbvToBangAGK{s_2}$. We set $\cbvFCtxt =
    \abs{y}{\cbvFCtxt'}$ concluding this cases since
    $\cbvToBangAGK{\cbvFCtxt} =
    \oc\abs{y}{\oc\cbvToBangAGK{{\cbvFCtxt'}}} =
    \oc\abs{y}{\oc\bangFCtxt'\bangCtxtPlug{\der{\Hole}}} =
    \bangFCtxt\bangCtxtPlug{\der{\Hole}}$.

\item[\bltI] $t = \app{t_1}{t_2}$: Let us distinguish two cases:
    \begin{itemize}
    \item[\bltII] $\bangBangPred{\cbvToBangAGK{t_1}}$: Then
        $\cbvToBangAGK{t} =
        \der{\app[\,]{\bangStrip{\cbvToBangAGK{t_1}}}{\cbvToBangAGK{t_2}}}$
        and $\cbvToBangAGK{t_1} = \bangLCtxt<\oc s'>$. Thus, by
        induction on $t_1$, necessarily $t_1 = \cbvLCtxt'<v>$ for some
        $\cbvLCtxt' \in \cbvLCtxtSet$ and $v \in \setCbvValues$. We
        distinguish three cases:
        \begin{itemize}
        \item[\bltIII] $\bangFCtxt = \der{\Hole}$, $\bangLCtxt =
            \cbvToBangAGK{{\cbvLCtxt'}}$, $v = \abs{x}{t'_1}$, $s'_1 =
            \oc\cbvToBangAGK{{t'_1}}$, $s'_2 = \cbvToBangAGK{t_2}$: We
            set $\cbvFCtxt = \Hole$, $\cbvLCtxt = \cbvLCtxt'$, $s_1 =
            t'_1$ and $s_2 = t_2$. Then $t = \app{t_1}{t_2} =
            \app{\cbvLCtxt'<v>}{t_2} =
            \app{\cbvLCtxt'<\abs{x}{t'_1}>}{t_2} =
            \app{\cbvLCtxt<\abs{x}{s_1}>}{s_2} =
            \cbvFCtxt<\app{\cbvLCtxt<\abs{x}{s_1}>}{s_2}>$ and
            $\cbvToBangAGK{\cbvFCtxt}\bangCtxtPlug{\der{\Hole}} =
            \cbvToBangAGK{\Hole}\bangCtxtPlug{\der{\Hole}} =
            \der{\Hole} = \bangFCtxt$.

        \item[\bltIII] $\bangFCtxt =
            \der{\app[\,]{\bangFCtxt'}{\cbvToBangAGK{t_2}}}$ for some
            $\bangFCtxt' \in \bangFCtxtSet$ such that
            $\bangStrip{\cbvToBangAGK{t_1}} =
            \bangFCtxt'<\app{\bangLCtxt<\abs{x}{s'_1}>}{s'_2}>$: By
            \ih on $t_1$, there exist $\cbvFCtxt' \in \cbvFCtxtSet$,
            $\cbvLCtxt \in \cbvLCtxtSet$ and $s_1, s_2 \in
            \setCbvTerms$ such that $t_1 =
            \cbvFCtxt'<\app{\cbvLCtxt<\abs{x}{s_1}>}{s_2}>$,
            $\bangBangPred{\cbvToBangAGK{{\cbvFCtxt'}}}$, $\bangFCtxt'
            =
            \bangStrip{\cbvToBangAGK{{\cbvFCtxt'}}\cbvCtxtPlug{\der{\Hole}}}$,
            $\bangLCtxt = \cbvToBangAGK{\cbvLCtxt}$, $s'_1 =
            \oc\cbvToBangAGK{s_1}$ and $s'_2 = \cbvToBangAGK{s_2}$. We
            set $\cbvFCtxt = \app[\,]{\cbvFCtxt'}{t_2}$ then $t =
            \app{t_1}{t_2} =
            \app{\cbvFCtxt'<\app{\cbvLCtxt<\abs{x}{s_1}>}{s_2}>}{t_2}
            =
            (\app{\cbvFCtxt'}{t_2})\cbvCtxtPlug{\app{\cbvLCtxt<\abs{x}{s_1}>}{s_2}}
            = \cbvFCtxt<\app[\,]{\cbvLCtxt<\abs{x}{s_1}>}{s_2}>$ and
            $\cbvToBangAGK{\cbvFCtxt}\cbvCtxtPlug{\der{\Hole}} =
            \cbvToBangAGK{(\app[\,]{\cbvFCtxt'}{t_2})}\cbvCtxtPlug{\der{\Hole}}
            =
            \der{\app[\,]{\bangStrip{\cbvToBangAGK{{\cbvFCtxt'}}}}{\cbvToBangAGK{t_2}}}\cbvCtxtPlug{\der{\Hole}}
            =
            \der{\app{\bangStrip{\cbvToBangAGK{{\cbvFCtxt'}}\cbvCtxtPlug{\der{\Hole}}}}{\cbvToBangAGK{t_2}}}
            = \der{\app[\,]{\bangFCtxt'}{\cbvToBangAGK{t_2}}} =
            \bangFCtxt$.

        \item[\bltIII] $\bangFCtxt =
            \der{\app[\,]{\bangStrip{\cbvToBangAGK{t_1}}}{\bangFCtxt'}}$
            for some $\bangFCtxt' \in \bangFCtxtSet$ such that
            $\bangStrip{\cbvToBangAGK{t_2}} =
            \bangFCtxt'<\app[\,]{\bangLCtxt<\abs{x}{s'_1}>}{s'_2}>$:
            By \ih on $t_2$, there exist $\cbvFCtxt' \in
            \cbvFCtxtSet$, $\cbvLCtxt \in \cbvLCtxtSet$ and $s_1, s_2
            \in \setCbvTerms$ such that $t_2 =
            \cbvFCtxt'<\app{\cbvLCtxt<\abs{x}{s_1}>}{s_2}>$,
            $\bangFCtxt =
            \cbvToBangAGK{{\cbvFCtxt'}}\cbvCtxtPlug{\der{\Hole}}$,
            $\bangLCtxt = \cbvToBangAGK{\cbvLCtxt}$, $s'_1 =
            \oc\cbvToBangAGK{s_1}$ and $s'_2 = \cbvToBangAGK{s_2}$. We
            set $\cbvFCtxt = \app{t_1}{\cbvFCtxt'}$ so that $t =
            \app{t_1}{t_2} =
            \app{t_1}{\cbvFCtxt'<\app{\cbvLCtxt<\abs{x}{s_1}>}{s_2}>}
            =
            (\app{t_1}{\cbvFCtxt'})\cbvCtxtPlug{\app{\cbvLCtxt<\abs{x}{s_1}>}{s_2}}
            = \cbvFCtxt<\app{\cbvLCtxt<\abs{x}{s_1}>}{s_2}>$ and
            $\cbvToBangAGK{(\cbvFCtxt<\der{\Hole}>)} =
            \cbvToBangAGK{(\app{t_1}{\cbvFCtxt'<\der{\Hole}>})} =
            \der{\app[\,]{\bangStrip{\cbvToBangAGK{t_1}}}{\cbvToBangAGK{\cbvFCtxt'<\der{\Hole}>}}}
            =
            \der{\app[\,]{\bangStrip{\cbvToBangAGK{t_1}}}{\bangFCtxt'}}
            = \bangFCtxt$.
        \end{itemize}

    \item[\bltII] $\neg\bangBangPred{\cbvToBangAGK{t_1}}$: Thus
        $\cbvToBangAGK{t} =
        \der{\app[\,]{\der{\cbvToBangAGK{t_1}}}{\cbvToBangAGK{t_2}}}$
        and since $\cbvToBangAGK{t} =
        \bangFCtxt<\app{\bangLCtxt<\abs{x}{s'_1}>}{s'_2}>$ then
        either:
        \begin{itemize}
        \item[\bltIII] $\bangFCtxt =
            \der{\app[\,]{\der{\bangFCtxt'}}{\cbvToBangAGK{t_2}}}$ for
            some $\bangFCtxt' \in \bangFCtxtSet$ such that
            $\cbvToBangAGK{t_1} =
            \bangFCtxt'<\app{\bangLCtxt<\abs{x}{s'_1}>}{s'_2}>$: By
            \ih on $t_1$, there exist $\cbvFCtxt' \in \cbvFCtxtSet$,
            $\cbvLCtxt \in \cbvLCtxtSet$ and $s_1, s_2 \in
            \setCbvTerms$ such that $t_1 =
            \cbvFCtxt'<\app{\cbvLCtxt<\abs{x}{s_1}>}{s_2}>$,
            $\bangFCtxt =
            \cbvToBangAGK{{\cbvFCtxt'}}\cbvCtxtPlug{\der{\Hole}}$,
            $\bangLCtxt = \cbvToBangAGK{\cbvLCtxt}$, $s'_1 =
            \oc\cbvToBangAGK{s_1}$ and $s'_2 = \cbvToBangAGK{s_2}$. We
            set $\cbvFCtxt = \der{\app[\,]{\der{\cbvFCtxt'}}{t_2}}$
            concluding this cases since $\cbvToBangAGK{\cbvFCtxt} =
            \der{\app[\,]{\der{\cbvToBangAGK{{\cbvFCtxt'}}}}{\cbvToBangAGK{t_2}}}
            =
            \der{\app[\,]{\der{\bangFCtxt'\bangCtxtPlug{\der{\Hole}}}}{\cbvToBangAGK{t_2}}}
            = \bangFCtxt\bangCtxtPlug{\der{\Hole}}$.

        \item[\bltIII] $\bangFCtxt =
            \der{\app[\,]{\der{\cbvToBangAGK{t_1}}}{\bangFCtxt'}}$ for
            some $\bangFCtxt' \in \bangFCtxtSet$ such that
            $\cbvToBangAGK{t_2} =
            \bangFCtxt'<\app{\bangLCtxt<\abs{x}{s'_1}>}{s'_2}>$: By
            \ih on $t_2$, there exist $\cbvFCtxt' \in \cbvFCtxtSet$,
            $\cbvLCtxt \in \cbvLCtxtSet$ and $s_1, s_2 \in
            \setCbvTerms$ such that $t_2 =
            \cbvFCtxt'<\app{\cbvLCtxt<\abs{x}{s_1}>}{s_2}>$,
            $\bangFCtxt =
            \cbvToBangAGK{{\cbvFCtxt'}}\cbvCtxtPlug{\der{\Hole}}$,
            $\bangLCtxt = \cbvToBangAGK{\cbvLCtxt}$, $s'_1 =
            \oc\cbvToBangAGK{s_1}$ and $s'_2 = \cbvToBangAGK{s_2}$. We
            set $\cbvFCtxt =
            \der{\app[\,]{\der{\cbvToBangAGK{t_1}}}{\cbvFCtxt'}}$
            concluding this cases since $\cbvToBangAGK{\cbvFCtxt} =
            \der{\app[\,]{\der{\cbvToBangAGK{t_1}}}{\cbvToBangAGK{{\cbvFCtxt'}}}}
            =
            \der{\app[\,]{\der{\cbvToBangAGK{t_1}}}{\bangFCtxt'\bangCtxtPlug{\der{\Hole}}}}
            = \bangFCtxt\bangCtxtPlug{\der{\Hole}}$.
        \end{itemize}
    \end{itemize}

\item[\bltI] $t = t_1\esub{y}{t_2}$: Then $\cbvToBangAGK{t} =
    \cbvToBangAGK{t_1}\esub{y}{\cbvToBangAGK{t_1}}$ and since
    $\cbvToBangAGK{t} =
    \bangFCtxt<\app{\bangLCtxt<\abs{x}{s'_1}>}{s'_2}>$ either:
    \begin{itemize}
    \item[\bltII] $\bangFCtxt =
        \bangFCtxt'\esub{y}{\cbvToBangAGK{t_2}}$ for some $\bangFCtxt'
        \in \bangFCtxtSet$ such that $\cbvToBangAGK{t_1} =
        \bangFCtxt'<\app{\bangLCtxt<\abs{x}{s'_1}>}{s'_2}>$. By \ih on
        $t_1$, there exist $\cbvFCtxt' \in \cbvFCtxtSet$, $\cbvLCtxt
        \in \cbvLCtxtSet$ and $s_1, s_2 \in \setCbvTerms$ such that
        $t_1 = \cbvFCtxt'<\app{\cbvLCtxt<\abs{x}{s_1}>}{s_2}>$,
        $\bangFCtxt =
        \cbvToBangAGK{{\cbvFCtxt'}}\cbvCtxtPlug{\der{\Hole}}$,
        $\bangLCtxt = \cbvToBangAGK{\cbvLCtxt}$, $s'_1 =
        \oc\cbvToBangAGK{s_1}$ and $s'_2 = \cbvToBangAGK{s_2}$. We set
        $\cbvFCtxt = \cbvFCtxt'\esub{y}{t_2}$ concluding this cases
        since $\cbvToBangAGK{\cbvFCtxt} =
        \cbvToBangAGK{{\cbvFCtxt'}}\esub{y}{\cbvToBangAGK{t_2}} =
        \bangFCtxt'\bangCtxtPlug{\der{\Hole}}\esub{y}{\cbvToBangAGK{t_2}}
        = \bangFCtxt\bangCtxtPlug{\der{\Hole}}$.

    \item[\bltII] $\bangFCtxt =
        \cbvToBangAGK{t_1}\esub{y}{\bangFCtxt'}$ for some $\bangFCtxt'
        \in \bangFCtxtSet$ such that $\cbvToBangAGK{t_2} =
        \bangFCtxt'<\app{\bangLCtxt<\abs{x}{s'_1}>}{s'_2}>$. By \ih on
        $t_2$, there exist $\cbvFCtxt' \in \cbvFCtxtSet$, $\cbvLCtxt
        \in \cbvLCtxtSet$ and $s_1, s_2 \in \setCbvTerms$ such that
        $t_2 = \cbvFCtxt'<\app{\cbvLCtxt<\abs{x}{s_1}>}{s_2}>$,
        $\bangFCtxt =
        \cbvToBangAGK{{\cbvFCtxt'}}\cbvCtxtPlug{\der{\Hole}}$,
        $\bangLCtxt = \cbvToBangAGK{\cbvLCtxt}$, $s'_1 =
        \oc\cbvToBangAGK{s_1}$ and $s'_2 = \cbvToBangAGK{s_2}$. We set
        $\cbvFCtxt = t_1\esub{y}{\cbvFCtxt'}$ concluding this cases
        since $\cbvToBangAGK{\cbvFCtxt} =
        \cbvToBangAGK{{\cbvFCtxt'}}\esub{y}{\cbvToBangAGK{t_2}} =
        \bangFCtxt'\bangCtxtPlug{\der{\Hole}}\esub{y}{\cbvToBangAGK{t_2}}
        = \bangFCtxt\bangCtxtPlug{\der{\Hole}}$.
    \end{itemize}
\end{itemize}

\noindent $2.$ Same principle as Point $1$.
\end{proof}

\begin{lemma}
    \label{lem:cbvAGK_s!_redex_Reverse_Translated}%
    Let $\cbvToBangAGK{t} = \bangFCtxt<s'_1\esub{x}{\bangLCtxt<\oc
    s'_2>}>$ for some $\bangFCtxt \in \bangFCtxtSet$, $\bangLCtxt \in
    \bangLCtxtSet$ and $s'_1, s'_2 \in \setBangTerms$. Then there
    exists $\cbvFCtxt \in \cbvFCtxtSet$, $\cbvLCtxt \in \cbvLCtxtSet$,
    $s \in \setCbvTerms$ and $v \in \setCbvValues$ such that $t =
    \cbvFCtxt<s\esub{x}{\cbvLCtxt<v>}>$, $\bangLCtxt =
    \cbvToBangAGK{\cbvLCtxt}$, $\oc s'_2 = \cbvToBangAGK{v}$ and
    either:
    \begin{itemize}
    \item[\bltII] $\cbvFuncPred{\cbvFCtxt}$ and
        $\bangBangPred{\cbvToBangAGK{s_1}}$ so that $\bangFCtxt =
        \cbvToBangAGK[*]{\cbvFCtxt}$ and $s'_1 =
        \bangStrip{\cbvToBangAGK{s}}$.

    \item[\bltII] $\neg\cbvFuncPred{\cbvFCtxt}$ or
        $\neg\bangBangPred{\cbvToBangAGK{s_1}}$ so that $\bangFCtxt =
        \cbvToBangAGK{\cbvFCtxt}$ and $s'_1 = \cbvToBangAGK{s}$.
    \end{itemize}
\end{lemma}
\begin{proof}
We strengthen the induction hypothesis as follows:
\begin{enumerate}
\item Let $\cbvToBangAGK{t} = \bangFCtxt<s'_1\esub{x}{\bangLCtxt<\oc
    s'_2>}>$ for some $\bangFCtxt \in \bangFCtxtSet$, $\bangLCtxt \in
    \bangLCtxtSet$ and $s'_1, s'_2 \in \setBangTerms$. Then there
    exists $\cbvFCtxt \in \cbvFCtxtSet$, $\cbvLCtxt \in \cbvLCtxtSet$,
    $s \in \setCbvTerms$ and $v \in \setCbvValues$ such that $t =
    \cbvFCtxt<s\esub{x}{\cbvLCtxt<v>}>$, $\bangLCtxt =
    \cbvToBangAGK{\cbvLCtxt}$, $\oc s'_2 = \cbvToBangAGK{v}$ and
    either:
    \begin{itemize}
    \item[\bltII] $\cbvFuncPred{\cbvFCtxt}$ and
        $\bangBangPred{\cbvToBangAGK{s_1}}$ so that $\bangFCtxt =
        \cbvToBangAGK[*]{\cbvFCtxt}$ and $s'_1 =
        \bangStrip{\cbvToBangAGK{s}}$.

    \item[\bltII] $\neg\cbvFuncPred{\cbvFCtxt}$ or
        $\neg\bangBangPred{\cbvToBangAGK{s_1}}$ so that $\bangFCtxt =
        \cbvToBangAGK{\cbvFCtxt}$ and $s'_1 = \cbvToBangAGK{s}$.
    \end{itemize}

\item Let $\bangBangPred{\cbvToBangAGK{t}}$ and
    $\bangStrip{\cbvToBangAGK{t}} =
    \bangFCtxt<s'_1\esub{x}{\bangLCtxt<\oc s'_2>}>$ for some
    $\bangFCtxt \in \bangFCtxtSet$, $\bangLCtxt \in \bangLCtxtSet$ and
    $s'_1, s'_2 \in \setBangTerms$. Then there exists $\cbvFCtxt \in
    \cbvFCtxtSet$, $\cbvLCtxt \in \cbvLCtxtSet$, $s \in \setCbvTerms$
    and $v \in \setCbvValues$ such that $t =
    \cbvFCtxt<s\esub{x}{\cbvLCtxt<v>}>$, $\bangLCtxt =
    \cbvToBangAGK{\cbvLCtxt}$, $\oc s'_2 = \cbvToBangAGK{v}$ with two
    possible subcases:
    \begin{itemize}
    \item[\bltII] $\bangBangPred{\cbvToBangAGK{\cbvFCtxt}}$, with two
    possible subcases:
        \begin{itemize}
        \item[\bltIII] $\cbvFuncPred{\cbvFCtxt}$ and
            $\bangBangPred{\cbvToBangAGK{s_1}}$ so that $\bangFCtxt =
            \bangStrip{\cbvToBangAGK[*]{\cbvFCtxt}}$ and $s'_1 =
            \bangStrip{\cbvToBangAGK{s}}$.

        \item[\bltIII] $\neg\cbvFuncPred{\cbvFCtxt}$ or
            $\neg\bangBangPred{\cbvToBangAGK{s_1}}$ so that
            $\bangFCtxt = \bangStrip{\cbvToBangAGK{\cbvFCtxt}}$ and
            $s'_1 = \cbvToBangAGK{s}$.
        \end{itemize}

    \item[\bltII] $\neg\bangBangPred{\cbvToBangAGK{\cbvFCtxt}}$, thus
        $\cbvFCtxt = \cbvLCtxt'$ (in particular
        $\neg\cbvFuncPred{\cbvFCtxt}$) and
        $\bangBangPred{\cbvToBangAGK{s}}$ so that $\bangFCtxt =
        \cbvToBangAGK{\cbvFCtxt}$ and $s'_1 = \cbvToBangAGK{s}$.
    \end{itemize}
\end{enumerate}

$1.$ By induction on $t$:
\begin{itemize}
\item[\bltI] $t = x$: Then $\cbvToBangAGK{t} = \oc x$ which
    contradicts the hypothesis.

\item[\bltI] $t = \abs{y}{t'}$: Then $\cbvToBangAGK{t} =
    \oc\abs{y}{\oc\cbvToBangAGK{t'}}$ and since $\cbvToBangAGK{t} =
    \bangFCtxt<s'_1\esub{x}{\bangLCtxt<\oc s'_2>}>$ then necessarily
    $\bangFCtxt = \oc\abs{y}{\oc\bangFCtxt'}$ for some $\bangFCtxt'
    \in \bangFCtxtSet$ such that $\cbvToBangAGK{t'} =
    \bangFCtxt'<s'_1\esub{x}{\bangLCtxt<\oc s'_2>}>$. By \ih on $t'$,
    there exists $\cbvFCtxt' \in \cbvFCtxtSet$, $\cbvLCtxt \in
    \cbvLCtxtSet$, $s \in \setCbvTerms$ and $v \in \setCbvValues$ such
    that $t' = \cbvFCtxt'<s\esub{x}{\cbvLCtxt<v>}>$, $\bangLCtxt =
    \cbvToBangAGK{\cbvLCtxt}$, $\oc s'_2 = \cbvToBangAGK{v}$ and
    either:
    \begin{itemize}
    \item[\bltII] $\cbvFuncPred{\cbvFCtxt'}$ and
        $\bangBangPred{\cbvToBangAGK{s_1}}$ so that $\bangFCtxt' =
        \cbvToBangAGK[*]{{\cbvFCtxt'}}$ and $s'_1 =
        \bangStrip{\cbvToBangAGK{s}}$. We set $\cbvFCtxt :=
        \abs{y}{\cbvFCtxt'}$ which concludes this case since from
        $\bangFuncPred{\cbvFCtxt'}$ one deduces that
        $\bangFuncPred{\cbvFCtxt}$ thus $\cbvToBangAGK[*]{\cbvFCtxt} =
        \oc\abs{y}{\oc\cbvToBangAGK[*]{{\cbvFCtxt'}}} =
        \oc\abs{y}{\oc\bangFCtxt'} = \bangFCtxt$.

    \item[\bltII] $\neg\cbvFuncPred{\cbvFCtxt'}$ or
        $\neg\bangBangPred{\cbvToBangAGK{s_1}}$ so that $\bangFCtxt' =
        \cbvToBangAGK{{\cbvFCtxt'}}$ and $s'_1 = \cbvToBangAGK{s}$. We
        set $\cbvFCtxt := \abs{y}{\cbvFCtxt'}$ which concludes this
        case since from $\neg\bangFuncPred{\cbvFCtxt'}$ one deduces
        that $\neg\bangFuncPred{\cbvFCtxt}$ thus
        $\cbvToBangAGK{\cbvFCtxt} =
        \oc\abs{y}{\oc\cbvToBangAGK{{\cbvFCtxt'}}} =
        \oc\abs{y}{\oc\bangFCtxt'} = \bangFCtxt$.
    \end{itemize}
    Moreover $t = \abs{y}{t'} =
    \abs{y}{\cbvFCtxt'<s\esub{x}{\cbvLCtxt<v>}>} =
    (\abs{y}{\cbvFCtxt'})\cbvCtxtPlug{s\esub{x}{\cbvLCtxt<v>}} =
    \cbvFCtxt<s\esub{x}{\cbvLCtxt<v>}>$

\item[\bltI] $t = \app{t_1}{t_2}$: We distinguish two cases:
    \begin{itemize}
    \item[\bltII] $\bangBangPred{\cbvToBangAGK{t_1}}$: Then
        $\cbvToBangAGK{t} =
        \der{\app[\,]{\bangStrip{\cbvToBangAGK{t_1}}}{\cbvToBangAGK{t_2}}}$
        and since $\cbvToBangAGK{t} =
        \bangFCtxt<s'_1\esub{x}{\bangLCtxt<\oc s'_2>}>$ then either:
        \begin{itemize}
        \item[\bltIII] $\bangFCtxt =
            \der{\app[\,]{\bangFCtxt'}{\cbvToBangAGK{t_2}}}$ for some
            $\bangFCtxt' \in \bangFCtxtSet$ such that
            $\bangStrip{\cbvToBangAGK{t_1}} =
            \bangFCtxt'<s'_1\esub{x}{\bangLCtxt<\oc s'_2>}>$: By \ih
            on $t_2$, there exists $\cbvFCtxt' \in \cbvFCtxtSet$,
            $\cbvLCtxt \in \cbvLCtxtSet$, $s \in \setCbvTerms$ and $v
            \in \setCbvValues$ such that $t_1 =
            \cbvFCtxt'<s\esub{x}{\cbvLCtxt<v>}>$, $\bangLCtxt =
            \cbvToBangAGK{\cbvLCtxt}$, $\oc s'_2 = \cbvToBangAGK{v}$
            and either:
            \begin{itemize}
            \item[\bltIV] $\bangBangPred{\cbvToBangAGK{{\cbvFCtxt'}}}$
                and either:
                \begin{enumerate}
                \item[\bltV] $\cbvFuncPred{\cbvFCtxt'}$ and
                    $\bangBangPred{\cbvToBangAGK{s_1}}$ so that
                    $\bangFCtxt' =
                    \bangStrip{\cbvToBangAGK[*]{{\cbvFCtxt'}}}$ and
                    $s'_1 = \bangStrip{\cbvToBangAGK{s}}$: We set
                    $\cbvFCtxt := \app[\,]{\cbvFCtxt'}{t_2}$. From
                    $\bangFuncPred{\cbvFCtxt'}$ one then deduces that
                    $\bangFuncPred{\cbvFCtxt}$. Since
                    $\bangBangPred{\cbvToBangAGK{{\cbvFCtxt'}}}$ then
                    necessarily $\cbvFCtxt' \notin \cbvLCtxtSet$ thus
                    $\cbvToBangAGK[*]{\cbvFCtxt} =
                    \der{\app[\,]{\bangStrip{\cbvToBangAGK[*]{{\cbvFCtxt'}}}}{\cbvToBangAGK{t_2}}}
                    = \der{\app{\bangFCtxt'}{\cbvToBangAGK{t_2}}} =
                    \bangFCtxt$.

                \item[\bltV] $\neg\cbvFuncPred{\cbvFCtxt'}$ or
                    $\neg\bangBangPred{\cbvToBangAGK{s_1}}$ so that
                    $\bangFCtxt' =
                    \bangStrip{\cbvToBangAGK{{\cbvFCtxt'}}}$ and $s'_1
                    = \cbvToBangAGK{s}$: We set $\cbvFCtxt :=
                    \app[\,]{\cbvFCtxt'}{t_2}$ thus
                    $\cbvToBangAGK{\cbvFCtxt} =
                    \der{\app[\,]{\bangStrip{\cbvToBangAGK{{\cbvFCtxt'}}}}{\cbvToBangAGK{t_2}}}
                    = \der{\app{\bangFCtxt'}{\cbvToBangAGK{t_2}}} =
                    \bangFCtxt$. We distinguish two cases:
                    \begin{enumerate}
                    \item[\bltVI] $\neg\cbvFuncPred{\cbvFCtxt'}$:
                        Since
                        $\bangBangPred{\cbvToBangAGK{{\cbvFCtxt'}}}$
                        then necessarily $\cbvFCtxt' \notin
                        \cbvLCtxtSet$ thus from
                        $\neg\bangFuncPred{\cbvFCtxt'}$ one then
                        finally deduce that
                        $\neg\bangFuncPred{\cbvFCtxt}$.

                    \item[\bltVI]
                        $\neg\bangBangPred{\cbvToBangAGK{s_1}}$: Thus
                        concluding this case.
                    \end{enumerate}
                \end{enumerate}

            \item[\bltIV] $\cbvFCtxt' = \cbvLCtxt'$,
                $\bangBangPred{\cbvToBangAGK{s_1}}$ and
                $\neg\cbvFuncPred{\cbvFCtxt'}$ so that $\bangFCtxt' =
                \cbvToBangAGK{{\cbvFCtxt'}}$ and $s'_1 =
                \cbvToBangAGK{s}$: We set $\cbvFCtxt :=
                \app[\,]{\cbvFCtxt'}{t_2}$ and since $\cbvFCtxt' \in
                \cbvLCtxtSet$, then $\cbvFuncPred{\cbvFCtxt}$.
                Finally, since $\cbvFCtxt' \in \cbvLCtxtSet$, then
                $\neg\bangBangPred{\cbvToBangAGK{{\cbvFCtxt'}}}$ thus
                $\cbvToBangAGK[*]{\cbvFCtxt} =
                \der{\app[\,]{\cbvToBangAGK{{\cbvFCtxt'}}}{\cbvToBangAGK{t_2}}}
                = \der{\app[\,]{\bangFCtxt'}{\cbvToBangAGK{t_2}}} =
                \bangFCtxt$.
            \end{itemize}
            Moreover, $t = \app{t_1}{t_2} =
            \app{\cbvFCtxt'<s\esub{x}{\cbvLCtxt<v>}>}{t_2} =
            (\app{\cbvFCtxt'}{t_2})\cbvCtxtPlug{s\esub{x}{\cbvLCtxt<v>}}
            = \cbvFCtxt<s\esub{x}{\cbvLCtxt<v>}>$.

        \item[\bltIII] $\bangFCtxt =
            \der{\app[\,]{\bangStrip{\cbvToBangAGK{t_1}}}{\bangFCtxt'}}$
            for some $\bangFCtxt' \in \bangFCtxtSet$ such that
            $\cbvToBangAGK{t_2} =
            \bangFCtxt'<s'_1\esub{x}{\bangLCtxt<\oc s'_2>}>$: By \ih
            on $t_2$, there exists $\cbvFCtxt' \in \cbvFCtxtSet$,
            $\cbvLCtxt \in \cbvLCtxtSet$, $s \in \setCbvTerms$ and $v
            \in \setCbvValues$ such that $t_2 =
            \cbvFCtxt'<s\esub{x}{\cbvLCtxt<v>}>$, $\bangLCtxt =
            \cbvToBangAGK{\cbvLCtxt}$, $\oc s'_2 = \cbvToBangAGK{v}$
            and either:
            \begin{itemize}
            \item[\bltIV] $\cbvFuncPred{{\cbvFCtxt'}}$ and
                $\bangBangPred{\cbvToBangAGK{s_1}}$ so that
                $\bangFCtxt' = \cbvToBangAGK[*]{{\cbvFCtxt'}}$ and
                $s'_1 = \bangStrip{\cbvToBangAGK{s}}$. From
                $\bangFuncPred{\cbvFCtxt'}$ one deduces that
                $\bangFuncPred{\cbvFCtxt}$. We set $\cbvFCtxt :=
                \app{t_1}{\cbvFCtxt'}$ which concludes this case since
                $\cbvToBangAGK[*]{\cbvFCtxt} =
                \der{\app[\,]{\bangStrip{\cbvToBangAGK{t_1}}}{\cbvToBangAGK[*]{{\cbvFCtxt'}}}}
                =
                \der{\app[\,]{\bangStrip{\cbvToBangAGK{t_1}}}{\bangFCtxt'}}
                = \bangFCtxt$.

            \item[\bltIV] $\neg\cbvFuncPred{{\cbvFCtxt'}}$ or
                $\neg\bangBangPred{\cbvToBangAGK{s_1}}$ so that
                $\bangFCtxt' = \cbvToBangAGK{{\cbvFCtxt'}}$ and $s'_1
                = \cbvToBangAGK{s}$. We set $\cbvFCtxt :=
                \app{t_1}{\cbvFCtxt'}$ thus $\cbvToBangAGK{\cbvFCtxt}
                =
                \der{\app[\,]{\bangStrip{\cbvToBangAGK{t_1}}}{\cbvToBangAGK{{\cbvFCtxt'}}}}
                =
                \der{\app[\,]{\bangStrip{\cbvToBangAGK{t_1}}}{\bangFCtxt'}}
                = \bangFCtxt$. We distinguish two cases:
                \begin{enumerate}
                \item[\bltV] $\neg\cbvFuncPred{{\cbvFCtxt'}}$: One
                    then deduces that $\neg\bangFuncPred{\cbvFCtxt}$
                    concluding this case.

                \item[\bltV] $\neg\bangBangPred{\cbvToBangAGK{s_1}}$:
                    Thus concluding this case.
                \end{enumerate}
            \end{itemize}
            Moreover, $t = \app{t_1}{t_2} =
            \app[\,]{t_1}{\cbvFCtxt'<s\esub{x}{\cbvLCtxt<v>}>} =
            \app[\,]{t_1}{\cbvFCtxt'}\cbvCtxtPlug{s\esub{x}{\cbvLCtxt<v>}}
            = \cbvFCtxt<s\esub{x}{\cbvLCtxt<v>}>$
        \end{itemize}

    \item[\bltII] $\neg\bangBangPred{\cbvToBangAGK{t_1}}$: Then
        $\cbvToBangAGK{t} =
        \der{\app[\,]{\der{\cbvToBangAGK{t_1}}}{\cbvToBangAGK{t_2}}}$
        and since $\cbvToBangAGK{t} =
        \bangFCtxt<s'_1\esub{x}{\bangLCtxt<\oc s'_2>}>$ then either:
        \begin{itemize}
        \item[\bltIII] $\bangFCtxt =
            \der{\app[\,]{\der{\bangFCtxt'}}{\cbvToBangAGK{t_2}}}$ for
            some $\bangFCtxt' \in \bangFCtxtSet$ such that
            $\cbvToBangAGK{t_1} =
            \bangFCtxt'<s'_1\esub{x}{\bangLCtxt<\oc s'_2>}>$: By \ih
            on $t_1$, there exists $\cbvFCtxt' \in \cbvFCtxtSet$,
            $\cbvLCtxt \in \cbvLCtxtSet$, $s \in \setCbvTerms$ and $v
            \in \setCbvValues$ such that $t_1 =
            \cbvFCtxt'<s\esub{x}{\cbvLCtxt<v>}>$, $\bangLCtxt =
            \cbvToBangAGK{\cbvLCtxt}$, $\oc s'_2 = \cbvToBangAGK{v}$
            and either:
            \begin{itemize}
            \item[\bltIV] $\cbvFuncPred{{\cbvFCtxt'}}$ and
                $\bangBangPred{\cbvToBangAGK{s_1}}$ so that
                $\bangFCtxt' = \cbvToBangAGK[*]{{\cbvFCtxt'}}$ and
                $s'_1 = \bangStrip{\cbvToBangAGK{s}}$. From
                $\bangFuncPred{\cbvFCtxt'}$ one deduces that
                $\bangFuncPred{\cbvFCtxt}$. We set $\cbvFCtxt :=
                \app[\,]{\cbvFCtxt'}{t_2}$ which concludes this case
                since $\bangFuncPred{\cbvFCtxt'}$ then $\cbvFCtxt'
                \notin \cbvLCtxtSet$ thus $\cbvToBangAGK[*]{\cbvFCtxt}
                =
                \der{\app[\,]{\der{\cbvToBangAGK[*]{{\cbvFCtxt'}}}}{\cbvToBangAGK{t_2}}}
                =
                \der{\app[\,]{\der{\bangFCtxt'}}{\cbvToBangAGK{t_2}}}
                = \bangFCtxt$.

            \item[\bltIV] $\neg\cbvFuncPred{{\cbvFCtxt'}}$ or
                $\neg\bangBangPred{\cbvToBangAGK{s_1}}$ so that
                $\bangFCtxt' = \cbvToBangAGK{{\cbvFCtxt'}}$ and $s'_1
                = \cbvToBangAGK{s}$. We set $\cbvFCtxt :=
                \app[\,]{\cbvFCtxt'}{t_2}$ thus
                $\cbvToBangAGK{\cbvFCtxt} =
                \der{\app[\,]{\der{\cbvToBangAGK{{\cbvFCtxt'}}}}{\cbvToBangAGK{t_2}}}
                =
                \der{\app[\,]{\der{\bangFCtxt'}}{\cbvToBangAGK{t_2}}}
                = \bangFCtxt$. Finally, we distinguish two cases:
                \begin{enumerate}
                \item[\bltV] $\neg\bangBangPred{\cbvToBangAGK{s_1}}$:
                    Thus concluding this case.

                \item[\bltV] $\neg\cbvFuncPred{{\cbvFCtxt'}}$: Without
                    loss of generality, one can suppose that
                    $\bangBangPred{\cbvToBangAGK{s_1}}$. Suppose by
                    absurd that $\cbvFCtxt' \in \cbvLCtxtSet$, then
                    $\bangBangPred{\cbvToBangAGK{(\cbvFCtxt'<s\esub{x}{\cbvLCtxt<v>}>)}}$
                    contradicting
                    $\neg\bangBangPred{\cbvToBangAGK{t_1}}$. Thus
                    $\cbvFCtxt' \notin \cbvLCtxtSet$ and therefore
                    $\neg\cbvFuncPred{\cbvFCtxt}$ conclusing this
                    case.
                \end{enumerate}
            \end{itemize}
            Moreover $t = \app{t_1}{t_2} =
            \app{\cbvFCtxt'<s\esub{x}{\cbvLCtxt<v>}>}{t_2} =
            (\app[\,]{\cbvFCtxt'}{t_2})\bangCtxtPlug{s\esub{x}{\cbvLCtxt<v>}}
            = \cbvFCtxt<s\esub{x}{\cbvLCtxt<v>}>$.

        \item[\bltIII] $\bangFCtxt =
            \cbvToBangAGK{s_1}\esub{y}{\bangFCtxt'}$ for some
            $\bangFCtxt' \in \bangFCtxtSet$ such that
            $\cbvToBangAGK{t_2} =
            \bangFCtxt'<s'_1\esub{x}{\bangLCtxt<\oc s'_2>}>$: By \ih
            on $t_2$, there exists $\cbvFCtxt' \in \cbvFCtxtSet$,
            $\cbvLCtxt \in \cbvLCtxtSet$, $s \in \setCbvTerms$ and $v
            \in \setCbvValues$ such that $t_2 =
            \cbvFCtxt'<s\esub{x}{\cbvLCtxt<v>}>$, $\bangLCtxt =
            \cbvToBangAGK{\cbvLCtxt}$, $\oc s'_2 = \cbvToBangAGK{v}$
            and either:
            \begin{itemize}
            \item[\bltIV] $\cbvFuncPred{{\cbvFCtxt'}}$ and
                $\bangBangPred{\cbvToBangAGK{s_1}}$ so that
                $\bangFCtxt' = \cbvToBangAGK[*]{{\cbvFCtxt'}}$ and
                $s'_1 = \bangStrip{\cbvToBangAGK{s}}$. From
                $\bangFuncPred{\cbvFCtxt'}$ one deduces that
                $\bangFuncPred{\cbvFCtxt}$. We set $\cbvFCtxt :=
                \cbvFCtxt'\esub{y}{t_2}$ which concludes this case
                since $\cbvToBangAGK[*]{\cbvFCtxt} =
                \cbvToBangAGK{t_1}\esub{y}{\cbvToBangAGK[*]{{\cbvFCtxt'}}}
                = \cbvToBangAGK{t_1}\esub{y}{\bangFCtxt'} =
                \bangFCtxt$.

            \item[\bltIV] $\neg\cbvFuncPred{{\cbvFCtxt'}}$ or
                $\neg\bangBangPred{\cbvToBangAGK{s_1}}$ so that
                $\bangFCtxt' = \cbvToBangAGK{{\cbvFCtxt'}}$ and $s'_1
                = \cbvToBangAGK{s}$. We set $\cbvFCtxt :=
                \cbvFCtxt'\esub{y}{t_2}$ thus
                $\cbvToBangAGK{\cbvFCtxt} =
                \cbvToBangAGK{t_1}\esub{y}{\cbvToBangAGK{{\cbvFCtxt'}}}
                = \cbvToBangAGK{t_1}\esub{y}{\bangFCtxt'} =
                \bangFCtxt$. We distinguish two cases:
                \begin{enumerate}
                \item[\bltV] $\neg\cbvFuncPred{{\cbvFCtxt'}}$: One
                    then deduces that $\neg\bangFuncPred{\cbvFCtxt}$
                    concluding this case.

                \item[\bltV] $\neg\bangBangPred{\cbvToBangAGK{s_1}}$:
                    Thus concluding this case.
                \end{enumerate}
            \end{itemize}
            Moreover, $t = t_1\esub{y}{t_2} =
            t_1\esub{y}{\cbvFCtxt'<s\esub{x}{\cbvLCtxt<v>}>} =
            t_1\esub{y}{\cbvFCtxt'}\cbvCtxtPlug{s\esub{x}{\cbvLCtxt<v>}}
            = \cbvFCtxt<s\esub{x}{\cbvLCtxt<v>}>$
        \end{itemize}
    \end{itemize}

\item[\bltI] $t = t_1\esub{y}{t_2}$: Then $\cbvToBangAGK{t} =
    \cbvToBangAGK{t_1}\esub{y}{\cbvToBangAGK{t_2}}$ and since
    $\cbvToBangAGK{t} = \bangFCtxt<s'_1\esub{x}{\bangLCtxt<\oc
    s'_2>}>$ then either:
    \begin{itemize}
    \item[\bltII] $\bangFCtxt = \Hole$, $x = y$, $\cbvToBangAGK{t_1} =
        s'_1$ and $\cbvToBangAGK{t_2} = \bangLCtxt<\oc s'_2>$: By
        induction on $t_2$, one has that $t_2 = \cbvLCtxt<v>$ for some
        $\cbvLCtxt \in \cbvLCtxtSet$ and $v \in \setCbvValues$ such
        that $\cbvToBangAGK{v} = \oc s'_2$ and
        $\cbvToBangAGK{\cbvLCtxt} = \bangLCtxt$. Taking $\cbvFCtxt =
        \Hole$ concludes this case since $t = s\esub{x}{\cbvLCtxt<v>}
        = \cbvFCtxt<s\esub{x}{\cbvLCtxt<v>}>$,
        $\neg\cbvFuncPred{\cbvFCtxt}$ and $\cbvToBangAGK{\cbvFCtxt} =
        \Hole = \cbvFCtxt$.

    \item[\bltII] $\bangFCtxt =
        \bangFCtxt'\esub{y}{\cbvToBangAGK{t_2}}$ for some $\bangFCtxt'
        \in \bangFCtxtSet$ such that $\cbvToBangAGK{t_1} =
        \bangFCtxt'<s'_1\esub{x}{\bangLCtxt<\oc s'_2>}>$: By \ih on
        $t_1$, there exists $\cbvFCtxt' \in \cbvFCtxtSet$, $\cbvLCtxt
        \in \cbvLCtxtSet$, $s \in \setCbvTerms$ and $v \in
        \setCbvValues$ such that $t_1 =
        \cbvFCtxt'<s\esub{x}{\cbvLCtxt<v>}>$, $\bangLCtxt =
        \cbvToBangAGK{\cbvLCtxt}$, $\oc s'_2 = \cbvToBangAGK{v}$ and
        either:
        \begin{itemize}
        \item[\bltIII] $\cbvFuncPred{{\cbvFCtxt'}}$ and
            $\bangBangPred{\cbvToBangAGK{s_1}}$ so that $\bangFCtxt' =
            \cbvToBangAGK[*]{{\cbvFCtxt'}}$ and $s'_1 =
            \bangStrip{\cbvToBangAGK{s}}$. From
            $\bangFuncPred{\cbvFCtxt'}$ one deduces that
            $\bangFuncPred{\cbvFCtxt}$. We set $\cbvFCtxt :=
            \cbvFCtxt'\esub{y}{t_2}$ which concludes this case since
            $\cbvToBangAGK[*]{\cbvFCtxt} =
            \cbvToBangAGK[*]{{\cbvFCtxt'}}\esub{y}{\cbvToBangAGK{t_2}}
            = \bangFCtxt'\esub{y}{\cbvToBangAGK{t_2}} = \bangFCtxt$.

        \item[\bltIII] $\neg\cbvFuncPred{{\cbvFCtxt'}}$ or
            $\neg\bangBangPred{\cbvToBangAGK{s_1}}$ so that
            $\bangFCtxt' = \cbvToBangAGK{{\cbvFCtxt'}}$ and $s'_1 =
            \cbvToBangAGK{s}$. We set $\cbvFCtxt :=
            \cbvFCtxt'\esub{y}{t_2}$ thus $\cbvToBangAGK{\cbvFCtxt} =
            \cbvToBangAGK{{\cbvFCtxt'}}\esub{y}{\cbvToBangAGK{t_2}} =
            \bangFCtxt'\esub{y}{\cbvToBangAGK{t_2}} = \bangFCtxt$. We
            distinguish two cases:
            \begin{enumerate}
            \item[\bltIV] $\neg\cbvFuncPred{{\cbvFCtxt'}}$: One then
                deduces that $\neg\bangFuncPred{\cbvFCtxt}$ concluding
                this case.

            \item[\bltIV] $\neg\bangBangPred{\cbvToBangAGK{s_1}}$:
                Thus concluding this case.
            \end{enumerate}
        \end{itemize}
        Moreover, $t = t_1\esub{y}{t_2} =
        \cbvFCtxt'<s\esub{x}{\cbvLCtxt<v>}>\esub{y}{t_2} =
        (\cbvFCtxt'\esub{y}{t_2})\cbvCtxtPlug{s\esub{x}{\cbvLCtxt<v>}}
        = \cbvFCtxt<s\esub{x}{\cbvLCtxt<v>}>$.

    \item[\bltII] $\bangFCtxt =
        \cbvToBangAGK{s_1}\esub{y}{\bangFCtxt'}$ for some $\bangFCtxt'
        \in \bangFCtxtSet$ such that $\cbvToBangAGK{t_2} =
        \bangFCtxt'<s'_1\esub{x}{\bangLCtxt<\oc s'_2>}>$: By \ih on
        $t_2$, there exists $\cbvFCtxt' \in \cbvFCtxtSet$, $\cbvLCtxt
        \in \cbvLCtxtSet$, $s \in \setCbvTerms$ and $v \in
        \setCbvValues$ such that $t_2 =
        \cbvFCtxt'<s\esub{x}{\cbvLCtxt<v>}>$, $\bangLCtxt =
        \cbvToBangAGK{\cbvLCtxt}$, $\oc s'_2 = \cbvToBangAGK{v}$ and
        either:
        \begin{itemize}
        \item[\bltIII] $\cbvFuncPred{{\cbvFCtxt'}}$ and
            $\bangBangPred{\cbvToBangAGK{s_1}}$ so that $\bangFCtxt' =
            \cbvToBangAGK[*]{{\cbvFCtxt'}}$ and $s'_1 =
            \bangStrip{\cbvToBangAGK{s}}$. From
            $\bangFuncPred{\cbvFCtxt'}$ one deduces that
            $\bangFuncPred{\cbvFCtxt}$. We set $\cbvFCtxt :=
            \cbvFCtxt'\esub{y}{t_2}$ which concludes this case since
            $\cbvToBangAGK[*]{\cbvFCtxt} =
            \cbvToBangAGK{t_1}\esub{y}{\cbvToBangAGK[*]{{\cbvFCtxt'}}}
            = \cbvToBangAGK{t_1}\esub{y}{\bangFCtxt'} = \bangFCtxt$.

        \item[\bltIII] $\neg\cbvFuncPred{{\cbvFCtxt'}}$ or
            $\neg\bangBangPred{\cbvToBangAGK{s_1}}$ so that
            $\bangFCtxt' = \cbvToBangAGK{{\cbvFCtxt'}}$ and $s'_1 =
            \cbvToBangAGK{s}$. We set $\cbvFCtxt :=
            \cbvFCtxt'\esub{y}{t_2}$ thus $\cbvToBangAGK{\cbvFCtxt} =
            \cbvToBangAGK{t_1}\esub{y}{\cbvToBangAGK{{\cbvFCtxt'}}} =
            \cbvToBangAGK{t_1}\esub{y}{\bangFCtxt'} = \bangFCtxt$. We
            distinguish two cases:
            \begin{enumerate}
            \item[\bltIV] $\neg\cbvFuncPred{{\cbvFCtxt'}}$: One then
                deduces that $\neg\bangFuncPred{\cbvFCtxt}$ concluding
                this case.

            \item[\bltIV] $\neg\bangBangPred{\cbvToBangAGK{s_1}}$:
                Thus concluding this case.
            \end{enumerate}
        \end{itemize}
    \end{itemize}
\end{itemize}

~ 

$2.$ Same principle as Point $1$.
\end{proof}

And similarly for the \CBVSymb one-step reverse simulation property:

\begin{lemma}[\CBVSymb One-Step Reverse Simulation]
    \label{lem:cbv detailed one step reverse simulation}
    Let $t \in \setCbvTerms, u' \in \setBangTerms$ and $\rel' \in
    \{\bangSymbBeta, \bangSymbSubs, \bangSymbBang\}$ where $u'$ is a
    \bangSetFNF<d!>, then%
    \begin{equation*}
        \cbvToBangAGK{t} \bangArr_F<R'>\bangArrSet*_F<d!> u'
        % \;\;\text{and}\;\; \bangPredFNF<d!>{u'}
            \quad \Rightarrow \quad
        \left\{\begin{array}{lr}
            \exists\; u \in \setCbvTerms,
                &\cbvToBangAGK{u} = u'
        \\
            \exists\; \rel \in \{\cbvSymbBeta, \cbvSymbSubs\},
                &\cbvToBangAGK{\rel} = \rel'
        \\
            \exists\; \cbvFCtxt \in \cbvFCtxtSet,
                % &\bangFCtxt = f_{\rel, t}(\cbvToBangAGK{\cbvFCtxt})
        \end{array}\right\}
         \text{
        such that } t \cbvArr_F<R> u
    \end{equation*}%
    Moreover, $\bangFCtxt$ is either $\cbvToBangAGK{\cbvFCtxt}$,
    $\cbvToBangAGK[*]{\cbvFCtxt}$ or
    $\cbvToBangAGK{\cbvFCtxt}\bangCtxtPlug{\der{\Hole}}$.
\end{lemma}
\begin{proof}
Let $t \in \setCbvTerms$ such that $\cbvToBangAGK{t}$ is not a
$\bangFCtxtSet$-normal form. We distinguish three cases:
\begin{itemize}
\item[\bltI] $\cbvToBangAGK{t} =
    \bangFCtxt<\app{\bangLCtxt<\abs{x}{s'_1}>}{s'_2}>$ thus
    $\cbvToBangAGK{t} \bangArr_F<dB>
    \bangFCtxt<\bangLCtxt<s'_1\esub{x}{s'_2}>> =: s'$. By
    \Cref{lem:cbvAGK_dB_redex_Reverse_Translated}, one has that $t =
    \cbvFCtxt<\app{\cbvLCtxt<\abs{x}{s_1}>}{s_2}>$ for some $\cbvFCtxt
    \in \cbvFCtxtSet$, $\cbvLCtxt \in \cbvLCtxtSet$ and $s_1, s_2 \in
    \setCbvTerms$ such that $\bangFCtxt =
    \cbvToBangAGK{\cbvFCtxt}\bangCtxtPlug{\der{\Hole}}$, $\bangLCtxt =
    \cbvToBangAGK{\cbvLCtxt}$, $s'_1 = \oc\cbvToBangAGK{s_1}$ and
    $s'_2 = \cbvToBangAGK{s_2}$. We set $u :=
    \cbvFCtxt<\cbvLCtxt<s_1\esub{x}{s_2}>>$ so that $t \cbvArr_F<dB>
    u$. Suppose $s' \bangArrSet*_F<d!> u'$ for some \bangSetFNF<d!>
    $u' \in \setBangTerms$. We distinguish the following two cases.
    \begin{itemize}
    \item[\bltII] $\cbvFuncPred{\cbvFCtxt}$ and
        $\bangBangPred{\cbvToBangAGK{s_1}}$: Then $\cbvToBangAGK{s_1}
        = \oc\bangStrip{\cbvToBangAGK{s_1}}$ thus $s' =
        \cbvToBangAGK{\cbvFCtxt}\bangCtxtPlug{\der{\cbvToBangAGK{\cbvLCtxt}\bangCtxtPlug{\oc\oc
        \bangStrip{\cbvToBangAGK{s_1}}\esub{x}{\cbvToBangAGK{s_2}}}}}$.
        One has that $s' \bangArrSet*_F<d!>
        \cbvToBangAGK[*]{\cbvFCtxt}\bangCtxtPlug{\cbvToBangAGK{\cbvLCtxt}\bangCtxtPlug{\bangStrip{\cbvToBangAGK{s_1}}\esub{x}{\cbvToBangAGK{s_2}}}}
        = \cbvToBangAGK{u}$ using
        \Cref{lem:cbvAGK_Contextual_Translation} thus being a
        $\bangFCtxtSet<\bangSymbBang>$-normal form. By confluence of
        $\bangArrSet_F<d!>$ (\Cref{lem:Confluence_Admin_Full}), one
        finally concludes that $\cbvToBangAGK{u} = u'$.

    \item[\bltII] $\neg\cbvFuncPred{\cbvFCtxt}$ or
        $\neg\bangBangPred{\cbvToBangAGK{s_1}}$: Then $s' =
        \cbvToBangAGK{\cbvFCtxt}\bangCtxtPlug{\der{\cbvToBangAGK{\cbvLCtxt}\bangCtxtPlug{\oc
        \cbvToBangAGK{s_1}\esub{x}{\cbvToBangAGK{s_2}}}}}$. One has
        that $s' \bangArrSet_F<d!>
        \cbvToBangAGK{\cbvFCtxt}\bangCtxtPlug{\cbvToBangAGK{\cbvLCtxt}\bangCtxtPlug{\cbvToBangAGK{s_1}\esub{x}{\cbvToBangAGK{s_2}}}}
        = \cbvToBangAGK{u}$ using
        \Cref{lem:cbvAGK_Contextual_Translation} thus being a
        $\bangFCtxtSet<\bangSymbBang>$-normal form. By confluence of
        $\bangArrSet_F<d!>$ (\Cref{lem:Confluence_Admin_Full}), one
        finally concludes that $\cbvToBangAGK{u} = u'$.
    \end{itemize}

\item[\bltI] $\cbvToBangAGK{t} =
    \bangFCtxt<s'_1\esub{x}{\bangLCtxt<\oc s'_2>}>$ thus
    $\cbvToBangAGK{t} \bangArr_F
    \bangFCtxt<\bangLCtxt<s'_1\isub{x}{s'_2}>> =: s'$. By
    \Cref{lem:cbvAGK_s!_redex_Reverse_Translated}, $t =
    \cbvFCtxt<s\esub{x}{\cbvLCtxt<v>}>$ for some $\cbvFCtxt \in
    \cbvFCtxtSet$, $\cbvLCtxt \in \cbvLCtxtSet$, $s \in \setCbvTerms$
    and $v \in \setCbvValues$ with $\bangLCtxt =
    \cbvToBangAGK{\cbvLCtxt}$, $\oc s'_2 = \cbvToBangAGK{v}$ and
    either:
    \begin{itemize}
    \item[\bltII] $\cbvFuncPred{\cbvFCtxt}$ and
        $\bangBangPred{\cbvToBangAGK{s}}$ so that $\bangFCtxt =
        \cbvToBangAGK[*]{\cbvFCtxt}$ and $s'_1 =
        \bangStrip{\cbvToBangAGK{s}}$. We set $u :=
        \cbvFCtxt<\cbvLCtxt<s_1\isub{x}{v}>>$ so that $t \cbvArr_F<sV>
        u$. Suppose $s' \bangArrSet_F<d!> u'$ for some \bangSetFNF<d!>
        $u' \in \setBangTerms$. Since
        $\bangBangPred{\cbvToBangAGK{s_1}}$ then by
        \Cref{lem:cbvAGK_BangPred_on_Isub,lem:cbvAGK_Contextual_Translation},
        one deduces that
        $\bangBangPred{\cbvToBangAGK{(\cbvLCtxt<s_1\isub{x}{s_2}>)}}$
        thus $\cbvToBangAGK{u} =
        \cbvToBangAGK[*]{\cbvFCtxt}\bangCtxtPlug{\bangStrip{\cbvToBangAGK{(\cbvLCtxt<s\isub{x}{v}>)}}}$
        using \Cref{lem:cbvAGK_Contextual_Translation}. Using
        \Cref{lem:cbvAGK_Translation_on_Isub,lem:bang_Strip_on_Isub},
        one has that $\cbvToBangAGK{u} =
        \bangFCtxt<\bangLCtxt<s'_1\isub{x}{s'_2}>> = u'$ thus $u'$ is
        a $\bangFCtxtSet<\bangSymbBang>$-normal form. By confluence of
        $\bangArrSet_F<d!>$ (\Cref{lem:Confluence_Admin_Full}), one
        deduces that $u' = s'$ thus concluding this case.

    \item[\bltII] $\neg\cbvFuncPred{\cbvFCtxt}$ or
        $\neg\bangBangPred{\cbvToBangAGK{s_1}}$ so that $\bangFCtxt =
        \cbvToBangAGK{\cbvFCtxt}$ and $s'_1 = \cbvToBangAGK{s}$. We
        set $u := \cbvFCtxt<\cbvLCtxt<s_1\isub{x}{v}>>$ so that $t
        \cbvArr_F<sV> u$. Suppose $s' \bangArrSet_F<d!> u'$ for some
        \bangSetFNF<d!> $u' \in \setBangTerms$. We distinguish two
        cases:
        \begin{itemize}
        \item[\bltIII] $\neg\cbvFuncPred{\cbvFCtxt}$: Then using
            \Cref{lem:cbvAGK_Contextual_Translation}, one has that
            $\cbvToBangAGK{u} =
            \cbvToBangAGK{\cbvFCtxt}\bangCtxtPlug{\cbvToBangAGK{(\cbvLCtxt<s\isub{x}{v}>)}}$
            hence $\cbvToBangAGK{u} = u'$ using
            \Cref{lem:cbvAGK_Contextual_Translation,lem:cbvAGK_Translation_on_Isub}
            thus $u'$ is a $\bangFCtxtSet<\bangSymbBang>$-normal form.
            Using \Cref{lem:Confluence_Admin_Full}, one deduces that
            $u' = s'$ thus concluding this case.

        \item[\bltIII] $\neg\bangBangPred{\cbvToBangAGK{s_1}}$: Since
            $\neg\bangBangPred{\cbvToBangAGK{s_1}}$ one deduces that
            $\neg\bangBangPred{(\cbvToBangAGK{s_1\isub{x}{v}})}$ using
            \Cref{lem:cbvAGK_BangPred_on_Isub}. Thus, $\cbvToBangAGK{u}
            =
            \cbvToBangAGK{\cbvFCtxt}\bangCtxtPlug{\cbvToBangAGK{(\cbvLCtxt<s_1\isub{x}{v}>)}}$
            and therefore $\cbvToBangAGK{u} = u'$ using
            \Cref{lem:cbvAGK_Contextual_Translation,lem:cbvAGK_Translation_on_Isub}.
            One deduces that $u'$ is a
            $\bangFCtxtSet<\bangSymbBang>$-normal form and using
            \Cref{lem:Confluence_Admin_Full}, one concludes that $u' =
            s'$ closing this case.
        \end{itemize}
    \end{itemize}

\item[\bltI] $\cbvToBangAGK{t} = \bangFCtxt<\der{\bangLCtxt<\oc s>}>$
    which is impossible by since the translation produces term in
    $\bangFCtxtSet<\bangSymbBang>$-normal form.
\end{itemize}
\end{proof}

This detailed simulation can be used to get iterated preservations
that

\begin{lemma}
    \label{lem: cbv preservation restricted}
    %\deliaLu \giulio{{\tt [G: NOT READ]}}%
    Let $\cbvECtxtSet \subseteq \cbvFCtxtSet$, $\bangECtxtSet
    \subseteq \bangFCtxtSet$ be two families of contexts such that:
    \begin{enumerate}
    \item Let $\cbvECtxt \in \cbvECtxtSet$, then %
        $\cbvToBangAGK{\cbvECtxt}, %
        \cbvToBangAGK{\cbvECtxt}\bangCtxtPlug{\der{\Hole}}, %
        \cbvToBangAGK[*]{\cbvECtxt}, %
        \cbvToBangAGK[*]{\cbvRmDst{\cbvECtxt}} \in \bangECtxtSet$.

    \item Let $\bangECtxt \in \bangECtxtSet$ and $\cbvFCtxt \in
        \cbvFCtxtSet$ such that either $\cbvToBangAGK{\cbvFCtxt} =
        \bangECtxt$ or $\cbvToBangAGK[*]{\cbvFCtxt} = \bangECtxt$ or
        $\cbvToBangAGK{\cbvFCtxt}\bangCtxtPlug{\der{\Hole}} =
        \bangECtxt$, then $\cbvFCtxt \in \cbvECtxtSet$.

    \item $\bangECtxtSet$ admits a diligence process.

    \item Let $t \in \setCbvTerms$ and $u' \in \setBangTerms$ such
        that $\cbvToBangAGK{t} \bangArrSet*_E u'$ where $u'$ is a
        \bangENF<d!>, then it is also a \bangSetFNF<d!>.
    \end{enumerate}
    Then, the following properties hold:
    \begin{itemize}
    \item[\bltI] \textbf{(Normal Forms)}: %
        Let $t \in \setCbvTerms$, then: \quad%
        \begin{equation*}
            t \text{ is a \cbvSetENF}
                \quad\Leftrightarrow\quad
            \cbvToBangAGK{t} \text{ is a \bangSetENF}
        \end{equation*}

    \item[\bltI] \textbf{(Stability)}: %
        Let $t, \in \setCbvTerms$ and $u \in \setBangTerms$ where $u$
        is a \bangENF<d!>, then:
        \begin{equation*}
            \cbvToBangAGK{t} \bangArrSet*_E u'
                \quad \Rightarrow \quad
            \exists\, u \in \setCbvTerms, \; \cbvToBangAGK{u} = u'
        \end{equation*}

    \item[\bltI] \textbf{(Simulation and Reverse Simulation)}: %
        Let $t, u \in \setCbvTerms$, then:
        \begin{equation*}
            t \cbvArrSet*_E u
                \quad \Leftrightarrow \quad
            \cbvToBangAGK{t} \bangArrSet*_E \cbvToBangAGK{u}
        \end{equation*}
        Moreover, the number of $\cbvSymbBeta/\cbvSymbSubs$-steps
        matches the number $\bangSymbBeta/\bangSymbSubs$-steps.
    \end{itemize}
\end{lemma}
\begin{proof}
%\deliaLu \giulio{{\tt [G: NOT READ]}}%
Let us first show that the following two properties hold:
\begin{itemize}
\item[\bltI] \textbf{(One-Step Simulation)}: %
    Let $t, u \in \setCbvTerms$, then:
    \begin{equation*}
        t \cbvArrSet_E u
            \quad \Rightarrow \quad
        \cbvToBangAGK{t} \bangArrSet*_E \cbvToBangAGK{u}
    \end{equation*}

\item[\bltI] \textbf{(One-Step Reverse Simulation)}: %
    Let $t \in \setCbvTerms$ and $u' \in \setBangTerms$ where $u'$ is
    a \bangENF<d!>, then:
    \begin{equation*}
        \cbvToBangAGK{t} \bangArrSet_E\bangArrSet*_E<d!> u'
            \quad \Rightarrow \quad
        \exists\, u \in \setCbvTerms, \;
            \cbvToBangAGK{u} = u'
        \text{ and }
            t \bangArrSet_E u
    \end{equation*}
\end{itemize}
In particular, $\cbvSymbBeta$-steps (\resp $\cbvSymbSubs$-steps) are
simulated by $\bangSymbBeta$-steps (\resp $\bangSymbSubs$-steps) with
possible administrative steps.

~

We distinguish the two cases:
\begin{itemize}
\item[\bltI] \textbf{(One-Step Simulation)}: Let $t, u \in
    \setCbnTerms$ such that $t \cbvArrSet_E u$. Thus $t =
    \cbvECtxt<t'> \cbvArr_E<R> \cbvECtxt<u'>$ for some $\cbvECtxt \in
    \cbvECtxtSet$, $\rel \in \{\cbvSymbBeta, \cbvSymbSubs\}$ and $t',
    u' \in \setCbvTerms$. Since $\cbvECtxtSet \subseteq \cbvFCtxtSet$
    then using \CBVSymb one-step simulation %
    (\Cref{lem:cbv detailed one step simulation}), we distinguish the
    following cases:
    \begin{itemize}
    \item[\bltI] When $\rel = \cbvSymbBeta$:
        \begin{itemize}
        \item[\bltII] If $\cbvFuncPred{\cbvECtxt}$ and
            $\bangBangPred{\cbvToBangAGK{u'}}$:
            $\cbvToBangAGK{(\cbvECtxt<t'>)}
            \bangArr_{\bangFCtxt_1\bangCtxtPlug{\cbvToBangAGK{\rel}}}
            \bangArr_{F_2<\bangSymbBang>}
            \bangArr_{F_3<\bangSymbBang>}
            \cbvToBangAGK{(\cbvECtxt<u'>)}$ with $\bangFCtxt_1 =
            \cbvToBangAGK{\cbvECtxt}\bangCtxtPlug{\der{\Hole}}$,
            $\bangFCtxt_2 = \cbvToBangAGK{\cbvECtxt}$ and
            $\bangFCtxt_3 = \cbvToBangAGK[*]{\cbvRmDst{\cbvECtxt}}$.
            Since $\cbvECtxt \in \cbvECtxtSet$, then by hypothesis,
            $\bangFCtxt_1, \bangFCtxt_2, \bangFCtxt_3 \in
            \bangECtxtSet$ so that $\cbvToBangAGK{(\cbvECtxt<t'>)}
            \bangArrSet_{E\bangCtxtPlug{\cbvToBangAGK{\rel}}}
            \bangArrSet_{E<d!>} \bangArrSet_{E<d!>}
            \cbvToBangAGK{(\cbvECtxt<u'>)}$ and by transitivity
            $\cbvToBangAGK{t} \bangArrSet*_E \cbvToBangAGK{u}$.

        \item[\bltII] Otherwise: $\cbvToBangAGK{(\cbvECtxt<t'>)}
            \bangArr_{\bangFCtxt_1\bangCtxtPlug{\cbvToBangAGK{\rel}}}
            \bangArr_{F_2<\bangSymbBang>}
            \cbvToBangAGK{(\cbvECtxt<u'>)}$ with $\bangFCtxt_1 =
            \cbvToBangAGK{\cbvECtxt}\bangCtxtPlug{\der{\Hole}}$ and
            $\bangFCtxt_2 = \cbvToBangAGK{\cbvECtxt}$. Since
            $\cbvECtxt \in \cbvECtxtSet$, then by hypothesis,
            $\bangFCtxt_1, \bangFCtxt_2 \in \bangECtxtSet$ so that
            $\cbvToBangAGK{(\cbvECtxt<t'>)}
            \bangArrSet_{E\bangCtxtPlug{\cbvToBangAGK{\rel}}}
            \bangArrSet_{E<d!>} \cbvToBangAGK{(\cbvECtxt<u'>)}$ and by
            transitivity $\cbvToBangAGK{t} \bangArrSet*_E
            \cbvToBangAGK{u}$.
        \end{itemize}

    \item[\bltI] When $\rel = \cbvSymbSubs$:
        \begin{itemize}
        \item[\bltII] If $\cbvFuncPred{\cbvECtxt}$ and
            $\bangBangPred{\cbvToBangAGK{u'}}$:
            $\cbvToBangAGK{(\cbvECtxt<t'>)}
            \bangArr_{\bangFCtxt\bangCtxtPlug{\cbvToBangAGK{\rel}}}
            \cbvToBangAGK{(\cbvECtxt<u'>)}$ with $\bangFCtxt =
            \cbvToBangAGK[*]{\cbvECtxt}$. Since $\cbvECtxt \in
            \cbvECtxtSet$, then by hypothesis, $\bangFCtxt \in
            \bangECtxtSet$ so that $\cbvToBangAGK{(\cbvECtxt<t'>)}
            \bangArrSet_{E\bangCtxtPlug{\cbvToBangAGK{\rel}}}
            \cbvToBangAGK{(\cbvECtxt<u'>)}$ thus $\cbvToBangAGK{t}
            \bangArrSet_E \cbvToBangAGK{u}$.

        \item[\bltII] Otherwise: $\cbvToBangAGK{(\cbvECtxt<t'>)}
            \bangArr_{F<\cbvToBangAGK{\rel}>}
            \cbvToBangAGK{(\cbvECtxt<u'>)}$ with $\bangFCtxt =
            \cbvToBangAGK{\cbvECtxt}$. Since $\cbvECtxt \in
            \cbvECtxtSet$, then by hypothesis, $\bangFCtxt \in
            \bangECtxtSet$ so that $\cbvToBangAGK{(\cbvECtxt<t'>)}
            \bangArrSet_{E\bangCtxtPlug{\cbvToBangAGK{\rel}}}
            \cbvToBangAGK{(\cbvECtxt<u'>)}$ thus $\cbvToBangAGK{t}
            \bangArrSet_E \cbvToBangAGK{u}$.
        \end{itemize}
    \end{itemize}

\item[\bltI] \textbf{(One-Step Reverse Simulation)}: Let $t \in
    \setCbvTerms$, $u' \in \setBangTerms$ and $\rel' \in
    \{\bangSymbBeta, \bangSymbSubs, \bangSymbBang\}$ where $u'$ is a
    \bangENF<d!> and $\cbvToBangAGK{t} \bangArrSet_E u'$. By
    hypothesis, one deduces that $u'$ is a \bangSetFNF<d!>. By
    definition, $\cbvToBangAGK{t} \bangArr_E<R'> u'$ for some
    $\bangECtxt \in \bangECtxtSet$. Since $\bangECtxtSet \subseteq
    \bangFCtxtSet$, then using \CBVSymb one-step reverse simulation
    (\Cref{lem:cbv detailed one step reverse simulation}), there
    exists $u \in \setBangTerms$, $\rel \in \{\cbvSymbBeta,
    \cbvSymbSubs\}$ and $\cbvFCtxt \in \cbvFCtxtSet$ such that $t
    \cbvArr_{F<R>} u$ with $\cbvToBangAGK{u} = u'$ and $\rel' =
    \cbvToBangAGK{\rel}$. Moreover, $\bangECtxt$ is either
    $\cbvToBangAGK{\cbvFCtxt}$, $\cbvToBangAGK[*]{\cbvFCtxt}$ or
    $\cbvToBangAGK{\cbvFCtxt}\bangCtxtPlug{\der{\Hole}}$ thus by
    hypothesis $\cbvFCtxt \in \cbvECtxtSet$ and therefore $t
    \cbvArrSet_E u$ which concludes this case.
\end{itemize}
Using these, let us now prove the three expected properties:
\begin{itemize}
\item[\bltI] \textbf{(Normal Forms)}: %
    Direct consequence of simulation and reverse simulation.

\item[\bltI] \textbf{(Stability)}: %
    Let $t \in \setCbvTerms$ and $u' \in \setBangTerms$ where $u'$ is
    a \bangENF<d!> and $\cbvToBangAGK{t} \bangArrSet*_E u'$. Then by
    hypothesis, one has that $u'$ is a \bangSetFNF<d!> and since
    $\bangECtxtSet$ admits a diligence process, we deduce that
    $\cbvToBangAGK{t} \bangArrSet*_{Eai} u'$. Let us proceed by
    induction on the length of $\cbvToBangAGK{t} \bangArrSet*_{Eai}
    u'$:
    \begin{itemize}
    \item[\bltII] $\cbvToBangAGK{t} \bangArrSet^0_E u'$: Then $u' =
        \cbvToBangAGK{t}$ and taking $u = t$ trivially concludes this
        case.
    \item[\bltII] $\cbvToBangAGK{t} \bangArrSet_E<R>\bangArrSet*_E<d!>
        s' \bangArrSet*_{Eai} u'$ with $s' \in \setBangTerms$ where
        $s'$ is a \bangENF<d!>: Using one-step reverse simulation,
        there exists $s \in \setCbvTerms$ such that $\cbvToBangAGK{s}
        = s'$ and thus by \ih on $\cbvToBangAGK{s} \bangArrSet*_E u'$,
        one concludes that there exists $u \in \setCbvTerms$ such that
        $\cbvToBangAGK{u} = u'$.
    \end{itemize}

\item[\bltI] \textbf{(Simulation)}: %
    Let $t, u \in \setCbvTerms$ such that $t \cbvArrSet*_E u$. By
    induction on the length of $t \cbvArrSet*_E u$:
    \begin{itemize}
    \item[\bltII] $t \cbvArrSet^0_E u$: Then $t = u$ thus
        $\cbvToBangAGK{t} = \cbvToBangAGK{u}$ and by reflexivity
        $\cbvToBangAGK{t} \bangArrSet*_E \cbvToBangAGK{u}$.
    \item[\bltII] $t \cbvArrSet_E s \cbvArrSet*_E u$: On one hand
      $\cbvToBangAGK{t} \cbvArrSet*_E \cbvToBangAGK{s}$ using one-step
      simulation. On the other hand, by \ih on $s
        \cbvArrSet*_E$, one has that $\cbvToBangAGK{s} \bangArrSet*_E
        \cbvToBangAGK{u}$. Finally, one concludes by transitivity that
        $\cbvToBangAGK{t} \bangArrSet*_E \cbvToBangAGK{u}$.
    \end{itemize}

\item[\bltI] \textbf{(Reverse Simulation)}: %
    Let $t, u \in \setCbvTerms$ such that $\cbvToBangAGK{t}
    \bangArrSet*_E \cbvToBangAGK{u}$. By contruction, one has that $u$
    is a \bangSetFNF<d!> and by diligence $\cbvToBangAGK{t}
    \bangArrSet*_{Eai} \cbvToBangAGK{u}$. By induction on the length
    of $\cbvToBangAGK{t} \bangArrSet*_{Eai} \cbvToBangAGK{u}$:
    \begin{itemize}
    \item[\bltII] $\cbvToBangAGK{t} \bangArrSet^0_{Eai}
        \cbvToBangAGK{u}$: Then $\cbvToBangAGK{t} = \cbvToBangAGK{u}$
        thus by injectivity $t = u$ and by reflexivity $t
        \cbvArrSet*_E u$.
    \item[\bltII] $\cbvToBangAGK{t} \bangArrSet_E \bangArrSet*_E<d!>
        s' \bangArrSet*_{Eai} \cbvToBangAGK{u}$ for some $s' \in
        \setBangTerms$ where $s'$ is a \bangENF<d!>: On one hand,
        using one-step reverse simulation, there exists $s \in
        \setCbvTerms$ such that $t \cbvArrSet_E s$ with
        $\cbvToBangAGK{s} = s'$. On the other hand, by \ih on
        $\cbvToBangAGK{s} \bangArrSet*_{Eai} \cbvToBangAGK{u}$, one
        has that $s \cbvArrSet*_E u$. Finally, one concludes by
        transitivity that $t \cbvArrSet*_E u$.
    \end{itemize}
\end{itemize}
\end{proof}

\begin{lemma}
    \label{lem:cbv_full_context_stability}
    The following property hold:
    \begin{enumerate}
    \item Let $\cbvFCtxt \in \cbvFCtxtSet$, then %
        $\cbvToBangAGK{\cbvFCtxt}, %
        \cbvToBangAGK{\cbvFCtxt}\bangCtxtPlug{\der{\Hole}}, %
        \cbvToBangAGK[*]{\cbvFCtxt}, %
        \cbvToBangAGK[*]{\cbvRmDst{\cbvFCtxt}} \in \bangFCtxtSet$.

    \item Let $\bangFCtxt \in \bangFCtxtSet$ and $\cbvFCtxt \in
        \cbvFCtxtSet$ such that either $\cbvToBangAGK{\cbvFCtxt} =
        \bangFCtxt$ or $\cbvToBangAGK[*]{\cbvFCtxt} = \bangFCtxt$ or
        $\cbvToBangAGK{\cbvFCtxt}\bangCtxtPlug{\der{\Hole}} =
        \bangFCtxt$, then $\cbvFCtxt \in \cbvFCtxtSet$.

    \item Let $t \in \setCbvTerms$ and $u' \in \setBangTerms$ such
        that $\cbvToBangAGK{t} \bangArrSet*_F u'$ where $u'$ is a
        \bangSetFNF<d!>, then it is also a \bangSetFNF<d!>.
    \end{enumerate}
\end{lemma}
\begin{proof} ~
    \begin{enumerate}
    \item By definition of $\cbvToBangAGK{\cdot},
        \cbvToBangAGK[*]{\cdot}$ and $\cbvRmDst{\cdot}$.

    \item By hypothesis.

    \item By hypothesis.
    \end{enumerate}
\end{proof}

\RecCbvPropFull*
\begin{proof}
    Using \Cref{lem: cbv preservation restricted} with
    \Cref{lem:cbv_full_context_stability} and
    \Cref{lem:Implicitation}.
    \qed
\end{proof}

\begin{lemma}
    \label{lem:cbv_surface_context_stability}
    The following property hold:
    \begin{enumerate}
    \item Let $\cbvSCtxt \in \cbvSCtxtSet$, then %
        $\cbvToBangAGK{\cbvSCtxt}, %
        \cbvToBangAGK{\cbvSCtxt}\bangCtxtPlug{\der{\Hole}}, %
        \cbvToBangAGK[*]{\cbvSCtxt}, %
        \cbvToBangAGK[*]{\cbvRmDst{\cbvSCtxt}} \in \bangSCtxtSet$.

    \item Let $\bangSCtxt \in \bangSCtxtSet$ and $\cbvFCtxt \in
        \cbvFCtxtSet$ such that either $\cbvToBangAGK{\cbvFCtxt} =
        \bangSCtxt$ or $\cbvToBangAGK[*]{\cbvFCtxt} = \bangSCtxt$ or
        $\cbvToBangAGK{\cbvFCtxt}\bangCtxtPlug{\der{\Hole}} =
        \bangSCtxt$, then $\cbvFCtxt \in \cbvSCtxtSet$.
    \end{enumerate}
\end{lemma}
\begin{proof} ~
    \begin{enumerate}
    \item Notice that if $\cbvSCtxt \in \cbvSCtxtSet$ then
        $\cbvRmDst{\cbvSCtxt} \in \cbvSCtxtSet$, and that if
        $\bangSCtxt \in \bangSCtxtSet$ then $\bangSCtxt<\der{\Hole}>
        \in \bangSCtxtSet$. Thanks to these observations, it becomes
        sufficient to show that if $\cbvSCtxt \in \cbvSCtxtSet$, then
        $\cbvToBangAGK{\cbvSCtxt} \in \bangSCtxt$ and
        $\cbvToBangAGK[*]{\cbvSCtxt} \in \bangSCtxt$.

        By induction on $\cbvSCtxt \in \cbvSCtxtSet$:
        \begin{itemize}
        \item[\bltI] $\cbvSCtxt = \Hole$: Then
            $\cbvToBangAGK{\cbvSCtxt} = \Hole \in \bangSCtxt$ and
            $\cbvToBangAGK[*]{\cbvSCtxt} = \Hole \in \bangSCtxt$.

        \item[\bltI] $\cbvSCtxt = \app[\,]{\cbvSCtxt'}{t}$: By \ih on
            $\cbvSCtxt'$, one has that $\cbvToBangAGK{{\cbvSCtxt'}},
            \cbvToBangAGK[*]{{\cbvSCtxt'}} \in \bangSCtxtSet$. Let us
            distinguish cases:
            \begin{itemize}
            \item[\bltII] $\cbvSCtxt' \in \cbvLCtxt$: Then
                $\cbvToBangAGK{\cbvSCtxt} =
                \der{\app[\,]{\der{\cbvToBangAGK{{\cbvSCtxt'}}}}{\cbvToBangAGK{t}}}
                \in
                \der{\app[\,]{\der{\bangSCtxtSet}}{\cbvToBangAGK{t}}}
                \subseteq \bangSCtxtSet$ and
                $\cbvToBangAGK[*]{\cbvSCtxt} =
                \der{\app[\,]{\cbvToBangAGK[*]{{\cbvSCtxt'}}}{\cbvToBangAGK{t}}}
                \in
                \der{\app[\,]{\bangSCtxtSet}{\cbvToBangAGK{t}}}
                \subseteq \bangSCtxtSet$.

            \item[\bltII] $\bangBangPred{\cbvToBangAGK{\cbvSCtxt}}$:
                By induction on $\bangSCtxt'$, one has that
                $\bangStrip{\cbvToBangAGK{{\bangSCtxt'}}},
                \bangStrip{\cbvToBangAGK[*]{{\bangSCtxt'}}} \in
                \bangSCtxt$ thus $\cbvToBangAGK{\cbvSCtxt} =
                \der{\app[\,]{\bangStrip{\cbvToBangAGK{{\cbvSCtxt'}}}}{\cbvToBangAGK{t}}}
                \in
                \der{\app[\,]{\bangSCtxtSet}{\cbvToBangAGK{t}}}
                \subseteq \bangSCtxtSet$ and
                $\cbvToBangAGK[*]{\cbvSCtxt} =
                \der{\app[\,]{\bangStrip{\cbvToBangAGK[*]{{\cbvSCtxt'}}}}{\cbvToBangAGK{t}}}
                \in
                \der{\app[\,]{\bangSCtxtSet}{\cbvToBangAGK{t}}}
                \subseteq \bangSCtxtSet$.

            \item[\bltII] Otherwise: Then $\cbvToBangAGK{\cbvSCtxt} =
                \der{\app[\,]{\der{\cbvToBangAGK{{\cbvSCtxt'}}}}{\cbvToBangAGK{t}}}
                \in
                \der{\app[\,]{\der{\bangSCtxtSet}}{\cbvToBangAGK{t}}}
                \subseteq \bangSCtxtSet$ and
                $\cbvToBangAGK[*]{\cbvSCtxt} =
                \der{\app[\,]{\der{\cbvToBangAGK[*]{{\cbvSCtxt'}}}}{\cbvToBangAGK{t}}}
                \in
                \der{\app[\,]{\der{\bangSCtxtSet}}{\cbvToBangAGK{t}}}
                \subseteq \bangSCtxtSet$.

            \end{itemize}

        \item[\bltI] $\cbvSCtxt = \app[\,]{t}{\cbvSCtxt'}$: By \ih on
            $\cbvSCtxt'$, one has that $\cbvToBangAGK{{\cbvSCtxt'}},
            \cbvToBangAGK[*]{{\cbvSCtxt'}} \in \bangSCtxtSet$. Let us
            distinguish cases:
            \begin{itemize}
            \item[\bltII] $\bangBangPred{\cbvToBangAGK{t}}$: Then
                $\cbvToBangAGK{\cbvSCtxt} =
                \der{\app[\,]{\bangStrip{\cbvToBangAGK{t}}}{\cbvToBangAGK{{\cbvSCtxt'}}}}
                \in
                \der{\app[\,]{\bangStrip{\cbvToBangAGK{t}}}{\bangSCtxtSet}}
                \subseteq \bangSCtxtSet$ and
                $\cbvToBangAGK[*]{\cbvSCtxt} =
                \der{\app[\,]{\bangStrip{\cbvToBangAGK{t}}}{\cbvToBangAGK[*]{{\cbvSCtxt'}}}}
                \in
                \der{\app[\,]{\bangStrip{\cbvToBangAGK{t}}}{\bangSCtxtSet}}
                \subseteq \bangSCtxtSet$.

            \item[\bltII] Otherwise: Then $\cbvToBangAGK{\cbvSCtxt} =
                \der{\app[\,]{\der{\cbvToBangAGK{t}}}{\cbvToBangAGK{{\cbvSCtxt'}}}}
                \in
                \der{\app[\,]{\der{\cbvToBangAGK{t}}}{\bangSCtxtSet}}
                \subseteq \bangSCtxtSet$ and
                $\cbvToBangAGK[*]{\cbvSCtxt} =
                \der{\app[\,]{\der{\cbvToBangAGK{t}}}{\cbvToBangAGK[*]{{\cbvSCtxt'}}}}
                \in
                \der{\app[\,]{\der{\cbvToBangAGK{t}}}{\bangSCtxtSet}}
                \subseteq \bangSCtxtSet$.
            \end{itemize}

        \item[\bltI] $\cbvSCtxt = \cbvSCtxt'\esub{x}{t}$: By \ih on
            $\cbvSCtxt'$, one has that $\cbvToBangAGK{{\cbvSCtxt'}},
            \cbvToBangAGK[*]{{\cbvSCtxt'}} \in \bangSCtxtSet$ so that
            $\cbvToBangAGK{\cbvSCtxt} =
            \cbvToBangAGK{{\cbvSCtxt'}}\esub{x}{\cbvToBangAGK{t}} \in
            \bangSCtxtSet\esub{x}{\cbvToBangAGK{t}} \subseteq
            \bangSCtxtSet$ and $\cbvToBangAGK[*]{\cbvSCtxt} =
            \cbvToBangAGK[*]{{\cbvSCtxt'}}\esub{x}{\cbvToBangAGK{t}}
            \in \bangSCtxtSet\esub{x}{\cbvToBangAGK{t}} \subseteq
            \bangSCtxtSet$.

        \item[\bltI] $\cbvSCtxt = t\esub{x}{\cbvSCtxt'}$: By \ih on
            $\cbvSCtxt'$, one has that $\cbvToBangAGK{{\cbvSCtxt'}},
            \cbvToBangAGK[*]{{\cbvSCtxt'}} \in \bangSCtxtSet$ so that
            $\cbvToBangAGK{\cbvSCtxt} =
            \cbvToBangAGK{t}\esub{x}{\cbvToBangAGK{{\cbvSCtxt'}}} \in
            \cbvToBangAGK{t}\esub{x}{\bangSCtxtSet} \subseteq
            \bangSCtxtSet$ and $\cbvToBangAGK[*]{\cbvSCtxt} =
            \cbvToBangAGK{t}\esub{x}{\cbvToBangAGK[*]{{\cbvSCtxt'}}}
            \in \cbvToBangAGK{t}\esub{x}{\bangSCtxtSet} \subseteq
            \bangSCtxtSet$.
        \end{itemize}

    \item Let $\bangSCtxt \in \bangSCtxtSet$ and $\cbvFCtxt \in
        \cbvFCtxtSet$ such that either $\cbvToBangAGK{\cbvFCtxt} =
        \bangSCtxt$, $\cbvToBangAGK[*]{\cbvFCtxt} = \bangSCtxt$ or
        $\cbvToBangAGK{\cbvFCtxt}\bangCtxtPlug{\der{\Hole}} =
        \bangSCtxt$. By induction on $\bangSCtxt \in \bangSCtxtSet$:
        \begin{itemize}
        \item[\bltI] $\bangSCtxt = \Hole$: We distinguish three cases:
            \begin{itemize}
            \item[\bltII] $\cbvToBangAGK{\cbvFCtxt} = \Hole$: Then
                necessarily $\cbvFCtxt = \Hole$ and thus $\cbvFCtxt
                \in \cbvSCtxtSet$.
            \item[\bltII] $\cbvToBangAGK[*]{\cbvFCtxt} = \Hole$: Same
                as the previous case.
                % Then necessarily $\cbvFCtxt = \Hole$ and thus
                % $\cbvFCtxt \in \cbvSCtxtSet$.
            \item[\bltII]
                $\cbvToBangAGK{\cbvFCtxt}\bangCtxtPlug{\der{\Hole}} =
                \Hole$: Impossible.
            \end{itemize}

        \item[\bltI] $\bangSCtxt = \abs{x}{\bangSCtxt'}$: Impossible
            by definition of $\cbvToBangAGK{\cdot}$ and
            $\cbvToBangAGK[*]{\cdot}$.

        \item[\bltI] $\bangSCtxt = \app[\,]{\bangSCtxt'}{t'}$:
            Impossible by definition of $\cbvToBangAGK{\cdot}$ and
            $\cbvToBangAGK[*]{\cdot}$.

        \item[\bltI] $\bangSCtxt = \app[\,]{t'}{\bangSCtxt'}$:
            Impossible by definition of $\cbvToBangAGK{\cdot}$ and
            $\cbvToBangAGK[*]{\cdot}$.

        \item[\bltI] $\bangSCtxt = \bangSCtxt'\esub{x}{t'}$: We
            distinguish three cases:
            \begin{itemize}
            \item[\bltII] $\cbvToBangAGK{\cbvFCtxt} =
                \bangSCtxt'\esub{x}{t'}$: Then necessarily $\cbvFCtxt
                = \cbvFCtxt'\esub{x}{t}$ for some $\cbvFCtxt' \in
                \cbvFCtxtSet$ and $t \in \setCbvTerms$ such that
                $\cbvToBangAGK{{\cbvFCtxt'}} = \bangSCtxt'$ and
                $\oc\cbvToBangAGK{t} = t'$. By \ih on $\bangSCtxt'$,
                one has that $\cbvFCtxt' \in \cbvSCtxtSet$ so that
                $\cbvFCtxt = \cbvFCtxt'\esub{x}{t} \in
                \cbvSCtxtSet\esub{x}{t} \subseteq \cbvSCtxtSet$.

            \item[\bltII] $\cbvToBangAGK[*]{\cbvFCtxt} =
                \bangSCtxt'\esub{x}{t'}$ or
                $\cbvToBangAGK{\cbvFCtxt}\bangCtxtPlug{\der{\Hole}} =
                \bangSCtxt'\esub{x}{t'}$: Similar to the previous
                case.
            \end{itemize}

        \item[\bltI] $\bangSCtxt = t'\esub{x}{\bangSCtxt'}$: We
            distinguish three cases:
            \begin{itemize}
            \item[\bltII] $\cbvToBangAGK{\cbvFCtxt} =
                t'\esub{x}{\bangSCtxt'}$: Then necessarily $\cbvFCtxt
                = t\esub{x}{\cbvFCtxt'}$ for some $t \in \setCbvTerms$
                and $\cbvFCtxt' \in \cbvFCtxtSet$ such that
                $\cbvToBangAGK{t} = t'$ and
                $\cbvToBangAGK{{\cbvFCtxt'}} = \bangSCtxt'$. By \ih on
                $\bangSCtxt'$, one has that $\cbvFCtxt' \in
                \cbvSCtxtSet$ so that $\cbvFCtxt =
                t\esub{x}{\cbvFCtxt'} \in t\esub{x}{\cbvSCtxtSet}
                \subseteq \cbvSCtxtSet$.

            \item[\bltII] $\cbvToBangAGK[*]{\cbvFCtxt} =
                t'\esub{x}{\bangSCtxt'}$ or
                $\cbvToBangAGK{\cbvFCtxt}\bangCtxtPlug{\der{\Hole}} =
                t'\esub{x}{\bangSCtxt'}$: Similar to the previous
                case.
            \end{itemize}

        \item[\bltI] $\bangSCtxt = \der{\bangSCtxt'}$: We distinguish
            three cases:
            \begin{itemize}
            \item[\bltII] $\cbvToBangAGK{\cbvFCtxt} =
                \der{\bangSCtxt'}$: Then $\cbvFCtxt$ is necessarily an
                application and there exists $\cbvFCtxt' \in
                \cbvFCtxtSet$ and $t \in \setCbvTerms$ such that
                either:
                \begin{itemize}
                \item[\bltIII] $\cbvFCtxt = \app[\,]{\cbvFCtxt'}{t}$:
                    Then necessarily $\bangSCtxt' =
                    \app[\,]{\bangSCtxt'}{t'}$ for some $\bangSCtxt'
                    \in \bangSCtxtSet$ and $t' \in \setBangTerms$ such
                    that $\cbvToBangAGK{t} = t'$. Two cases must be
                    distinguished on $\cbvFCtxt'$:
                    \begin{itemize}
                    \item[\bltIV]
                        $\bangBangPred{\cbvToBangAGK{{\cbvFCtxt'}}}$
                        and $\bangStrip{\cbvToBangAGK{{\cbvFCtxt'}}} =
                        \bangSCtxt'$: (See $(*)$ below).

                    \item[\bltIV]
                        $\neg\bangBangPred{\cbvToBangAGK{{\cbvFCtxt'}}}$
                        and $\der{\cbvToBangAGK{{\cbvFCtxt'}}} =
                        \bangSCtxt'$: Thus $\bangSCtxt' =
                        \der{\bangSCtxt''}$ for some $\bangSCtxt'' \in
                        \bangSCtxtSet$ such that
                        $\cbvToBangAGK{{\cbvFCtxt'}} = \bangSCtxt''$.
                        By \ih on $\bangSCtxt''$, one has that
                        $\cbvFCtxt' \in \cbvSCtxtSet$ so that
                        $\cbvFCtxt = \app[\,]{\cbvFCtxt'}{t} \in
                        \app[\,]{\cbvSCtxtSet}{t} \subseteq
                        \cbvSCtxtSet$.
                    \end{itemize}

                \item[\bltIII] $\cbvFCtxt = \app[\,]{t}{\cbvFCtxt'}$:
                    Then necessarily $\bangSCtxt' =
                    \app[\,]{t'}{\bangSCtxt'}$ for some $\bangSCtxt'
                    \in \bangSCtxtSet$ and $t' \in \setBangTerms$ such
                    that $\cbvToBangAGK{{\cbvFCtxt'}} = \bangSCtxt'$.
                    By \ih on $\bangSCtxt'$, one has that $\cbvFCtxt'
                    \in \cbvSCtxtSet$ so that $\cbvFCtxt =
                    \app[\,]{t}{\cbvFCtxt'} \in
                    \app[\,]{t}{\cbvSCtxtSet} \subseteq \cbvSCtxtSet$.
                \end{itemize}

            \item[\bltII] $\cbvToBangAGK[*]{\cbvFCtxt} =
                \der{\bangSCtxt'}$: Same as the previous case (with
                one additional trivial subsubcase)

            \item[\bltII]
                $\cbvToBangAGK{\cbvFCtxt}\bangCtxtPlug{\der{\Hole}} =
                \der{\bangSCtxt'}$: Then either $\cbvFCtxt = \Hole$
                thus $\cbvFCtxt \in \cbvSCtxtSet$ or $\cbvFCtxt$ is an
                application an the case is similar to the first case.
            \end{itemize}
        \end{itemize}
    \end{enumerate}

    $(*)$ We distinguish two cases on
    $\bangBangPred{\cbvToBangAGK{{\cbvFCtxt'}}}$:
    \begin{itemize}
    \item[\bltI] $\cbvToBangAGK{{\cbvFCtxt'}} = \bangLCtxt<\oc
        \bangSCtxt''>$ for some $\bangLCtxt \in \bangLCtxtSet$ and
        some $\bangSCtxt'' \in \bangSCtxtSet$. Let us show by
        induction on $\bangLCtxt$ that this cases is impossible:
        \begin{itemize}
        \item[\bltII] $\bangLCtxt = \Hole$: Then
            $\cbvToBangAGK{{\cbvFCtxt'}} = \oc\bangSCtxt''$ thus
            necessarily $\cbvFCtxt' = \abs{x}{\cbvFCtxt''}$ for some
            $\cbvFCtxt'' \in \cbvFCtxtSet$ such that
            $\abs{x}{\oc\cbvToBangAGK{{\cbvFCtxt''}}} = \bangSCtxt''$
            which is impossible since no $\bangSCtxt'' \in
            \bangSCtxtSet$ can have a hole under a bang.

        \item[\bltII] $\bangLCtxt = \bangLCtxt'\esub{x}{s'}$ for some
            $\bangLCtxt' \in \bangLCtxtSet$ and $s' \in
            \setBangTerms$: Then $\cbvToBangAGK{{\cbvFCtxt'}} =
            \bangLCtxt'<\oc\bangSCtxt''>\esub{x}{s'}$ thus necessarily
            $\cbvFCtxt' = \cbvFCtxt''\esub{x}{s}$ for some
            $\cbvFCtxt'' \in \cbvFCtxtSet$ and $s \in \setCbvTerms$
            such that $\cbvToBangAGK{{\cbvFCtxt''}} =
            \bangLCtxt'<\oc\bangSCtxt''>$ which is impossible by \ih.
        \end{itemize}
    \item[\bltI] $\cbvToBangAGK{{\cbvFCtxt'}} =
        \bangLCtxt_1<\bangLCtxt_2<\oc u'>\esub{x}{\bangSCtxt''}>$ for
        some $\bangLCtxt_1, \bangLCtxt_2 \in \bangLCtxtSet$,
        $\bangSCtxt'' \in \bangSCtxtSet$ and $u' \in \setBangTerms$.
        By induction on $\bangLCtxt_1$:
        \begin{itemize}
        \item[\bltII] $\bangLCtxt_1 = \Hole$: Then
            $\cbvToBangAGK{{\cbvFCtxt'}} = \bangLCtxt_2<\oc
            u'>\esub{x}{\bangSCtxt''}$ thus necessarily $\cbvFCtxt' =
            u\esub{x}{\cbvFCtxt''}$ for some $u \in \setCbvTerms$ and
            $\cbvFCtxt'' \in \cbvFCtxtSet$ such that $\cbvToBangAGK{u}
            = \bangLCtxt_2<\oc u'>$ and $\cbvToBangAGK{{\cbvFCtxt''}}
            = \bangSCtxt''$. By \ih on $\bangSCtxt''$, one has that
            $\cbvToBangAGK{{\cbvFCtxt''}} \in \cbvSCtxtSet$ so that
            $\cbvFCtxt = \app[\,]{\cbvFCtxt'}{t} =
            \app[\,]{u\esub{x}{\cbvFCtxt''}}{t} \in
            \app[\,]{u\esub{x}{\bangSCtxtSet}}{t} \subseteq
            \bangSCtxtSet$.

        \item[\bltII] $\bangLCtxt_1 = \bangLCtxt_1'\esub{y}{s'}$ for
            some $\bangLCtxt'_1 \in \bangLCtxtSet$ and $s' \in
            \setBangTerms$. Then $\cbvToBangAGK{{\cbvFCtxt'}} =
            \bangLCtxt'_1<\bangLCtxt_2<\oc
            u'>\esub{x}{\bangSCtxt''}>\esub{y}{s'}$ thus necessarily
            $\cbvFCtxt' = \cbvFCtxt''\esub{y}{s}$ for some
            $\cbvFCtxt'' \in \cbvFCtxtSet$ and $s \in \setCbvTerms$
            such that $\cbvToBangAGK{{\cbvFCtxt''}} =
            \bangLCtxt'_1<\bangLCtxt_2<\oc u'>\esub{x}{\bangSCtxt''}>$
            and $\cbvToBangAGK{s} = s'$. By \ih on $\bangLCtxt'_1$,
            one has that $\cbvFCtxt'' \in \bangSCtxtSet$ so that
            $\cbvFCtxt = \app[\,]{\cbvFCtxt'}{t} =
            \app[\,]{\cbvFCtxt''\esub{y}{s}}{t} \in
            \app[\,]{\cbvSCtxtSet\esub{y}{s}}{t} \subseteq
            \cbvSCtxtSet$.
            \qed
        \end{itemize}
    \end{itemize}
\end{proof}

\begin{lemma}
    \label{lem:cbv_SNF<d!>_is_FNF<d!>}
    Let $t \in \setCbvTerms$ and $u' \in \setBangTerms$ such that
    $\cbvToBangAGK{t} \bangArrSet*_S u'$ where $u'$ is a
    \bangSetSNF<d!>, then it is also a \bangSetFNF<d!>.
\end{lemma}
\begin{proof}
    From the observation that $\cbvToBangAGK{t}$ is a \bangSetINF<d!>,
    we show by induction on the length of the derivation
    $\cbvToBangAGK{t} \bangArrSet*_S u'$ that $u'$ is a
    \bangSetINF<d!> which concludes this case since $u'$ is a
    \bangSetFNF<d!> if and only if it is both a \bangSetSNF<d!> and
    \bangSetINF<d!>.
    \qed
\end{proof}

\begin{lemma}
    \label{lem:cbv_INF<d!>_is_FNF<d!>}
    Let $t \in \setCbvTerms$ and $u' \in \setBangTerms$ such that
    $\cbvToBangAGK{t} \bangArrSet*_S u'$ where $u'$ is a
    \bangSetSNF<d!>, then it is also a \bangSetFNF<d!>.
\end{lemma}
\begin{proof}
    From the observation that $\cbvToBangAGK{t}$ \bangSetSNF<d!>, we
    show by induction on the length of the derivation
    $\cbvToBangAGK{t} \bangArrSet*_S u'$ that $u'$ is a
    \bangSetSNF<d!> which concludes this case since $u'$ is a
    \bangSetFNF<d!> if and only if it is both a \bangSetSNF<d!> and
    \bangSetINF<d!>.
\end{proof}

\RecCbvPropSurface*
\begin{proof}
    Apply \Cref{lem: cbv preservation restricted} with
    \Cref{lem:cbv_surface_context_stability,lem:Implicitation,lem:cbv_INF<d!>_is_FNF<d!>}.
        \qed
\end{proof}

\section{Appendix: Proofs of \Cref{sec:Factorization}}
\label{secproofs:Factorization}

	\subsection{Internal Diligence for \BANGSymb}
	\label{prfSec:Internal_Diligence}

  \begin{definition}
    \label{d:Internal_context_reduction}
    \emph{Internal  reduction} $\bangArrSet_I<R>$ is extended to contexts
as expected:
\begin{itemize}
\item either the redex occurs in a \emph{subterm} of some internal  context,  
\item or the redex occurs in an internal \emph{subcontext}, where, in
  particular, $\bangSymbSubs$-redexes are only defined as
  $t\esub{x}{(\oc u)\esub{x_1}{u_1} \ldots \esub{x_i}{\bangFCtxt_i}
    \ldots \esub{x_n}{u_n}}$, but not as $t\esub{x}{\bangLCtxt<\oc
    \bangFCtxt>}$ or $\bangFCtxt\esub{x}{\bangLCtxt<\oc u>}$.
\end{itemize}   
\end{definition}

Examples of the first case are  $\Hole\esub{x}{\oc (\app{(\abs{w}{w})}{z})} \bangArrSet_I<dB> \Hole \esub{x}{\oc (w\esub{w}{z})}$, $\app{\Hole\, }{\oc \der{\oc z}} \bangArrSet_I<d!> \app{\Hole\, }{\oc z} $,  and $\Hole\esub{x}{\oc (z\esub{z}{\oc w}) }\bangArrSet_I<s!>  \Hole\esub{x}{\oc w}$
while examples of the second case are
$\oc (\app{(\abs{x}{x})}{\Hole})
\bangArrSet_I<dB> \oc (x\esub{x}{\Hole})$,
$\oc \der{\oc \Hole}
\bangArrSet_I<dB> \oc \Hole$, 
and $\oc (x\esub{x}{(\oc z)\esub{z}{\Hole} })\bangArrSet_I<s!>
  \oc( z\esub{z}{\Hole})$ but  $\oc (\app{(xx)\esub{x}{\oc \Hole}}{z})$
  %\not \bangArrSet_I<s!>  \oc (\app{\Hole\Hole}{z})$
  and $\oc(\Hole\esub{x}{!y})$ %
  %\not \bangArrSet_I<s!> \oc (\ldots)$
are $\bangArrSet_I<R>$-irreducible.     

\begin{lemma}
\label{l:Internal_Context_Stability}
Let $t \in \setBangTerms$. 
\begin{enumerate}
	% \item $t\isub{x}{u} \bangArrSet_I<R> t'\isub{x}{u}$;
	% \item $u\isub{x}{t} \bangArrSet*_I<R> u\isub{x}{t'}$;

\item For every $\bangICtxt \in \bangICtxtSet$, if $\bangICtxt  \bangArrSet_I<R> \bangFCtxt$ then $\bangFCtxt \in \bangICtxtSet$.

\item For every $\bangICtxt \in \bangICtxtSet$, if $t \bangArrSet_F<R> t'$ then $\bangICtxt<t>  \bangArrSet_I<R> \bangICtxt<t'>$.
	\item For every $\bangFCtxt \in \bangFCtxtSet$, if $\bangFCtxt \bangArrSet_I<\rel>
	\bangFCtxt'$ then $\bangFCtxt<t>  \bangArrSet_I<R> \bangFCtxt'<t>$.
\end{enumerate}
\end{lemma}
	
\begin{proof} \mbox{} 
  \begin{enumerate}
  \item   By induction on the definition of $\bangArrSet_I<\rel>$ on contexts.
  \item Since $t \bangArrSet_F<R> t'$, there is a full context $\bangFCtxt$ and $r,r' \in \setBangTerms$ such that $t = \bangFCtxt<r> \bangArrSet_F<R> \bangFCtxt<r'> = t'$ with $r \mapsto_\rel r'$. 
			As $\bangICtxt<\bangFCtxt>$ is an internal context, $\bangICtxt<t> = \bangICtxt<\bangFCtxt<r>> \bangArrSet_I<R> \bangICtxt<\bangFCtxt<r'>> = \bangICtxt<t'>$.

\item By induction on the definition of $\bangArrSet_I<\rel>$ on contexts.
		\qed
\end{enumerate}
\end{proof}

\begin{lemma}
		\label{lem:Commutation_I<d!>_I<R>}
		Let $\rel \in \{\bangSymbBeta, \bangSymbSubs\}$ and $t, u_1, u_2
		\in \setBangTerms$ such that $t \bangArrSet_I<d!> u_1$ and $t
		\bangArrSet_I<\rel> u_2$. Then there exists $s \in \setBangTerms$
		such that the diagram below commutes.
        \begin{equation*}
            \begin{array}{ccc}
                t                               &\bangLongArrSet{1cm}_I<R>  &u_2
            \\
                \bangLongDownArrBangISet{1cm}   &                           &\bangLongDownArrsBangISet{1cm}
            \\
                u_1                             &\bangLongArrSet{1cm}_I<R>  &s
            \end{array}
        \end{equation*}
\end{lemma}
\begin{proof} 
    We  analyze different (potentially overlapping) situations that at the
end cover all  possible cases. To close the diagrams we
use~\Cref{l:stability} and ~\Cref{l:Internal_Context_Stability}.

\begin{enumerate}
\item $t = \bangICtxt<\der{\bangLCtxt<\oc s_1>}>$ and $u_1 =
    \bangICtxt<\bangLCtxt<s_1>>$  and the step $t \bangArrSet_I<R>
    u_2$ occurs internally inside $\bangICtxt$, or fully inside
    $\bangLCtxt$, or $s_1$. We analyze all the possible cases.
    \begin{itemize}
    \item[\bltI] If $\bangICtxt \bangArrSet_I<R> \bangFCtxt$, then we
      know by \Cref{l:Internal_Context_Stability}.1 that there exist
      $\bangICtxt' \in \bangICtxtSet$ such that $\bangFCtxt =
      \bangICtxt'$. Then:
      \begin{equation*}
            \begin{array}{ccc}
                \bangICtxt<\der{\bangLCtxt<\oc s_1>}>       &\bangLongArrSet{1cm}_I<R>  &\bangICtxt'<\der{\bangLCtxt<\oc s_1>}> 
            \\[0.2cm]
                \bangLongDownArrBangISet{1cm}               &                           &\bangLongDownArrBangISet{1cm}
            \\[0.2cm]
                \bangICtxt<\bangLCtxt<s_1>>                 &\bangLongArrSet{1cm}_I<R>  &\bangICtxt'<\bangLCtxt<s_1>>
            \end{array}
        \end{equation*}

    \item[\bltI] If $\bangLCtxt \bangArrSet_F<R> \bangLCtxt'$, then:
      \begin{equation*}
            \begin{array}{ccc}
                \bangICtxt<\der{\bangLCtxt<\oc s_1>}>       &\bangLongArrSet{1cm}_I<R>  &\bangICtxt<\der{\bangLCtxt'<\oc s_1>}> 
            \\[0.2cm]
                \bangLongDownArrBangISet{1cm}               &                           &\bangLongDownArrBangISet{1cm}
            \\[0.2cm]
                \bangICtxt<\bangLCtxt<s_1>>                 &\bangLongArrSet{1cm}_I<R>  &\bangICtxt<\bangLCtxt'<s_1>>
            \end{array}
        \end{equation*}

    \item[\bltI] If $s_1 \bangArrSet_F<R> s'_1$, then:
        \begin{equation*}
            \begin{array}{ccc}
                \bangICtxt<\der{\bangLCtxt<\oc s_1>}>   &\bangLongArrSet{1cm}_I<R>  &\bangICtxt<\der{\bangLCtxt<\oc s'_1>}> 
            \\[0.2cm]
                \bangLongDownArrBangISet{1cm}           &                           &\bangLongDownArrBangISet{1cm}
            \\[0.2cm]
                \bangICtxt<\bangLCtxt<s_1>>             &\bangLongArrSet{1cm}_I<R>  &\bangICtxt<\bangLCtxt<s'_1>>
            \end{array}
          \end{equation*}
        \end{itemize}

\item $t = \bangICtxt<\app{\bangLCtxt<\abs{x}{s_1}>}{s_2}>$ and $u_2 =
    \bangICtxt<\bangLCtxt<s_1\esub{x}{s_2}>>$  and the step $t
    \bangArrSet_I<d!> u_1$ occurs internally inside $\bangICtxt$, or
    fully inside $\bangLCtxt$, $s_1$, or $s_2$. We analyze all the
    possible cases.
    \begin{itemize}
    \item[\bltI] If $\bangICtxt \bangArrSet_I<d!> \bangFCtxt$, then we
        know from \Cref{l:Internal_Context_Stability}.1 that there
        exist $\bangICtxt' \in \bangICtxtSet$ such that $\bangFCtxt =
        \bangICtxt'$. Then:
        \begin{equation*}
            \begin{array}{ccc}
                \bangICtxt<\app{\bangLCtxt<\abs{x}{s_1}>}{s_2}>     &\bangLongArrSet{1cm}_I<dB> &\bangICtxt<\bangLCtxt<s_1\esub{x}{s_2}>> 
            \\[0.2cm]
                \bangLongDownArrBangISet{1cm}                       &                           &\bangLongDownArrBangISet{1cm}
            \\[0.2cm]
                \bangICtxt'<\app{\bangLCtxt<\abs{x}{s_1}>}{s_2}>    &\bangLongArrSet{1cm}_I<dB> &\bangICtxt'<\bangLCtxt<s_1\esub{x}{s_2}>>
            \end{array}
        \end{equation*}

    \item[\bltI] If $\bangLCtxt \bangArrSet_F<d!> \bangLCtxt'$, then:
        \begin{equation*}
            \hspace{-1.5cm}
            \begin{array}{ccc}
                \bangICtxt<\app{\bangLCtxt<\abs{x}{s_1}>}{s_2}>     &\bangLongArrSet{1cm}_I<dB> &\bangICtxt<\bangLCtxt<s_1\esub{x}{s_2}>>
            \\[0.2cm]
                \bangLongDownArrBangISet{1cm}                       &                           &\bangLongDownArrBangISet{1cm}
            \\[0.2cm]
                \bangICtxt<\app{\bangLCtxt'<\abs{x}{s_1}>}{s_2}>    &\bangLongArrSet{1cm}_I<dB> &\bangICtxt<\bangLCtxt'<s_1\esub{x}{s_2}>>
            \end{array}
        \end{equation*}

    \item[\bltI] If $s_1 \bangArrSet_F<d!> s'_1$, then:
        \begin{equation*}
            \begin{array}{ccc}
                \bangICtxt<\app{\bangLCtxt<\abs{x}{s_1}>}{s_2}>     &\bangLongArrSet{1cm}_I<dB> &\bangICtxt<\bangLCtxt<s_1\esub{x}{s_2}>>
            \\[0.2cm]
                \bangLongDownArrBangISet{1cm}                       &                           &\bangLongDownArrBangISet{1cm}
            \\[0.2cm]
                \bangICtxt<\app{\bangLCtxt<\abs{x}{s'_1}>}{s_2}>    &\bangLongArrSet{1cm}_I<dB> &\bangICtxt<\bangLCtxt<s'_1\esub{x}{s_2}>>
            \end{array}
        \end{equation*}

    \item[\bltI] If $s_2 \bangArrSet_F<d!> s'_2$, then:
        \begin{equation*}
            \begin{array}{ccc}
                \bangICtxt<\app{\bangLCtxt<\abs{x}{s_1}>}{s_2}>     &\bangLongArrSet{1cm}_I<dB> &\bangICtxt<\bangLCtxt<s_1\esub{x}{s_2}>>
            \\[0.2cm]
                \bangLongDownArrBangISet{1cm}                       &                           &\bangLongDownArrBangISet{1cm}
            \\[0.2cm]
                \bangICtxt<\app{\bangLCtxt<\abs{x}{s_1}>}{s'_2}>    &\bangLongArrSet{1cm}_I<dB> &\bangICtxt<\bangLCtxt<s_1\esub{x}{s'_2}>> 
            \end{array}
        \end{equation*}
    \end{itemize}

\item $t = \bangICtxt<s_1\esub{x}{\bangLCtxt<\oc s_2>}>$ and $u_2 =
    \bangICtxt<\bangLCtxt<s_1\isub{x}{s_2}>>$  and the step $t
    \bangArrSet_I<d!> u_1$ occurs internally inside $\bangICtxt$, or
    fully inside $\bangLCtxt$, $s_1$ or $s_2$. We analyze all the
    possible  cases.
    \begin{itemize}
    \item[\bltI] If $\bangICtxt \bangArrSet_I<d!> \bangFCtxt$, then we
        know from \Cref{l:Internal_Context_Stability}.1 that there
        exist $\bangICtxt' \in \bangICtxtSet$ such that $\bangFCtxt =
        \bangICtxt'$. Then:
        \begin{equation*}
            \hspace{-1cm}
            \begin{array}{ccc}
                \bangICtxt<s_1\esub{x}{\bangLCtxt<\oc s_2>}>    &\bangLongArrSet{1cm}_I<s!> &\bangICtxt<\bangLCtxt<s_1\isub{x}{s_2}>>
            \\[0.2cm]
                \bangLongDownArrBangISet{1cm}                   &                           &\bangLongDownArrBangISet{1cm}
            \\[0.2cm]
                \bangICtxt'<s_1\esub{x}{\bangLCtxt<\oc s_2>}>   &\bangLongArrSet{1cm}_I<s!> &\bangICtxt'<\bangLCtxt<s_1\isub{x}{s_2}>>
            \end{array}
        \end{equation*}

    \item[\bltI] If $s_1 \bangArrSet_F<d!> s'_1$, then:  
        \begin{equation*}
            \hspace{-1cm}
            \begin{array}{ccc}
                \bangICtxt<s_1\esub{x}{\bangLCtxt<\oc s_2>}>    &\bangLongArrSet{1cm}_I<s!> &\bangICtxt<\bangLCtxt<s_1\isub{x}{s_2}>>
            \\[0.2cm]
                \bangLongDownArrBangISet{1cm}                   &                           &\bangLongDownArrBangISet{1cm}
            \\[0.2cm]
                \bangICtxt<s'_1\esub{x}{\bangLCtxt<\oc s_2>}>   &\bangLongArrSet{1cm}_I<s!> &\bangICtxt<\bangLCtxt<s'_1\isub{x}{s_2}>>
            \end{array}
        \end{equation*}

    \item[\bltI] If $\bangLCtxt \bangArrSet_F<d!> \bangLCtxt'$, then:
        \begin{equation*}
            \hspace{-1cm}
            \begin{array}{ccc}
                \bangICtxt<s_1\esub{x}{\bangLCtxt<\oc s_2>}>    &\bangLongArrSet{1cm}_I<s!> &\bangICtxt<\bangLCtxt<s_1\isub{x}{s_2}>>
            \\[0.2cm]
                \bangLongDownArrBangISet{1cm}                   &                           &\bangLongDownArrBangISet{1cm}
            \\[0.2cm]
                \bangICtxt<s_1\esub{x}{\bangLCtxt'<\oc s_2>}>   &\bangLongArrSet{1cm}_I<s!> &\bangICtxt<\bangLCtxt'<s_1\isub{x}{s_2}>>
            \end{array}
        \end{equation*}

    \item[\bltI] If $s_2 \bangArrSet_F<d!> s'_2$, then:
        \begin{equation*}
            \hspace{-1cm}
            \begin{array}{ccc}
                \bangICtxt<s_1\esub{x}{\bangLCtxt<\oc s_2>}>    &\bangLongArrSet{1cm}_I<s!> &\bangICtxt<\bangLCtxt<s_1\isub{x}{s_2}>>
            \\[0.2cm]
                \bangLongDownArrBangISet{1cm}                   &                           &\bangLongDownArrsBangISet{1cm}
            \\[0.2cm]
                \bangICtxt<s_1\esub{x}{\bangLCtxt<\oc s'_2>}>   &\bangLongArrSet{1cm}_I<s!> &\bangICtxt<\bangLCtxt<s_1\isub{x}{s'_2}>>
            \end{array}
        \end{equation*}
    \qed
    \end{itemize}
\end{enumerate}

\end{proof}

%\giulio{[G: In the proofs of \Cref{l:Internal_Context_Stability,lem:Local Confluence ->I<d!>,lem:Commutation_I<d!>_I<R>}, many lemmas are missing (for instance that $\bangICtxt<\bangFCtxt>$ and $\bangFCtxt<\bangICtxt>$ are internal). But they are rather tedious and obvious.]}

\begin{lemma}[Diamond of $\bangArrSet_I<d!>$]
		\label{lem:Local Confluence ->I<d!>}
		Let $t, u_1, u_2 \in \setBangTerms$. 
		If $t \bangArrSet_I<d!>
		u_1$, $t \bangArrSet_I<d!> u_2$ and $u_1 \neq u_2$, then there is
		$s \in \setBangTerms$ such that the diagram below commutes.
		\begin{equation*}
			\begin{array}{ccc}
				t                               &\bangLongArrSet{1cm}_I<d!> &u_2
				\\
				\bangLongDownArrBangISet{1cm}   &                           &\bangLongDownArrBangISet{1cm}
				\\
				u_1                             &\bangLongArrSet{1cm}_I<d!> &s
			\end{array}
		\end{equation*}
\end{lemma}
\begin{proof}
As $t \bangArrSet_I<d!> u_2$, then $t = \bangICtxt<\der{\bangLCtxt<\oc
s_1>}>$ and $u_2 = \bangICtxt<\bangLCtxt<s_1>>$. Different cases have
to be distinguished, since the reduction step $t \bangArrSet_I<d!>
u_1$ may occur internally inside $\bangICtxt$, or fully inside
$\bangLCtxt$, or $s_1$; in all these cases we use
\Cref{l:stability,l:Internal_Context_Stability} to conclude. 
\begin{itemize}
\item[\bltI] If $\bangICtxt \bangArrSet_I<d!> \bangFCtxt$, then we
    know from \Cref{l:Internal_Context_Stability}.1 that there exist
    $\bangICtxt' \in \bangICtxtSet$ such that $\bangFCtxt =
    \bangICtxt'$. Then:
    \begin{equation*}
        \hspace{-1cm}
        \begin{array}{ccc}
            \bangICtxt<\der{\bangLCtxt<\oc s_1>}>   &\bangLongArrSet{1cm}_I<d!>     &\bangICtxt<\bangLCtxt<s_1>>
        \\[0.2cm]
            \bangLongDownArrBangISet{1cm}           &                               &\bangLongDownArrBangISet{1cm}
        \\[0.2cm]
            \bangICtxt'<\der{\bangLCtxt<\oc s_1>}>  & \bangLongArrSet{1cm}_I<d!>    &\bangICtxt'<\bangLCtxt<s_1>>
        \end{array}
    \end{equation*}

\item[\bltI] If $\bangLCtxt \bangArrSet_F<d!> \bangLCtxt'$, then:
    \begin{equation*}
        \hspace{-1cm}
        \begin{array}{ccc}
            \bangICtxt<\der{\bangLCtxt<\oc s_1>}>   & \bangLongArrSet{1cm}_I<d!>    &\bangICtxt<\bangLCtxt<s_1>>
        \\[0.2cm]
            \bangLongDownArrBangISet{1cm}           &                               &\bangLongDownArrBangISet{1cm}
        \\[0.2cm]
            \bangICtxt<\der{\bangLCtxt'<\oc s_1>}>  &\bangLongArrSet{1cm}_I<d!>     &\bangICtxt<\bangLCtxt'<s_1>>
        \end{array}
    \end{equation*}

\item[\bltI] If $s_1 \bangArrSet_F<d!> s'_1$, then:
    \begin{equation*}
        \hspace{-1cm}
        \begin{array}{ccc}
            \bangICtxt<\der{\bangLCtxt<\oc s_1>}>   &\bangLongArrSet{1cm}_I<d!>     &\bangICtxt<\bangLCtxt<s_1>>
        \\[0.2cm]
            \bangLongDownArrBangISet{1cm}           &                               &\bangLongDownArrBangISet{1cm}
        \\[0.2cm]
            \bangICtxt<\der{\bangLCtxt<\oc s'_1>}>  &\bangLongArrSet{1cm}_I<d!>     &\bangICtxt<\bangLCtxt<s'_1>>
        \end{array}
    \end{equation*}
\qed
\end{itemize}
\end{proof}

\begin{corollary}[Confluence of $\bangArr_I<d!>$]
    \label{lem:Confluence_I<d!>}
    Let $t, u_1, u_2 \in \setBangTerms$ such that $t
    \bangArrSet*_I<d!> u_1$ and $t \bangArrSet*_I<d!> u_2$. Then there
    is $s \in \setBangTerms$ such that the diagram below commutes.
    \begin{equation*}
        \begin{array}{ccc}
            t                               &\bangLongArrSet{1cm}*_I<d!>    &u_1
        \\
            \bangLongDownArrsBangISet{1cm}  &                               &\bangLongDownArrsBangISet{1cm}
        \\
            u_2                             &\bangLongArrSet{1cm}*_I<d!>    &s
        \end{array}
    \end{equation*}
\end{corollary}
\begin{proof} 
    Immediate consequence of the diamond for $\bangArrSet_I<d!>$
    (\Cref{lem:Local Confluence ->I<d!>}).
    \qed
\end{proof}

\begin{lemma}[Internal Diligence]
    \label{lem:Internal_Diligence}%
    Let $t, u \in \setBangTerms$ such that $t \bangArrSet*_I u$.
    If $u$ is a \bangSetINF<d!>, then $t\bangArrSet*_{Iai} u$.
\end{lemma}
\begin{proof} \label{prf:Internal_Diligence}%
    Just apply \Cref{lem: Abstract Implicitation} to $\bangArrSet_I$,
    since its hypotheses are fulfilled
    (\Cref{lem:Confluence_I<d!>,lem:Commutation_I<d!>_I<R>}).
    \qed
\end{proof}

\subsection{Bang Factorization}

\begin{definition}
    \label{d:surface context-reduction}%
     \emph{Surface reduction} $\bangArrSet_S<R>$ is extended to contexts
as expected:
\begin{itemize}
\item either the redex occurs in a \emph{subterm} of some surface context,  
\item or the redex occurs in a surface \emph{subcontext}, where, in particular,
  $\bangSymbSubs$-redexes are only defined
as   
$t\esub{x}{(\oc u)\esub{x_1}{u_1} \ldots \esub{x_i}{\bangFCtxt_i}
\ldots \esub{x_n}{u_n}}$, but not
as $t\esub{x}{\bangLCtxt<\oc \bangFCtxt>}$ or % otherwise the hole would be duplicated
  $\bangFCtxt\esub{x}{\bangLCtxt<\oc u>}$.  
\end{itemize} 
\end{definition}

Remark that in particular $\bangArrSet_S<s!>$ reduction cannot make
the bang affecting the substituted subcontext disappear. Examples of
the first case are $\Hole\esub{x}{\app{(\abs{w}{w})}{z}}
\bangArrSet_S<dB> \Hole \esub{x}{w\esub{w}{z}}$, $\app{\Hole\,
}{\der{\oc z}} \bangArrSet_S<d!> \app{\Hole\, }{z} $ and
$\abs{x}{\Hole z\esub{z}{\oc x}} \bangArrSet_S<s!>\abs{x}{\Hole x}$,
while examples of the second case are $\app{(\abs{x}{x})}{\Hole}
\bangArrSet_S<dB> x\esub{x}{\Hole}$ and $ \Hole (x\esub{x}{\oc
y})\bangArrSet_S<s!> \Hole y$, but $(xx)\esub{x}{\oc \Hole}$ and
$(\Hole x)\esub{x}{\oc y}$ are $\bangArrSet_S<R>$-irreducible.

%    \victorTodo[Ne garder que les cas importants]
\begin{lemma}
    \label{l:surface_context_reduction}
    Let $t \in \setBangTerms$. Then, 
    \begin{enumerate}
%<<<<<<< HEAD
%     \item For every $\bangICtxt \in \bangICtxtSet$, if $\bangICtxt  \bangArrSet_S<R> \bangFCtxt$ then $\bangFCtxt \in \bangICtxtSet$.
%     \giulio{[G: FALSE! Take $\bangICtxtSet \ni \bangICtxt = \bangSCtxt<\der{\bangLCtxt<\oc \bangSCtxt'>}> \bangArrSet_S<d!> \bangSCtxt<\bangLCtxt< \bangSCtxt'>> \notin \bangICtxtSet$ or .$\bangICtxtSet \ni \bangICtxt = \bangSCtxt<t\esub{x}{\bangLCtxt<\oc \bangSCtxt'>}> \bangArrSet_S<s!> \bangSCtxt<\bangLCtxt< t\i`sub{x}{\bangSCtxt'}>> \notin \bangICtxtSet$.]}
%   \item For every $\bangSCtxt \in \cgiulio{\bangICtxtSet}{\bangSCtxtSet}$, if $\bangSCtxt  \bangArrSet_S<R> \bangFCtxt$ then $\bangFCtxt \in \bangSCtxtSet$.
%=======
    % \item For every $\bangICtxt \in \bangICtxtSet$, if $\bangICtxt  \bangArrSet_S<R> \bangFCtxt$ then $\bangFCtxt \in \bangICtxtSet$.
      % \giulio{[G: FALSE! Take $\bangICtxtSet \ni \bangICtxt = \bangSCtxt<\der{\bangLCtxt<\oc \bangSCtxt'>}> \bangArrSet_S<d!> \bangSCtxt<\bangLCtxt< \bangSCtxt'>> \notin \bangICtxtSet$.]}
   \item \label{item:ISdBI}%
        For every $\bangICtxt \in \bangICtxtSet$, if $\bangICtxt
        \bangArrSet_S<dB> \bangFCtxt$ then $\bangFCtxt \in
        \bangICtxtSet$.

   \item For every $\bangSCtxt \in \bangSCtxtSet$, if $\bangSCtxt
        \bangArrSet_S<R> \bangFCtxt$ then $\bangFCtxt \in
        \bangSCtxtSet$.
   % \item $t\isub{x}{u} \bangArrSjet_F<R> t'\isub{x}{u}$;
        % \item $u\isub{x}{t} \bangArrSet*_F<R> u\isub{x}{t'}$;
        % \item $\bangFCtxt<t>  \bangArrSet_F<R> \bangFCtxt<t'>$;
    \item\label{p:surface_context_reduction}%
        If $\bangFCtxt \bangArrSet_S<\rel> \bangFCtxt'$, then
        $\bangFCtxt<t>  \bangArrSet_S<R> \bangFCtxt'<t>$.

    \item If $t \bangArrSet_S<R> t'$, then $\bangSCtxt<t>
        \bangArrSet_S<R> \bangSCtxt<t'>$.
    \end{enumerate}
\end{lemma}
\begin{proof}
    All the points are by straightforward induction on $
    \bangArrSet_S<R>$.
    \qed
\end{proof}

\BangSFS*
\begin{proof} \label{prf:BangSFS}%
%\delia{{\tt [D: NOT READ]}}
We prove the following properties:
\begin{enumerate}
\item \textbf{(Termination)} The reductions
$\bangArrSet_S<dB>$, $\bangArrSet_S<s!>$ and $\bangArrSet_S<d!>$ are
terminating.

\item \textbf{(Row-swaps)}
    \begin{enumerate}
    \item $\bangArrSet_I<dB> \bangArrSet_S<dB>%
            \ \subseteq \ %
        \bangArrSet_S<dB> \bangArrSet_I<dB>$;

    \item $\bangArrSet_I<s!> \bangArrSet_S<s!>%
            \ \subseteq \ %
        \bangArrSet+_S<s!>\bangArrSet*_I<s!>$;

    \item $\bangArrSet_I<d!> \bangArrSet_S<d!>%
            \ \subseteq \ %
        \bangArrSet_S<d!> \bangArrSet_I<d!>%
        \,\cup\,%
        \bangArrSet_S<d!> \bangArrSet_S<d!>$.
    \end{enumerate}

\item \textbf{(Diagonal-swaps)}
    \begin{enumerate}
    \item $\bangArrSet_I<dB> \bangArrSet_S<s!>%
            \ \subseteq \
        \bangArrSet_S<s!> \bangArrSet*_F<dB>$;

    \item $\bangArrSet_I<dB> \bangArrSet_S<d!>%
            \ \subseteq \ %
        \bangArrSet_S<d!> \bangArrSet_I<dB>%
%        \,\cup\,
        \bangArrSet_S<d!> \bangArrSet_S<dB>$;

    \item $\bangArrSet_I<s!> \bangArrSet_S<dB>%
            \ \subseteq \ 
        \bangArrSet_S<dB> \bangArrSet_I<s!>$;

    \item $\bangArrSet_I<s!> \bangArrSet_S<d!>%
            \ \subseteq \
        \bangArrSet_S<d!> \bangArrSet_I<s!>%
        \;\cup\;%
        \bangArrSet_S<d!> \bangArrSet_S<s!>$;

    \item $\bangArrSet_I<d!> \bangArrSet_S<dB>%
            \ \subseteq \
        \bangArrSet_S<dB>\bangArrSet_I<d!>$;

    \item $\bangArrSet_I<d!> \bangArrSet_S<s!>%
            \ \subseteq \ 
        \bangArrSet_S<s!>\bangArrSet*_F<d!>$.
    \end{enumerate}
\end{enumerate}
thus deducing that $(\setBangTerms, \;
\{\bangFCtxt\bangCtxtPlug{\bangSymbBeta}\!,\,
\bangFCtxt\bangCtxtPlug{\bangSymbSubs}\!,\,
\bangFCtxt\bangCtxtPlug{\bangSymbBang}\})$ is an $SFS$.

\begin{enumerate}
	
\item \textbf{(Termination)} Reduction $\bangArrSet_S<d!>$ is terminating by \Cref{lemma:RecFullDistantBangSN}.
Reduction $\bangArrSet_S<dB>$ is terminating by \Cref{lemma:RecFullDistantBetaSN}.
The proof that reduction $\bangArrSet_S<s!>$ is terminating is an easy adaptation of the proof in \cite[Lemma 2.9]{AccattoliKesner12} to \BANGSymb.

\item \textbf{(Row-swaps)} Let $t, u, s \in \setBangTerms$.
    \begin{enumerate}
    \item If $t \bangArrSet_I<dB> u \bangArrSet_S<dB> s$ then $t =
        \bangICtxt<\app{\bangLCtxt<\abs{x}{t_1}>}{t_2}>$ and $u =
        \bangICtxt<\bangLCtxt<t_1\esub{x}{t_2}>>$ for some $\bangICtxt
        \in \bangICtxtSet$. Note that the step $u \bangArrSet_S<dB> s$
        can only occur inside $\bangICtxt$\ so that $\bangICtxt
        \bangArrSet_S<dB> \bangFCtxt$ and $\bangFCtxt = \bangICtxt'
        \in \bangICtxtSet$ by
        \Cref{l:surface_context_reduction}.\ref{item:ISdBI}. Thus
        $s=\bangICtxt'<\bangLCtxt<t_1\esub{x}{t_2}>>$. Using
        \Cref{l:surface_context_reduction}.\ref{p:surface_context_reduction}:
        \begin{equation*}
            \begin{array}{ccccc}
                t = &\bangICtxt<\app{\bangLCtxt<\abs{x}{t_1}>}{t_2}>    &\bangLongArrSet{1cm}_I<dB>     &\bangICtxt<\bangLCtxt<t_1\esub{x}{t_2}>> & = u
            \\[0.2cm]
                &\bangLongDownArrBetaSSet{1cm}                          &                               &\bangLongDownArrBetaSSet{1cm}
            \\[0.2cm]
                &\bangICtxt'<\app{\bangLCtxt<\abs{x}{t_1}>}{t_2}>       &\bangLongArrSet{1cm}_I<dB>     &\bangICtxt'<\bangLCtxt<t_1\esub{x}{t_2}>> & = s
            \end{array}
        \end{equation*}

    \item If $t \bangArrSet_I<s!> u \bangArrSet_S<s!> s$ then $t =
        \bangICtxt<t_1\esub{x}{ \bangLCtxt<\oc t_2>}>$ and $u =
        \bangICtxt<\bangLCtxt<t_1\isub{x}{t_2}>>$ for some $\bangICtxt
        \in \bangICtxtSet$. 
%        Note that the step $u \bangArrSet_S<s!> s$ can only occur inside
%        $\bangICtxt$, thus
%        $\bangICtxt \bangArrSet_S<dB>
%        \bangFCtxt$, and by \Cref{l:surface_context_reduction}.1
%        $\bangFCtxt =  \bangICtxt' \in \bangICtxtSet$.
%      Thus $s=\bangICtxt'<\bangLCtxt<t_1\isub{x}{t_2}>>$.
%      Using \Cref{l:surface_context_reduction}.3:
%        \begin{equation*}
%            \begin{array}{ccccc}
%               t = &\bangICtxt<t_1\esub{x}{\bangLCtxt<\oc t_2>}>        &\bangLongArrSet{1cm}_I<s!>     &\bangICtxt<\bangLCtxt<t_1\isub{x}{t_2}>> & = u
%            \\[0.2cm]
%                &\bangLongDownArrSubsSSet{1cm}                       &                               &\bangLongDownArrSubsSSet{1cm}
%            \\[0.2cm]
%                &\bangICtxt'<t_1\esub{x}{\bangLCtxt<\oc t_2>}>       &\bangLongArrSet{1cm}_I<s!>     &\bangICtxt'<\bangLCtxt<t_1\isub{x}{t_2}>> & = s
%            \end{array}
%        \end{equation*}
		Note that the step $u \bangArrSet_S<s!> s$ can occur in two different ways:
		\begin{itemize}
        \item either inside $\bangICtxt$ so that $\bangICtxt
			\bangArrSet_S<s!> \bangICtxt' \in \bangICtxtSet$; thus
			$s=\bangICtxt'<\bangLCtxt<t_1\isub{x}{t_2}>>$; using
			\Cref{l:surface_context_reduction}.\ref{p:surface_context_reduction}:
			\begin{equation*}
				\begin{array}{ccccc}
					t = &\bangICtxt<t_1\esub{x}{\bangLCtxt<\oc t_2>}>        &\bangLongArrSet{1cm}_I<s!>     &\bangICtxt<\bangLCtxt<t_1\isub{x}{t_2}>> & = u
					\\[0.2cm]
					&\bangLongDownArrSubsSSet{1cm}                       &                               &\bangLongDownArrSubsSSet{1cm}
					\\[0.2cm]
					&\bangICtxt'<t_1\esub{x}{\bangLCtxt<\oc t_2>}>       &\bangLongArrSet{1cm}_I<s!>     &\bangICtxt'<\bangLCtxt<t_1\isub{x}{t_2}>> & = s
				\end{array}
			\end{equation*}
		
        \item $\bangICtxt = \bangSCtxt<r\esub{y}{\bangLCtxt'<\oc
			\bangSCtxt'>}>$ with $\bangSCtxt, \bangSCtxt' \!\in
			\bangSCtxtSet$; so $u = \bangSCtxt <r
			\esub{y}{\bangLCtxt'<\oc \bangSCtxt'<\bangLCtxt<
			t_1\isub{x}{t_2}>>>}>$ and $s =
			\bangSCtxt<\bangLCtxt'<r\isub{y}{\bangSCtxt'<\bangLCtxt<t_1\isub{x}{t_2}>>}>>$;
			using
			\Cref{l:surface_context_reduction}.\ref{p:surface_context_reduction}:
			\begin{equation*}
				\begin{array}{ccccc}
					t = & \bangSCtxt <r \esub{y}{\bangLCtxt'<\oc \bangSCtxt'<t_1\esub{x}{ \bangLCtxt<\oc t_2>}>>}>    &\bangLongArrSet{.3cm}_I<s!>     & \bangSCtxt <r \esub{y}{\bangLCtxt'<\oc \bangSCtxt'<\bangLCtxt<t_1 \isub{x}{t_2}>>>}> & = u
					\\[0.2cm]
					&\bangLongDownArrSubsS{.5cm}              &                               &\bangLongDownArrSubsS{.5cm}
					\\[0.2cm]
					u' = & \bangSCtxt <\bangLCtxt'<r \isub{y}{\bangSCtxt'< t_1\esub{x}{\bangLCtxt<\oc t_2>}>}>>   &\bangLongArrSet{.3cm}^*_F<s!>     & \bangSCtxt<\bangLCtxt'<r\isub{y}{\bangSCtxt'<\bangLCtxt<t_1\isub{x}{t_2}>>}>> & = s
				\end{array}
			\end{equation*}
			We conclude by observing that the reduction sequence $u'
			\bangArrSet*_F<s!> s$ can be rearranged so that $u'
			\bangArrSet*_S<s!> \bangArrSet*_I<s!> s$ because all
			$\bangSymbSubs$-redexes $t_1\esub{x}{\bangLCtxt<\oc t_2>}$
			occurring in $r \isub{y}{\bangSCtxt'<
			t_1\esub{x}{\bangLCtxt<\oc t_2>}>}$ are not overlapping
			(as they replace different occurrences of $y$ in $r$),
			hence they can be fired independently, first the surface
			ones and then the internal ones.

        \item[\bltI] $\bangICtxt =
            \bangSCtxt<\bangICtxt'\esub{y}{\bangLCtxt'<\oc s'>}>$ for
            some $\bangICtxt' \in \bangICtxtSet$ and $s' \in
            \setBangTerms$: Thus
            \begin{equation*}
                \begin{array}{rcl}
                    t &=& \bangSCtxt<\bangICtxt'<t_1\esub{x}{\bangLCtxt<\oc t_2>}>\esub{y}{\bangLCtxt'<\oc s'>}>
                \\
                    u &=& \bangSCtxt<\bangLCtxt<\bangICtxt'<t_1\isub{x}{\oc t_2}>>\esub{y}{\bangLCtxt'<\oc s'>}>
                \\
                    s &=& \bangSCtxt<\bangLCtxt'<\bangLCtxt\isub{y}{s'}\bangCtxtPlug{\bangICtxt'\isub{y}{s'}\bangCtxtPlug{t_1\isub{y}{s'}\isub{x}{\oc t_2\isub{y}{s'}}}}>>
                \end{array}
            \end{equation*}
            Then, by taking $u' =
            \bangSCtxt<\bangLCtxt'<\bangICtxt'\isub{y}{s'}\bangCtxtPlug{t_1\isub{y}{s'}\esub{x}{\bangLCtxt\isub{y}{s'}\bangCtxtPlug{\oc
            t_2\isub{y}{s'}}}}>>$, one has:
            \begin{equation*}
        		\begin{array}{ccccc}
        			t                               &\bangLongArrSet{1cm}_I<s!>     &u
                \\[0.2cm]
                    \bangLongDownArrSubsSSet{1cm}   &                               &\bangLongDownArrSubsSSet{1cm}
                \\[0.2cm]
                    u'                              &\bangLongArrSet{1cm}_I<s!>     &s
        		\end{array}
        	\end{equation*}
		\end{itemize} 

    \item If $t \bangArrSet_I<d!> u \bangArrSet_S<d!> s$ then $t =
        \bangICtxt<\der{\bangLCtxt<\oc t'>}>$ and $u =
        \bangICtxt<\bangLCtxt<t'>>$ for some $\bangICtxt \!\in\! \bangICtxtSet$. 
%        Note that the step $u
%        \bangArr_S<d!> s$ can only occur inside $\bangICtxt$,
%        thus $\bangICtxt \bangArrSet_S<dB>
%        \bangFCtxt$, and by \Cref{l:surface_context_reduction}.1
%        we have  $\bangFCtxt =  \bangICtxt' \in \bangICtxtSet$.
%      Thus, $s=\bangICtxt'<\bangLCtxt<t'>>$.
%      Using \Cref{l:surface_context_reduction}.3:
%        \begin{equation*}
%            \begin{array}{ccccc}
%                t =  & \bangICtxt<\der{\bangLCtxt<\oc t'>}>                &\bangLongArrSet{1cm}_I<d!>     &\bangICtxt<\bangLCtxt<t'>> & = u
%            \\[0.2cm]
%                &\bangLongDownArrBangSSet{1cm}                       &                               &\bangLongDownArrBangSSet{1cm}
%            \\[0.2cm]
%                &\bangICtxt'<\der{\bangLCtxt<\oc t'>}>               &\bangLongArrSet{1cm}_I<d!>     &\bangICtxt'<\bangLCtxt<t'>> & = s
%            \end{array}
%        \end{equation*}
	Note that the step $u  \bangArr_S<d!> s$ can occur inside $\bangICtxt$ in two different~ways:
	\begin{itemize}
		\item either $\bangICtxt \bangArrSet_S<d!> \bangICtxt' \in \bangICtxtSet$; 
		thus, $s=\bangICtxt'<\bangLCtxt<t'>>$;
		      using \Cref{l:surface_context_reduction}.\ref{p:surface_context_reduction}:
		        \begin{equation*}
			            \begin{array}{ccccc}
				                t =  & \bangICtxt<\der{\bangLCtxt<\oc t'>}>                &\bangLongArrSet{1cm}_I<d!>     &\bangICtxt<\bangLCtxt<t'>> & = u
				            \\[0.2cm]
				                &\bangLongDownArrBangSSet{1cm}                       &                               &\bangLongDownArrBangSSet{1cm}
				            \\[0.2cm]
				                &\bangICtxt'<\der{\bangLCtxt<\oc t'>}>               &\bangLongArrSet{1cm}_I<d!>     &\bangICtxt'<\bangLCtxt<t'>> & = s
				            \end{array}
			        \end{equation*}
		        
		 \item or $\bangICtxt = \bangSCtxt<\der{\bangLCtxt<\oc \bangSCtxt'>}> \bangArr_S<d!> \bangSCtxt<\bangLCtxt<\bangSCtxt'>> \in \bangSCtxtSet$;
		 using \Cref{l:surface_context_reduction}.\ref{p:surface_context_reduction}:
		 \begin{equation*}
		 	\begin{array}{ccccc}
		 		t =  & \bangSCtxt<\der{\bangLCtxt'<\oc \bangSCtxt'<\der{\bangLCtxt<\oc t'>}> > }>            &\bangLongArrSet{.3cm}_I<d!>     & \bangSCtxt<\der{\bangLCtxt'<\oc \bangSCtxt'<\bangLCtxt<t'>> > }>   & = u
		 		\\[0.2cm]
		 		&\bangLongDownArrBangSSet{.5cm}                       &                               &\bangLongDownArrBangSSet{.5cm}
		 		\\[0.2cm]
		 		&\bangSCtxt<\bangLCtxt'<\bangSCtxt'<\der{\bangLCtxt<\oc t'>}> > >                &\bangLongArrSet{.3cm}_S<d!>     &\bangSCtxt<\bangLCtxt'< \bangSCtxt'<\bangLCtxt<t'>> > > & = s
		 	\end{array}
		 \end{equation*}
	\end{itemize}
	\end{enumerate}

\item \textbf{(Diagonal-swaps)}  Let $t, u, s \in \setBangTerms$.
    \begin{enumerate}
    \item Let $t
        \bangArr_I<dB> u \bangArr_S<s!> s$ then $t =
        \bangICtxt<\app{\bangLCtxt<\abs{x}{t_1}>}{t_2}>$ and $u =
        \bangICtxt<\bangLCtxt<t_1\esub{x}{t_2}>>$ for some $\bangICtxt \in \bangICtxtSet$. 
        Note that the step $u \bangArr_S<s!> s$ can occur in two different ways:
        \begin{itemize}
        	\item either $\bangICtxt \bangArrSet_S<s!> \bangICtxt'\in
			\bangICtxtSet$; thus, $s =
			\bangICtxt'<\bangLCtxt<t_1\esub{x}{t_2}>>$; using
			\Cref{l:surface_context_reduction}.\ref{p:surface_context_reduction}:
        	\begin{equation*}
        		\begin{array}{ccccc}
        			t = & \bangICtxt<\app{\bangLCtxt<\abs{x}{t_1}>}{t_2}>    &\bangLongArrSet{1cm}_I<dB>     & \bangICtxt<\bangLCtxt<t_1\esub{x}{t_2}>> & = u
        			\\[0.2cm]
        			&\bangLongDownArrSubsSSet{1cm}              &                               &\bangLongDownArrSubsSSet{1cm}
        			\\[0.2cm]
        			& \bangICtxt'<\app{\bangLCtxt<\abs{x}{t_1}>}{t_2}>   &\bangLongArrSet{1cm}_I<d!>     & \bangICtxt'<\bangLCtxt<t_1\esub{x}{t_2}>> & = s
        		\end{array}
        	\end{equation*}
        	
        	\item or $\bangICtxt = \bangSCtxt<r\esub{y}{\bangLCtxt'<\oc \bangSCtxt'>}>$ with $\bangSCtxt, \bangSCtxt' \!\in \bangSCtxtSet$; so, $u = \bangSCtxt <r \esub{y}{\bangLCtxt'<\oc \bangSCtxt'<\bangLCtxt< t_1\esub{x}{t_2}>>>}>$ and
        	$s = \bangSCtxt<\bangLCtxt'<r\isub{y}{\bangLCtxt<t_1\esub{x}{t_2}>}>>$;
        	using \Cref{l:surface_context_reduction}.\ref{p:surface_context_reduction}:        
        \end{itemize}
		\begin{equation*}
			\begin{array}{ccccc}
				t = & \bangSCtxt <r \esub{y}{\bangLCtxt'<\oc \bangSCtxt'<\bangLCtxt< \abs{x}{ t_1}>t_2>>}>    &\bangLongArrSet{.3cm}_I<dB>     & \bangSCtxt <r \esub{y}{\bangLCtxt'<\oc \bangSCtxt'<\bangLCtxt<t_1 \esub{x}{t_2}>>>}> & = u
				\\[0.2cm]
				&\bangLongDownArrSubsS{.5cm}              &                               &\bangLongDownArrSubsS{.5cm}
				\\[0.2cm]
				& \bangSCtxt <\bangLCtxt'<r \isub{y}{\bangSCtxt'<\app{\bangLCtxt< \abs{x}{ t_1}>}{t_2}>}>>   &\bangLongArrSet{.3cm}^*_F<dB>     & \bangSCtxt<\bangLCtxt'<r\isub{y}{\bangSCtxt'<\bangLCtxt<t_1\esub{x}{t_2}>>}>> & = s
			\end{array}
		\end{equation*}

    \item Let $t \bangArr_I<dB> u \bangArr_S<d!> s$ then $t =
        \bangFCtxtInt\bangCtxtPlug{\app{\bangLCtxt\bangCtxtPlug{\abs{x}{t_1}}}{t_2}}$ and 
        $u =
        \bangFCtxtInt\bangCtxtPlug{\bangLCtxt\bangCtxtPlug{t_1\esub{x}{t_2}}}$ for some $\bangICtxt \in \bangICtxtSet$.
	Note that the step $u \bangArr_S<d!> s$ can occur in two different ways:
	\begin{itemize}
		\item either $\bangICtxt \bangArrSet_S<s!> \bangICtxt'\in
			\bangICtxtSet$;  
		thus,
		$s = \bangICtxt'<\bangLCtxt<t_1\esub{x}{t_2}>>$;
		using \Cref{l:surface_context_reduction}.\ref{p:surface_context_reduction}:
		\begin{equation*}
			\begin{array}{ccccc}
				t = & \bangICtxt<\app{\bangLCtxt<\abs{x}{t_1}>}{t_2}>    &\bangLongArrSet{1cm}_I<dB>     & \bangICtxt<\bangLCtxt<t_1\esub{x}{t_2}>> & = u
				\\[0.2cm]
				&\bangLongDownArrBangSSet{1cm}              &                               &\bangLongDownArrBangSSet{1cm}
				\\[0.2cm]
				& \bangICtxt'<\app{\bangLCtxt<\abs{x}{t_1}>}{t_2}>   &\bangLongArrSet{1cm}_I<dB>     & \bangICtxt'<\bangLCtxt<t_1\esub{x}{t_2}>> & = s
			\end{array}
		\end{equation*}
	
		\item or $\bangICtxt = \bangSCtxt<\der{\bangLCtxt<\oc \bangSCtxt'>}> \bangArr_S<d!> \bangSCtxt<\bangLCtxt<\bangSCtxt'>> \in \bangSCtxtSet$;
		using \Cref{l:surface_context_reduction}.\ref{p:surface_context_reduction}:
		\begin{equation*}
			\begin{array}{ccccc}
				t =  & \bangSCtxt<\der{\bangLCtxt'<\oc \bangSCtxt'<\app{\bangLCtxt\bangCtxtPlug{\abs{x}{t_1}}}{t_2}> > }>            &\bangLongArrSet{.3cm}_I<dB>     & \bangSCtxt<\der{\bangLCtxt'<\oc \bangSCtxt'<\bangLCtxt<t_1 \esub{x}{t_2}>> > }>   & = u
				\\[0.2cm]
				&\bangLongDownArrBangSSet{0.5cm}                       &                               &\bangLongDownArrBangSSet{.5cm}
				\\[0.2cm]
				&\bangSCtxt<\bangLCtxt'<\bangSCtxt'<\app{\bangLCtxt\bangCtxtPlug{\abs{x}{t_1}}}{t_2}> > >                &\bangLongArrSet{.3cm}_S<dB>     &\bangSCtxt<\bangLCtxt'< \bangSCtxt'<\bangLCtxt<t_1\esub{x}{t_2}>> > > & = s
			\end{array}
		\end{equation*}
	\end{itemize}

    \item  Let $t \bangArr_I<s!> u \bangArr_S<dB> s$ then $t =
    \bangFCtxtInt\bangCtxtPlug{t_1\esub{x}{\bangLCtxt\bangCtxtPlug{\oc
    t_2}}}$ and $u =
    \bangFCtxtInt\bangCtxtPlug{\bangLCtxt\bangCtxtPlug{t_1\isub{x}{t_2}}}$
    for some $\bangICtxt \in \bangICtxtSet$. Note that the step $u
    \bangArr_S<dB> s$ can only occur inside $\bangICtxt$, thus
    $\bangICtxt \bangArr_S<dB> \bangFCtxt$ and by
    \Cref{l:surface_context_reduction}.\ref{item:ISdBI}  $\bangFCtxt =
    \bangICtxt' \in \bangICtxtSet$. Thus, 	$s =
    \bangICtxt'<\bangLCtxt<t_1\isub{x}{t_2}>>$. Using
    \Cref{l:surface_context_reduction}.\ref{p:surface_context_reduction}:
    	\begin{equation*}
    		\begin{array}{ccccc}
    			t = & \bangICtxt<t_1 \esub{x}{\bangLCtxt<\oc t_2>}>    &\bangLongArrSet{1cm}_I<s!>     & \bangICtxt<\bangLCtxt<t_1\isub{x}{t_2}>> & = u
    			\\[0.2cm]
    			&\bangLongDownArrBetaSSet{1cm}              &                               &\bangLongDownArrBetaSSet{1cm}
    			\\[0.2cm]
    			& \bangICtxt'<t_1 \esub{x}{\bangLCtxt<\oc t_2>}>   &\bangLongArrSet{1cm}_I<s!>     & \bangICtxt'<\bangLCtxt<t_1\isub{x}{t_2}>> & = s
    		\end{array}
    	\end{equation*}

    \item  Let $t \bangArr_I<s!> u \bangArr_S<d!> s$ then $t =
        \bangFCtxtInt\bangCtxtPlug{t_1\esub{x}{\bangLCtxt\bangCtxtPlug{\oc
        t_2}}}$ and $u =
        \bangFCtxtInt\bangCtxtPlug{\bangLCtxt\bangCtxtPlug{t_1\isub{x}{t_2}}}$ for some $\bangICtxt \in \bangICtxtSet$.
	Note that the step $u \bangArr_S<d!> s$ can occur inside $\bangICtxt$ in two different ways:
	\begin{itemize}
		\item either $\bangICtxt \bangArr_S<d!> \bangICtxt' \in \bangICtxtSet$; 
		thus,
		$s = \bangICtxt'<\bangLCtxt<t_1\isub{x}{t_2}>>$;
		using \Cref{l:surface_context_reduction}.\ref{p:surface_context_reduction}:
		\begin{equation*}
			\begin{array}{ccccc}
				t = & \bangICtxt<t_1 \esub{x}{\bangLCtxt<\oc t_2>}>    &\bangLongArrSet{1cm}_I<s!>     & \bangICtxt<\bangLCtxt<t_1\isub{x}{t_2}>> & = u
				\\[0.2cm]
				&\bangLongDownArrBangSSet{1cm}              &                               &\bangLongDownArrBangSSet{1cm}
				\\[0.2cm]
				& \bangICtxt'<t_1 \esub{x}{\bangLCtxt<\oc t_2>}>   &\bangLongArrSet{1cm}_I<s!>     & \bangICtxt'<\bangLCtxt<t_1\isub{x}{t_2}>> & = s
			\end{array}
		\end{equation*}
		
		\item or $\bangICtxt = \bangSCtxt<\der{\bangLCtxt<\oc \bangSCtxt'>}> \bangArr_S<d!> \bangSCtxt<\bangLCtxt<\bangSCtxt'>> \in \bangSCtxtSet$;
		using \Cref{l:surface_context_reduction}.\ref{p:surface_context_reduction}:
		\begin{equation*}
			\begin{array}{ccccc}
				t =  & \bangSCtxt<\der{\bangLCtxt'<\oc \bangSCtxt'<t_1{\esub{x}{\bangLCtxt<\oc t_2>}}> > }>            &\bangLongArrSet{.3cm}_I<s!>     & \bangSCtxt<\der{\bangLCtxt'<\oc \bangSCtxt'<\bangLCtxt<t_1 \isub{x}{t_2}>> > }>   & = u
				\\[0.2cm]
				&\bangLongDownArrBangSSet{0.5cm}                       &                               &\bangLongDownArrBangSSet{.5cm}
				\\[0.2cm]
				&\bangSCtxt<\bangLCtxt'<\bangSCtxt'<t_1\esub{x}{\bangLCtxt<\oc t_2>}> > >                &\bangLongArrSet{.3cm}_S<s!>     &\bangSCtxt<\bangLCtxt'< \bangSCtxt'<\bangLCtxt<t_1\isub{x}{t_2}>> > > & = s
			\end{array}
		\end{equation*}
	\end{itemize}

    \item Let $t \bangArr_I<d!> u \bangArr_S<dB> s$ then $t =
        \bangFCtxtInt\bangCtxtPlug{\der{\bangLCtxt\bangCtxtPlug{\oc
        t_1}}}$ and $u =
        \bangFCtxtInt\bangCtxtPlug{\bangLCtxt\bangCtxtPlug{t_1}}$ for some $\bangICtxt \in \bangICtxtSet$.
		Note that the step $u
		\bangArrSet_S<dB> s$ can only occur inside $\bangICtxt$, so
		$\bangICtxt \bangArrSet_S<dB>
		\bangFCtxt$, and by \Cref{l:surface_context_reduction}.\ref{item:ISdBI} 
		$\bangFCtxt =  \bangICtxt' \in \bangICtxtSet$.
		Thus, $s=\bangICtxt'<\bangLCtxt<t_1>>$.
		Using \Cref{l:surface_context_reduction}.\ref{p:surface_context_reduction}:
		\begin{equation*}
			\begin{array}{ccccc}
				t = &\bangICtxt<\der{\bangLCtxt<\oc t_1>}>     &\bangLongArrSet{1cm}_I<d!>     &\bangICtxt<\bangLCtxt<t_1>> & = u
				\\[0.2cm]
				&\bangLongDownArrBetaSSet{1cm}                       &                               &\bangLongDownArrBetaSSet{1cm}
				\\[0.2cm]
				&\bangICtxt'<\der{\bangLCtxt<\oc t_1>}>    &\bangLongArrSet{1cm}_I<d!>     &\bangICtxt'<\bangLCtxt<t_1>> & = s
			\end{array}
		\end{equation*}

		\item Let $t \bangArr_I<d!> u \bangArr_S<s!> s$ then $t =
		        \bangFCtxtInt\bangCtxtPlug{\der{\bangLCtxt\bangCtxtPlug{\oc
					        t_1}}}$ and $u =
		        \bangFCtxtInt\bangCtxtPlug{\bangLCtxt\bangCtxtPlug{t_1}}$ for some $\bangICtxt \in \bangICtxtSet$. 
			Note that the step $u \bangArrSet_S<s!> s$ can occur in two different ways:
			\begin{itemize}
				\item either $\bangICtxt \bangArrSet_S<s!> \bangICtxt'
				\in  \bangICtxtSet$; thus,
				$s=\bangICtxt'<\bangLCtxt<t_1\isub{x}{t_2}>>$; using
				\Cref{l:surface_context_reduction}.\ref{p:surface_context_reduction}:
				\begin{equation*}
					\begin{array}{ccccc}
						t = &\bangICtxt<\der{\bangLCtxt<\oc t_1>}>        &\bangLongArrSet{1cm}_I<d!>     &\bangICtxt<\bangLCtxt<t_1>> & = u
						\\[0.2cm]
						&\bangLongDownArrSubsSSet{1cm}                       &                               &\bangLongDownArrSubsSSet{1cm}
						\\[0.2cm]
						&\bangICtxt'<\der{\bangLCtxt<\oc t_1>}>       &\bangLongArrSet{1cm}_I<d!>     &\bangICtxt'<\bangLCtxt<t_1>> & = s
					\end{array}
				\end{equation*}
				
				\item or $\bangICtxt = \bangSCtxt<r\esub{y}{\bangLCtxt'<\oc \bangSCtxt'>}>$ with $\bangSCtxt, \bangSCtxt' \!\in \bangSCtxtSet$; 
				so $u = \bangSCtxt <r \esub{y}{\bangLCtxt'<\oc \bangSCtxt'<\bangLCtxt<t_1>>>}>$ and
				$s = \bangSCtxt<\bangLCtxt'<r\isub{y}{\bangSCtxt'<\bangLCtxt<t_1>>}>>$;
				using \Cref{l:surface_context_reduction}.\ref{p:surface_context_reduction}:
			\end{itemize} 
			\begin{equation*}
				\begin{array}{ccccc}
					t = & \bangSCtxt <r \esub{y}{\bangLCtxt'<\oc \bangSCtxt'<\der{ \bangLCtxt<\oc t_1>}>>}>    &\bangLongArrSet{.3cm}_I<d!>     & \bangSCtxt <r \esub{y}{\bangLCtxt'<\oc \bangSCtxt'<\bangLCtxt<t_1 >>>}> & = u
					\\[0.2cm]
					&\bangLongDownArrSubsS{.5cm}              &                               &\bangLongDownArrSubsS{.5cm}
					\\[0.2cm]
					& \bangSCtxt <\bangLCtxt'<r \isub{y}{\bangSCtxt'< \der{\bangLCtxt<\oc t_1>}>}>>   &\bangLongArrSet{.3cm}^*_F<d!>     & \bangSCtxt<\bangLCtxt'<r\isub{y}{\bangSCtxt'<\bangLCtxt<t_1>>}>> & = s
				\end{array}
			\end{equation*}
    \end{enumerate}
\end{enumerate}
\end{proof}

\subsection{Call-by-Name Factorization}

\begin{lemma}
    \label{lem: Cbn Internal mapped to Bang Internal}%
    The following two properties hold:
    \begin{enumerate}
    \item \textbf{(\CBNSymb $\rightarrow$ \BANGSymb)}
    Let $\cbnICtxt \in \cbnICtxtSet$, then
        $\cbnToBangAGK{\cbnICtxt} \in \bangICtxtSet$.%
    \item \textbf{(\BANGSymb $\rightarrow$ \CBNSymb)} %
        Let $\bangICtxt \in \bangICtxtSet$ and $\cbnFCtxt \in
        \cbnFCtxtSet$ such that $\cbnToBangAGK{\cbnFCtxt} =
        \bangICtxt$, then $\cbnFCtxt \in \cbnICtxtSet$.
    \end{enumerate}
\end{lemma}
\begin{proof} ~%
    \begin{enumerate}
\item \textbf{(\CBNSymb $\rightarrow$ \BANGSymb)} By induction on
    $\cbnICtxt \in \cbnICtxtSet$:
    \begin{itemize}
    \item[\bltI] $\cbnICtxt = \app[\,]{s}{\cbnFCtxt}$ for some $s \in
        \setCbnTerms$ and $\cbnFCtxt \in \cbnFCtxtSet$, then 
        $\cbnToBangAGK{\cbnICtxt} =
        \cbnToBangAGK{(\app[\,]{s}{\cbnFCtxt})} =
        \app[\,]{\cbnToBangAGK{s}}{\oc\cbnToBangAGK{\cbnFCtxt}} =
        (\app[\,]{\cbnToBangAGK{s}}{\Hole})\bangCtxtPlug{\oc\cbnToBangAGK{\cbnFCtxt}}
        \in \bangSCtxtSet<\oc\bangFCtxtSet> \subseteq \bangICtxtSet$.

    \item[\bltI] $\cbnICtxt = s\esub{x}{\cbnFCtxt}$ for some $s \in
        \setCbnTerms$ and $\cbnFCtxt \in \cbnFCtxtSet$, then 
        $\cbnToBangAGK{\cbnICtxt} =
        \cbnToBangAGK{(s\esub{x}{\cbnFCtxt})} =
        \cbnToBangAGK{s}\esub{x}{\oc\cbnToBangAGK{\cbnFCtxt}} =
        (\cbnToBangAGK{s}\esub{x}{\Hole})\bangCtxtPlug{\oc\cbnToBangAGK{\cbnFCtxt}}
        \in \bangSCtxtSet<\oc\bangFCtxtSet> \subseteq \bangICtxtSet$.

    \item[\bltI] $\cbnICtxt = \cbnSCtxt<\cbnICtxt'>$ for some
        $\cbnSCtxt \in \cbnSCtxtSet \backslash \{\Hole\}$ and
        $\cbnICtxt' \in \cbnICtxtSet$: By induction on $\cbnSCtxt$:
        \begin{itemize}
        \item[\bltII] $\cbnSCtxt = \Hole$: Then by \ih on
            $\cbnICtxt'$, one has that $\cbnToBangAGK{{\cbnICtxt'}}
            \in \bangICtxtSet$ which concludes this case since
            $\cbnICtxt = \cbnSCtxt<\cbnICtxt'> = \cbnICtxt'$.

        \item[\bltII] $\cbnSCtxt = \abs{x}{\cbnSCtxt'}$: Then by \ih
            on $\cbnSCtxt'$, one has that
            $\cbnToBangAGK{(\cbnSCtxt'<\cbnICtxt'>)} \in
            \bangICtxtSet$ and thus $\cbnToBangAGK{\cbnICtxt} =
            \cbnToBangAGK{(\abs{x}{\cbnSCtxt'<\cbnICtxt'>})} =
            \abs{x}{\cbnToBangAGK{(\cbnSCtxt'<\cbnICtxt'>)}} =
            (\abs{x}{\Hole})\bangCtxtPlug{\cbnToBangAGK{(\cbnSCtxt'<\cbnICtxt'>)}}
            \in \bangSCtxtSet<\bangICtxtSet> \subseteq \bangICtxtSet$.

        \item[\bltII] $\cbnSCtxt = \app[\,]{\cbnSCtxt'}{t}$: Then by
            \ih on $\cbnSCtxt'$, one has that
            $\cbnToBangAGK{(\cbnSCtxt'<\cbnICtxt'>)} \in
            \bangICtxtSet$ and thus $\cbnToBangAGK{\cbnICtxt} =
            \cbnToBangAGK{(\app{\cbnSCtxt'<\cbnICtxt'>}{t})} =
            \app[\,]{\cbnToBangAGK{(\cbnSCtxt'<\cbnICtxt'>)}}{\oc\cbnToBangAGK{t}}
            =
            (\app[\,]{\Hole}{\oc\cbnToBangAGK{t}})\bangCtxtPlug{\cbnToBangAGK{(\cbnSCtxt'<\cbnICtxt'>)}}
            \in \bangSCtxtSet<\bangICtxtSet> \subseteq \bangICtxtSet$.

        \item[\bltII] $\cbnSCtxt = \cbnSCtxt'\esub{x}{t}$: Then by \ih
            on $\cbnSCtxt'$, one has that
            $\cbnToBangAGK{(\cbnSCtxt'<\cbnICtxt'>)} \in
            \bangICtxtSet$ and thus $\cbnToBangAGK{\cbnICtxt} =
            \cbnToBangAGK{(\cbnSCtxt'<\cbnICtxt'>\esub{x}{t})} =
            \cbnToBangAGK{(\cbnSCtxt'<\cbnICtxt'>)}\esub{x}{\oc\cbnToBangAGK{t}}
            =
            (\Hole\esub{x}{\oc\cbnToBangAGK{t}})\bangCtxtPlug{\cbnToBangAGK{(\cbnSCtxt'<\cbnICtxt'>)}}
            \in \bangSCtxtSet<\bangICtxtSet> \subseteq \bangICtxtSet$.
        \end{itemize}
    \end{itemize}

\item \textbf{(\BANGSymb $\rightarrow$ \CBNSymb)} Let $\bangICtxt \in
    \bangICtxtSet$ and $\cbnFCtxt \in \cbnFCtxtSet$ such that
    $\cbnToBangAGK{\cbnFCtxt} = \bangICtxt$. By induction on
    $\bangICtxt \in \bangICtxtSet$:
    \begin{itemize}
    \item[\bltI] $\bangICtxt = \oc\bangFCtxt$ for some $\bangFCtxt \in
        \bangFCtxtSet$. Impossible since one cannot have
        $\cbnToBangAGK{\cbnFCtxt} = \oc\bangFCtxt$ by definition of
        the \CBNSymb embedding on contexts.

    \item[\bltI] $\bangICtxt = \bangSCtxt<\bangICtxt'>$ for some $\bangSCtxt
        \in \bangSCtxtSet\backslash\{\Hole\}$: By induction on
        $\bangSCtxt \in \bangSCtxtSet$:
        \begin{itemize}
        \item[\bltII] $\bangSCtxt = \Hole$: Then
            $\cbnToBangAGK{\cbnFCtxt} = \bangSCtxt<\bangICtxt'> =
            \bangICtxt'$ and by \ih on $\bangICtxt'$, one concludes
            that $\cbnFCtxt \in \cbnICtxtSet$.

        \item[\bltII] $\bangSCtxt = \abs{x}{\bangSCtxt'}$: Then
            $\bangSCtxt<\bangICtxt'> =
            \abs{x}{\bangSCtxt'<\bangICtxt'>}$ thus necessarily
            $\cbnFCtxt = \abs{x}{\cbnFCtxt'}$ for some $\cbnFCtxt' \in
            \cbnFCtxtSet$ such that $\cbnToBangAGK{{\cbnFCtxt'}} =
            \bangSCtxt'<\bangICtxt'>$. By \ih on $\bangSCtxt'$, one
            has that $\cbnFCtxt' \in \cbnICtxtSet$ so that $\cbnFCtxt
            = \abs{x}{\bangFCtxt'} =
            (\abs{x}{\Hole})\cbnCtxtPlug{\cbnFCtxt'} \in
            \cbnSCtxtSet<\cbnICtxtSet> \subseteq \cbnICtxtSet$.

        \item[\bltII] $\bangSCtxt = \app{\bangSCtxt'}{s'}$: Then
            $\bangSCtxt<\bangICtxt'> =
            \app{\bangSCtxt'<\bangICtxt'>}{s'}$ thus necessarily
            $\cbnFCtxt = \app[\,]{\cbnFCtxt'}{s}$ for some $\cbnFCtxt'
            \in \cbnFCtxtSet$ and $s \in \setCbnTerms$ such that
            $\cbnToBangAGK{{\cbnFCtxt'}} = \bangSCtxt'<\bangICtxt'>$
            and $\oc\cbnToBangAGK{s} = s'$. By \ih on $\bangSCtxt'$,
            one has that $\cbnFCtxt' \in \cbnICtxtSet$ so that
            $\cbnFCtxt = \app[\,]{\bangFCtxt'}{s} =
            (\app[\,]{\Hole}{s})\cbnCtxtPlug{\cbnFCtxt'} \in
            \cbnSCtxtSet<\cbnICtxtSet> \subseteq \cbnICtxtSet$.

        \item[\bltII] $\bangSCtxt = \app[\,]{s'}{\bangSCtxt'}$: Then
            $\bangSCtxt<\bangICtxt'> =
            \app[\,]{s'}{\bangSCtxt'<\bangICtxt'>}$ thus necessarily
            $\cbnFCtxt = \app[\,]{s}{\cbnFCtxt'}$ for some $s \in
            \setCbnTerms$ and $\cbnFCtxt' \in \cbnFCtxtSet$ such that
            $\cbnToBangAGK{s} = s'$ and
            $\oc\cbnToBangAGK{{\cbnFCtxt'}} =
            \bangSCtxt'<\bangICtxt'>$, which concludes this case since
            $\cbnFCtxt = \app[\,]{s}{\cbnFCtxt'} \in
            \app[\,]{s}{\cbnFCtxtSet} \subseteq \cbnICtxtSet$.

        \item[\bltII] $\bangSCtxt = \bangSCtxt'\esub{x}{s'}$: Then
            $\bangSCtxt<\bangICtxt'> =
            \bangSCtxt'<\bangICtxt'>\esub{x}{s'}$ thus necessarily
            $\cbnFCtxt = \cbnFCtxt'\esub{x}{s}$ for some $\cbnFCtxt'
            \in \cbnFCtxtSet$ and $s \in \setCbnTerms$ such that
            $\cbnToBangAGK{{\cbnFCtxt'}} = \bangSCtxt'<\bangICtxt'>$
            and $\oc\cbnToBangAGK{s} = s'$. By \ih on $\bangSCtxt'$,
            one has that $\cbnFCtxt' \in \cbnICtxtSet$ so that
            $\cbnFCtxt = \bangFCtxt'\esub{x}{s} =
            (\Hole\esub{x}{s})\cbnCtxtPlug{\cbnFCtxt'} \in
            \cbnSCtxtSet<\cbnICtxtSet> \subseteq \cbnICtxtSet$.

        \item[\bltII] $\bangSCtxt = s'\esub{x}{\bangSCtxt'}$: Then
            $\bangSCtxt<\bangICtxt'> =
            s'\esub{x}{\bangSCtxt'<\bangICtxt'>}$ thus necessarily
            $\cbnFCtxt = s\esub{x}{\cbnFCtxt'}$ for some $s \in
            \setCbnTerms$ and $\cbnFCtxt' \in \cbnFCtxtSet$ such that
            $\cbnToBangAGK{s} = s'$ and
            $\oc\cbnToBangAGK{{\cbnFCtxt'}} =
            \bangSCtxt'<\bangICtxt'>$, which concludes this case since
            $\cbnFCtxt = s\esub{x}{\cbnFCtxt'} \in
            s\esub{x}{\cbnFCtxtSet} \subseteq \cbnICtxtSet$.

        \item[\bltII] $\bangSCtxt = \der{\bangSCtxt'}$: Impossible
            since one cannot have $\cbnToBangAGK{\cbnFCtxt} =
            \der{\bangSCtxt}$ by definition of the \CBNSymb embedding
            on contexts.
            \qed
        \end{itemize}
    \end{itemize}
\end{enumerate}%
\end{proof}

% \begin{corollary}
%     \label{lem:cbn_prop_internal_extended}
%     The following properties hold:
%     \begin{itemize}
%     \item[\bltI] \textbf{(Normal Forms)}: $\forall\, t \in
%         \setCbnTerms,$ \quad%
%         \begin{equation*}
%             t \text{ is a \cbnSetINF}
%                 \quad\Leftrightarrow\quad
%             \cbnToBangAGK{t} \text{ is a \bangSetINF}
%         \end{equation*}

%     \item[\bltI] \textbf{(Stability)}: $\forall\, t, \in
%         \setCbnTerms,\; \forall\, u \in \setBangTerms,$
%         \begin{equation*}
%             \cbnToBangAGK{t} \bangArrSet*_I u'
%                 \quad \Rightarrow \quad
%             \exists\, u \in \setCbnTerms, \; \cbnToBangAGK{u} = u'
%         \end{equation*}

%     \item[\bltI] \textbf{(Simulation and Reverse Simulation)}:
%         $\forall\, t, u \in \setCbnTerms,$
%         \begin{equation*}
%             t \cbnArrSet*_I u
%                 \quad \Leftrightarrow \quad
%             \cbnToBangAGK{t} \bangArrSet*_I \cbnToBangAGK{u}
%         \end{equation*}
%         Moreover, the number of $\cbnSymbBeta/\cbnSymbSubs$-steps
%         matches the number $\bangSymbBeta/\bangSymbSubs$-steps.
%     \end{itemize}
% \end{corollary}
\RecCbnPropInternal*
\begin{proof}
    Immediate consequence of \Cref{lem: cbn preservation restricted,lem: Cbn
    Internal mapped to Bang Internal}.
\qed
\end{proof}

% \RecCbnPropInternal*
% \begin{proof}
%     Third item of \Cref{lem:cbn_prop_internal_extended}.
% \end{proof}

\subsection{Call-by-Value Factorization}

\begin{lemma}
    \label{lem:cbv_internal_context_stability}
    The following property hold:
    \begin{enumerate}
    \item Let $\cbvICtxt \in \cbvICtxtSet$, then %
        $\cbvToBangAGK{\cbvICtxt}, %
        \cbvToBangAGK{\cbvICtxt}\bangCtxtPlug{\der{\Hole}}, %
        \cbvToBangAGK[*]{\cbvICtxt}, %
        \cbvToBangAGK[*]{\cbvRmDst{\cbvICtxt}} \in \bangICtxtSet$.

    \item Let $\bangICtxt \in \bangICtxtSet$ and $\cbvFCtxt \in
        \cbvFCtxtSet$ such that either $\cbvToBangAGK{\cbvFCtxt} =
        \bangICtxt$ or $\cbvToBangAGK[*]{\cbvFCtxt} = \bangICtxt$ or
        $\cbvToBangAGK{\cbvFCtxt}\bangCtxtPlug{\der{\Hole}} =
        \bangICtxt$, then $\cbvFCtxt \in \cbvICtxtSet$.
    \end{enumerate}
\end{lemma}
\begin{proof} ~
    \begin{enumerate}
    \item Notice that if $\cbvICtxt \in \cbvICtxtSet$ then
        $\cbvRmDst{\cbvICtxt} \in \cbvICtxtSet$, and that if
        $\bangICtxt \in \bangICtxtSet$ then $\bangICtxt<\der{\Hole}>
        \in \bangICtxtSet$. Thanks to these observations, it becomes
        sufficient to show that if $\cbvICtxt \in \cbvICtxtSet$, then
        $\cbvToBangAGK{\cbvICtxt} \in \bangICtxt$ and
        $\cbvToBangAGK[*]{\cbvICtxt} \in \bangICtxt$.

        Let $\cbvICtxt \in \cbnICtxtSet$, then by construction
        $\cbvICtxt = \cbvSCtxt<\abs{x}{\cbvFCtxt}>$ for some
        $\cbvSCtxt \in \cbvSCtxtSet$ and $\cbvFCtxt \in \cbvFCtxtSet$.
        We distinguish two cases:
        \begin{itemize}
        \item[\bltI] $\cbvFuncPred{\cbvSCtxt}$: Using \Cref{lem: cbv
            embedding with
            contexts,lem:cbv_surface_context_stability}:
            \begin{equation*}
                \begin{array}{r cl cl}
                    \cbvToBangAGK{\cbvICtxt}
                    &=& \cbvToBangAGK[*]{\cbvSCtxt}<\bangStrip{\cbvToBangAGK{(\abs{x}{\cbvFCtxt})}}>
                \\
                    &=& \cbvToBangAGK[*]{\cbvSCtxt}<\bangStrip{\oc\abs{x}{\oc\cbvToBangAGK{\cbvFCtxt}}}>
                \\
                    &=& \cbvToBangAGK[*]{\cbvSCtxt}<\abs{x}{\oc\cbvToBangAGK{\cbvFCtxt}}>
                    &\in& \bangSCtxtSet<\abs{x}{\oc\bangFCtxtSet}>
                    \;\subseteq\; \bangSCtxtSet<\oc\bangFCtxtSet>
                    \;\subseteq\; \bangICtxtSet
                \end{array}
            \end{equation*}

        \item[\bltI] $\neg\cbvFuncPred{\cbvSCtxt}$: Using \Cref{lem:
            cbv embedding with
            contexts,lem:cbv_surface_context_stability}:
            \begin{equation*}
                \begin{array}{r cl cl}
                    \cbvToBangAGK{\cbvICtxt}
                    &=& \cbvToBangAGK{\cbvSCtxt}<\cbvToBangAGK{(\abs{x}{\cbvFCtxt})}>
                \\
                    &=& \cbvToBangAGK{\cbvSCtxt}<\oc\abs{x}{\oc\cbvToBangAGK{\cbvFCtxt}}>
                    &\in& \bangSCtxtSet<\oc\abs{x}{\oc\bangFCtxtSet}>
                    \;\subseteq\; \bangSCtxtSet<\oc\bangFCtxtSet>
                    \;\subseteq\; \bangICtxtSet
                \end{array}
            \end{equation*}
        \end{itemize}

    \item By straightforward induction on $\bangICtxt \in \bangICtxtSet$.
    \qed

    \end{enumerate}

\end{proof}

\begin{corollary}
    \label{cor:cbv prop-internal_extended}
    Let $t, u \in \setCbvTerms$ and $u' \in \setBangTerms$ such that
    $\bangPredINF<d!>{u'}$. Then,
    \begin{itemize}
    \item \textbf{(Stability)}: \hspace{1.86cm}%
            $\cbvToBangAGK{t} \bangArrSet*_I u'%
                \quad \Rightarrow \quad%
            \exists\, u \in \setCbvTerms, \; \cbvToBangAGK{u} = u'$

    \item \textbf{(Normal Forms)}: \hspace{0.2cm}%
            $t \text{ is a \cbvSetINF}%
                \quad\Leftrightarrow\quad%
            \cbvToBangAGK{t} \text{ is a \bangSetINF}$

    \item \textbf{(Simulations)}: \hspace{1.32cm}%
            $t \cbvArrSet*_I u%
                \quad \Leftrightarrow \quad%
            \cbvToBangAGK{t} \bangArrSet*_I \cbvToBangAGK{u}$

            Moreover, the number of $\cbvSymbBeta/\cbvSymbSubs$-steps
            on the left matches the number
            $\bangSymbBeta/\bangSymbSubs$-steps on the right.
    \end{itemize}
\end{corollary}
% \RecCbvPropInternal*
\begin{proof}
    Using \Cref{lem: cbv preservation restricted} with
    \Cref{lem:cbv_internal_context_stability,lem:Internal_Diligence,lem:cbv_INF<d!>_is_FNF<d!>}.
    \qed
\end{proof}

% \RecCbvPropInternal*
% \begin{proof}
%     Third item of \Cref{cor:cbv prop-internal_extended}.
% \end{proof}

%\input{5_Proofs.tex}
}

\end{document}